%% file: arxiv_2.tex
\documentclass[showpacs,amsmath,amssymb,twocolumn,aps,pra,superscriptaddress,notitlepage,floatfix,10pt]{revtex4-2}

\usepackage{stylesetting}
\usepackage{orcidlink}
 \usepackage[normalem]{ulem}

\begin{document}
\let\oldacl\addcontentsline
\renewcommand{\addcontentsline}[3]{}
\title{Phase shadow: A noise-tolerant path to global quantum property estimation}

 \date{\today}
\author{Qingyue Zhang~\orcidlink{0009-0000-5638-9746}}
\affiliation{Fudan University, Shanghai 200433, China}
\author{Dayue Qin~\orcidlink{0000-0001-9225-8129}}
\affiliation{Fudan University, Shanghai 200433, China}
\author{Zhou You~\orcidlink{0000-0002-6140-2092}}
\affiliation{Fudan University, Shanghai 200433, China}
\author{Feng Xu~\orcidlink{0000-0002-7015-1467}}
\email{fengxu@fudan.edu.cn}
\affiliation{Fudan University, Shanghai 200433, China}

\author{Jens Eisert~\orcidlink{0000-0003-3033-1292}}
\email{jense@zedat.fu-berlin.de}
\affiliation{Freie Universit\"at Berlin, 14195 Berlin, Germany}
\affiliation{Helmholtz-Zentrum Berlin f\"ur Materialien und Energie, 14109 Berlin, Germany}

\author{You Zhou~\orcidlink{0000-0003-0886-077X}}
\email{you\_zhou@fudan.edu.cn}
\affiliation{Fudan University, Shanghai 200433, China}

\begin{abstract}

Measuring global quantum properties—such as the fidelity to complex multipartite states—is both an essential and experimentally challenging task. Classical shadow estimation offers favorable sample complexity, but typically relies on many-qubit circuits that are difficult to realize on current platforms. We propose the robust phase shadow scheme, a measurement framework based on random circuits with controlled-$Z$ as the unique entangling gate type, tailored to architectures such as trapped ions and neutral atoms. Leveraging tensor diagrammatic reasoning, we rigorously analyze the induced circuit ensemble and show that phase shadows match the performance of full Clifford-based ones. Importantly, our approach supports a noise-robust extension via purely classical post-processing, enabling reliable estimation under gate-dependent noise where existing techniques often fail. Additionally, by exploiting structural properties of random stabilizer states, we design an efficient post-processing algorithm that resolves a key computational bottleneck in previous shadow protocols. Our results enhance the practicality of shadow-based techniques, providing a robust and scalable route for estimating global properties in noisy quantum systems.

\textbf{As compared with the previous version of this manuscript, this version introduces a novel Generalized RPS framework to estimate arbitrary gate-dependent noises for general Clifford circuits.}


\end{abstract}

\maketitle

Learning properties 
of quantum systems constitutes both a fundamental challenge and a 
key enabler 
for the development of quantum technologies~\cite{eisert2020quantum,kliesch2021theory}. 
Recent advances 
in randomized measurements~\cite{Efficient,elben2023randomized, cieslinski2024analysing} and shadow estimation~\cite{aaronson2019shadow,huang2020predicting} have enabled sample-efficient characterization of 
large quantum systems, bypassing the need for impractical full tomography. In 
this framework, a randomly chosen unitary is applied to the state, followed by a projective measurement. Repeated experiments yield unbiased estimators—classical shadows—that 
can be reused to efficiently predict many observables of interest.

The sampling performance of shadow estimation depends critically on two factors: the choice of the random unitary ensemble and the target observable~\cite{huang2020predicting,Hu2022Hamiltonian,hu2023classical,zhang2024minimal,park2023resource}. For local observables, such as few-body Pauli operators, independent single-qubit rotations (or their variants) are sufficient~\cite{zhang2021experimental,hadfield2022measurements,nguyen2022optimizing,wu2023overlapped,van2022hardware,liu2023group,gresch2025guaranteed}. In contrast, estimating global observables—such as fidelities to highly entangled many-body states—typically requires more structured ensembles, such as global Clifford circuits~\cite{huang2020predicting,helsen2023thrifty,mao2025qudit}. While Clifford measurements can significantly suppress statistical errors via entanglement, their practical use is hindered by the substantial experimental overhead of implementing a large number of many-qubit gates. On noisy quantum platforms, gate imperfections can accumulate and bias the outcome, offsetting the theoretical advantages and posing a major obstacle for reliable estimation. 

\begin{figure}
    \centering
    \includegraphics[width=\linewidth]{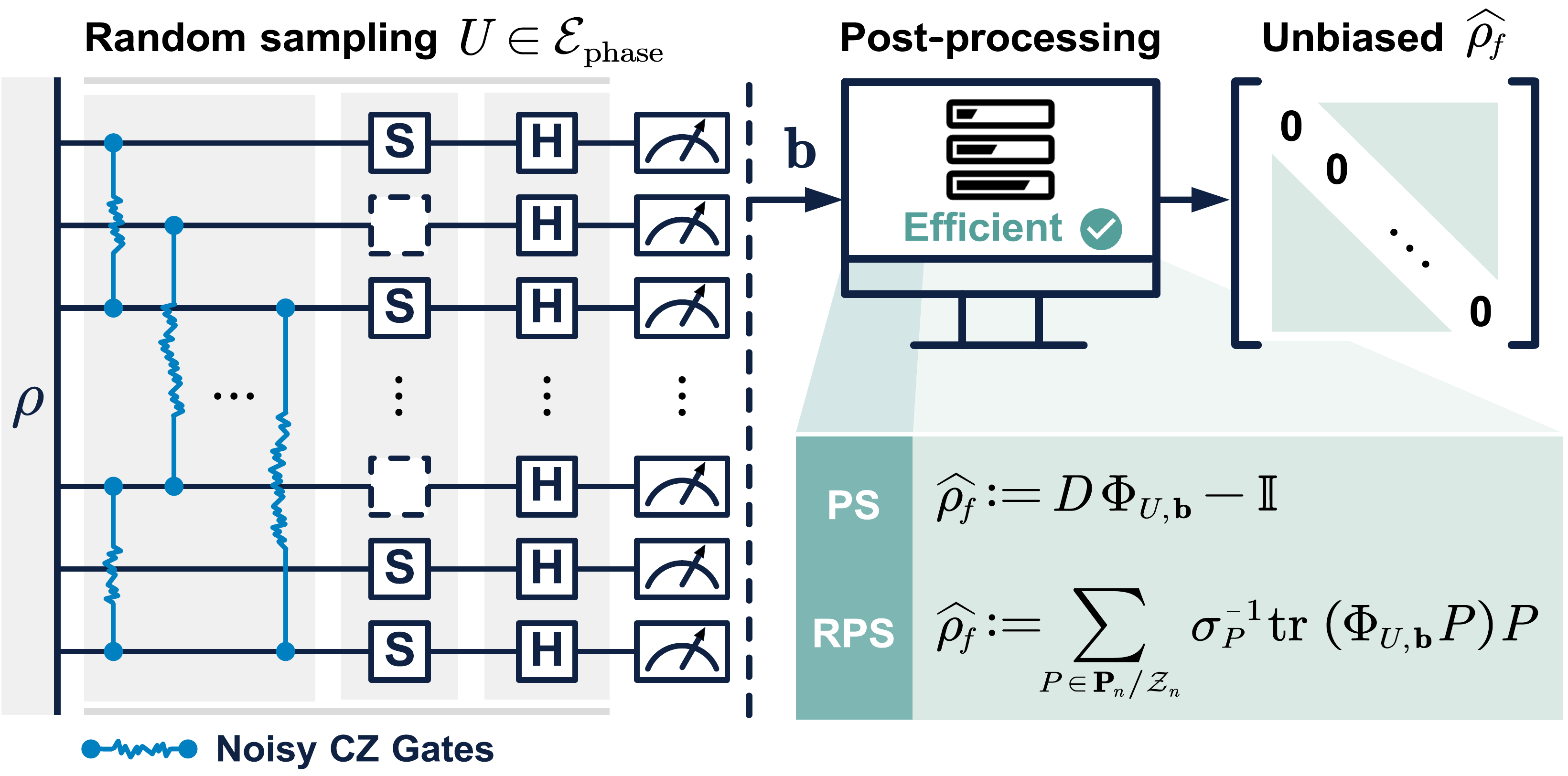}
    \caption{
    Outline of the (robust) phase shadow framework. The random phase circuit is in the form 
    ${CZ}$-${S}$-${H}$
before the final computational basis measurement, where the layer of two-qubit ${CZ}$ gates may suffer gate-dependent noise. For both noiseless (\cref{th:PshadowMain}) and noisy scenarios (\cref{th:PshadowNoise}), via proper and efficient post-processing, 
our framework can return unbiased estimator of the off-diagonal part of $\rho$, say $\rho_f$, with favorable statistical performance.
    }
    \label{fig:outline1}
\end{figure}


While several approaches have been proposed to improve the practicality of shadow estimation,
key limitations of the scheme remain. Some approximate global Clifford ensembles with low-depth circuits~\cite{bertoni2024shallow, akhtar2023scalable, schuster2024random}, while others mitigate noise via calibration and sophisticated classical post-processing~\cite{chen2021robust, koh2022classical, onorati2024noise}, or combine both strategies~\cite{hu2025demonstration, farias2024robust}. These methods often rely on idealized noise models—particularly gate independence—that are not always guaranteed or expected to hold in 
realistic settings~\cite{helsen2022general, liu2024group}. Moreover, post-processing in shallow-circuit schemes typically depends on matrix-product-state heuristics, which fail for generic global observables. Indeed, biases may persist even under these strategies~\cite{schuster2024random,cioli2024approximate,chen2021robust,brieger2025stability}. Thus, a practically robust and efficient protocol for estimating global properties under realistic quantum noise remains elusive~\cite{brieger2025stability}.

In this work, we propose a systematic framework referred to as \emph{robust phase shadow} (RPS) based on random circuits composed of randomly placed controlled-${Z}$ (${CZ}$) gates and single-qubit measurements (\cref{fig:outline1}). As the sole entangling gate, ${CZ}$ is native to a number of platforms like those of 
cold atoms in controlled collisions \cite{Mandel-Nature-2003}, Rydberg atoms~\cite{evered2023high} and trapped ions~\cite{figgatt2019parallel}, enabling hardware-efficient implementation while naturally incorporating characteristic noise features of these architectures. 
Exploiting tensor diagrammatic reasoning, we construct an unbiased estimator of the target state with rigorous performance guarantees, and systematically incorporate gate-dependent noise into the analysis—allowing us to develop the robust version that remains accurate under realistic experimental conditions. Additionally, we provide an efficient post-processing algorithm for 
estimating fidelities to arbitrary stabilizer states, 
without relying on restrictive assumptions~\cite{bertoni2024shallow,akhtar2023scalable}. These features make RPS a powerful tool for benchmarking quantum platforms and a cornerstone for 
robust quantum learning and phase-gate-based algorithms.

Beyond the RPS protocol mentioned above, which primarily focuses on the dominant Pauli-$Z$-type errors of the diagonal circuits, we also develop a framework termed generalized robust phase shadow, which can handle arbitrary gate-dependent  noises. This framework is applicable not only to phase circuits, but also to general Clifford circuits, such as random Clifford circuits and the shallow invariants. To provide theoretical guarantees, we propose a strictly unbiased estimator, accompanied by explicit variance bounds and efficient post-processing.

\section*{Results}
\subsection*{Statistical property of phase circuits} 
At the heart of shadow estimation lies the choice of the random unitary ensemble. Before presenting our main framework, we first define the ensemble employed in this work and analyze its key statistical properties, with a particular focus on deriving the second and third moments.
Consider an $n$-qubit 
system
equipped with a Hilbert space $\mc{H}_D=\mc{H}_2^{\otimes n}$ and denote with $\{\ket{\mb{b}}\}$ the computational basis specified by the binary string $\mb{b}\in\{0,1\}^n$.
The underlying quantum circuit structure of the random unitary takes the ${CZ-S-H}$ form, where ${CZ}=\operatorname{diag(1, 1, 1, -1)}$, ${S} = \operatorname{diag}(1, i)$ and ${H}$ is the 
Hadamard gate. As shown in \cref{fig:outline1}, the 2-qubit ${CZ}$ and 1-qubit ${S}$  gate are randomly chosen, and the final ${H}$  is fixed, so 
that   $U={H}^{\otimes n}U_A\in \mc{E}_{\text{phase}}$
with 
\begin{equation}\label{Eq:diagonalcircuits}
U_A=\prod_{0\leq i<j\leq n-1}{CZ}_{i,j}^{A_{i,j}}\prod_{k\in{[n]}}{S}_k^{A_{k,k}},
\end{equation} 
where $A$ is a random symmetric binary matrix encoding the interaction graph, with each entry $A_{i,j} \in \{0,1\},i\leq j$ \emph{independently and identically distributed} (i.i.d.) with equal probability. Note $U_A$ only contains diagonal gates that would change the phase on computational bases~\cite{nakata2014generating,nechita2021graphical,liu2024predicting}.

To proceed, 
let $\Phi_{U,\mb{b}}:= U^{\dag} \ket{\mb{b}}\bra{\mb{b}}U$, 
and define the $m$-th moment function of an unitary ensemble $\mc{E}$ as $\mb{M}_{\mc{E}}^{(m)} := D^{-1}\mbb{E}_{U\sim \mc{E}} \sum_{\mb{b}} \Phi_{U,\mb{b}}^{\otimes m}$. For this, one finds the following properties for the introduced phase circuit ensemble $\mc{E}_{\text{phase}}$. 

\begin{prop}[2nd and 3rd moments]\label{prop:23moment}
The $m=2,3$-th moment functions of the diagonal ensemble $\mc{E}_{{\text{phase}}}$ defined around \cref{Eq:diagonalcircuits} reads
\begin{equation}\label{eq:23moment}
    \mb{M}_{\mc{E}_{\text{phase}}}^{(m)}: = D^{-m} \bigcup_{\pi \in S_m}V_n(\pi).
    \end{equation}
\end{prop}
Here, $\pi$ is the element of the $m$-th order symmetric group $S_m$, with its unitary representation on $\mc{H}_D^{\otimes m}$ as $V_n(\pi)$ (i.e., $V_n(\pi)\bigotimes_{i=1}^{m} \ket{\mb{b}_{i}}=\bigotimes_{i=1}^{m} \ket{\mb{b}_{\pi^{-1}(i)}}$~\cite{mele2024introduction}). The union operation here is defined as the bitwise OR operation for operators written in the computational basis.
The full proof of \cref{prop:23moment} is left to Supplementary Note 3, which also presents the generalization to $m>3$.  
\begin{figure*}
    \centering
    \includegraphics[width=\linewidth]{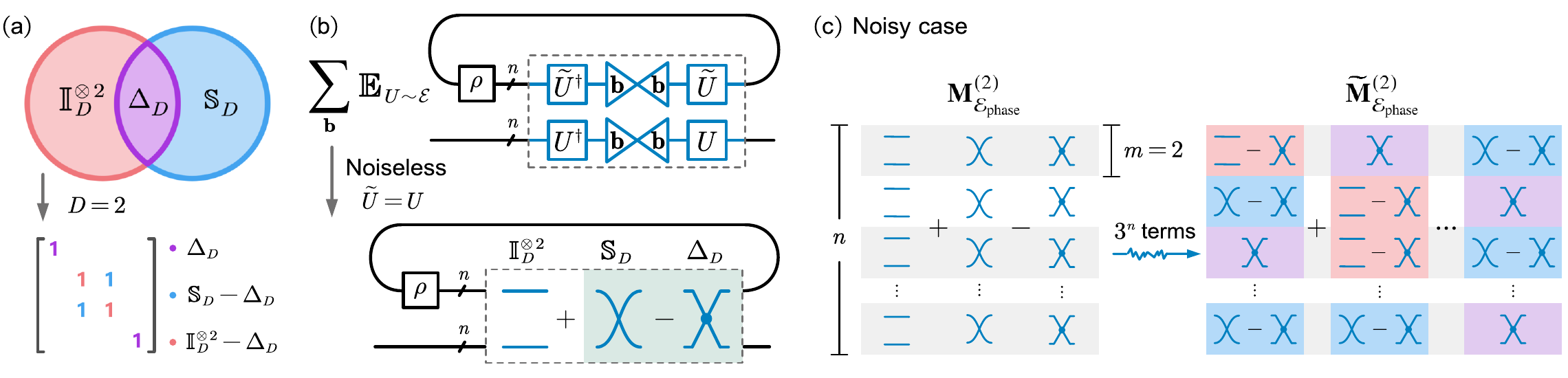}
    \caption{
    (a) An illustration of \cref{eq:23moment} of $m=2$ case. The union of the two permutation unitaries 
    $\mbb{I}_D^{\otimes 2}\cup \mbb{S}_D=\mbb{I}_D^{\otimes 2}+ \mbb{S}_D-\Delta_D$, with the parity projector $\Delta_D$ being their intersection, say $\mbb{I}_D^{\otimes 2}\cap \mbb{S}_D=\Delta_D$. For $D=2$, the corresponding matrix representation is also shown. 
    (b) A tensor diagram illustration of the (noiseless) measurement channel in \cref{Eq:idealchannel}. By 
    tracing the first copy, one obtains the output proportional to $\id_D+\rho_f$. (c) In the presence of noise, the moment function expands from a compact form (left) to a sum of $3^n$ tensor-product terms (right), each built from qubit-level operators ${\Delta_2,(\sw_2 - \Delta_2),(\mathbb{I}_2^{\otimes 2} - \Delta_2)}$, as shown by the colored blocks.
}
    \label{fig:2copy}
\end{figure*}
For comparison, the moment of the Haar ensemble (or any unitary $m$-design) reads 
$\mb{M}^{(m)}_{\mc{E}_\text{Haar}}=\mc{O}(D^{-m})\sum_{\pi \in S_m}V_n(\pi)$ (see, e.g., Ref.~\cite{mele2024introduction}).
Neglecting for a moment the coefficients of the same order of $D$, the essential difference lies in the operations `$\bigcup$' and `$\sum$' to organize all permutation unitaries. For instance of $m=2$, \cref{eq:23moment} reads $\mb{M}_{\mc{E}_{\text{phase}}}^{(2)}=D^{-2}(\mbb{I}_D^{\otimes 2}\cup \mbb{S}_D)=D^{-2}(\mbb{I}_D^{\otimes 2}+\mbb{S}_D-\Delta_D)$, with two permutation unitaries, the identity $\mbb{I}_D^{\otimes 2}$ 
and the swap operator $\mbb{S}_D$. Here, 
$\Delta_D=\Delta_2^{\otimes n}$ the 
parity projector with $\Delta_2=\ketbra{0}{0}^{\otimes 2}+\ketbra{1}{1}^{\otimes 2}$. The minus sign of $\Delta_D$ reflects the union operation since it is shared by $\mbb{I}_D^{\otimes 2}$ and $\mbb{S}_D$, as shown in \cref{fig:2copy}(a).
By contrast, the counterpart of the 
Haar ensemble (or the Clifford) shows 
$\mb{M}^{(2)}_{\mc{E}_\text{Haar}}=2(\mbb{I}_D^{\otimes 2}+\mbb{S}_D)/(D(D+1))$~\cite{mele2024introduction}. 


We remark that there are previous works~\cite{nakata2014generating,bremner2016average,nechita2021graphical} studying the statistical properties of diagonal quantum circuits. Our result, however, provides the most compact circuit implementation using only random \( {CZ} \) and ${S}$  gates. Notably, the union operation `$\bigcup$' introduced in \cref{eq:23moment} clearly highlights both the connection and distinction between random diagonal circuits and their Haar-random counterpart, and plays a crucial role in the development of the subsequent shadow estimation protocol. In particular, the approach developed in the proof 
(presented in the Methods section), grounded in this structural insight, enables the construction of unbiased estimators, effective variance control and a natural extension to general noisy scenarios.

\subsection*{Phase shadow framework}
Equipped with phase circuit ensemble $\mc{E}_{\text{phase}}$, we can introduce the \emph{phase shadow} (PS) 
framework first for the noiseless scenario. 
Shadow estimation is a quantum-classical hybrid process. On the quantum side, a random unitary $U\in\mc{E}$ is applied to an unknown quantum state $\rho\mapsto U\rho U^{\dagger}$ and one measures it in the computational basis to get the
classical outcome word $\mb{b}$, 
with conditional probability $\Pr(\mb{b}|U)=\langle \mb{b}| U\rho U^{\dag} |\mb{b}\rangle=\tr(\rho \Phi_{U,\mb{b}})$. On the classical side, one `prepares' a quantum snapshot $\Phi_{U,\mb{b}}$ on the classical computer.

Taking the expectation on both $U$ and $\mb{b}$, the whole process could be effectively written as a channel 
\begin{equation}\label{Eq:idealchannel}
\begin{aligned}
\mc{M}_{\mc{E}}(\rho):
    &=\mbb{E}_{U\sim\mc{E},\mb{b}}\ \Phi_{U,\mb{b}}=\sum_{\mb{b}}\mbb{E}_{U\sim\mc{E}} \tr_{1}(\rho\otimes \id \text{ }\Phi_{U,\mb{b}}^{\otimes 2})\\
    &=D \tr_1(\rho\otimes \id\ \mb{M}_{\mc{E}}^{(2)}),
\end{aligned}
\end{equation}
where $\tr_{1}$ denotes the (partial) trace out the first copy, and in the second line we relate 
it to 
the predefined moment 
function.  

At first glance, one might expect that replacing the full Clifford ensemble~\cite{huang2020predicting} in \cref{Eq:idealchannel} with \(\mathcal{E}_{\text{phase}}\) would introduce considerable bias, since \(\mathcal{E}_{\text{phase}}\) is only an approximate design and deviates significantly 
from the full Clifford group~\cite{nakata2014generating}. 
However, we show that an \textit{unbiased estimator} is still achievable—not 
for the entire quantum state \(\rho\), but specifically for its \textit{off-diagonal component}, 
denoted by \(\rho_f\).

As illustrated  in \cref{fig:2copy}(b), 
by using \cref{prop:23moment} with $m=2$, one has the output of the channel $\mc{M}_{\mc{E}_{\text{phase}}}(\rho)=D^{-1}(\id_D+\rho_f)$, where $\rho_f=\rho-\rho_d$ and the diagonal part is $\rho_d=\tr_1(\Delta_D\rho)=\sum_{\mb{b}}\ketbra{\mb{b}}{\mb{b}}\rho\ketbra{\mb{b}}{\mb{b}}$. 
While it is true that \(\mathcal{E}_{\text{phase}}\) is not tomographically complete—the channel output lacks information about the diagonal part \(\rho_d\)—this is not a fundamental limitation. The diagonal component \(\rho_d\) can be efficiently estimated via direct measurement in the computational basis. 
This leads to a practical unbiased reconstruction protocol, as formalized in the theorem below.


\begin{theorem}[Unbiased and efficient recovery with phase shadow estimation]\label{th:PshadowMain}
Suppose that one conducts phase shadow estimation using the random unitary ensemble 
$\mc{E}_{{\text{phase}}}$, the unbiased estimator of $\rho_f$ in a single-shot 
gives 
\begin{equation}\label{eq:shadow2}
\widehat{\rho_f}:=D\Phi_{U,\bb}-\id,
\end{equation}
with $U$ the applied 
unitary and $\bb$ the measurement result,
such that $\mathbb{E}_{\{U,\mb{b}\}}\widehat{\rho_f} = \rho_f$. 
The expectation value $o_f=\tr(O \rho_f)$ can be estimated for observables $O$ as $\widehat{o_f}=\tr(O\widehat{\rho_f})$
with 
\begin{equation}\label{eq:varm}
\begin{aligned}
\text{Var}(\widehat{o_f}) \leq 3\|O_f\|_2^2
\end{aligned}
\end{equation}
being its variance 
and $O_f$  the off-diagonal part of $O$. 
\end{theorem}
Here, $\|A\|_2\coloneq{\tr(AA^{\dag})}^{1/2}$ denotes the 
Frobenius norm. 
The variance bound stated in \cref{eq:varm} is derived from the \(m=3\) case of \cref{prop:23moment}. A sketch of the proof is provided in the Methods section, where we present an elegant derivation using tensor diagrammatic reasoning, grounded in the union (and intersection) operation of tensor indices introduced before.

For $\rho_d$, the computational basis result $\ketbra{\bb}$ can faithfully act as an estimator $\widehat{\rho_d}=\ketbra{\bb}{\bb}$. The corresponding variance can be bounded by the spectral norm $\text{Var}(\widehat{o_d})\leq \|O_d\|_{\infty}^2$ with $\widehat{o_d}=\tr(O\widehat{\rho_d})$. 
Suppose that one repeats the PS protocol for $N_f$ rounds to get the shadow set $\{\rho_f^{(i)}\}$, and the computational-basis measurement for $N_d$ rounds to get $\{\rho_d^{(j)}\}$, the full estimator of $\rho$ thus delivers $\widehat{\rho}=N_f^{-1}\sum_i \rho_f^{(i)}+N_d^{-1}\sum_j \rho_d^{(j)}$. At the same time, $\widehat{o}=\tr(O\widehat{\rho})$ has the variance $\text{Var}(\widehat{o})= N_f^{-1}\text{Var}(\widehat{o_f})+N_d^{-1}\text{Var}(\widehat{o_d})$
\comments{
\begin{equation}\label{eq:varMfMd}
\begin{aligned}
\text{Var}(\widehat{o})= N_f^{-1}\text{Var}(\widehat{o_f})+N_d^{-1}\text{Var}(\widehat{o_d})\\
\leq 3N_f^{-1}\|O_f\|^2_2+N_d^{-1}\|O\|_{\infty}^2,
\end{aligned}
\end{equation}
}
under the total measurement time $N=N_f+N_d$. By properly choosing $N_f/N_d$, one further has $\text{Var}(\widehat{o})\leq 4N^{-1}\|O\|^2_2$, with the detailed analysis left to the Supplementary Note 3. The bound is almost the same to the previous result $3N^{-1}\|O\|^2_2$ resorting to the full Clifford group~\cite{huang2020predicting}.

As such, PS can predict many nonlocal observables $\{O_k\}_{k=1}^L$ with a total sample complexity $N\sim \log(L)\max \text{Var}(\widehat{o})$ using the collected shadow sets if $\|O\|^2_2$ is bounded, just like the original shadow protocol~\cite{huang2020predicting}. One important special case is constituted by the fidelity estimation of many multipartite entangled states $O=\Psi:=\ketbra{\Psi}$, with $\|\Psi\|^2_2=1$. We also numerically confirm such an advantage for the GHZ state fidelity in \cref{fig:Hadamard-bias}(a).
Note that Ref.~\cite{park2023resource} employs equatorial stabilizer measurements for shadow estimation, resulting in a circuit structure similar to ours. However, our diagrammatic analysis of the phase circuit ensemble provides an explicit expression for an unbiased estimator of \(\rho_f\) with significantly simpler post-processing in  \cref{eq:shadow2} and a much tighter variance bound in \cref{eq:varm}. 
Furthermore, leveraging the introduced diagrammatic framework, our PS admits a natural extension to realistic noisy scenarios, making it both theoretically insightful and practically robust.

\begin{figure}
    \centering
    \includegraphics[width=\linewidth]{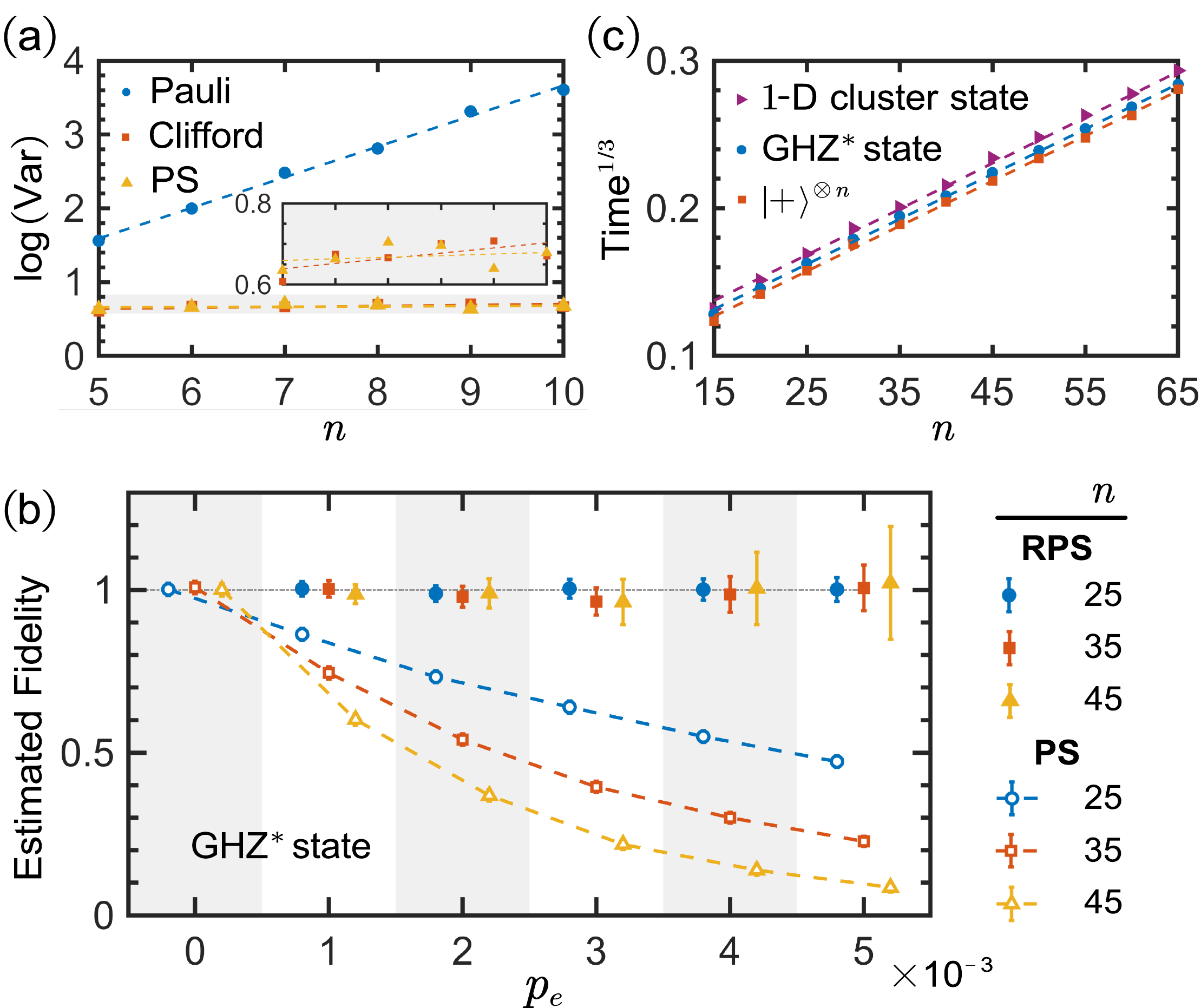}
    \caption{(a) Variance of fidelity estimation for the GHZ* state (noiseless case), comparing Pauli, Clifford measurements, and PS~\cite{huang2020predicting} under different qubit number $n$ of the system. The GHZ* state is locally equivalent to the \emph{canonical} GHZ state~\cite{zhou2019detecting}. (b) Fidelity estimation of the GHZ* state using  RPS and PS measurements for $n \in \{25,35,45\}$ qubits under different noise level $p_e$ ($N = 50,\!000$ snapshots for each data point).  (c) 
    Cubic root of the post-processing time for fidelity estimation of three graph states using RPS ($n = 15$--$65$, $N = 10,\!000$). The details on the numerical simulations are left to Supplementary Note 8.}
    \label{fig:Hadamard-bias}
\end{figure}

\subsection*{Robust estimation under realistic gate noise}
Noise poses a major challenge in implementing any shadow estimation protocol, particularly in Clifford measurement schemes that rely heavily on two-qubit gates. Indeed, gate-dependent noise remains a significant obstacle in notions of quantum learning~\cite{chen2021robust,brieger2025stability}. We show that the proposed PS can be further enhanced to mitigate such noise in practical scenarios, namely \emph{robust phase shadow} (RPS).

Note that the layer of random ${S}$  gates and the final fixed  ${H}$ gates in \cref{fig:outline1} consist solely of single-qubit operations, which is more accurate compared to the 2-qubit gates. Indeed, they serve to select the measurement basis as either $X$ or  $Y$, and several methods exist for mitigating the noise effects in the final readout~\cite{gluza2020quantum}. 
Therefore, our focus here is on the noise characteristics of the complex ${CZ}$ layer. The noisy ${CZ}$ gate acting on a qubit pair $(i,j)$ can reasonably realistically be modeled as  
    $\widetilde{{CZ}}_{i,j}:={CZ}_{i,j}\times{ZZ}( \theta_{i,j})$, with ${ZZ}( \theta_{i,j})=\mathrm{exp}(-{i} \frac{\theta_{i,j}}{2} {Z}_i {Z}_j)$
and each $\theta_{i,j} \sim \mc{N}(0,\sigma^2)$ and 
being i.i.d.~\cite{lotshaw2023modeling}.  
As we show in Supplementary Note 4, the 
noise
effect results in a Pauli channel, $\mc{D}(\rho) = (1-p_e)\rho +p_e {Z}_i {Z}_j \rho {Z}_i {Z}_j $ with $p_e\approx {\sigma^2}/{4}$, depending on whether the ${CZ}$ gate is operated. As such, the total noise is gate-dependent on the given layout pattern of the ${CZ}$ layer. 
We additionally consider a more general noise model 
including ${IZ}$ and ${ZI}$ errors as well, common to Rydberg atom architectures~\cite{cong2022hardware,evered2023high}.


Due to impact of noise, the 
Born probability observed in the experiment is changed to \(\Pr(\mathbf{b}|U) = \text{Tr}(\Phi_{\widetilde{U},\mathbf{b}} \rho)\), with \(\widetilde{U}\) denoting the noisy version of the intended unitary \(U\). Consequently, the measurement channel in \cref{Eq:idealchannel} would undergo a significant 
and intricate transformation. Specifically, the related moment function $\mb{M}_{\mc{E}_{\text{phase}}}^{(2)}$ which originally admits a compact form in \cref{prop:23moment}, now decomposes into a summation of $3^n$ terms involving 
qubit-level operators, like 
\begin{equation*}
    \Delta_2^{\otimes n_1} \otimes (\mathbb{I}_2^{\otimes 2} - \Delta_2)^{\otimes n_2} \otimes (\sw_2 - \Delta_2)^{\otimes n_3}, 
\end{equation*}
with \( n_1 + n_2 + n_3 = n \), as shown diagrammatically in \cref{fig:2copy}(c).

If one continues to use the noiseless post-processing strategy as \cref{eq:shadow2} of \cref{th:PshadowMain}, it would introduce significant bias in the estimation due to accumulated noise effects, as observed numerically in \cref{fig:Hadamard-bias}(b). As such, accurately determining the inverse of the current complex quantum channel becomes essential. Fortunately, the channel (or equivalently the moment function) can be decomposed into the Pauli basis, with its coefficients being computable in a scalable manner.

\begin{prop}[Pauli decomposition]\label{prop:noisyChannel}
In the RPS protocol, suppose that the noisy $\text{CZ}$ gates 
are applied during the experiment. The second-order moment function in \cref{Eq:idealchannel} in the noisy scenario admits a Pauli decomposition as
\begin{equation}\label{Eq:P2P}  
\widetilde{\mathbf{M}}_{\mc{E}_{{\text{phase}}}}^{(2)} = D^{-3} \sum_{{P} \in \mb{P}_n} \sigma_{P} \, {P} \otimes {P},  
\end{equation}  
and the coefficient of the Pauli operator in the form \( {P} = \id_2^{\otimes n_1}\otimes  {Z}^{\otimes n_2}\otimes \{{X},{Y}\}^{\otimes n_3} \) reads 
\begin{equation}\label{Eq:P2Pha}
\begin{aligned}
\sigma_P = \sum_{s=0}^{n_1 + n_2} \sum_{t=\max(0,s-n_1)}^{\min(s, n_2)} &(-1)^t \binom{n_1}{s-t} \binom{n_2}{t}   \\
&\times(1 - p_e)^{(n_1 + n_2 - s)n_3} (p_e)^{sn_3},
    \end{aligned}
\end{equation}
with $p_e$ being the noise rate of the ${CZ}$ gates.
\end{prop}

The proof of \cref{prop:noisyChannel} is left to Supplementary Note 5, along with more detailed information about the following generalization.
Here 
we assume that all ${CZ}$ gates share the same noise rate $p_e$. Such uniform assumption can be relaxed 
and the Pauli coefficients $\sigma_P$ 
can still be efficiently calculated using a low-order approximation, which matches well with the exact values. 
Moreover, in the Methods section, we also 
extend the result to noisy models with additional single-qubit $Z$ error, which happens naturally in Rydberg atom systems~\cite{cong2022hardware,evered2023high}.

Denote the set of Pauli operators in the $Z$ basis as $\mc{Z}_n=\{\id_2,Z\}^{\otimes n}$. We show in Supplementary Note 5 that $\sigma_P=0$, if $P\in \mc{Z}_n$, except $\sigma_{\id_D}=D$ in \cref{Eq:P2Pha}. Note that each $\sigma_P\in \mb{P}_n/\mc{Z}_n$ can be calculated in $\mc{O}(n^2)$.  By inserting \cref{Eq:P2P} into \cref{Eq:idealchannel}, it is direct to see that the noisy quantum channel shows $\widetilde{\mc{M}}_{\mc{E}_{\text{phase}}}(\rho) = D^{-1}\id_D+D^{-2} \sum_{P \in \mb{P}_n/\mc{Z}_n} \sigma_P \tr(\rho P)P$, and we remark that all $\sigma_P=1$ for the noiseless case ($p_e=0$).

This decomposition enables the practical inversion of the channel, facilitating unbiased estimation in the realistic noisy scenario, 
equipped with rigorous performance guarantees.


\begin{theorem}[Guarantees of unbiased and efficient recovery under realistic noise]\label{th:PshadowNoise} 
Suppose that one conducts 
robust phase shadow estimation using the random unitary ensemble $\mc{E}_{{\text{phase}}}$ with noisy ${CZ}$ gates, then the unbiased estimator of \(\rho_f\) in a single-shot is given by
\begin{equation} \label{Eq:InvNoisychannelPauli}  
\widehat{\rho_f}_{\text{robust}} := \sum_{P \in \mb{P}_n/\mathcal{Z}_n} \sigma_{P}^{-1} \, \tr(\Phi_{U,\mathbf{b}} {P}) \, {P},  
\end{equation}

where \(\sigma_P\) is the Pauli coefficient of the forward noisy channel from \cref{prop:noisyChannel}, ensuring \(\mathbb{E}_{\{U,\mathbf{b}\}} \widehat{\rho_f}_{\text{robust}} = \rho_f\).  
The expectation value  $o_f=\tr(O\rho_f)$ can be estimated as \(\widehat{o_f}_{\text{robust}} = \tr(O \widehat{\rho_f}_{\text{robust}})\), and its variance is
\begin{equation} \label{Eq:Var_first}  
\text{Var}(\widehat{o_f}_{\text{robust}}) \lesssim \Theta(1)  \exp(\frac{n^2 p_e}{2})  \|O_f\|_2^2  
\end{equation}  
when $np_e\ll 1$, where \(O_f\) is the off-diagonal part of \(O\).  
\end{theorem}

To the best of our knowledge, the estimator in \cref{Eq:InvNoisychannelPauli} is the first to incorporate challenging yet natural gate-dependent noise and achieve strict unbiasedness. 

As shown in \cref{Eq:Var_first}, 
the variance scales as $\exp(n^2 p_e/2)$, which is, indeed, a universal feature reflecting the fundamental bias-variance trade-off relation~\cite{cai2023quantum}, with supporting numerical results  presented in Supplementary Note 8. Importantly, for current quantum platforms typically with about $n=50$ qubits and gate error rates $p_e$ in the range of $10^{-3} \sim 10^{-2}$, one has $np_e \ll 1$ and $n^2p_e/2 = \mc{O}(1)$. As such, the statistical variance introduced by error-mitigation-like post-processing remains moderate in the current regime, thus ensuring estimation efficiency. 
In \cref{fig:Hadamard-bias}(b), we compare the performance of RPS and PS.  For PS, the total noise accumulates over ${CZ}$ gates, leading to a significant bias. For RPS, under feasible sampling time, it can perform unbiased estimation with acceptable statistical uncertainty, demonstrating strong potential for near-term quantum applications. 


\subsection*{Efficient post-processing}
By \cref{th:PshadowNoise}, the estimator for an observable \( O \) is given by
\begin{equation} \label{Eq:noiseOf}  
\widehat{o_f}_{\text{robust}} = \sum_{P \in \mb{P}_n/\mathcal{Z}_n} \sigma_P^{-1} \, \tr(\Phi_{U,\mathbf{b}} P) \, \tr(O P),
\end{equation}
where the summation of $P$ in \cref{Eq:noiseOf} nominally involves \( 4^n - 2^n \) Pauli terms. The estimation remains efficient when $O$ admits a polynomial-size Pauli decomposition (e.g., for local observables). At first glance, when the observable is stabilizer states $O = \Psi = V^\dagger \ket{\mb{0}}\bra{\mb{0}} V$ where $V$ is Clifford, it seems that exact post-processing is intractable in general, because $\Psi$ typically contains exponentially many ($2^n$) Pauli terms. That said, while \cref{th:PshadowNoise} guarantees favorable sampling complexity, computing individual estimators remains challenging. This mirrors limitations in existing protocols--for instance, tensor-network approaches often require heuristic approximations of $\sigma_P^{-1}$~\cite{bertoni2024shallow,akhtar2023scalable}.

Nevertheless, we demonstrate that the classical post-processing remains efficient in our protocol ($\mathcal{O}(n^3)$) even for stabilizer states. This tractability result, proven in Supplementary Note 7, is further supported by numerical simulations across various target states (see \cref{fig:Hadamard-bias}(c)).

\begin{prop}[Efficient classical post-processing (informal)]\label{Prop:PostProcess} Suppose that the observable $O$ is a stabilizer state, the post-processing of robust phase estimation using \cref{Eq:InvNoisychannelPauli} is efficient, specifically  $\mathcal{O}(n^3)$, in expectation.
\end{prop}

The result relies on a key observation about shadow estimation: since unitaries $U$ are uniformly sampled from $\mathcal{E}$, it is more reasonable to quantify the expected post-processing cost by averaging over $U \in \mathcal{E}$. Crucially, we observe that the average number of shared Pauli operators between the stabilizer groups of $\Phi_{U,\mathbf{b}}$ and $\Psi$ in \cref{Eq:noiseOf} remains $\mathcal{O}(1)$ under uniform sampling of $U\in\mathcal{E}_{\text{phase}}$ (or the full Clifford group).

The concrete post-processing workflow for stabilizer-state observables is also provided in \cref{Algo:EMeffpostproc} in
the Methods section, leveraging the overlap property and stabilizer formalism~\cite{aaronson2004improved} to maintain efficiency, 
which extends to other protocols relying on random Clifford measurements~\cite{huang2020predicting,bertoni2024shallow,akhtar2023scalable}.

\subsection*{Generalized robust phase shadow under gate-dependent noise}
The original RPS construction establishes a foundational framework for robust estimation, successfully revealing the intricate twirling structure of the ${\rm S}$--${\rm CZ}$--${\rm H}$ architecture under Pauli-$Z$-type noise. This model is highly effective for many realistic platforms where $Z$-type dephasing is the dominant error source, allowing for the elimination of a significant portion of experimental noise. To further enhance this capability and address residual errors from other noise channels, we here extend the protocol to a \textit{generalized robust phase shadow (generalized RPS)} framework. This generalization applies the RPS framework to arbitrary Clifford circuits subject to arbitrary gate-dependent noise on all quantum gates, ensuring comprehensive noise resilience. Full proofs and details are provided in Supplementary Note 10.

In this setting, we first construct a modified estimator by replacing the coefficient $\sigma_{P}$ in \cref{Eq:InvNoisychannelPauli} with a \emph{circuit-dependent} factor $\sigma(P,U)$. This modification exploits the knowledge of the sampled Clifford unitary $U$ and corresponding noise channel  during classical post-processing. Specifically, suppose the circuit implementation of $U$ be decomposed into a sequence of gates $g_1,g_2, \dots, g_m$.  After each gate $g_j$, a gate-dependent Pauli channel $\Lambda_j$ acts on the support of $g_j$,
\begin{equation}
  \Lambda_j(Q)=\alpha_{j,Q} Q,\qquad Q\in\mathbb{P}_k,
\end{equation}
where $k$ is the number of qubits on which the noise channel $\Lambda_j$ acts and the real coefficients satisfy $0<\alpha_{j,Q}\le 1$. In the ideal noiseless case, the propagation of a Pauli operator $P\in\mathbb{P}_n$ through the circuit is given by 
\begin{equation}
P_0=P,\qquad P_j=g_j P_{j-1} g_j^\dagger,\quad j=1,\dots,m.
\end{equation} 
When there is noise, we obtain a multiplicative attenuation factor for noise inversion
\begin{equation}\label{Eq:sigmaPU}
    \sigma(P,U) = \prod_{j=1}^{m}\alpha_{j,P_j},
\end{equation}
which depends both on the input Pauli $P$ and on the specific circuit $U$. Note that the dependence on $U$ arises from a circuit-specific noise model.

This leads to the following generalization of \cref{th:PshadowNoise}.

\begin{theorem}[Generalized RPS under gate-dependent Pauli noise]\label{thm:genRPS}

Consider robust phase shadow estimation based on an ensemble $\mc{E}_{phase}$ as above, where each gate $g_j$ is followed by a Pauli channel $\Lambda_j$ with eigenvalues $\alpha_{j,Q}$ and $\sigma(P,U)$ is defined in \cref{Eq:sigmaPU}. Then the generalized single-shot estimator
\begin{equation}\label{Eq:InvNoisychannelPauliU}
  \widehat{\rho_f}_{\mathrm{gen}}
  =
  \sum_{P\in\mathbb{P}_n/\mathcal{Z}_n}
      \sigma(P,U)^{-1}\,\tr(\Phi_{U,\mathbf{b}}P)\,P
\end{equation}
is unbiased on the off-diagonal part, i.e., $\mathbb{E}_{U,\mathbf{b}}\bigl[\widehat{\rho_f}_{\mathrm{gen}}\bigr] = \rho_f.$

For stabilizer-state observables $O=V^\dagger\ket{\mathbf{0}}\bra{\mathbf{0}}V$ with Clifford $V$, the corresponding estimator $\widehat{o_f}_{\mathrm{gen}}={\rm tr}(O\,\widehat{\rho_f}_{\mathrm{gen}})$ has variance bounded by
\begin{equation}
\mathrm{Var}\bigl(\widehat{o_f}_{\mathrm{gen}}\bigr)
  \;\le\;
  \Theta(1)\,
  \sqrt{\,
    \mathbb{E}_{U}
    \Bigl[
      \max_{P\in\mathbb{P}_n/\mathcal{Z}_n}\sigma(P,U)^{-4}
    \Bigr]
  }.
  \label{Eq:stab-var-final-main}
\end{equation}
\end{theorem}

Besides, the post-processing of  \cref{Eq:InvNoisychannelPauliU} for stabilizer states is also efficient along with \cref{Prop:PostProcess}, as the calculation of the coefficient $\sigma(P,U)$ costs only $O(n^2)$ for phase shadow.

Here, we present some remarks on generalized RPS. First, we emphasize that the generalized RPS framework is capable of handling general gate-dependent noise models. While Theorem~\ref{thm:genRPS} addresses Pauli noise, by incorporating randomized compiling techniques~\cite {wallman2016noise}, which
convert arbitrary noise into Pauli noise, the generalized formalism can handle arbitrary noise. Second, unlike prior assumptions of uniform noise scale, here every gate $g_j$, including single-qubit operations, is assigned a distinct Pauli channel $\Lambda_j$. This flexibility allows the protocol to leverage gate-specific error profiles obtained via standard gate calibration~\cite{blume2017demonstration,erhard2019characterizing,magesan2012efficient}, rather than calibrating a global parameter in~\cite{chen2021robust}. Furthermore, the framework can be extended to arbitrary circuit structures, including random Clifford measurement and its shallow invariants~\cite{schuster2024random,bertoni2024shallow}, thereby providing a unified and strictly unbiased estimation strategy for randomized measurement.


To validate the efficiency of the generalized RPS framework, we perform comprehensive numerical simulations in \cref{fig:comparison}, comparing it against standard protocols under realistic, gate-dependent noise conditions. The simulations utilize the gate-dependent depolarizing noise model for all gates, where error rates vary across different gates and qubits, mimicking the non-uniform noise level of the realistic hardware (see Supplementary Note~10 for details).

The numerical results, visualized in~\cref{fig:comparison}(a), reveal the contrast of the estimation bias. The generalized RPS estimator maintains strict unbiasedness across all noise levels. By accurately inverting the specific noise channel $\Lambda_j$ for each gate, it recovers the ideal expectation values with no systematic error, confirming the theoretical predictions of \cref{thm:genRPS}. In contrast, the \emph{robust shadow estimation protocol}  (RSE) fails to achieve unbiased estimation. This comparison underscores that precise, gate-dependent noise inversion by the generalized RPS framework is essential for high-precision shadow estimation for global properties.

To demonstrate the universality of our framework, we extend the numerical validation beyond the native ${CZ-S-H}$ structure to standard random Clifford measurements.  By applying our generalized formalism to the random Clifford ensemble, we verify that the property of strict unbiasedness holds regardless of the underlying circuit structure. Furthermore, as shown in Fig.~\ref{fig:comparison}(b), the generalized RPS exhibits a significantly lower variance compared to the generalized Clifford estimator. This advantage arises due to the constant-factor reduction in circuit depth enabled by the ${CZ-S-H}$ circuit structure, highlighting the practical utility of the phase shadow for randomized measurements.


\begin{figure*}[t]
    \centering
    \includegraphics[width=0.95\linewidth]{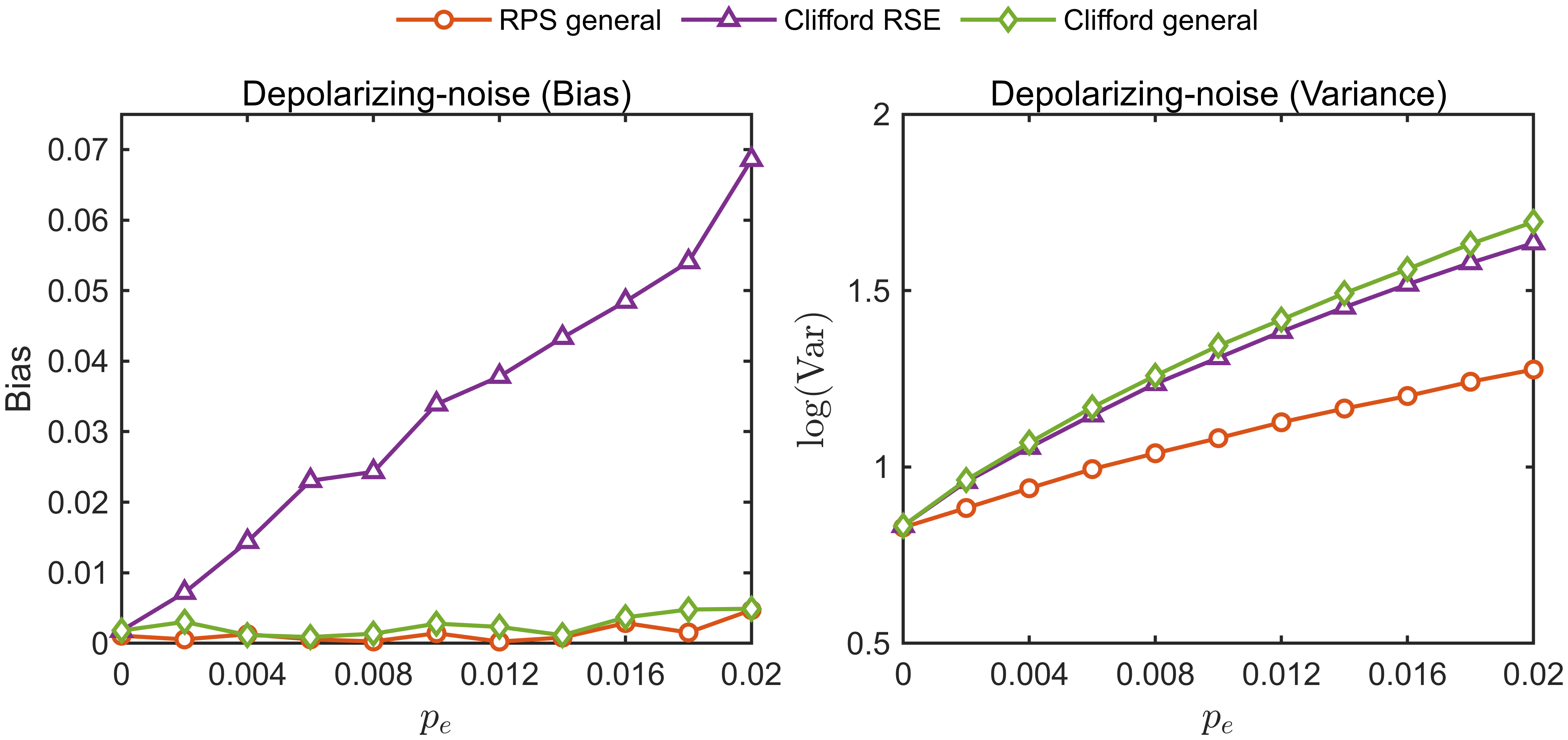}
    \caption{(Left) Estimation bias and (right) logarithmic estimation variance for 10-qubit random stabilizer states, comparing generalized RPS (orange) and its extension to random Clifford circuits (green) with the Clifford RSE protocol (purple)~\cite{chen2021robust}. Each data point is computed using $N=10^4$ measurement snapshots and averaged over 100 random stabilizer states. The RSE protocol is calibrated using $10^6$ runs for each noise level. The details on the numerical simulations are left to Supplementary Note 10.}
    \label{fig:comparison}
\end{figure*}

\section*{Discussion}
Our work shows that robust and scalable estimation of global quantum properties can be achieved through hardware-native circuits in realistic quantum platforms and suitable structure-aware design. RPS strikes a natural balance between local random measurements and global Clifford schemes, avoiding exponential overhead while ensuring statistical guarantees under challenging gate-dependent noise. Furthermore, the generalized RPS framework further extends this versatility to general gate-dependent noise and arbitrary randomized measurement schemes. These strengths render RPS a highly promising framework for robust benchmarking, effective error mitigation, and hybrid quantum-classical learning at scale.  


Our results demonstrate that the proposed method establishes practical feasibility for deployment on state-of-the-art devices, effectively bridging the gap between theoretical performance and hardware constraints. Quantitatively, considering a 50-qubit system with state-of-the-art gate fidelities exceeding $99.5\%$ (as demonstrated in recent Rydberg atom arrays~\cite{evered2023high,bluvstein2024logical}), the sampling overhead due to variance scaling is approximately a factor of $10^2-10^3$ scale. This overhead is well within the capabilities of current or near-future experimental platforms. 

Looking ahead, several promising directions arise: generalizing phase-random circuits to other native two-qubit gates like iSWAP and $\sqrt{\text{SWAP}}$ to better match other hardware platforms~\cite{gao2025establishing,zhang2023scalable}, developing shallow-depth circuit variants to further reduce experimental complexity~\cite{dalzell2024random, schuster2024random}, and exploring applications like cross-platform verification~\cite{CrossDevice,elben2020cross} and the verification of quantum error correction~\cite{gullans2021quantum}. Such extensions can further enhance the practical impact of RPS and contribute to the co-design of noise-resilient protocols tailored to specific quantum architectures, which seems urgently needed to make further progress in the field. 
On a higher level, we expect such increasingly simple and hardware-tailored robust schemes of read-out and benchmarking to become important when scaling up quantum hardware while maintaining predictive
power, and eventually reaching a regime of \emph{fault-tolerant application-scale quantum} (FASQ) \cite{MindTheGaps} machines.


\comments{\red{long version:}
Our work demonstrates that robust and scalable estimation of global quantum properties is possible using hardware-native resources and structure-aware design. The RPS framework offers a middle ground between the simplicity of local random measurements and the expressiveness of global Clifford schemes. By leveraging ${CZ}$-only entangling layers and diagrammatic analysis, RPS avoids the exponential cost typically associated with global measurements while maintaining rigorous statistical guarantees even in the presence of gate-dependent noise.
These feaurures makes RPS broadly applicable to quantum benchmarking, error mitigation, and hybrid quantum-classical learning protocols.

Looking forward, the structural simplicity and robustness of RPS suggest several promising directions for experimental realization and generalization. First, the ${CZ}$-based design may be naturally extended to other native entangling gates such as iSWAP or $\sqrt{\text{SWAP}}$, broadening compatibility across different hardware platforms. Second, beyond fidelity estimation, the RPS framework may be adapted for other tasks including cross-platform state comparison, overlap estimation, and hybrid quantum benchmarking. Third, shallow-circuit variants of RPS could further reduce experimental cost while retaining robustness, especially when combined with variational or noise-aware sampling strategies. Finally, the underlying phase-random circuit ensemble may prove useful in broader contexts of quantum information theory, such as the construction of random codes for quantum error correction or randomized decoupling schemes for noise suppression. These directions highlight the versatility of phase-based random circuits and their potential as foundational tools in near-term quantum technologies.
}

\section*{Methods}

Here, we present important information about the arguments in the main text. We start by providing some definitions. For an operator $A$ represented in the computational basis as a Boolean matrix, we say that $\ket{i}\bra{j} \in A$ if and only if the matrix element $A_{i,j} = 1$.
The \textit{union} (\textit{intersection}) 
of two operators $A$ and $B$, denoted $A \cup(\cap) B$, 
is the one whose matrix representation is the element-wise logical OR(AND) of $A$ and $B$:
$(A \cup(\cap) B)_{i,j} = A_{i,j} \lor(\land) B_{i,j}$. Note that $A\cap B$ is indeed the Hadamard product of the corresponding matrices.

\subsection*{Moment function of phase shadow}

Here, we provide a proof sketch of \cref{prop:23moment} with $m=2$, i.e., $\mb{M}_{\mc{E}_{\text{phase}}}^{(2)}=D^{-2}(\id_D^{\otimes 2}\cup \mbb{S}_D^{\otimes 2})$ with the diagrammatic derivation. The proof of $m=3$ and more general cases are left 
to be shown in Supplementary Note 3.
Recall the definition of the moment function 
$\mb{M}_{\mc{E}_{\text{phase}}}^{(2)}= D^{-1}\mbb{E}_{U\sim \mc{E}_{\text{phase}}} \sum_{\mb{b}} \Phi_{U,\mb{b}}^{\otimes 2}$, with $\Phi_{U,\mb{b}}=U^{\dag} \ket{\mb{b}}\bra{\mb{b}}U$ and $U={H}^{\otimes n}U_A\in \mc{E}_{\text{phase}}$.

Let us first consider a single qubit, say $n=1$, one has $\Phi_{U,b}= U_A^{\dag}{H} {X}^{b}\ket{{0}}\bra{0} {X}^{b} {H}U_A=U_A^{\dag}{Z}^{b}\ket{+}\bra{+} {Z}^{{b}}U_A=\frac1{2}\sum_{x,y\in\{0,1\}}U_A^{\dag}Z^{{b}}\ket{{x}}\bra{{y}}Z^{{b}}U_A$. The measurement result $b$ acts effectively as a random ${Z}$ gate, and we expand $\ket{+}\bra{+}$ in the computational basis. The moment function now becomes 
\begin{equation}
\mb{M}_{\mc{E}_{\text{phase}}}^{(2)}\!\!=\!\frac1{4}\mbb{E}_{\{U_A, b\}} \!\!\!\sum_{x,y,w,s} \!\! (U_A^{\dag}{Z}^{b})^{\otimes 2} \!\ket{x}\!\bra{y}\otimes \ket{w}\!\bra{s}\! ({Z}^{b}U_A)^{\otimes 2},\end{equation}
and $\ket{w}\bra{s}$ is for the second copy. In the $n=1$ case, $U_A$ is just from $\{\id_2,{S}\}$, and together with ${Z}^{b}$, composes a random gate set $\{{S}^k\}_{k=0}^3$. It is clear that $U_A$ is diagonal and only changes the phase of the operator $\ket{x,w}\bra{y,s}$ in the 2-copy computational basis. So the average or twirling effect of random $U_A {Z}^b$ would turn the coefficient of some operators into zero.  Here, we find that there are six out of $16$ survive, and they belongs to three non-overlapping operators, which are
$\Delta_2 = \ket{0,0}\bra{0,0} + \ket{1,1}\bra{1,1}$,
$\id_4 - \Delta_2 = \ket{0,0}\bra{1,1} + \ket{1,1}\bra{0,0}$ and
$\mathbb{S}_2 - \Delta_2 = \ket{0,1}\bra{1,0} + \ket{1,0}\bra{0,1}$. Consequently, the result shows $1/4(\Delta_2+(\id_4 - \Delta_2)+(\mathbb{S}_2 - \Delta_2))=1/4(\id_4 + \mathbb{S}_2 - \Delta_2)=(1/4)(\id_4\cup \mbb{S}_2)$, which finishes the proof for $n=1$. 

One can generalize it to the $n$-qubit case as follows. Since the operation of random ${S}$ gates and the projective measurement to get $\mb{b}$ are both independently conducted for each qubit. So without the twirling effect of random $CZ$ layer in $U_A$, one has the $n$-qubit two-copy result directly as the tensor product of single-qubit ones,  $D^{-2}(\id_4 + \mathbb{S}_2 - \Delta_2)^{\otimes n}$. The effect of ${CZ}$ twirling is to transform $(\id_4+ \mbb{S}_2-\Delta_2)^{\otimes n }$ to $\id_4^{\otimes n}+ \mbb{S}_2^{\otimes n}-\Delta_2^{\otimes n}$, which is explained as follows.

Let $\mc{B} = \{\Delta_2,\id_4-\Delta_2,\mbb{S}_2-\Delta_2\}$. For any two-qubit two-copy operator $B_1 \otimes B_2$ with $B_1, B_2 \in \mc{B}$. Here, the twirling result under a random ${CZ}$ gate reads $\mbb{E}_{U\sim\{\id^{\otimes 2},{CZ}\}} U^{\dagger\otimes2} B_1\otimes B_2 U^{\otimes2}=0$, iff $\{B_1,B_2\} =\{\id_4-\Delta_2,\mbb{S}_2-\Delta_2\}$, otherwise it keeps unchanged. 
\comments{
\begin{equation}\label{Eq:EMCZTwir}
    \begin{split}
        &\mbb{E}_{U\sim\{\id^{\otimes 2},{CZ}\}} U^{\dagger\otimes2} B_1\otimes B_2 U^{\otimes2} \\&= 
    \begin{cases}
        B_1 \otimes B_2 & \text{if } \{B_1,B_2\} \neq \{\id_4-\Delta_2,\mbb{S}_2-\Delta_2\}, \\
        0               & \text{otherwise}.
    \end{cases}\\
    \end{split}
\end{equation}
}
That is, among the nine possible combinations of the two-copy operators, two of them are twirled to be zero.

One can take $(\id_4+ \mbb{S}_2-\Delta_2)^{\otimes n}$ as a trinomial expansion of the three possible qubit-level operators in $\mc{B}$, thus containing all possible terms in the form of \(
\Delta_2^{\otimes n_1} \otimes (\mathbb{I}_4 - \Delta_2)^{\otimes n_2} \otimes (\sw_2 - \Delta_2)^{\otimes n_3}
\). By random ${CZ}$ gates on all qubit pairs, such operators would be twirled to zero iff $n_2n_3\neq 0$. The operators that does not satisfy this condition must be from two binomials $\{\mathbb{I}_4 - \Delta_2,\Delta_2\}^{\otimes n}$ or $\{\mathbb{S}_2 - \Delta_2,\Delta_2\}^{\otimes n}$. 
As a result, one has
\begin{equation}
    (\id_4+ \mbb{S}_2-\Delta_2)^{\otimes n }\xrightarrow{\text{CZ}} \id_4^{\otimes n}+ \mbb{S}_2^{\otimes n}-\Delta_2^{\otimes n}=\id_D^{\otimes 2}\cup\mbb{S}_D,
\end{equation}
which concludes the proof.

\subsection*{Estimation variance of phase shadow}
We show in Supplementary Note 3 that 
the variance of the  estimator in \cref{eq:shadow2} can be related to the $m=3$ moment function $\mb{M}_{\mc{E}_{\text{phase}}}^{(3)}$ as $\text{Var}(\widehat{o_f}) 
\leq D^3 {\rm tr} (\rho\otimes O_f \otimes O_f\text{ } \mb{M}_{\mc{E}_{\text{phase}}}^{(3)})={\rm tr}(\rho\otimes O_f \otimes O_f\text{ } \bigcup_{\pi \in S_3}V_n(\pi))$, by \cref{prop:23moment}.
    
    \comments{
    \begin{equation}
    \begin{split}
\text{Var}(\widehat{o_f}) &\leq \mathbb{E}\ [\tr(O_f\widehat{\rho_f})]^2
=D^3 \tr(\rho\otimes O_f \otimes O_f\text{ } \mb{M}_{\mc{E}_{\text{phase}}}^{(3)}) \\&= \tr[\rho\otimes O_f \otimes O_f\text{ } \bigcup_{\pi \in S_3}V_n(\pi)],
    \end{split}
    \end{equation}
where in the last we insert the result of $\mb{M}_{\mc{E}_{\text{phase}}}^{(3)}$ in \cref{prop:23moment}.
}

Since the operator $O_f$ is off-diagonal, one does not need to consider $V_n(\pi)$ that does not permute the second or the third copy. In particular, $\tr[\rho\otimes O_f \otimes O_f\text{ }\ket{\mb{x,w,z}}\bra{\mb{y,s,t}}]=0$ for any $\ket{\mb{x,w,z}}\bra{\mb{y,s,t}}\in V_n(\pi_{()})\cup  V_n(\pi_{(12)})\cup V_n(\pi_{(13)})$, because the elements of this set satisfy either $\mb{w}=\mb{s}$ or $\mb{z}=\mb{t}$. As such, only elements from $V_n(\pi_{(23)})\cup V_n(\pi_{(123)})\cup V_n(\pi_{(132)})$ are worth considering. As shown in \cref{fig:EMm3} (a), 
by using the principle of inclusion-exclusion, one finds 
\begin{equation}\label{Eq:EMvarfinal}
   \begin{split}
\text{Var}(\widehat{o_f}) &\leq  \tr\big\{\rho\otimes O_f \otimes O_f [V_n(\pi_{(23)})+V_n(\pi_{(123)})+V_n(\pi_{(132)})\\
&\ \ \ \ -V_n(\pi_{(23)})\cap V_n(\pi_{(123)})-V_n(\pi_{(23)})\cap V_n(\pi_{(132)})]\big\}\\
&=\tr(\rho)\tr(O_f^2)+2\tr(\rho O_f^2)-2\sum_{\mb{x}}\bra{\mb{x}}\rho\ket{\mb{x}}\bra{\mb{x}}O_f^2\ket{\mb{x}}\\
& \leq \tr(O_f^2)+2\tr(\rho O_f^2)-2\tr(\rho_d O_f^2)\leq 3\tr(O_f^2).
\end{split}
\end{equation}
Indeed, all common elements of  $V_n(\pi_{(123)})$ and $V_n(\pi_{(132)})$ are contained in $V_n(\pi_{(23)})$.
The intersections for the last two terms in the second line like $V_n(\pi_{(23)})\cap V_n(\pi_{(123)})$ equal the Hadamard product of them, leading to a new tensor (shown in \cref{fig:EMm3} (b)) and thus the last term in the third line (see \cref{fig:EMm3} (c)).

\begin{figure}
    \centering
    \includegraphics[width=\linewidth]{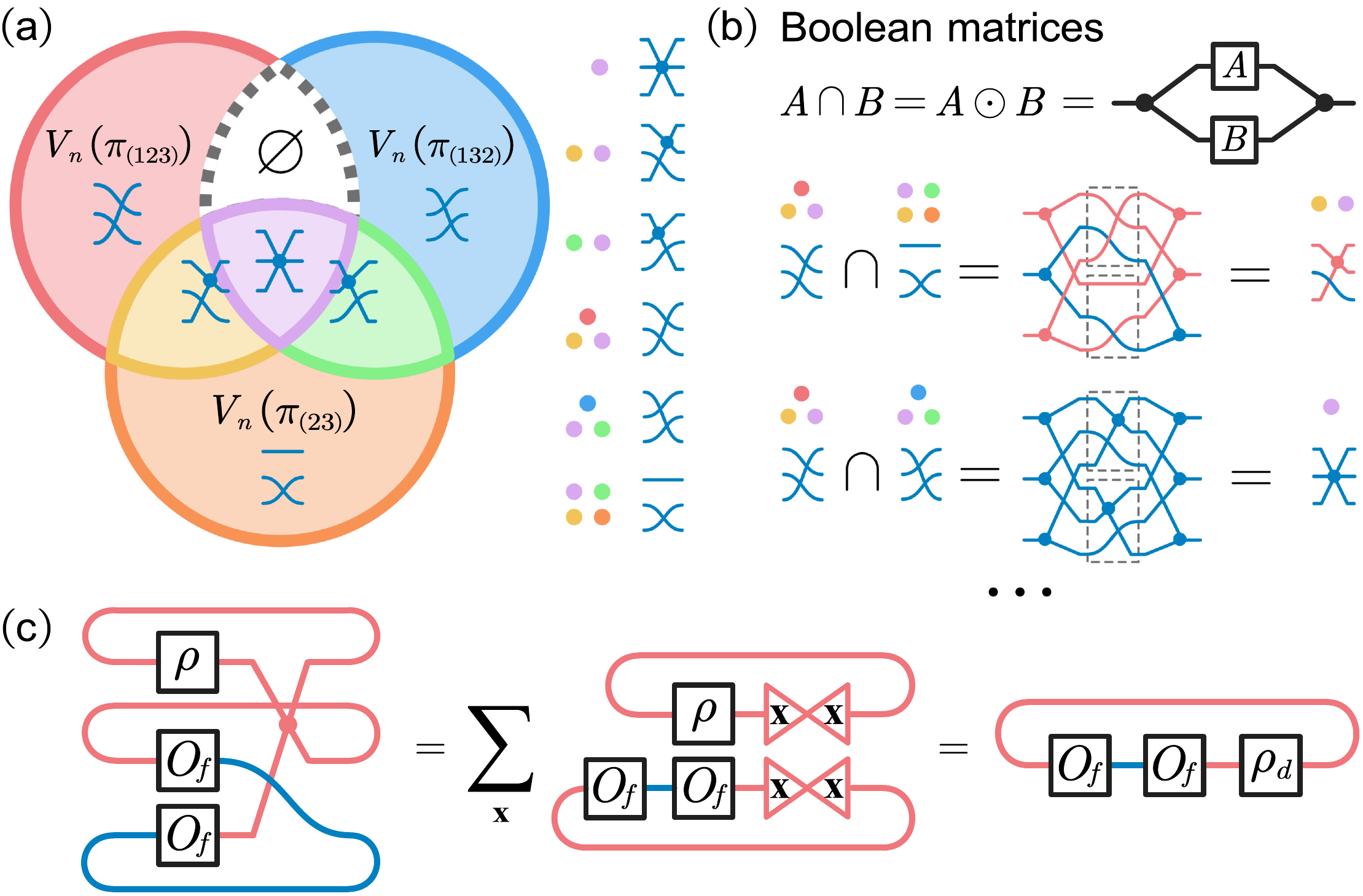}    
    \caption{A tensor diagrammatic illustration of \cref{Eq:EMvarfinal}. (a) Using the Principle of Inclusion-Exclusion, we transform $V_n(\pi_{(23)})\cup V_n(\pi_{(123)})\cup V_n(\pi_{(132)})$ to an alternating sum of Boolean tensors. We observe that the common elements of  $V_n(\pi_{(123)})$ and $V_n(\pi_{(132)})$ are all contained in $V_n(\pi_{(23)})$. (b) The intersection like $V_n(\pi_{(123)})\cap V_n(\pi_{(23)})$ and $V_n(\pi_{(123)})\cap V_n(\pi_{(132)})$ equal to the Hadamard product of them, respectively. (c) The contribution of the tensor $V_n(\pi_{(123)})\cap V_n(\pi_{(23)})$ to the estimation variance. By reconfiguring the tensor network (untwist the blue line), we derive the simplification $\sum_{\mb{x}}\bra{\mb{x}}\rho\ket{\mb{x}}\bra{\mb{x}}O_f^2\ket{\mb{x}} = \tr(\rho_dO_f^2)$.}

    \label{fig:EMm3}
\end{figure}

\subsection*{Efficient post-processing scheme} 
Here, we give a proof sketch of \cref{Prop:PostProcess} and an efficient post-processing scheme of robust phase shadow. First, consider the observable being a Pauli operator $O=Q \in \mb{P}_n/\mathcal{Z}_n$ in \cref{Eq:noiseOf}, one only need consider the term $Q$ in the summation and thus has $\widehat{Q} 
=D~\sigma_Q^{-1}\bra{\mb{b}}UQU^{\dagger}\ket{\mb{b}}=D~\sigma_Q^{-1}\bra{\mb{0}}{X}^{\mb{b}}UQU^{\dagger}{X}^{\mb{b}}\ket{\mb{0}}$, with ${X}^{\mb{b}}:=\bigotimes_i {X}_i^{b_i}$.  Note that $|\mb{0}\rangle \langle\mb{0}| = D^{-1}\sum_{\mb{a}}Z^{\mb{a}}$ as a stabilizer state \cite{aaronson2004improved}. As a result, one can apply the Tableau formalism \cite{aaronson2004improved} for Clifford operation of ${X}^{\mb{b}}U$ on $Q$ to check whether the resulting Pauli operator has overlap with $Z^{\mb{a}}$, which takes $\mc{O}(n^3)$ time. And it takes $\mc{O}(n^2)$ time to compute the value of $\sigma_Q$ by \cref{prop:noisyChannel}. As such, $\mc{O}(n^3)$ in total is sufficient. This efficient approach naturally extends to polynomial terms of Pauli operators in $O$.

Second, for the challenging stabilizer state, $O =\Psi=V^{\dagger}|\mb{0}\rangle \langle\mb{0}|V$, with $V$ being some Clifford unitary. Denote
the stabilizers of $\Psi$ and $\Phi_{U,\mb{b}} = U^{\dag}{X}^{\mb{b}}\ket{\mb{0}}\bra{\mb{0}}{X}^{\mb{b}} U$ as $S_V$ and $S_{X^{\mb{b}}U}$, and the stabilizer group as $\mb{S}_V$ and $\mb{S}_{X^{\mb{b}}U}$. And one has $\Psi=D^{-1}\sum S_V$, similar for $\Phi_{U,\mb{b}}$.  In this way, the estimator in \cref{Eq:noiseOf} reads 
\begin{equation}\label{eq:Mstab}
\tr(\widehat{\Psi_f}) = D^{-2} \sum_{P\in \mb{P}_n/\mathcal{Z}_n} \sum_{\substack{S_V\in \mb{S}_V \\ S_{X^{\mb{b}}U}\in \mb{S}_{X^{\mb{b}}U}}} \sigma_P^{-1}\tr( S_{X^{\mb{b}}U} P)\tr( S_V P)
.    
\end{equation}
Let us use $[S_V]$ to denote the corresponding phaseless Pauli operator of $S_V$, and the phaseless stabilizer group $[\mb{S}_V]$. In fact, one has $[\mb{S}_{X^{\mb{b}}U}]=[\mb{S}_U]$. In this way, it is clear that one should consider $P\in[\mb{S}_U]\cap[\mb{S}_V]\cap \mb{P}_n/\mathcal{Z}_n$ in \cref{eq:Mstab}, otherwise the term returns zero. To determine all possible $P$ in this subset, which in general quantifies the post-processing cost,
equivalently, one only needs to find all $\mb{a}$ such that $[UV^{\dag}(Z^{\mb{a}})VU^{\dag}]\in \mc{Z}_n$. Indeed, such $\mb{a}$ compose a subspace in the binary field (say such $Z^{\mb{a}}$ as a group), we denote the rank of it as $n_g(U,V)$, and we prove in Supplementary Note 7 the following.

\begin{lemma}[Time complexity]\label{prop:time_complexity}
Let $V$ be any fixed Clifford unitary and $U$ be sampled uniformly from the phase unitary ensemble $\mathcal{E}_{\text{phase}}$. The expectation of the random variable $2^{n_g(U,V)}$ which quantifies the number of the shared stabilizers (without phase) of $U$ and $V$, and thus the post-processing cost, shows $\mathbb{E}_{U\sim\mathcal{E}_{\text{phase}}}[2^{n_g(U,V)}] = \mathcal{O}(1)$.
\end{lemma}

In the following, we use the 
(${Z}$)-Tableau formalism \cite{aaronson2004improved} to efficiently determine all such ${Z}^{\mb{a}}$ and summarize the whole post-precessing  in \cref{Algo:EMeffpostproc}. For a Clifford unitary $W$, the operation on single-qubit operator ${Z}_i$ neglecting the phase shows $[W^{\dagger} {Z}_{i}W] = \prod_{j=0}^{n-1} \sqrt{-1}^{\gamma_{i,j}\delta_{i,j}}{X}_{j}^{\gamma_{i,j}} {Z}_{j}^{\delta_{i,j}}$, and the matrix $[{C},{D}]=[{{[{{\gamma }_{i,j}}]}_{n\times n}},{{[{{\delta }_{i,j}}]}_{n\times n}}]$ are called the ${Z}$-Tableau without phase.

\begin{algorithm}[H]
    \caption{\justifying Efficient post-processing of the robust phase shadow of stabilizer-state observables}
    \label{Algo:EMeffpostproc}

    \KwInput{\justifying Random unitary $U \in \mathcal{E}_{\text{phase}}$, projective measurement result $\lvert \mb{b} \rangle$, and stabilizer-state observable $\Psi=V^\dagger \lvert \mb{0} \rangle \langle \mb{0} \rvert V$ with $V$ being Clifford.}
    \KwOutput{Fidelity estimator $\widehat{\Psi_f}$ in \cref{Eq:noiseOf}.}

    \vspace{0.5em} 

    $F \leftarrow 0$\;
    Represent the Z-Tableau of $W=VU^\dag$ as $[{C}, {D}]$\;
    Perform Gaussian elimination on $[{C} \ {D}]$ to obtain the row echelon form $\begin{bmatrix}
    {C}_1 & {D}_1 \\
    \varnothing & {D}_2
    \end{bmatrix}$, with $\varnothing$ being an $n_g\times n$ null matrix and let ${D}_2=[\mb{c}_0,\mb{c}_1,\dots, \mb{c}_{n_g-1}]^T$\;
    
    \For{$\mb{v} \in \{0,1\}^{n_g}$}{
        $\mb{a} \leftarrow \sum_{j=0}^{n_g-1}v_j\mb{c}_j$\;
        \If{$P=[V^{\dag}Z^{\mb{a}}V]\in \mb{P}_n/\mathcal{Z}_n$}{
            Calculate the value $\sigma_P$ via \cref{Eq:InvNoisychannelPauli}\;
            $F \leftarrow F + \sigma_P^{-1}\bra{\mb{b}}U^{\dagger}PU\ket{\mb{b}}\bra{\mb{0}}V^{\dagger}PV\ket{\mb{0}}$\;
        }
    }
    \KwRet{$F$}\;
\end{algorithm}
The computational costs for steps 2, 3, 7 (single-Pauli case) and 8 are all $O(n^3)$ using the Tableau formalism. Given that steps 7 and 8 are iterated $2^{n_g(U,V)}$ times, 
the overall cost of \cref{Algo:EMeffpostproc} is consequently $\mathcal{O}(n^3) + \mathcal{O}(n^3 \cdot 2^{n_g(U,V)})$, which is $\mathcal{O}(n^3)$ by expectation based on \cref{prop:time_complexity}. In addition, for the estimator of diagonal part $\widehat{\Psi_d}=\bra{\mb{b}}\Psi\ket{\mb{b}}$, it can also be calculated in $\mathcal{O}(n^3)$ similarly. More detailed information about the efficient post-processing scheme is left to Supplementary Note 7.

\subsection*{Extended noise model} In the RPS protocol, we utilize the ${ZZ}$ Pauli error of $CZ$ gates as the gate-dependent noise. Here, we present an extended noise model that also considers ${IZ}$ and $ZI$ errors, which are native in the neutral atom platforms ~\cite{wu2022erasure,cong2022hardware}. In this way, we denote the noisy two-qubit gates $\widetilde{{CZ}}= {CZ}\circ\Lambda$, where $\Lambda(\rho) = (1-\frac{3}{4}p_e)\rho +\frac{p_e}{4} ZI \rho ZI +\frac{p_e}{4} IZ \rho IZ +\frac{p_e}{4} ZZ \rho ZZ$. Using the noise model, one can also achieve unbiased estimation using the estimator in \cref{Eq:InvNoisychannelPauli}, where the coefficient changed to $\sigma_P =  \sum_{s=0}^{n_1 + n_2} \sum_{t=\max(0,s-n_1)}^{\min(s, n_2)} (-1)^t \binom{n_1}{s-t} \binom{n_2}{t}  (1 - \frac{p_e}{2})^{(n_1 + n_2 - s)n_3+\frac{n_3(n_3-1)}{2}} (\frac{p_e}{2})^{sn_3}$ for Pauli in the form of $ P = \id_2^{\otimes n_1}
\otimes  {Z}^{\otimes n_2}\otimes \{{X},{Y}\}^{\otimes n_3}$. See the detailed proof in Supplementary Note 9. Hence, the post-processing scheme is also efficient for Pauli operators and stabilizer states as shown in the previous section. Finally, numerical results show that the estimation variance under the updated noise model also scales exponentially with $n^2p_e$, just as \cref{Eq:Var_first}.

\section*{Data availability} The data that support the findings of
this study are openly available \cite{RPSData}.

\section*{Code availability} The code used to create the data is openly available \cite{RPSData}.





\section*{Acknowledgments}
This work has been supported by the Innovation Program for Quantum Science and Technology Grant Nos.~2024ZD0301900 and 2021ZD0302000, the National Natural Science Foundation of China (NSFC) Grant No.~12205048, the Shanghai Science and Technology Innovation Action Plan Grant No.~24LZ1400200, Shanghai Pilot Program for Basic Research - Fudan University 21TQ1400100 (25TQ003), and the start-up funding of Fudan University.
It has also been supported by the BMFTR (MuniQCAtoms), the Munich Quantum Valley, Berlin Quantum, the Clusters of Excellence (ML4Q, MATH+), and the DFG (CRC 183).
\section*{Author contributions}
Y.Z. and Q.Z. conceived the initial idea for the project, with substantial contributions from J.E. and F.X. Q.Z. and Y.Z. developed the theory of robust phase shadow, with substantial inputs from J.E., D.Q. and Z.Y. Moreover, Z.Y. and Q.Z. performed the numerical analysis and prepared the figures. All authors contributed substantially to the final manuscript.

\section*{Competing interests}
The authors declare no competing interests.

\section*{Correspondence} Correspondence and requests for materials should be addressed to F.~Xu, J.~Eisert, or Y.~Zhou.






\newpage
\onecolumngrid

\begin{appendix}
\bigskip

\section*{Supplementary Note 1. --Notations and preliminaries}\label{App:UaA}

At the beginning of this appendix, we introduce some 
basic notations and conventions. Throughout this work, we employ set operations—including intersection, union, and belonging—to describe the relationship between Boolean matrices. Let $A$ and $B$ be two Boolean matrices or operators with compatible dimensions. The union and intersection of $A$ and $B$ 
are, respectively, defined as bit-wise OR and AND operations as
\begin{eqnarray}
    C = A \cup B:\; c_{i,j} = a_{i,j} \vee b_{i,j},\\
    C = A\cap B:\;c_{i,j} = a_{i,j}\wedge b_{i,j},
\end{eqnarray} where $a_{i,j},b_{i,j}$ and $c_{i,j}$ are matrix elements of $A,B$ and $C$. We illustrate these operations in \cref{fig:prelim}(a). The intersection is nothing but the Hadamard product, whose diagrammatic representation is shown in \cref{fig:prelim}(c). Here, the concept of the Boolean matrices can be extended to operators that are Boolean in the computational basis. We say $\ket{i}\bra{j}\in A$, if $\bra{i}A\ket{j} = 1$, i.e., $a_{i,j}=1$. Also, given two matrices $A,B$ with compatible dimensions, we say $A\subseteq B$ if $b_{i,j}=1$ for all $(i,j)$ such that $a_{i,j}=1$.

%
Second, we introduce the notation for sets of finite integers. Let $[n]=\{0,1,\dots ,n-1\}$, then we denote with $I_p \subset [n]$ a subset of $[n]$ with $p$ elements, while $\mc{I}_p$ denotes the collection of all subsets of $[n]$ with $p$ elements. Consequently, when $I_p$ refers to a set of qubits, $A^{I_p}$ implies applying the operator $A$ to every qubit indexed in $I_p$. In addition to the conventional notation of the union operation of sets, we introduce the \emph{summation} of sets $I_p \in \mathcal{I}_p$. When $I_p^{(1)}\cap I_p^{(2)}  = 0$, we say that  
\begin{equation}
I_p^{(1)}+I_p^{(2)} \coloneq I_p^{(1)}\cup I_p^{(2)}. 
\end{equation}
In the remainder of this manuscript, whenever we use $I_p^{(1)}+I_p^{(2)}$, it means  that $I_p^{(1)}\cap I_p^{(2)}  = 0$. 
%
%
Third, we introduce the representation of the symmetry group on the Hilbert space of the qubit system. Supposing that $\pi$ is an element of the symmetric group $S_m$, we denote with $V_n(\pi)$ the representation of $\pi$ on the $n$-qubit system, i.e., 
\begin{equation}
V_n(\pi)\ket{\psi_0}\otimes\ket{\psi_1}\otimes\cdot\cdot\cdot\otimes\ket{\psi_{m-1}} = \ket{\psi_{\pi^{-1}(0)}}\otimes\ket{\psi_{\pi^{-1}(1)}}\otimes\cdot\cdot\cdot\otimes\ket{\psi_{{\pi^{-1}(m-1)}}},
\end{equation}
for any $\ket{\psi_i}\in \mc{H}^{D}$ with $D=2^n$. We illustrate the 2nd and 3rd order symmetric group in \cref{fig:prelim}(b).

Finally, we introduce the group of Pauli operators. The Pauli group for $n$-qubit quantum system is $\mbb{P}_n=\{\pm \sqrt{\pm{1}}\times\{\id_2,
X,{Y},{Z}\}\}^{\otimes n}$, where $\id_2$, ${X,Y,Z}$ are the single-qubit Pauli operators. We refer to $\id_d$ as the $d$-dimensional identity matrix, while $\id_D$ will sometimes be abbreviated as $\id$. We will  also both use $I$ or $\id$, depending on whether we would like to emphasize whether we mean the identity matrix or the identity element of the Pauli group. 
For convenience, we denote the phaseless Pauli operators as $\mb{P}_n=\{\id_2,{X},{Y},{Z}\}^{\otimes n}$, and $\mb{P}^*_n={P}_n \setminus \id$ as the non-identity phaseless Pauli operators. The Pauli 
 Clifford group is defined as the normalizer of the Pauli group, i.e., $U^{\dag}PU\in \mbb{P}_n, \forall P\in \mathbb{P}_n$. A Clifford operator is determined by its action on $U^{\dag}{X}_iU$ and $U^{\dag}{Z}_iU$ for $i\in[n]$, because ${X}_i$ and ${Z}_i$ are the generators of the Pauli group, where ${X}_i = \id_2^{\otimes i-1}\otimes {X}\otimes \id_2^{\otimes n-i}$ and similarly for ${Z}_i$. 

\begin{figure}
    \centering
    \includegraphics[width=.85\linewidth]{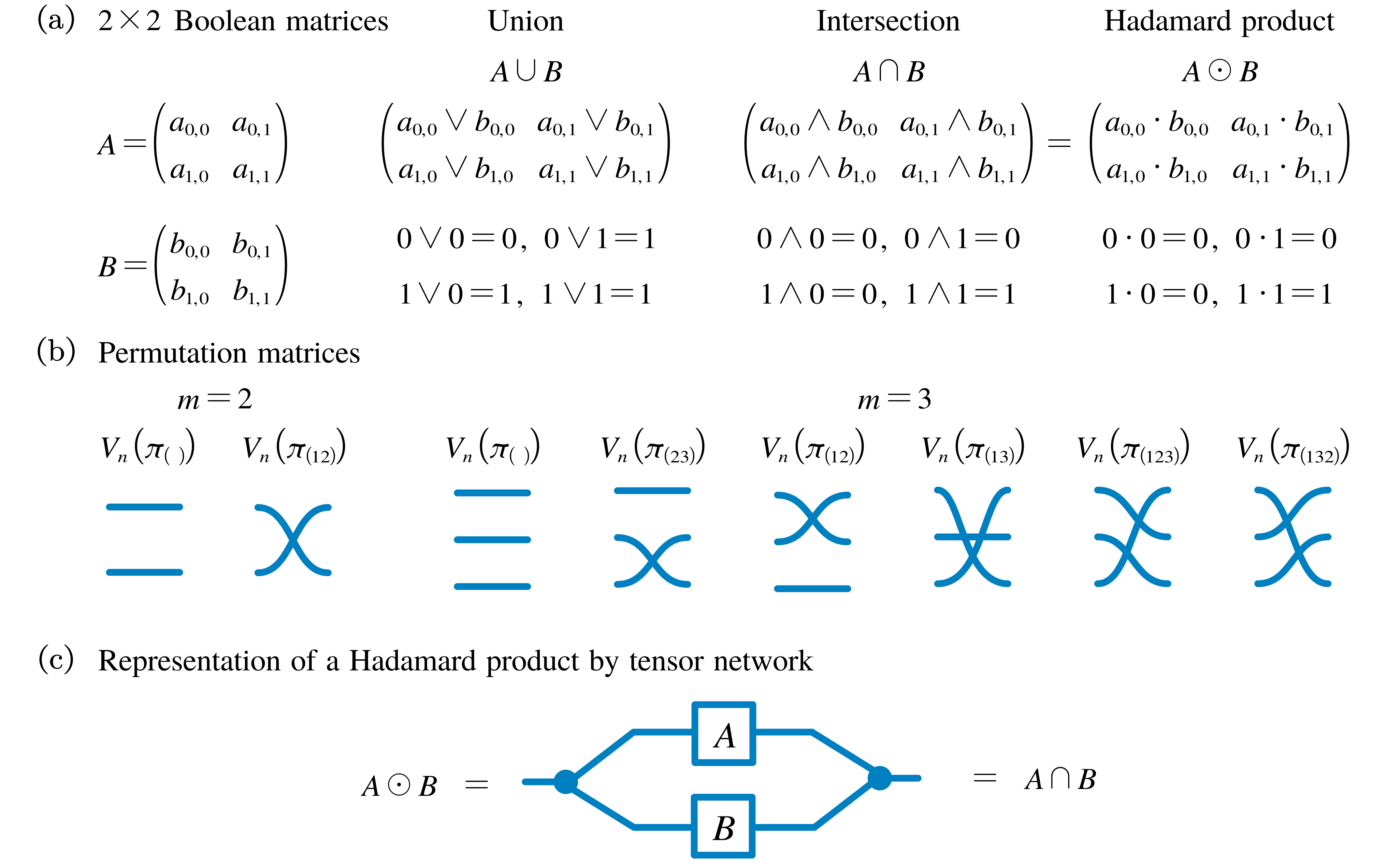}
    \caption{Graphical illustration of the notations used in this manuscript.}
    \label{fig:prelim}
\end{figure}




\section*{Supplementary Note 2. --Ensemble of general diagonal circuits}\label{Ap:DiagEnsem}

In this work, we employ circuits with a ${CZ-S-H}$ structure, which we refer to as phase circuits, as our main tool for achieving randomized measurements. In this section, we introduce a more general diagonal circuit ensemble, of which the phase circuit ensemble is a special case. A random diagonal circuit incorporates three levels of randomness. The first level corresponds to applying single random diagonal gates. The second level corresponds to applying diagonal gates on all possible qubit subsets, which gives rise to a layer of random diagonal circuits. The third level corresponds to concatenating multiple layers. 
Here, we provide a detailed introduction to the structure of the ensemble of general diagonal circuits via the three levels of randomness. At the first level, we uniformly sample $p$-qubit diagonal gates from the ensemble 
\begin{equation}
    \mc{E}_{p, q, I_p} = \{\operatorname{diag} (1, 1, 1, \dots, e^{i\frac{2\pi}{2^q}k})\mid k = 0, 1, \cdots, 2^q-1\},
\end{equation}
where $p<n$ is the number of qubits, $q\geq 1$ is an integer, and $I_p$ denotes the subset of $[n]$ of size $p$. $I_{p}$ specifies the qubits on which the gate is being applied. For example, when $p=2, q=1$, the ensemble becomes the $\{\mbb{I}_4,{CZ}\}$. When $p=1$ and $q=2$, the ensemble becomes $\{\mbb{I}_2, {S, S^2=Z, S^3=S}^\dagger\}$. 

At the second level, we apply a random gate from the ensemble $\mc{E}_{p,q,I_p}$ to all possible subsets of $[n]$ with size $p$. As a result, we obtain a layer of diagonal gates uniformly sampled from the ensemble
\begin{equation}\label{Eq:blockDiagonal}
    \mc{E}_{p,q} = \left\{ \prod_{I_p \in \mathcal{I}_p} U_{p,q,I_p} \mid U_{p,q,I_p} \in \mc{E}_{p,q,I_p} \right\}.
\end{equation}
For example, the ensemble $\mc{E}_{1,2}$ means that a single-qubit gate is sampled uniformly and independently from the set $\{\mathbb{I}_2, {S}, {Z}, {S}^{\dagger}\}$. Likewise, the ensemble $\mc{E}_{2,1}$ means that a two-qubit gate is sampled uniformly and independently from the set $\{\mathbb{I}_2^{\otimes 2},{CZ}\}$.

At the third level, we concatenate multiple layers of diagonal circuits and an additional layer of Hadamard gates on every qubits to obtain the final circuit $U_{\text{diag}} = U_{p,p}U_{p-1,p}\dots U_{1,p}H^{\otimes n}$, where $U_{p,p}\in\mc{E}_{p,p}, U_{p-1,p}\in \mc{E}_{p-1,p},\dots, U_{1,p}\in \mc{E}_{1,p}$. The collection of all $U_{\text{diag}}
$ constitutes the unitary ensemble of general diagonal circuits given by
\begin{equation}\label{Eq:defGeneralEns}
    \mc{E}_{\text{diag}}^{(p)} = \{ U_{p,p}U_{p-1,p}\dots U_{1,p}{H}^{\otimes n}|U_{p,p}\in\mc{E}_{p,p},U_{p-1,p}\in \mc{E}_{p-1,p},\dots ,U_{1,p}\in \mc{E}_{1,p}\}. 
\end{equation}

As a special example, the unitary ensemble $\mc{E}_{\text{diag}}^{(2)}$ follows the circuit structure ${CS-S-H}$, which is not Clifford due to the ${CS}$ gates. Classical shadow tomography, in general, requires that the random circuit can be efficiently simulated classically. Therefore, we replace the ${CS}$  gates with the ${CZ}$ gates, obtaining the ensemble of phase circuits $\mc{E}_\text{phase}$ introduced in the main text. As we will demonstrate in Supplementary Note 3(B), this replacement leaves the 2nd and 3rd moments of the ${CS-S-H}$ ensemble and ${CZ-S-H}$ ensemble unchanged.

\section*{Supplementary Note 3. --Moment functions of diagonal circuit ensembles}\label{ap:preProp-23moment}

The moment functions, particularly, 
the 2nd and 3rd moments, are crucial for computing the inverse map of randomized measurements and conducting theoretical performance analysis. In this section, we introduce the definition of the moment function and derive the moment function of the first-level ensemble $\mc{E}_{p,q,I_p}$. 
Denote the $m$-copy twirling channel over a unitary ensemble $\mathcal{E}$ as
\begin{equation}
\Lambda_{\mathcal{E}}^{(m)}(\cdot) \coloneq \mathbb{E}_{U \in \mathcal{E}} U^{\dagger\otimes m} (\cdot) U^{ \otimes m},
\end{equation}
the definition of the moment function is as follows.
\begin{definition}[Moment functions]\label{def:n-moment}
For an $n$-qubit system, the $m$-th moment  function of $\mc{E}$ is defined as\begin{equation}
   \mb{M}_{\mc{E}}^{(m)} \coloneq\frac{1}{2^n}\sum_{\mb{b}}\Lambda_{\mc{E}}^{(m)}(|\mb{b}\rangle\langle \mb{b}|^{\otimes m}) .
\end{equation}  
\end{definition}
Introducing $\Phi_{U,\mb{b}} \coloneq U^{\dagger}|\mb{b}\rangle\langle \mb{b}| U$, we can rewrite the $m$-th moment function as 
\begin{equation}
\mb{M}_{\mc{E}}^{(m)} = \frac{1}{2^n}\mbb{E}_{U\sim \mc{E}} \sum_{\mb{b}} \Phi_{U,\mb{b}}^{\otimes m}. 
\end{equation} 
Before we are in the position to explore the moment function $\mb{M}_{\mc{E}_\text{phase}}^{(m)}$ of the phase circuit ensemble $\mathcal{E}_\text{phase}$, it is necessary to introduce the following lemma. 


\begin{lemma}[Twirling channels]\label{lem:twirling_diagonal}
Let $\mathbf{\hat{x}}$ and $\mathbf{\hat{y}}$ be two $mn$-bit binary numbers. The twirling channel of the first-level ensemble $\mc{E}_{p,q,I_p}$, when applied to an $m$-copy $n$-qubit operator $|\mathbf{\hat{x}}\rangle\langle \mathbf{\hat{y}}|$ in the computational basis, is given by
\begin{equation}
    \Lambda_{p,q,I_p}^{(m)}(|\mathbf{\hat{x}}\rangle\langle \mathbf{\hat{y}}|) =
\begin{cases}
|\mathbf{\hat{x}}\rangle\langle \mathbf{\hat{y}}|, & \sum\limits_{t=0}^{m-1} \prod\limits_{j\in I_p} x_{t,j} \equiv \sum\limits_{t=0}^{m-1} \prod\limits_{j\in I_p} y_{t,j} \pmod{2^q}, \\
0, & \text{otherwise}.
\end{cases}
\end{equation} where $x_{t,j}$ and $y_{t,j}$, respectively, label the binary element on the $j$-th qubit of the $t$-th copy of $\mathbf{\hat{x}},\mathbf{\hat{y}}$.
\end{lemma}

\begin{proof}
Consider a unitary operator $U = \text{diag}(1,1,1,\dots,e^{i\frac{2\pi}{2^q}k})$ acting on the qubit set $I_p$. Its action on $|\mathbf{x}\rangle$ is given by
\begin{equation}
    U^{\otimes m} |\mathbf{\hat{x}}\rangle = e^{i\frac{2\pi}{2^q} k \sum_{t=0}^{m-1} \prod_{j\in I_p} x_{t,j}} |\mathbf{\hat{x}}\rangle.
\end{equation}
Thus, applying the twirling channel on $|\mathbf{\hat{x}}\rangle\langle \mathbf{\hat{y}}|$ yields
\begin{equation}
\begin{aligned}
\Lambda_{p,q,I_p}^{(m)}(|\mathbf{\hat{x}}\rangle\langle \mathbf{\hat{y}}|)
&= \mathbb{E}_{U \in \mc{E}_{p,q,I_p}} U^{\dagger\otimes m} |\mathbf{\hat{x}}\rangle\langle \mathbf{\hat{y}}| U^{ \otimes m} \\
&= \frac{1}{2^q} \sum_{k=0}^{2^q-1} e^{i\frac{2\pi}{2^q} k \left( \sum_{t=0}^{m-1} \prod_{j\in I_p} y_{t,j} - \sum_{t=0}^{m-1} \prod_{j\in I_p} x_{t,j} \right)} |\mathbf{\hat{x}}\rangle\langle \mathbf{\hat{y}}| \\
&=
\begin{cases}
|\mathbf{\hat{x}}\rangle\langle \mathbf{\hat{y}}|, & \sum\limits_{t=0}^{m-1} \prod\limits_{j\in I_p} x_{t,j} \equiv \sum\limits_{t=0}^{m-1} \prod\limits_{j\in I_p} y_{t,j} \pmod{2^q}, \\
0, & \text{otherwise}.
\end{cases}
\end{aligned}
\end{equation}
This completes the proof. 
\end{proof}

Using \cref{lem:twirling_diagonal}, one can show that 
\begin{eqnarray}\label{Eq: twirling_composition_1}
    \Lambda_{p, q, I_p}^{(m)} \circ \Lambda_{p', q', I_p'}^{(m)}  (O) &=& \Lambda_{p, q, I_p}^{(m)}(O) \cap \Lambda_{p', q', I_p'}^{(m)} (O),\\
     \Lambda_{p, q, I_p}^{(m)} (O) \cap \Lambda_{p, q, I_p}^{(m)} (O') &=& 0,\; \forall O\cap O' = 0,
\end{eqnarray} 
for $m$-copy Boolean operators $O$ and $O^\prime$ in the computational basis with the compatible dimension, and we can extend these results to the twirling channel $\Lambda_{p,q}^{(m)}$ with the second-level ensemble $\mc{E}_{p,q}$ in \cref{Eq:blockDiagonal}
\begin{eqnarray}
    \Lambda^{(m)}_{p, q} \circ \Lambda^{(m)}_{p', q'}  (O) &=& \Lambda^{(m)}_{p, q}(O) \cap \Lambda^{(m)}_{p', q'} (O),\\
     \Lambda^{(m)}_{p, q} (O) \cap \Lambda^{(m)}_{p, q} (O') &=& 0,\; \forall O\cap O' = 0.
\end{eqnarray} 
Since the intersection operation is commutative, the composition of twirling on the $m$-copy operator is also mutually commutative.
%
%
Equipped with these results, we can proceed to prove the moment functions. In Supplementary Note 3(A), we derive the moment function for phase circuit ensembles (Proposition 1 in the main text), and in Supplementary Note 3(B), we prove it for general diagonal circuit ensembles.

\subsection*{A. Proof of Proposition 1}\label{ap:Prop-23moment-1}

Using $\ket{\mb{b}} = X^{\mb{b}}\ket{\mb{0}}\coloneq(\prod_{i:b_i=1}X_i)\ket{\mb{0}}$, we can write $\Phi_{U,\mb{b}}$ for $U\in \mc{E}_\text{phase}$ as
\begin{equation}\label{eq:PhiAb}
\begin{split}
    \Phi_{U,\mb{b}}= U^{\dagger}\ket{\mb{b}}\bra{\mb{b}} U&=U_A^{\dag}\ {H}^{\otimes n} X^{\mb{b}}\ket{\mb{0}}\bra{\mb{0}} X^{\mb{b}} {H}^{\otimes n}  U_A\\
    &=U_A^{\dag}\ Z^{\mb{b}} {H}^{\otimes n} \ket{\mb{0}}\bra{\mb{0}}  {H}^{\otimes n} Z^{\mb{b}}\ U_A\\
    &=\frac{1}{2^n}U_A^{\dag}\ Z^{\mb{b}} \sum_{\mb{x},\mb{y}}\ket{\mb{x}}\bra{\mb{y}} Z^{\mb{b}}\ U_A.
\end{split}
 \end{equation}
Taking the mean of $\Phi_{U,\mb{b}}$ over all $\mb{b}$, we obtain 
\begin{equation}\label{eq:sumPhiAb}
\frac{1}{2^n}\sum_\mb{b}{\Phi_{U,\mb{b}} } =\frac{1}{2^n}\sum_{\mb{x},\mb{y}} U_A^{\dag} \Lambda_{1,1}(\ket{\mb{x}}\bra{\mb{y}}) U_A,
\end{equation} where we have used $\Lambda_{1,1} (\bullet) = \frac{1}{2^n}\sum_{\mb{b}} Z^{\mb{b}}(\bullet) Z^{\mb{b}}$. Extending this result to the $m$-copy case, we arrive at 
\begin{equation}\label{eq:sumPhiAb_mcopy}
\frac{1}{2^n}\sum_{\mb{b}}{\Phi_{U,\mb{b}}^{\otimes m} } =\frac{1}{2^{mn}}\sum_{\mb{\hat{x}},\mb{\hat{y}}\in\{0,1\}^{mn}} U_A^{\dagger\otimes m}  \Lambda_{1,1}^{(m)} (|\mathbf{\hat{x}}\rangle\langle \mathbf{\hat{y}}|) U_A^{ \otimes{m}}.
\end{equation}
Taking the average of \cref{eq:sumPhiAb_mcopy} over $U\sim\mc{{E}}_\text{phase}$, we obtain the moment function of the phase circuit ensemble as
\begin{equation}\label{eq:momentAb}
\begin{split}
        \mb{M}_{\mc{E}_\text{phase}}^{(m)}\coloneq\frac{1}{2^n}\mathbb{E}_{U\sim\mc{{E}}_\text{phase}} \sum_\mb{b}{\Phi_{U,\mb{b}}^{\otimes m} } &=\frac{1}{2^{mn}}\sum_{\mb{\hat{x}},\mb{\hat{y}}\in\{0,1\}^{mn}}  \Lambda_{A}^{(m)} \circ \Lambda_{1,1}^{(m)}(|\mathbf{\hat{x}}\rangle\langle \mathbf{\hat{y}}|)\\
        &=\frac{1}{2^{mn}}\sum_{\mb{\hat{x}},\mb{\hat{y}}\in\{0,1\}^{mn}}  \Lambda_{2,1}^{(m)} \circ \Lambda_{1,2}^{(m)}\circ \Lambda_{1,1}^{(m)}(|\mathbf{\hat{x}}\rangle\langle \mathbf{\hat{y}}|)\\
        &=\frac{1}{2^{mn}}\sum_{\mb{\hat{x}},\mb{\hat{y}}\in\{0,1\}^{mn}}  \Lambda_{2,1}^{(m)}\circ \Lambda_{1,2}^{(m)}(|\mathbf{\hat{x}}\rangle\langle \mathbf{\hat{y}}|),
\end{split}
\end{equation}
where we have used $\Lambda_A^{(m)}=\Lambda_{2,1}^{(m)}\circ \Lambda_{1,2}^{(m)}$. From this derivation, we observe that sampling single-qubit gates independently of the set $\{\id_2, {S}\}$ for each qubit achieves the same effect as $\Lambda_{1,2}^{(m)}$ twirling, since there is already $\Lambda_{1,1}^{(m)}$ twirling stemming from the measurement result in the computational basis. Hence, we can simplify and set the single-qubit gates sampled from $\{\id_2, {S}\}$ to achieve the same effect as $\Lambda_{1,2}^{(m)}$. Next, we prove the results of Proposition 1 for the cases  $m=2$ and $m=3$ separately.

First, we derive the moment function for the $m=2$ case. We denote the two-copy operators as $ \ket{\mb{\hat{x}}} = \ket{\mb{x}}\otimes \ket{\mb{w}} = \ket{\mb{x,w}}$, and $ \bra{\mb{\hat{y}}} = \bra{\mb{y}}\otimes \bra{\mb{s}} = \bra{\mb{y,s}}$. Then  \cref{eq:momentAb} becomes
\begin{equation}
\begin{split}
    \mb{M}_{\mc{E}_\text{phase}}^{(2)}&=\frac{1}{2^{2n}}\sum_{\mb{x},\mb{w},\mb{y},\mb{s}\in\{0,1\}^{n}}  \Lambda_{2,1}^{(2)}\circ \Lambda_{1,2}^{(2)}(\ket{\mb{x,w}}\bra{\mb{y,s}}).
\end{split}
\end{equation}
 Then, utilizing \cref{lem:twirling_diagonal}, we obtain  
\begin{equation}\label{Eq:P2tem}
\begin{split}
    \mb{M}_{\mc{E}_\text{phase}}^{(2)}&=\frac{1}{2^{2n}}\sum_{(\mb{x},\mb{w},\mb{y},\mb{s})\in C_1\cap C_2}  \ket{\mb{x,w}}\bra{\mb{y,s}},\\
\end{split}
\end{equation}
where the sets (constraints on the matrix elements) $C_1, C_2$ are  given by
\begin{equation}
    \begin{split}
        C_1 &= \{(\mb{x},\mb{w},\mb{y},\mb{s})|x_l+w_l\equiv y_l+s_l\pmod4, \forall l\in [n]\},\\
        C_2 &= \{(\mb{x},\mb{w},\mb{y},\mb{s})|x_ix_j+w_iw_j\equiv y_iy_j+s_is_j\pmod2, \forall i\neq j\in[n]\}.
    \end{split}
\end{equation}
 $C_1$ and $C_2$ are introduced, respectively, due to the twirling $\Lambda^{(2)}_{1,2}$ and $\Lambda_{2,1}^{ (2)}$. 
We first analyze the single-qubit twirling $\Lambda_{1,2}^{(2)}$. Suppose that
\begin{equation}
  \ket{\mb{x,w}}\bra{\mb{y,s}} = \bigotimes_{l=0}^{n-1}\ket{x_l ,w_l }\bra{y_l ,s_l},
\end{equation}
we enumerate all of six elements $\ket{x_l, w_l }\bra{y_l, s_l}$ that satisfy the constraint $C_1$ in  \cref{fig:p=2noiseless}(a), finding that they are the matrix elements of $\id_2^{\otimes 2}$ or $\mbb{S}_2$. 
Therefore, 
\begin{equation}\label{Eq:2temp1}
    \sum_{(\mb{x},\mb{w},\mb{y},\mb{s})\in C_1}  \ket{\mb{x,w}}\bra{\mb{y,s}} = (\id_2^{\otimes 2}\cup \mbb{S}_2)^{\otimes n} = (\id_2^{\otimes 2}+ \mbb{S}_2-\Delta_2)^{\otimes n}.
\end{equation} 
The derivation of the second equality in \cref{Eq:2temp1} is illustrated as a tensor diagram in \cref{fig:p=2noiseless}(b). 

Moreover, these six one-qubit operators can be grouped into three mutually orthogonal operators: These are $\Delta_2 = \ket{0,0}\bra{0,0}+\ket{1,1}\bra{1,1}$, $\id_4-\Delta_2 = \ket{0,0}\bra{1,1}+\ket{1,1}\bra{0,0}$, and $\mbb{S}_2-\Delta_2 = \ket{0,1}\bra{1,0}+\ket{1,0}\bra{0,1}$. Then, we can write \cref{Eq:2temp1} 
as 
\begin{equation}\label{Eq:2temp2}
\begin{split}
    \sum_{(\mb{x},\mb{w},\mb{y},\mb{s})\in C_1}  \ket{\mb{x,w}}\bra{\mb{y,s}}  &= (\id_2^{\otimes 2}+ \mbb{S}_2-\Delta_2)^{\otimes n}\\
    &=[\Delta_2 +(\id_4-\Delta_2)+(\mbb{S}_2-\Delta_2)]^{\otimes n}\\
   &=\sum_{I_i+I_j+I_k=[n]} \Delta_2^{I_i}\otimes (\id_4-\Delta_2)^{I_j} \otimes (\mbb{S}_2-\Delta_2)^{I_k},
\end{split}
\end{equation}
where $I_i\in\mc{I}_i,I_j\in\mc{I}_j,I_k\in\mc{I}_k$ are mutually disjoint and $I_i+I_j+I_k = [n]$. 
Given two operators $\ket{x_i, w_i} \bra{y_i , s_i}$ and $\ket{x_j, w_j} \bra{y_j , s_j}$ from $\id_2^{\otimes 2}\cup \mbb{S}_2$, the constraint $C_2$ is violated ($x_ix_j+w_iw_j \not\equiv  y_iy_j+s_is_j \pmod2 $) if and only if
\begin{itemize}
    \item  $\ket{x_i, w_i} \bra{y_i , s_i} \in \id_4-\Delta_2 $ and 
    $\ket{x_j, w_j} \bra{y_j , s_j} \in \mathbb{S}_2-\Delta_2 $, or 
    \item $ \ket{x_i, w_i} \bra{y_i , s_i} \in \mathbb{S}_2-\Delta_2 $ and  $\ket{x_j, w_j} \bra{y_j , s_j} \in \id_4-\Delta_2 $.
\end{itemize}
As a result, the operator $ \ket{\mb{x,w}}\bra{\mb{y,s}} \in \Delta_2^{I_i}\otimes (\id_4-\Delta_2)^{I_j} \otimes (\mbb{S}_2-\Delta_2)^{I_k}$ violates the constraint $C_2$ if and only if $|I_j|\neq0$ and $|I_k|\neq0$ together. Therefore, imposing the constraint $C_2$ on \cref{Eq:2temp2}, we obtain 
\begin{equation}\label{Eq:M2phase_final}
\begin{split}
    \sum_{(\mb{x},\mb{w},\mb{y},\mb{s})\in C_1\cap C_2}  \ket{\mb{x,w}}\bra{\mb{y,s}}&=\sum_{I_i+I_j+I_k=[n]} \Delta_2^{I_i}\otimes (\id_4-\Delta_2)^{I_j} \otimes (\mbb{S}_2-\Delta_2)^{I_k}\delta_{\{|I_j||I_k|=0\}}\\
    &=\id_4^{\otimes n}+ \mbb{S}_2^{\otimes n}-\Delta_2^{\otimes n},
\end{split}
\end{equation}
where $\delta_{\{|I_j||I_k|=0\}}$ equals to 1 if and only if $|I_j||I_k|=0$. By plugging \cref{Eq:M2phase_final} into  \cref{Eq:P2tem}, we finally obtain the moment function for the $m=2$ case as 
\begin{equation}\label{Eq:Phi_2}
    \mb{M}_{\mc{E}_\text{phase}}^{(2)}= \frac{1}{2^{2n}}(\id_4^{\otimes n}+ \mbb{S}_2^{\otimes n}-\Delta_2^{\otimes n})=\frac{1}{2^{2n}}(\id_4^{\otimes n}\cup \mbb{S}_2^{\otimes n}).
\end{equation}

Second, we prove the moment function in the $m=3$ case. Denote the 3-copy operators $\ket{\mb{\hat{x}}}\bra{\mb{\hat{y}}}$ as $\ket{\mb{x}}\bra{\mb{y}}\otimes\ket{\mb{w}}\bra{\mb{s}}\otimes\ket{\mb{z}}\bra{\mb{t}}=\ket{\mb{x,w,z}}\bra{\mb{y,s,t}}$ ,   \cref{eq:momentAb} becomes 
\begin{equation}\label{Eq:Phi_3}
    \mb{M}_{\mc{E}_\text{phase}}^{(3)}  = \frac{1}{2^{3n}}\sum_{\mb{x},\mb{w},\mb{z},\mb{y},\mb{s},\mb{t}\in\{0,1\}^{n}}  \Lambda_{1,2}^{(3)}\circ \Lambda_{2,1}^{(3)}(\ket{\mb{x,w,z}}\bra{\mb{y,s,t}}). 
\end{equation} 
Using  \cref{lem:twirling_diagonal}, we find that 
\begin{equation}\label{Eq:P3tem}
\begin{split}
    \mb{M}_{\mc{E}_\text{phase}}^{(3)}&=\frac{1}{2^{3n}}\sum_{(\mb{x},\mb{w},\mb{z},\mb{y},\mb{s},\mb{t})\in C_1\cap C_2}  \ket{\mb{x,w,z}}\bra{\mb{y,s,t}},
\end{split}
\end{equation}
where the sets $C_1,C_2$ are updated to the 3-copy case, showing
\begin{equation}
    \begin{split}
        C_1 &= \{(\mb{x},\mb{w},\mb{z},\mb{y},\mb{s},\mb{t})|x_l+w_l+z_l\equiv y_l+s_l+t_l\pmod 4, \forall l\in [n]\},\\
        C_2 &= \{(\mb{x},\mb{w},\mb{z},\mb{y},\mb{s},\mb{t})|x_ix_j+w_iw_j+z_iz_j\equiv y_iy_j+s_is_j+t_it_j\pmod2, \forall i\neq j\in[n]\}.\\
    \end{split}
\end{equation} 
Suppose that
\begin{equation}
    \ket{\mb{x,w,z}}\bra{\mb{y,s,t}} = \bigotimes_{l=0}^{n-1}\ket{x_l , w_l , z_l}\bra{y_l , s_l , t_l},
\end{equation}
we enumerate the matrix elements $\ket{x_l , w_l , z_l}\bra{y_l , s_l , t_l}$ belonging to $C_1$ along with their corresponding permutation group elements in \cref{fig:p=3noiseless}(a), which shows that these matrix elements simultaneously belong to either 2 or 6 permutation operations from $\{V_n(\pi)|\pi\in S_3\}$, and their summation constitutes the union of the permutation operations, i.e., 
$\sum_{(x_l,w_l,z_l,y_l,s_l,t_l)\in C_1}\ket{x_l , w_l , z_l}\bra{y_l , s_l , t_l} = \bigcup_{\pi\in S_3 }V_1(\pi)$. Therefore, 
\begin{equation}\label{Eq:3temp1}
    \sum_{(\mb{x},\mb{w},\mb{z},\mb{y},\mb{s},\mb{t})\in C_1}  \ket{\mb{x,w,z}}\bra{\mb{y,s,t}} 
    = [\cup_{\pi\in S_3}V_1(\pi)]^{\otimes n}.
\end{equation}
Moreover, we find that given the operators $\ket{x_i , w_i , z_i}\bra{y_i , s_i , t_i}$
and $\ket{x_j ,w_j ,z_j}\bra{y_j ,s_j ,t_j}$ from $\bigcup_{\pi\in S_3 }V_1(\pi)$, the second constraint is satisfied  ($x_ix_j+w_iw_j+z_iz_j\equiv y_iy_j+s_is_j+t_it_j\pmod 2$) if and only if they belong to a common one-qubit permutation operator from $\{V_1(\pi)|\pi\in S_3\}$, as indicated by  \cref{fig:p=3noiseless}. 
In this way, the effect of
the constraint $C_2$ enforces all $\ket{x_l , w_l , z_l}\bra{y_l , s_l , t_l}$
with $(l\in [n])$ to belong to a common $V_1(\pi)$, and thus any $(\mb{x},\mb{w},\mb{z},\mb{y},\mb{s},\mb{t})\in C_1\cap C_2$ should satisfy $\ket{\mb{x,w,z}}\bra{\mb{y,s,t}}\in V_n(\pi)=\bigotimes_{l=0}^{n-1} V^{(l)}_1(\pi)$. Hence,  we conclude that $\sum_{(\mb{x},\mb{w},\mb{z},\mb{y},\mb{s},\mb{t})\in C_1\cap C_2}\ket{\mb{x,w,z}}\bra{\mb{y,s,t}} \subseteq \bigcup_{\pi \in S_3}V_n(\pi)$.  

At the same time, if one selects some operators from the permutation operators, say $\ket{\mb{x,w,z}}\bra{\mb{y,s,t}}\in V_n(\pi)$, it is clear that both $C_1$ and $C_2$ is satisfied, that is, $\bigcup_{\pi \in S_3}V_n(\pi) \subseteq \sum_{(\mb{x},\mb{w},\mb{z},\mb{y},\mb{s},\mb{t})\in C_1\cap C_2}\ket{\mb{x,w,z}}\bra{\mb{y,s,t}}$. The completes the proof for the case of $m=3$.


\begin{figure}
    \centering
    \includegraphics[width=.85\linewidth]{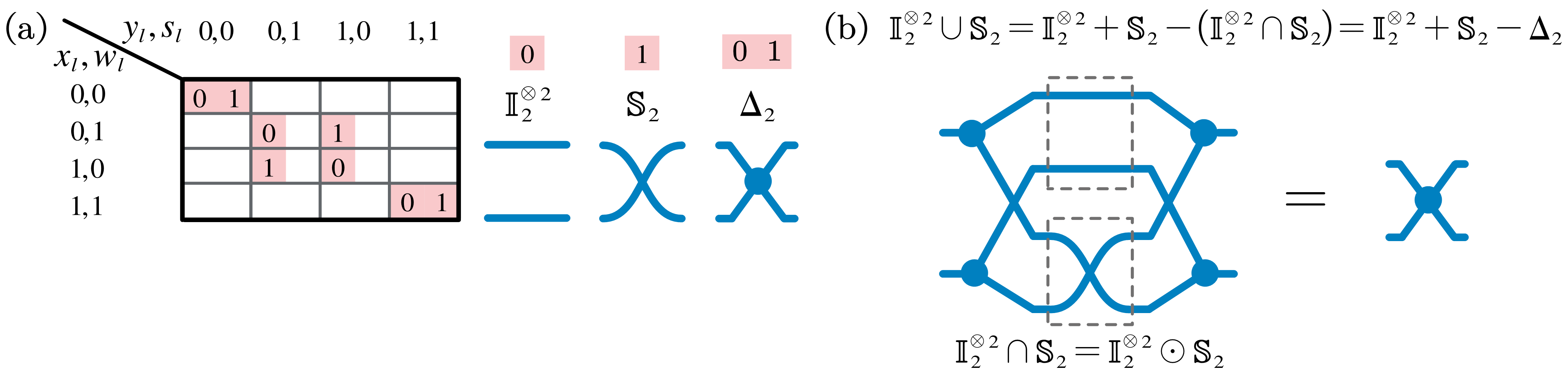}
    \caption{(a) An illustration of the matrix elements $\ket{x_l, w_l }\bra{y_l, s_l }$ that lie in $C_1$ and their corresponding permutation operators. For each element, the number marked in red indicates the permutation operator(s) to which $\ket{x_l, w_l }\bra{y_l, s_l }$ 
    belongs.  (b) a tensor diagram showing that the intersection and the union of $\id_2^{\otimes 2} $ and $\mbb{S}_2$. 
}
    \label{fig:p=2noiseless}
\end{figure}

\begin{figure}
    \centering
    \includegraphics[width=\linewidth]{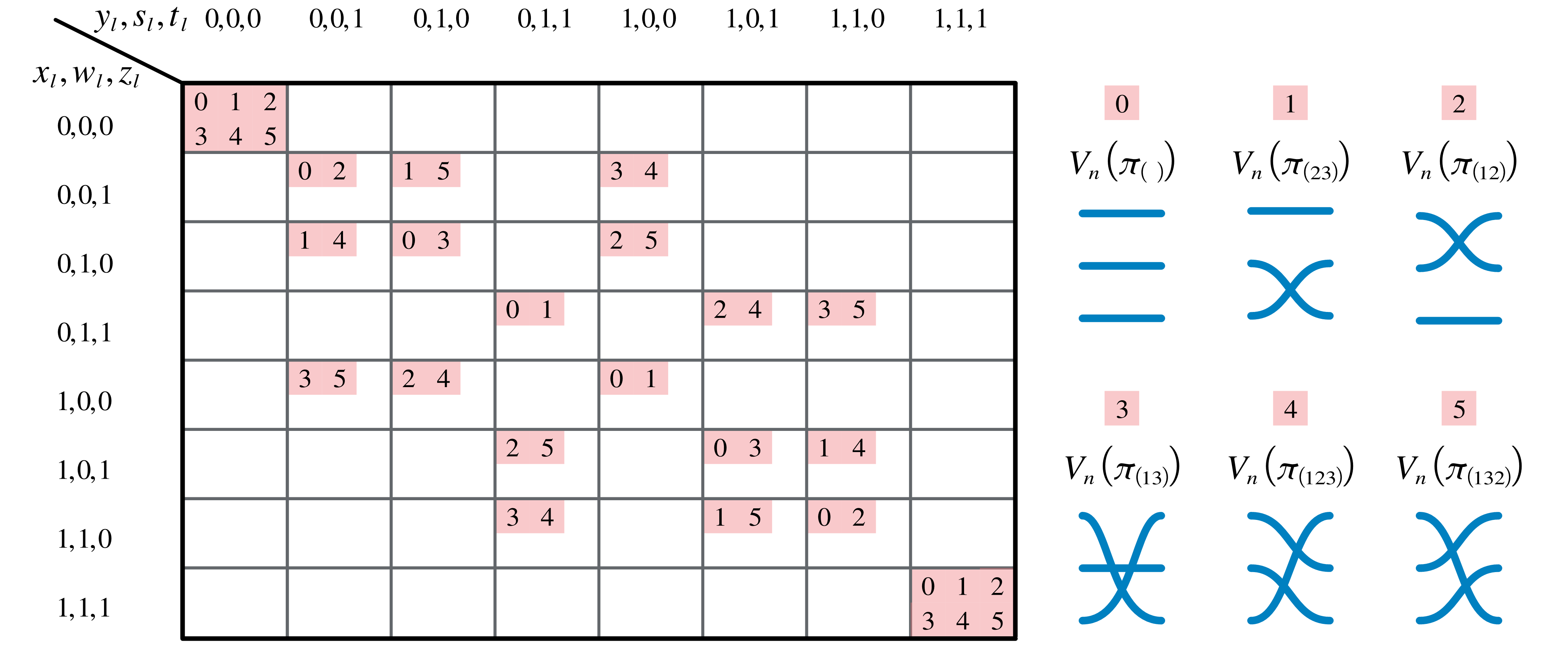}
    \caption{An illustration of the 20 possible strings of 
    $(x_l, w_l, z_l, y_l, s_l, t_l)$ is shown. For each string, the number marked in red indicates the permutation operator(s) to which $\ket{x_l , w_l , z_l}\bra{y_l , s_l , t_l}$ belongs. For example, $
\ket{0,0,1}\bra{0,0,1} \in V_1(\pi_{()}),\ V_1(\pi_{(12)})$.   As an example, consider  $\ket{x_i, w_i, z_i } \bra{y_i, s_i, t_i} = \ket{0,0,1}\bra{0,0,1}, \quad
\ket{x_j, w_j, z_j} \bra{y_j, s_j, t_j} = \ket{0,1,1}\bra{1,0,1}$.  Since both operators belongs to $\ V_1(\pi_{(12)})$, we have  $x_ix_j+w_iw_j+z_iz_j-y_iy_j-s_is_j-t_it_j = 0$.
}
    \label{fig:p=3noiseless}
\end{figure}

\subsection*{B. Moment function of general diagonal circuits}\label{Ap:MomentGeneral}

In the previous subsection, we
have proven Proposition 1 which concerns the 2nd and 3rd moment functions of the phase circuit ensemble $\mc{E}_\text{phase}$. In this subsection, we present and prove the $m$-th moment functions for the general diagonal circuit ensemble $\mc{E}^{(p)}_{\text{diag}}$ defined in \cref{Eq:defGeneralEns}. In this proof, we will not resort to enumerating all the matrix elements in diagrams like the previous proof of $\mc{E}_{\text{phase}}$.

\begin{lemma}[$m$-th moment function]\label{lem:m-moment}
The $m$-th moment function of the general diagonal circuit ensemble $\mc{E}_{\text{diag}}^{(p)}$ is given by
\begin{equation}
    \mb{M}_{\mc{E}_{\text{diag}}^{(p)}}^{(m)} =  \frac{1}{2^{mn}}\bigcup_{\pi \in S_m} V_n(\pi),
    \label{Eq: Nakata_spirit}
\end{equation}
     where $m<2^p$ and $n>p$.
\end{lemma}

\begin{proof}
First, similar to the proof in  Supplementary Note 3(A), we  simplify the moment function $\mb{M}_{\mc{E}_{\text{diag}}^{(p)}}^{(m)}$ by taking the mean of $\Phi_{U,\mb{b}}^{\otimes m}$ over all choices of $\mb{b}$, and express the average over $U \sim \mc{E}_{\text{diag}}^{(p)}$ as the $m$-copy twirling channel, i.e.,
\begin{equation}
\begin{split}
\frac{1}{2^n}\mbb{E}_{U\sim\mc{E}_{\text{diag}}^{(p)}} \sum_{\mb{b}} \Phi_{U,\mb{b}}^{\otimes m}&=\frac{1}{2^{mn}}\sum_{\mb{\hat{x}},\mb{\hat{y}}\in\{0,1\}^{mn}}   \Lambda_{p,p}^{^{(m)}} \circ\Lambda_{p-1,p}^{{(m)}}\circ\dots \circ\Lambda_{1,p}^{(m)}\circ \Lambda_{1,1}^{(m)}(|\mathbf{\hat{x}}\rangle\langle \mathbf{\hat{y}}|)\\
&=\frac{1}{2^{mn}}\sum_{\mb{\hat{x}},\mb{\hat{y}}\in\{0,1\}^{mn}}    \Lambda_{p,p}^{(m)}(|\mathbf{\hat{x}}\rangle\langle \mathbf{\hat{y}}|) \cap\Lambda_{p-1,p}^{{(m)}}(|\mathbf{\hat{x}}\rangle\langle \mathbf{\hat{y}}|)\cap\dots \cap\Lambda_{1,p}^{{(m)}} (|\mathbf{\hat{x}}\rangle\langle \mathbf{\hat{y}}|)\cap\Lambda_{1,1}^{{(m)}}(|\mathbf{\hat{x}}\rangle\langle \mathbf{\hat{y}}|)\\
&=\frac{1}{2^{mn}}\sum_{ (\mathbf{\hat{x}},\mathbf{\hat{y}})\in C_{p,p}\cap C_{p-1,p}\cap\dots \cap C_{1,p}}   |\mathbf{\hat{x}}\rangle\langle \mathbf{\hat{y}}|,
\end{split}
\label{Eq: Nakata_pre}
\end{equation}
where, according to \cref{lem:twirling_diagonal}, the sets $C_{p,q}$ are defined as
\begin{equation}\label{Eq:Cpq_generral}
    C_{p,q} \coloneq \{(\mathbf{\hat{x}},\mathbf{\hat{y}})|\sum_{t=0}^{m-1} \prod_{j\in I_p}x_{t,j} = \sum_{t=0}^{m-1} \prod_{j\in I_p}y_{t,j}, \forall I_p\in\mc{I}_p\}.
\end{equation}
In \cref{Eq:Cpq_generral}, we omit the modulo $2^q$ operation because $m<2^q$.

The procedures of proving $\sum_{ (\mathbf{\hat{x}},\mathbf{\hat{y}})\in C_{p,p}\cap C_{p-1,p}\cap\dots \cap C_{1,p}}   \ket{\mb{\hat{x}}}\bra{\mb{\hat{y}}} = \bigcup_{\pi \in S_m} V_n(\pi)$ can be divided into two parts. 

\textbf{Part 1:} The proof of $\bigcup_{\pi \in S_m} V_n(\pi)\subseteq\sum_{ (\mathbf{\hat{x}},\mathbf{\hat{y}})\in C_{p,p}\cap C_{p-1,p}\cap\dots \cap C_{1,p}}   \ket{\mb{\hat{x}}}\bra{\mb{\hat{y}}}$. 
This is equivalent to saying that every element of the permutation operators $\ket{\mb{\hat{x}}}\bra{\mb{\hat{y}}} \in V_n(\pi)$ will obey the constraints 
$ (\mathbf{\hat{x}},\mathbf{\hat{y}})\in C_{p,p}\cap C_{p-1,p}\cap\dots \cap C_{1,p}$, which can be checked since for all $|\mathbf{\hat{x}}\rangle\langle \mathbf{\hat{y}}| \in V_n(\pi)$. We find that $y_{t,j} = x_{\pi^{-1}(t),j}$, and then 
\begin{equation}
   \sum_{t=0}^{m-1} \prod_{j\in I_k}x_{t,j} = \sum_{t=0}^{m-1} \prod_{j\in I_k}x_{\pi^{-1}(t),j} = \sum_{t=0}^{m-1} \prod_{j\in I_k}y_{t,j} 
\end{equation}
for all $k\leq p$.

\textbf{Part 2:} The proof of $\sum_{ (\mathbf{\hat{x}},\mathbf{\hat{y}})\in C_{p,p}\cap C_{p-1,p}\cap\dots \cap C_{1,p}}   |\mathbf{\hat{x}}\rangle\langle \mathbf{\hat{y}}|\subseteq\bigcup_{\pi \in S_m} V_n(\pi)$. 
We now regard $\mathbf{\hat{x}}$ as an $m\times n$ matrix, whose rows correspond to copies and columns correspond to qubit indices.
Define, for any subset $I_k \subseteq [n]$, the matrix $\mathbf{\hat{x}}_{I_k}$ as the $m \times k$ matrix formed by the columns(qubits) indexed by $I_k$ from $\mathbf{\hat{x}}$. For a binary string $\mathbf{b}_k \in \{0,1\}^k$, define $N_{\mathbf{\hat{x}}}(I_k, \mathbf{b}_k)$ as the number of rows(copies) in $\mathbf{\hat{x}}_{I_k}$ that equal to $\mathbf{b}_k$. With these notations, we can restate the conditions in Eq.~\eqref{Eq:Cpq_generral}. For example,  the set $C_{1,q}$ ($q>1$) includes the elements $(\mathbf{\hat{x}},\mathbf{\hat{y}})$ that 
\begin{equation}
    \forall I_1\in\mc{I}_1,\text{  }N_{\mb{\hat{x}}}(I_1, \{1\}) = N_{\mb{\hat{y}}}(I_1, \{1\}).
    \label{Eq:Nakata_ex1}
\end{equation}
Likewise, for $m< 2^q$, the set $C_{k,q}$ includes the elements $(\mathbf{\hat{x}},\mathbf{\hat{y}})$ that 
\begin{equation}
    \forall I_k\in \mc{I}_k,\text{  }N_{\mb{\hat{x}}}(I_k, \mb{1}_k) = N_{\mb{\hat{y}}}(I_k, \mb{1}_k), 
    \label{Eq:Nakata_exk} 
\end{equation}
where $\mb{1}_k$ denotes a $k$-bit row vector with all elements being 1. This equation holds because the term $\prod_{j\in I_k}x_{t,j}$ in Eq.~\eqref{Eq:Cpq_generral} is nonzero only if $x_{t,j}=1$ for all $j\in I_k$.

As such, we now aim to prove that every element $ (\mathbf{\hat{x}},\mathbf{\hat{y}})\in C_{p,p}\cap C_{p-1,p}\cap\dots \cap C_{1,p}$ satisfies
\begin{equation}
    \forall k \leq p, \mb{b}_k \in \{0,1\}^{k}, I_k\subseteq[n], N_{\mb{\hat{x}}}(I_k , \mb{b}_k) = N_{\mb{\hat{y}}}(I_k , \mb{b}_k). 
    \label{Eq: Nakata_prop1}
\end{equation} 
We establish this by induction: if \cref{Eq: Nakata_prop1} holds when the weight of $\mb{b}_k$ is $l (l\leq k)$, then it also holds when the weight of $\mb{b}_k$ is $l-1$. Initially,  \cref{Eq: Nakata_prop1} holds when $\mb{b}_k$ is an all-ones row string, i.e., $\mb{b}_k = \mb{1}_k$, which corresponds directly to the constraint set $C_{k,p}$. 

Now we consider an arbitrary $(k-1)$-weight binary string 
$\mb{1}_k^{(k-1)}$, 
where all but one bit are 1 and the remaining bit is 0. Let 
$I_{k-1}$ 
denote the subset of indices corresponding to the ones in 
$\mb{1}_k^{(k-1)}$. From the constraint $C_{k-1,p}$, 
we have
\begin{equation}
    \begin{split}
        N_{\mb{\hat{x}}}(I_{k-1} , \mb{1}_{k-1}) &= N_{\mb{\hat{y}}}(I_{k-1} , \mb{1}_{k-1})\\
       \Rightarrow N_{\mb{\hat{x}}}(I_k , \mb{1}_k) + N_{\mb{\hat{x}}}(I_k , \mb{1}_k^{(k-1)}) &=  N_{\mb{\hat{y}}}(I_k , \mb{1}_k)+N_{\mb{\hat{y}}}(I_k , \mb{1}_k^{(k-1)}).
    \end{split}
\end{equation}
From \cref{Eq:Nakata_exk} we know that $N_{\mb{\hat{x}}}(I_k = \mb{1}_k) =  N_{\mb{\hat{y}}}(I_k = \mb{1}_k) $. This equality implies that $N_{\mb{\hat{x}}}(I_k = \mb{1}_k^{(k-1)}) =  N_{\mb{\hat{y}}}(I_k = \mb{1}_k^{(k-1)}) $. Thus, \cref{Eq: Nakata_prop1} has been proven for all $(k-1)$-weight strings $\mb{b}_k$.

We now proceed inductively for general $l$-weight strings $\mb{1}_k^{(l)}$, where $l < k-1$. 
Let  $I_l$ be the set of indices corresponding to the ones in $\mb{1}_k^{(l)}$. Suppose that we already have  \begin{equation}
   N_{\mb{\hat{x}}}(I_k , \mb{1}_k^{(l^{'})}) =  N_{\mb{\hat{y}}}(I_k , \mb{1}_k^{(l^{'})})
\end{equation} 
for all $\mb{1}_k^{(l^{'})}$ and $l^{'}>l$. From the 
constraint $C_{l,p}$, it follows that
\begin{equation}
    \begin{split}
        N_{\mb{\hat{x}}}(I_{l} , \mb{1}_{l}^{ }) &=  N_{\mb{\hat{y}}}(I_{l} , \mb{l}_l^{})\\
        \Rightarrow \sum_{\mb{1}_{k,+}^{(l)}-\mb{1}_{k}^{(l)}\succeq\mb{0}}N_{\mb{\hat{x}}}(I_{k} , \mb{1}_{k,+}^{(l) }) &=  \sum_{\mb{1}_{k,+}^{(l)}-\mb{1}_{k}^{(l)}\succeq\mb{0}}N_{\mb{\hat{y}}}(I_{k} , \mb{1}_{k,+}^{(l) }),
    \end{split}
\end{equation}
where $\mb{1}_{k,+}^{(l)} - \mb{1}_k^{(l)} \succeq \mb{0}$ means that every 1 in $\mb{1}_k^{(l)}$ remains 1 in 
$\mb{1}_{k,+}^{(l)}$. 
$\succ \mb{0} $ further requires that at least one 0 in $\mb{1}_k^{(l)}$ 
becomes 1 in $\mb{1}_{k,+}^{(l)}$.
Using the induction hypothesis
\begin{equation}
N_{\mb{\hat{x}}}(I_{k} , \mb{1}_{k,+}^{(l) }) =  N_{\mb{\hat{y}}}(I_{k} , \mb{1}_{k,+}^{(l) }) \text{,  }\forall \mb{1}_{k,+}^{(l)}-\mb{1}_{k}^{(l)}\succ\mb{0}.
\end{equation}
We can cancel all terms in the above sum except the one corresponding to $\mb{1}_k^{(l)}$, proving that
\begin{equation}
    N_{\mb{\hat{x}}}(I_k , \mb{1}_k^{(l)}) =  N_{\mb{\hat{y}}}(I_k , \mb{1}_k^{(l)}). 
\end{equation}
Hence, \cref{Eq: Nakata_prop1} is established for all $l$-weight strings $\mb{b}_k$. 
Therefore, the claim in \cref{Eq: Nakata_prop1} holds universally.

In the next, we use \cref{Eq: Nakata_prop1} to prove that $N_{\mathbf{\hat{x}}}(I_n , \mathbf{b}_n) = N_{\mathbf{\hat{y}}}(I_n , \mathbf{b}_n)
\quad \text{for all } \mathbf{b}_n \in \{0,1\}^n$, which directly entails the validity 
of $\ket{\mathbf{\hat{x}}}\bra{\mathbf{\hat{y}}} \in \bigcup_{\pi \in S_m} V_n(\pi)$.
Note that the proof has already been carried out in Lemma 1 of Ref.~\cite{nakata2014generating} by proving its contraposition. For completeness, we restate the proof here.

Assume there exist  $\mathbf{\hat{x}}, \mathbf{\hat{y}}, \mathbf{b}_n$ such that
\begin{equation}
    N_{\mathbf{\hat{x}}}(I_n , \mathbf{b}_n) \neq N_{\mathbf{\hat{y}}}(I_n , \mathbf{b}_n).
\end{equation}
Then, by \cref{Eq: Nakata_prop1}, there must exist an integer $p' > p$, a subset $I_{p'} \in \mathcal{I}_{p'}$, and a binary string $\mathbf{b}_{p'} \in \{0,1\}^{p'}$ such that
\begin{equation}
    N_{\mathbf{\hat{x}}}(I_{p'} , \mathbf{b}_{p'}) \neq N_{\mathbf{\hat{y}}}(I_{p'} , \mathbf{b}_{p'}),
\end{equation}
while for all 
$I_{p'-1} \subset I_{p'}$ and all corresponding $\mathbf{b}_{p'-1} \in \{0,1\}^{p'-1}$, it holds that
\begin{equation}
    N_{\mathbf{\hat{x}}}(I_{p'-1} , \mathbf{b}_{p'-1}) = N_{\mathbf{\hat{y}}}(I_{p'-1} , \mathbf{b}_{p'-1}).
\end{equation}
We now consider $\hat{\mathbf{x}}_{I_{p'}}$ and $\hat{\mathbf{y}}_{I_{p'}}$, each of size 
$m \times p'$. 
Suppose some rows in $\hat{\mathbf{x}}_{I_{p'}}$ and 
$\hat{\mathbf{y}}_{I_{p'}}$ are identical, we can remove the identical rows from both matrices. Let $\tilde{\mathbf{x}}_{I_{p'}}$ and 
$\tilde{\mathbf{y}}_{I_{p'}}$ denote the resulting
$m' \times p'$ 
matrices, where the remaining rows are distinct between the two.
Without loss of generality, assume $\tilde{\mathbf{x}}_{I_{p'}}$ 
contains 
$ h$ all-zero rows (i.e., 
$\mathbf{0}_{p'}$), while $\tilde{\mathbf{y}}_{I_{p'}}$
contains no such row. Since
\begin{equation}
N_{\mathbf{\hat{x}}}(I_{p'-1} ,\mathbf{b}_{p'-1}) = N_{\mathbf{\hat{y}}}(I_{p'-1} , \mathbf{b}_{p'-1}) \quad \forall I_{p'-1} \subset I_{p'},
\end{equation}
each missing all-zero row in $ \tilde{\mathbf{y}}_{I_{p'}}$ must be compensated by a weight-1 row. Thus, 
$ \tilde{\mathbf{y}}_{I_{p'}}$ must contain at least $ h$ weight-1 rows.
Continuing this reasoning, since all weight-1 patterns in $\tilde{\mathbf{x}}_{I_{p'}}$ are absent while $\tilde{\mathbf{y}}_{I_{p'}}$ must have them, $ \tilde{\mathbf{x}}_{I_{p'}}$ must contain at least $ h$ weight-2 rows. Inductively, we find
$\tilde{\mathbf{x}}_{I_{p'}}$ contains 
$h$ copies of each even-weight binary vector of length $p'$,
 and zero copies of all odd-weight vectors.
Hence, the total number of rows in $\tilde{\mathbf{x}}_{I_{p'}}$ is given by 
\begin{equation}
m' = 2^{p'-1} \cdot h,
\end{equation}
which implies
\begin{equation}
m \geq m' = 2^{p'-1}h \geq 2^p,
\end{equation}
contradicting the assumption of Lemma~\ref{lem:m-moment} that $m < 2^p$.
This contradiction completes the proof.
\end{proof}

We remark that fully continuously random diagonal circuits \cite{nakata2014generating,nechita2021graphical} also yield the same result with $\mathcal{E}_{\text{diag}}^{(p)}$, as shown in \cref{lem:m-moment}. Furthermore, the discrete variability (or degree of discretization) of the gates can be further reduced, similar to the method proposed in Ref.~\cite{nakata2014generating}.

Finally, we compare the results between $\mc{E}_{\text{diag}}^{(p)}$ and $\mc{E}_{\text{phase}}$ presented in the previous two subsections. 
Note that $\mc{E}_{\text{diag}}^{(p)}$ is not merely a simple generalization of $\mc{E}_{\text{phase}}$: one can not directly obtain $\mc{E}_\text{phase}$ by letting $p=2$. However, we can still prove the moment function of $\mc{E}_\text{phase}$ using \cref{lem:m-moment} with a little bit extra work as follows. For $m=2 \; \text{and }3$, the moment function of $\mc{E}_\text{phase}$ reads 
\begin{equation}
    \mb{M}_{\mc{E}_\text{phase}}^{(m)} =\frac{1}{2^{mn}}\sum_{\mb{\hat{x}},\mb{\hat{y}}\in C_{1,2}\cap C_{2,1}}  \ket{\mb{\hat{x}}}\bra{\mb{\hat{y}}},
\end{equation}
while the moment function of $\mc{E}_{{\text{diag}}}^{(p)}$ becomes 
\begin{equation}
    \mb{M}_{\mc{E}_{{\text{diag}}}^{(p)}}^{(m)} =\frac{1}{2^{mn}}\sum_{\mb{\hat{x}},\mb{\hat{y}}\in C_{1,2}\cap C_{2,2}}  \ket{\mb{\hat{x}}}\bra{\mb{\hat{y}}}, 
\end{equation} where the sets $C_{1, 2}$, $C_{2, 1}$ and $C_{2, 2}$ depend on $m$ implicitly as in Eq.~\eqref{Eq:Cpq_generral}. The difference between $C_{2, 1}$ and $C_{2, 2}$ is that: $C_{2, 1}$ requires $ \sum_{t=0}^{m-1} \prod_{j\in I_p}x_{t,j} \equiv \sum_{t=0}^{m-1} \prod_{j\in I_p}y_{t,j} \pmod2, \forall I_p\in\mc{I}_p\ $, while $C_{2, 2}$ requires $ \sum_{t=0}^{m-1} \prod_{j\in I_p}x_{t,j} \equiv \sum_{t=0}^{m-1} \prod_{j\in I_p}y_{t,j} \pmod4, \forall I_p\in\mc{I}_p\ $. It is easy to show that $C_{1, 2}\cap C_{2, 2} \subseteq C_{1, 2}\cap C_{2, 1}$. In the next, we show that $C_{1, 2}\cap C_{2, 1} \subseteq C_{1, 2}\cap C_{2, 2}$ in the cases of $m=2$ and $3$. 

When $m=2$, both $\mb{\hat{x}}$ and $\mb{\hat{y}}$ have a size of $2\times n$. We consider two columns $\mathbf{x}_1, \mathbf{x}_2$ from $\mathbf{\hat{x}}$ and $\mathbf{y}_1, \mathbf{y}_2$ from $\mathbf{\hat{y}}$. Here, we denote $|\mb{x}|$ as the number of nonzero elements in $\mb{x}$. The difference between $C_{2, 1}$ and $C_{2, 2}$ is that $C_{2, 1}$ allows $|\mathbf{x}_1\cdot \mathbf{x}_2|=2$ and $|\mathbf{y}_1\cdot \mathbf{y}_2|=0$ while $C_{2, 2}$ do not (we have assumed $|\mathbf{x}_1\cdot \mathbf{x}_2|\geq |\mathbf{y}_1\cdot \mathbf{y}_2|$ without loss of generality).  Due to $C_{1, 2}$, $|\mathbf{x}_1| = |\mathbf{y}_1|$ and $|\mathbf{x}_2| = |\mathbf{y}_2|$, which ensures that the case where $|\mathbf{x}_1\cdot \mathbf{x}_2|=2$ and $|\mathbf{y}_1\cdot \mathbf{y}_2|=0$ does not occur. Therefore, $C_{1, 2}\cap C_{2, 2} =  C_{1, 2}\cap C_{2, 1}$ in the case of $m=2$.

When $m=3$,  the only possible case where $C_{2,1}$ is satisfied and $C_{2, 2}$ is violated is when there exist two $m$-bit columns $\mathbf{x}_1, \mathbf{x}_2$ from $\mathbf{\hat{x}}$ and $\mathbf{y}_1, \mathbf{y}_2$ from $\mathbf{\hat{y}}$ such that $|\mb{x}_1\cdot\mb{x}_2 - \mb{y}_1\cdot\mb{y}_2| = 2$ (we have ruled out $|\mb{x}_1\cdot\mb{x}_2 - \mb{y}_1\cdot\mb{y}_2| = 3$ for the same reason in the case of $m=2$). Without loss of generality, we assume that $\mb{x}_1\cdot\mb{x}_2>\mb{y}_1\cdot\mb{y}_2$. If $\mb{x}_1\cdot\mb{x}_2 = 3$, then $\mb{x}_1 = \mb{x}_2 = [1,1,1]^T$. Since $|\mb{x}_1| = |\mb{y}_1|, |\mb{x}_2| = |\mb{y}_2|$, $\mb{y}_1\cdot\mb{y}_2 = 3$ $> 1$. If $\mb{x}_1\cdot\mb{x}_2 = 2$, then $|\mb{y}_1|,|\mb{y}_2|\geq 2$, hence $\mb{y}_1\cdot\mb{y}_2 \neq 0$. Therefore $|\mb{x}_1\cdot\mb{x}_2 - \mb{y}_1\cdot\mb{y}_2| \neq 2$ , and we obtain $C_{1, 2}\cap C_{2, 2} =  C_{1, 2}\cap C_{2, 1}$ in the case of $m=3$.  

In conclusion, we obtain for $m=2,3$,
\begin{equation}\label{eq:momentAb3}
    \mb{M}_{\mc{E}_\text{phase}}^{(m)} = \mb{M}_{\mc{E}_{{\text{diag}}}^{(p)}}^{(m)}= \frac{1}{2^{mn}}\bigcup_{\pi \in S_m} V_n(\pi). 
\end{equation}

\subsection*{C. Proof of  Theorem 1}\label{Ap:ProofP1}

{First, we prove  Eq. (4) in the main text
as follows.} By plugging the result of \cref{Eq:Phi_2} into  Eq. (3) in the main text, we obtain 
\begin{equation}\label{Eq:idealchannelproof}
    \begin{split}
       \mc{M}(\rho)
       & =2^n\tr_{1}(\rho\otimes \id \text{ } \mb{M}_{\mc{E}_\text{phase}}^{(2)})\\
       & =2^{-n}\tr_{1}[\rho\otimes \id \text{ } (\mbb{I}_2^{\otimes n}+\mbb{S}_2^{\otimes n}-\Delta_2^{\otimes n})].
    \end{split}
\end{equation}
Then, based on the derivation using tensor diagrams in \cref{fig:channeloriginal}, we obtain $\mc{M}(\rho) = 2^{-n} (\id + \rho_f)$. Noting that $\Phi_{U,\bb}$ is an estimation of $\mc{M}(\rho)$, we find that $\widehat{\rho_f}=2^{n}\Phi_{U,\bb}-\id$ by linearity. 
\begin{figure}
    \centering
    \includegraphics[width=\linewidth]{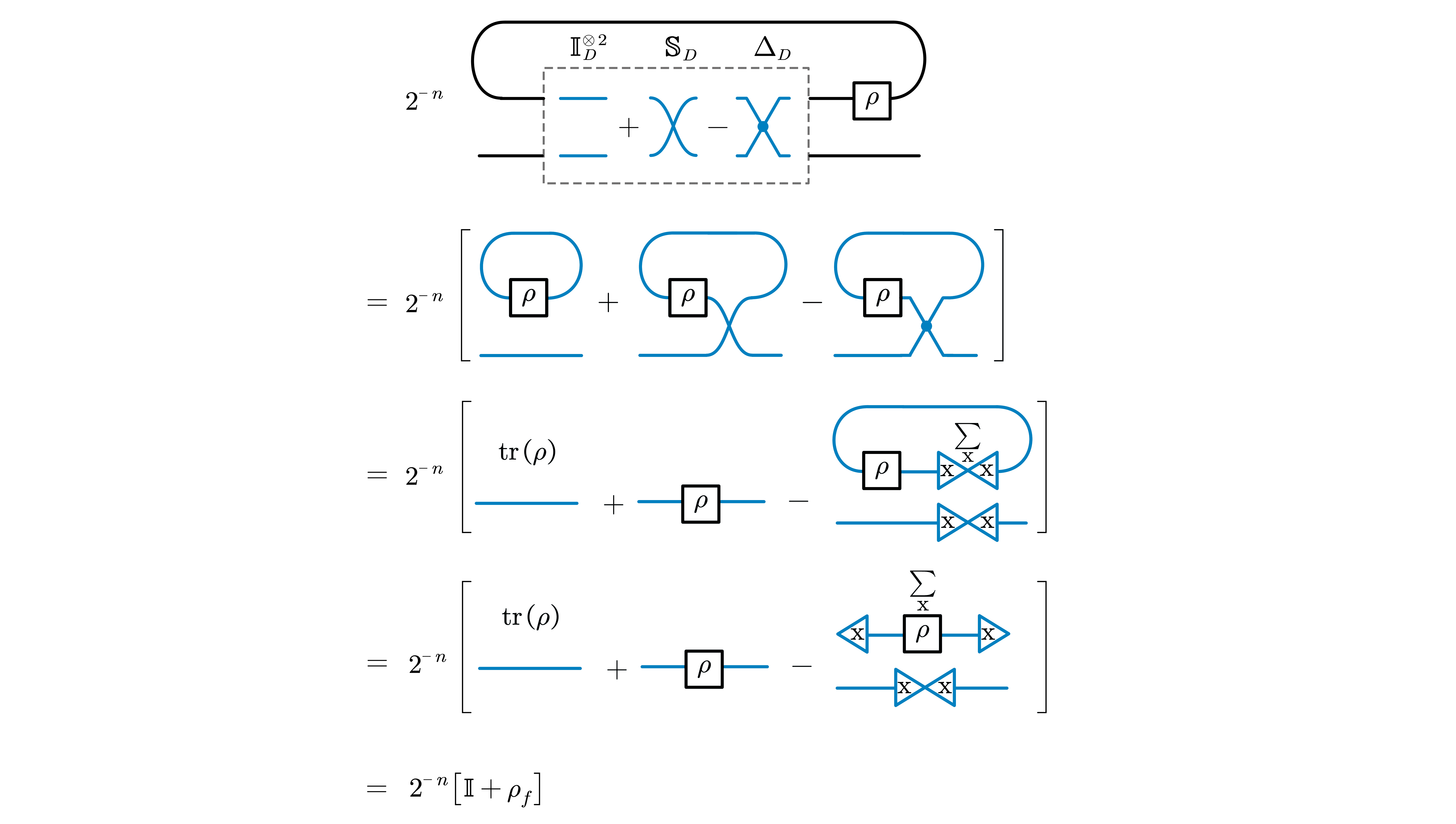}
    \caption{The tensor diagram on the proof of  Eq. (4) in the main text.}
    \label{fig:channeloriginal}
\end{figure}

{Second, we prove  Eq. (5) in the main text
as follows.}
By plugging the linear channel $\widehat{\rho_f}=2^{n}\Phi_{U,\bb}-\id$ into the estimator $\widehat{o_f} = \tr(O_f\widehat{\rho_f})$, we obtain the upper bound of the variance that
\begin{equation}\label{Eq:varfin}
\begin{split}
\text{Var}(\widehat{o_f}) \leq\mathbb{E}\ [\tr(O_f\widehat{\rho_f})]^2&=\mbb{E}_{U\sim \mc{E}_\text{phase}}\sum_{\mb{b}} \Pr(\mb{b}|U) \tr(O_f \widehat{\rho_f})^2\\
&=2^{2n} \mbb{E}_{U\sim \mc{E}_\text{phase}}\sum_{\mb{b}}\tr(\rho\otimes O_f \otimes O_f\text{ } \Phi_{U,\mb{b}}^{\otimes 3})\\
&=2^{3n} \tr(\rho\otimes O_f \otimes O_f\text{ } \mb{M}_{\mc{E}_\text{phase}}^{(3)}) = \tr(\rho\otimes O_f \otimes O_f\text{ } \bigcup_{\pi \in S_3}V_n(\pi)),
    \end{split}
    \end{equation}
where we have applied the result of the third moment function from Proposition 1.
    
Notice that $\tr[\rho\otimes O_f \otimes O_f\text{ }\ket{\mb{x,w,z}}\bra{\mb{y,s,t}}]=0$ for all $\ket{\mb{x,w,z}}\bra{\mb{y,s,t}}\in V_n(\pi_{()})\cup  V_n(\pi_{(12)})\cup V_n(\pi_{(13)})$ because the diagonal elements of $O_f$ are zero and the elements $\ket{\mb{x,w,z}}\bra{\mb{y,s,t}}$ of this set satisfy either  $\mb{w}=\mb{s}$ or $\mb{z}=\mb{t}$. As a result, we only need to consider the elements from (even only a part of ) $V_n(\pi_{(23)})\cup V_n(\pi_{(123)})\cup V_n(\pi_{(132)})$ in \cref{Eq:varfin}. As is shown in \cref{fig:m3}, the inclusion-exclusion principle transforms $V_n(\pi_{(23)})\cup V_n(\pi_{(123)})\cup V_n(\pi_{(132)})$ to an alternating sum of five Boolean tensors, as can be seen from
\begin{equation}\label{Eq:varfinal}
    \begin{split}
\text{Var}(\widehat{o_f}) &\leq \tr(\rho\otimes O_f \otimes O_f\text{ } \bigcup_{\pi \in S_3}V_n(\pi))\\
&=\tr[\rho\otimes O_f \otimes O_f\text{ } V_n(\pi_{(23)})\cup V_n(\pi_{(123)})\cup V_n(\pi_{(132)})]\\
&=\tr[\rho\otimes O_f \otimes O_f\text{ }(V_n(\pi_{(23)})+V_n(\pi_{(123)})+V_n(\pi_{(132)})-V_n(\pi_{(23)})\cap V_n(\pi_{(123)})-V_n(\pi_{(23)})\cap V_n(\pi_{(132)}))]\\
&= \tr(\rho)\tr(O_f^2)+2\tr(\rho O_f^2)-\sum_{\mb{x},\mb{z}}\bra{\mb{x}}\rho\ket{\mb{x}}\bra{\mb{z}}O_f\ket{\mb{x}}\bra{\mb{x}}O_f\ket{\mb{z}}-\sum_{\mb{x},\mb{w}}\bra{\mb{x}}\rho\ket{\mb{x}}\bra{\mb{x}}O_f\ket{\mb{w}}\bra{\mb{w}}O_f\ket{\mb{x}}\\
& \leq \tr(O_f^2)+2\tr(\rho O_f^2)\leq 3\tr(O_f^2)=3\|O_f\|_2^2,    
\end{split}
    \end{equation}
where we have used the fact that $O_f$ is Hermitian, which results in 
\begin{equation}
\begin{split}
&\sum_{\mb{x},\mb{z}}\bra{\mb{x}}\rho\ket{\mb{x}}\bra{\mb{z}}O_f\ket{\mb{x}}\bra{\mb{x}}O_f\ket{\mb{z}}\geq 0, \\
&\sum_{\mb{x},\mb{w}}\bra{\mb{x}}\rho\ket{\mb{x}}\bra{\mb{x}}O_f\ket{\mb{w}}\bra{\mb{w}}O_f\ket{\mb{x}}\geq 0,
\end{split}
\end{equation}
and 
\begin{equation}
    \tr(\rho O_f^2)\leq \|O_f^2\|_\infty\leq \tr(O_f^2)=\|O_f\|_2^2.
\end{equation}
\begin{figure}
    \centering
    \includegraphics[width=\linewidth]{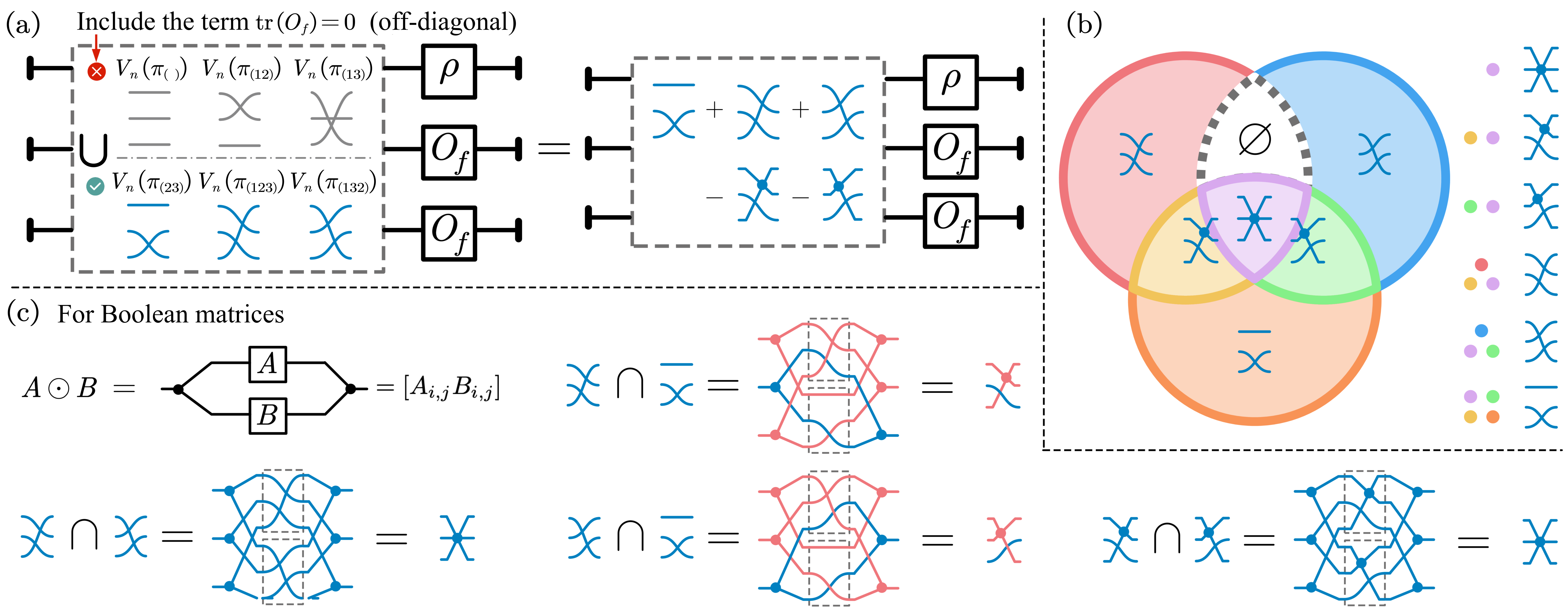}    
    \caption{The tensor diagrammatic illustration explaining   \cref{Eq:varfinal}. (a) A diagrammatic description that only three of six symmetric elements are applied in the equation, and the union of these three elements can be transformed to a summation of five operators. (b) We provide a Venn diagram to visualize this transformation. (c) The intersections of the symmetric elements are computed, where a key point is that for Boolean matrices, the intersection operation is equivalent to the Hadamard (elementwise) product.    }
    \label{fig:m3}
\end{figure}

Having proven Theorem 1, we add a comment on the total estimation variance of the phase shadow protocol. As shown in the main text, one can directly measure the computational basis to obtain $\widehat{\rho_d}=\ket{\mb{b}}\bra{\mb{b}}$. In this way, we obtain an unbiased estimator of $\tr(O\rho_d)$: $\widehat{o_d}=\tr(O\widehat{\rho_d})$. Thus, the upper bound of the estimation variance is given by
\begin{equation}
    \text{Var}(\widehat{o_d})\leq \mbb{E}_{\mb{b}} \widehat{o_d}^2 = \sum_{\mb{b}}\bra{\mb{b}}\rho\ket{\mb{b}}\bra{\mb{b}}O\ket{\mb{b}}^2\leq\max_{\mb{b}}\bra{\mb{b}}O\ket{\mb{b}}^2 = \|O_d\|_\infty^2=\|O_d^2\|_\infty\leq \|O_d\|_2^2.
\end{equation}
Suppose one repeats the phase shadow protocol for $N_f$ rounds to get the shadow data $\{\rho_f^{(i)}\}$, and computational-basis measurement for $N_d$ rounds to get $\{\rho_d^{(j)}\}$, the full estimator shows $\widehat{\rho}=N_f^{-1}\sum_i \rho_f^{(i)}+N_d^{-1}\sum_j \rho_d^{(j)}$, and $\widehat{o}=\tr(O\widehat{\rho})$ has the variance $\text{Var}(\widehat{o})= N_f^{-1}\text{Var}(\widehat{o_f})+N_d^{-1}\text{Var}(\widehat{o_d})$ under the total measurement time $N=N_f+N_d$. By choosing $N_f=3N_d$, one further has 
\begin{equation}
\begin{split}
        \text{Var}(\widehat{o})=N_f^{-1}\text{Var}(\widehat{o_f})+N_d^{-1}\text{Var}(\widehat{o_d})\leq N_f^{-1}\times 3\|O_f\|_2^2+N_d^{-1}\|O_d\|_2^2= \frac{4}{3}N^{-1}\times 3\|O_f\|_2^2+4N^{-1}\|O_d\|_2^2=4\|O\|_2^2,
\end{split}
\end{equation}
where the last equality holds because 
\begin{equation}
    \|O\|_2^2=\tr(O^2) =\tr[(O_f+O_d)^2]= \tr(O_f^2)+\tr(O_d^2)=\|O_f\|_2^2+\|O_d\|_2^2.
\end{equation}

\section*{Supplementary Note 4. --Realistic noise models}\label{Ap:rnm}

In the RPS protocol, we only need to account for the noise on controlled-${Z}$ (${CZ}$) gates, as other noises on the single qubits at the end of the circuit can be addressed through standard measurement error mitigation techniques~\cite{gluza2020quantum}.
On realistic noisy quantum computers, the noise on the ${CZ}$ gate is the \textit{biased} noise~\cite{roffe2023bias,bonilla2021xzzx}, where certain types of noise are dominant. In the RPS protocol, we model the noise on the ${CZ}$ gate as ${ZZ}$-type noise. We have also considered an extended biased noise model, which includes all two-qubit ${Z}$-type Pauli errors, and further provides a robust approach for this extended noise model in Supplementary Note 9. We present the mathematical modeling and physical origination of the noise models in the following.
Here, we model a noisy ${CZ}$ gate acting on a qubit pair $(i,j)$ as
\begin{equation}\label{Eq:noisemodelMain}
    \widetilde{{CZ}}_{i,j}\coloneq{CZ}_{i,j}{ZZ}( \theta_{i,j}), \ {ZZ}( \theta_{i,j})=\exp(-{i} \frac{\theta_{i,j}}{2} {Z}_i {Z}_j),
\end{equation}
where $\theta_{i,j} \sim \mc{N}(0,\sigma^2)$ are \emph{independently and identically distributed} (i.i.d.) random variables. This noise model is realistic, as two-qubit gates are typically implemented using two-qubit Ising couplings, where the rotation angle behaves as a random variable~\cite{maslov2018use,bravyi2022constant}. 
The average effect of this random rotation results in a noisy quantum channel
\begin{equation}\label{Eq:noimodel}
\begin{split}
        \mbb{E}_{\theta_{i,j}}  \widetilde{{CZ}}_{i,j}\rho \widetilde{{CZ}}_{i,j}^{\dagger} &={CZ} \cdot \mbb{E}_{\theta_{i,j}}[\exp(-{i} \frac{\theta_{i,j}}{2} {Z}_i {Z}_j) \rho \exp({i} \frac{\theta_{i,j}}{2} {Z}_i {Z}_j)] \cdot {CZ}^{\dag}\\
        &={CZ} \cdot \mbb{E}_{\theta_{i,j}}[(\cos\frac{\theta_{i,j}}{2}\id_2-{i}\sin\frac{\theta_{i,j}}{2}{Z}_i{Z}_j) \rho  (\cos\frac{\theta_{i,j}}{2}\id_2+{i}\sin\frac{\theta_{i,j}}{2}{Z}_i{Z}_j) ] \cdot{CZ}^{\dag}\\
        &={CZ} \cdot [\mbb{E}_{\theta_{i,j}} \cos^2\frac{\theta_{i,j}}{2}\rho+\mbb{E}_{\theta_{i,j}} \sin^2\frac{\theta_{i,j}}{2}{Z}_i{Z}_j\rho {Z}_i{Z}_j]  \cdot{CZ}^{\dag}\\
        &={CZ}\cdot \mc{D}(\rho)\cdot {CZ}^{\dag},\\
\end{split}
\end{equation}
where $\mc{D}(\rho) \coloneq \frac{1}{2}(1+e^{-\sigma^2/2})\rho +\frac{1}{2}(1-e^{-\sigma^2/2}) {Z}_i{Z}_j \rho {Z}_i {Z}_j $. In this way, we define the error rate of ${CZ}$ gates to be 
\begin{equation}
p_{e} \coloneq \frac{1}{2}(1-e^{-\sigma^2/2})\approx \frac{\sigma^2}{4}.
\end{equation}

Our focus on this type of noise stems from its physical foundations in realistic platforms like trapped ions and Rydberg atoms. Here, we provide a brief description of the noise source in these platforms.

\textit{Trapped ion platform.}
In trapped ion systems, ${CZ}$ gates can be implemented using \emph{global Mølmer–Sørensen} (GMS) gates~\cite{maslov2018use}. These gates offer native implementation of entangling operations and support efficient execution of Clifford circuits~\cite{bravyi2022constant}. \comments{A global ${CZ}$ layer may be realized as
\begin{equation}
  \prod_{0\leq i<j\leq n-1}{{CZ}_{i,j}^{A_{i,j}}},  
\end{equation}
where $A_{i,j}$ are i.i.d.\  Bernoulli random variables. However, assigning these independently is challenging in practice~\cite{van2021constructing}.} Specifically, we consider the implementation of ${CZ}$ gates based on GMS interaction
according to ${CZ}_{i,j} ={S}_i^{\dag}{S}_j^{\dag} {H}^{\otimes 2} {GMS}_{i,j}(\frac{\pi}{2}){H}^{\otimes 2}$, and the ${GMS}$ gate corresponds to an XX-type interaction of the form
%
\begin{equation}
   {GMS}_{i,j}(\theta) \coloneq {XX}_{i,j}(\theta) =\exp(-{i}\frac{\theta}{2}{X}_i {X}_j) .
\end{equation}
We discuss the noise of the GMS gate~\cite{zhang2025robust,martinez2022analytical,manovitz2017fast} as follows. Experimental study~\cite{lotshaw2023modeling} shows that the dominant gate-dependent noise source of the MS gates is the vibrational mode frequency fluctuation. These fluctuations translate into  over-rotations, which are effectively modeled by the same noisy ${CZ}$ form $ \widetilde{{CZ}} = {CZ} \times{ZZ}( \theta_{i,j})$, where $\theta_{i,j} \sim \mc{N}(0,\sigma^2)$ and are i.i.d.\  for every ${CZ}$ gate, as mentioned in \cref{Eq:noisemodelMain}.

\textit{Rydberg atom platform.}  Rydberg atom arrays are known to exhibit dominant ${Z}$-type errors and leakage~\cite{evered2023high,cong2022hardware}. Crucially, atom loss can be directly detected and converted into an erasure error~\cite{wu2022erasure}, which can in turn be converted into ${Z}$-type Pauli errors~\cite{sahay2023high}, especially in alkaline-earth Rydberg atom arrays.  For example, it is estimated that for metastable ${}^{171}{\text{Yb}}$~\cite{wu2022erasure}, $98\%$ of errors can be converted to erasure errors, resulting in a Pauli-${Z}$-type error imposing on the ${CZ}$ gates.

For two-qubit ${CZ}$ gates, the dominant Pauli-${Z}$-type errors include IZ, ZI, and ZZ errors.  Specifically, we denote the noisy two-qubit gates $\widetilde{{CZ}}= {CZ}\circ\Lambda$, where
\begin{equation}
    \Lambda(\rho) = (1-\frac{3}{4}p_e)\rho +\frac{p_e}{4} {ZI} \rho {ZI} +\frac{p_e}{4} {IZ} \rho {IZ} +\frac{p_e}{4} {ZZ} \rho {ZZ}.
\end{equation}
In Supplementary Note 9, we provide a detailed discussion of this extended noise model, including its measurement channel, the estimation variance, and the efficient post-processing procedure in phase shadow estimation.

In summary, both the trapped ion platform and the Rydberg atom platform naturally give rise to gate-dependent biased Pauli ${Z}$-type noise, validating the noise model used throughout our analysis and simulation of the RPS protocol.

\section*{Supplementary Note 5. --Proof of  Proposition 2}\label{ap:Prop-noisyChannel}

To obtain the measurement channel of the RPS protocol, it is important to consider a noisy version of the moment function. Here, we define and calculate the moment function of noisy phase circuits in Supplementary Note 5(A), and complete the proof of Proposition 2 in Supplementary Note 5(B). 

\subsection*{A. Second moment function of noisy phase circuits}\label{ap:Prop-23moment-noisy}
When subjected to noise, the randomized measurement channel becomes 
\begin{equation}\label{Eq:C1}
\mc{\widetilde{M}}(\rho)\coloneq\sum_{\mb{b}}\mbb{E}_{U\sim\mc{E}_\text{phase}} \Pr(\mb{b}|
\widetilde{U})\Phi_{U,\mb{b}}= 2^n \tr_{1}(\rho \otimes \id \text{ }\widetilde{\mb{M}}_{\mc{E}_\text{phase}}^{(2)}),
\end{equation}
where  $\widetilde{U}$ represents the noisy random circuit and the second noisy moment function $\widetilde{\mb{M}}_{\mc{E}_\text{phase}}^{(2)}$ is defined as 
\begin{equation}
    \widetilde{\mb{M}}_{\mc{E}_\text{phase}}^{(2)}  \coloneq2^{-n}\mathbb{E}_{U\sim \mc{E}_\text{phase}}   \sum_\mb{b}\widetilde{U}^{\dagger}\ket{\mb{b}}\bra{\mb{b}}\widetilde{U} \otimes U^{\dagger}\ket{\mb{b}}\bra{\mb{b}}U.
\end{equation}
The following lemma presents the second moment function of phase circuits under the noise given by \cref{Eq:noimodel}.

\begin{lemma}[Second moment functions of 
$\mc{E}_\text{phase}$]\label{thm: noisy 2-moment}
The second moment function of the phase circuit ensemble $\mc{E}_\text{phase}$ with noisy ${CZ}$ gate is 
\begin{equation}
    \widetilde{\mathbf{M}}_{\mc{E}_{\text{phase}}}^{(2)} =2^{-2n} \sum_{I_i,I_j,I_k} (1-p_{e})^{i\times k}(p_{e})^{j\times k}\Delta_2^{I_i}\otimes (\id_4-\Delta_2)^{I_j} \otimes (\mbb{S}_2-\Delta_2)^{I_k},
\label{Eq: channel_noise}
\end{equation}
where $I_i\in\mc{I}_i,I_j\in\mc{I}_j,I_k\in\mc{I}_k$ are mutually disjoint and $I_i+I_j+I_k = [n]$. 
\end{lemma}
 \begin{proof}
 
Under the noise model introduced in  \cref{Eq:noimodel}, the noisy gate reads $\widetilde{U} = {H}^{\otimes n} \widetilde{U_A}$, and the second moment function is given by
\begin{equation}
\begin{aligned}
    \widetilde{\mathbf{M}}_{\mc{E}_\text{phase}}^{(2)}  &=2^{-n}\mathbb{E}_{U\sim\mc{E}_\text{phase}}   \sum_\mb{b}\widetilde{U}^{\dagger}\ket{\mb{b}}\bra{\mb{b}}\widetilde{U} \otimes U^{\dagger}\ket{\mb{b}}\bra{\mb{b}}U\\
    &=2^{-n}\mathbb{E}_{U\sim\mc{E}_\text{phase}}   \sum_\mb{b}\widetilde{U_A}^{\dagger}{H}^{\otimes n}\ket{\mb{b}}\bra{\mb{b}}{H}^{\otimes n}\widetilde{U_A} \otimes U_A^{\dagger}{H}^{\otimes n}\ket{\mb{b}}\bra{\mb{b}}{H}^{\otimes n}U_A\\
    &=2^{-n}\mathbb{E}_{U_A}   \sum_\mb{b} \widetilde{U_A}^{\dagger}Z^{\mb{b}}{H}^{\otimes n}\ket{\mb{0}}\bra{\mb{0}}{H}^{\otimes n}Z^{\mb{b}}\widetilde{U_A} \otimes U_A^{\dagger}Z^{\mb{b}}{H}^{\otimes n}\ket{\mb{0}}\bra{\mb{0}}{H}^{\otimes n}Z^{\mb{b}}U_A\\
    &=2^{-3n}\mathbb{E}_{U_A}   \sum_\mb{b} \sum_{\mb{x,w,y,s}}\widetilde{U_A}^{\dagger}Z^{\mb{b}}\ket{\mb{x}}\bra{\mb{y}}Z^{\mb{b}}\widetilde{U_A} \otimes U_A^{\dagger}Z^{\mb{b}}\ket{\mb{w}}\bra{\mb{s}}Z^{\mb{b}}U_A\\
    &=2^{-3n}\mathbb{E}_{U_A}\sum_\mb{b} \sum_{\mb{x,w,y,s}}(\widetilde{U_A}^{\dagger}Z^{\mb{b}}\otimes U_A^{\dagger}Z^{\mb{b}})\text{ }\ket{\mb{x,w}}\bra{\mb{y,s}}\text{ }(Z^{\mb{b}}\widetilde{U_A}\otimes Z^{\mb{b}}U_A)\\
    &=2^{-2n}\mathbb{E}_{U_A} \sum_{\mb{x,w,y,s}}\widetilde{U_A}^{\dagger}\otimes U_A^{\dagger}\text{ }\Lambda_{1,1}^{(2)}(\ket{\mb{x,w}}\bra{\mb{y,s}})\text{ }\widetilde{U_A}\otimes U_A.\\
\end{aligned}
\label{eq:noise_channel}
\end{equation}

The twirling channel $\mathbb{E}_{U_A}\widetilde{U_A}^\dagger \otimes U_A^\dagger (\cdot) \widetilde{U_A}\otimes U_A$ is the composition of $\widetilde{\Lambda}_{{A,ij}}^{(2)}(\cdot)\coloneq\mathbb{E}_{A_{i,j}\in\{0,1\}} \widetilde{{CZ}}^{\dagger A_{i,j}}_{i,j}\otimes {CZ}_{i,j}^{\dagger A_{i,j}}(\cdot)\widetilde{{CZ}}_{i,j}^{A_{i,j}}\otimes {CZ}_{i,j}^{A_{i,j}}$ for all $i<j$. The twirling of $\widetilde{\Lambda}_{{A,ij}}^{(2)}$ on $\ket{\mb{x,w}}\bra{\mb{y,s}}$ is

\begin{equation}\label{Eq: channel_noise_ij}
\begin{aligned}
\widetilde{\Lambda}^{(2)}_{{A,ij}}(\ket{\mb{x,w}}\bra{\mb{y,s}})
&=\mathbb{E}_{A_{i,j}} \widetilde{{CZ}}^{\dagger A_{i,j}}_{i,j}\otimes {CZ}_{i,j}^{\dagger A_{i,j}}(\ket{\mb{x,w}}\bra{\mb{y,s}})\widetilde{{CZ}}_{i,j}^{A_{i,j}}\otimes {CZ}_{i,j}^{A_{i,j}}\\
&=\mathbb{E}_{A_{i,j}}  \widetilde{{CZ}}_{i,j}^{\dagger A_{i,j}}(\ket{\mb{x}}\bra{\mb{y}})\widetilde{{CZ}}_{i,j}^{A_{i,j}}\otimes {{CZ}}_{i,j}^{\dagger A_{i,j}}(\ket{\mb{w}}\bra{\mb{s}}){{CZ}}_{i,j}^{A_{i,j}}\\
&=\mathbb{E}_{A_{i,j}}  {{CZ}}_{i,j}^{\dagger A_{i,j}}\mc{D}(\ket{\mb{x}}\bra{\mb{y}}){{CZ}}_{i,j}^{A_{i,j}}\otimes {{CZ}}_{i,j}^{\dagger A_{i,j}}(\ket{\mb{w}}\bra{\mb{s}}){{CZ}}_{i,j}^{A_{i,j}}\\
&=\mathbb{E}_{A_{i,j}}  {{CZ}}_{i,j}^{\dagger A_{i,j}}[1-p_e+p_e(-1)^{x_i+x_j-y_i-y_j}](\ket{\mb{x}}\bra{\mb{y}}){{CZ}}_{i,j}^{A_{i,j}}\otimes {{CZ}}_{i,j}^{\dagger A_{i,j}}(\ket{\mb{w}}\bra{\mb{s}}){{CZ}}_{i,j}^{A_{i,j}}\\
&=\frac{1}{2}  \{1+(-1)^{x_i x_j +w_i w_j-y_i y_j -s_i s_j}\{1+p_{e}[(-1)^{x_i+x_j-y_i-y_j}-1]\}\}\ket{\mb{x,w}}\bra{\mb{y,s}}\\
&=\begin{cases}
    \ket{\mb{x,w}}\bra{\mb{y,s}} & \text{if } T_{i,j}^{(1)}\equiv0 \pmod{2}, T_{i,j}^{(2)}\equiv0 \pmod{2} \\
(1-p_{e})\ket{\mb{x,w}}\bra{\mb{y,s}} & \text{if } T_{i,j}^{(1)}\equiv0\pmod{2}, T_{i,j}^{(2)}\equiv1 \pmod{2}\\
0& \text{if } T_{i,j}^{(1)}\equiv1\pmod{2}, T_{i,j}^{(2)}\equiv0 \pmod{2}\\
p_{e}\ket{\mb{x,w}}\bra{\mb{y,s}}& \text{if } T_{i,j}^{(1)}\equiv1\pmod{2}, T_{i,j}^{(2)}\equiv1 \pmod{2}\\
\end{cases},\\
\end{aligned}
\end{equation}
where we have denoted $T_{i,j}^{(1)} \coloneq x_i x_j +w_i w_j-y_i y_j -s_i s_j$ and $T_{i,j}^{(2)} \coloneq x_i+x_j-y_i-y_j$. In the noiseless case ($p_e=0$), $c_{i,j}$ is binary (0 or 1), determined by $T_{i,j}^{(1)}$ alone. In contrast, $T_{i,j}^{(2)}$ modulates outcomes—some $0$ become $p_e$ and some $1$ become $1-p_e$ in the noisy case. One can observe that the result is either zero or the multiplication of the elements from $\{1-p_{e},p_{e}\}$.
With the results in \cref{Eq: channel_noise_ij}, the noisy second moment function reads
\begin{equation}
    \begin{aligned}
        \widetilde{\mathbf{M}}_{\mc{E}_\text{phase}}^{(2)}  &=2^{-2n}\mathbb{E}_{U_A} \sum_{\mb{x,w,y,s}}\widetilde{U_A}^{\dagger}\otimes U_A^{\dagger}\text{ }\Lambda_{1,1}^{(2)}(\ket{\mb{x,w}}\bra{\mb{y,s}})\text{ }\widetilde{U_A}\otimes U_A\\
        &=2^{-2n}\sum_{\mb{x,w,y,s}}  \widetilde{\Lambda}^{(2)}_{{A,01}}\circ\widetilde{\Lambda}^{(2)}_{{A,02}}\circ\dots \circ\widetilde{\Lambda}^{(2)}_{{A,n-2,n-1}}\circ\Lambda_{1,2}^{(2)}(\ket{\mb{x,w}}\bra{\mb{y,s}})\\
        &=2^{-2n}\sum_{\mb{x+w\equiv y+s} \pmod{4}} \widetilde{\Lambda}^{(2)}_{{A,01}}\circ\widetilde{\Lambda}^{(2)}_{{A,02}}\circ\dots \circ\widetilde{\Lambda}^{(2)}_{{A,n-2,n-1}}(\ket{\mb{x,w}}\bra{\mb{y,s}})\\
        &=2^{-2n} \sum_{\mb{x}+\mb{w}\equiv\mb{y}+\mb{s} \pmod{4}} \ket{\mb{x,w}}\bra{\mb{y,s}}\prod_{i<j}c_{i,j},
    \end{aligned}
    \label{Eq: channel_noise_all}
\end{equation}
where the coefficient is given by
\begin{equation}
    c_{i,j} =\begin{cases}
    1 & \text{if } T_{i,j}^{(1)}\equiv0 \pmod{2}, T_{i,j}^{(2)}\equiv0 \pmod{2} \\
(1-p_{e}) & \text{if } T_{i,j}^{(1)}\equiv0\pmod{2}, T_{i,j}^{(2)}\equiv1 \pmod{2}\\
0& \text{if } T_{i,j}^{(1)}\equiv1\pmod{2}, T_{i,j}^{(2)}\equiv0 \pmod{2}\\
p_{e}& \text{if } T_{i,j}^{(1)}\equiv1\pmod{2}, T_{i,j}^{(2)}\equiv1 \pmod{2}\\
\end{cases}. \label{eq: c_ij}
\end{equation}
For an operator $\ket{\mb{x,w}}\bra{\mb{y,s}}$, let $N^{(1)}_{\mb{x,w,y,s}}$, $N^{(2)}_{\mb{x,w,y,s}}$, and $N^{(3)}_{\mb{x,w,y,s}}$ denote the number of times $(1-p_{e})$, $p_{e}$, and $0$ appear in the set $\{c_{i,j}\mid i<j\}$, respectively. Then, we can write the noisy moment function as 
\begin{equation}
    \widetilde{\mathbf{M}}_{\mc{E}_\text{phase}}^{(2)} =2^{-2n} \sum_{\mb{x}+\mb{w}\equiv\mb{y}+\mb{s} \pmod{4}} \ket{\mb{x,w}}\bra{\mb{y,s}} (1-p_{e})^{N^{(1)}_{\mb{x,w,y,s}}} (p_{e})^{N^{(2)}_{\mb{x,w,y,s}}} 0^{N^{(3)}_{\mb{x,w,y,s}}}.
\label{Eq: channel_noise_ij_3}
\end{equation}

\begin{figure}
    \centering
    \includegraphics[width=\linewidth]{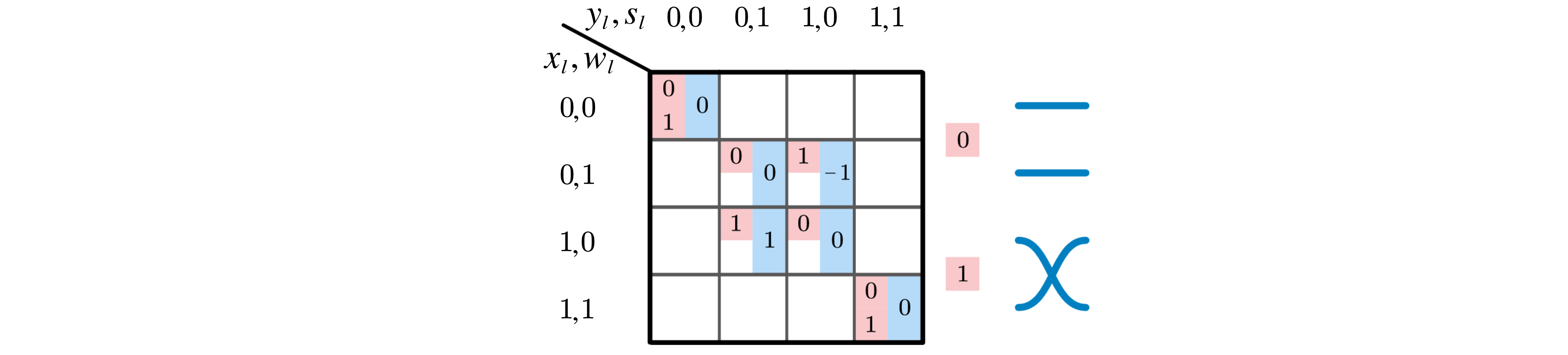}
    \caption{An illustration of six possible strings of $(x_l, w_l, y_l, s_l)$. For each string, the number marked in red indicates the permutation operator(s) to which $\ket{x_l, w_l }\bra{y_l, s_l}$ belongs. For example, $\ket{0,0}\bra{0,0} \in \id_2^{\otimes 2}, \mbb{S}_2$.  The number marked in blue represents the value of $x_l - y_l$.}
    \label{fig:p=2connection2}
\end{figure}

From now on, the key task is to determine the values of $N^{(1)}_{\mb{x,w,y,s}}$, $ N^{(2)}_{\mb{x,w,y,s}} $ and $N^{(3)}_{\mb{x,w,y,s}}$. We denote 
\begin{equation}    
\ket{\mb{x,w}}\bra{\mb{y,s}} = \bigotimes_{l=0}^{n-1} \ket{x_l , w_l }\bra{y_l , s_l}.
\end{equation}
The summation in \cref{Eq: channel_noise_ij_3} is over the matrix elements $\ket{\mb{x,w}}\bra{\mb{y,s}}$ satisfying $\mb{x}+\mb{w}\equiv\mb{y}+\mb{s} \pmod{4}$, therefore it requires that $x_l + w_l  = y_l + s_l $ for all $l$. As highlighted in red in \cref{fig:p=2connection2}, such  $\ket{x_l, w_l} \bra{y_l, s_l }$ are matrix elements of $\id_2^{\otimes 2}$ or $\mbb{S}_2$
Note that if there exists at least one permutation element $\pi$ such that both operators $\ket{x_i, w_i} \bra{y_i , s_i}$ and $\ket{x_j, w_j} \bra{y_j , s_j}$ belong to $V_1(\pi)$, then $T_{i,j}^{(1)} = 0 \ (\text{mod }\, 2)$. Otherwise, $T_{i,j}^{(1)} = 1 \ (\text{mod }\, 2)$.
 The value of $T_{i,j}^{(2)}$ depends on the differences $x_i - y_i$ and $x_j - y_j$,  and the values of $x_i - y_i$ is highlighted in blue in \cref{fig:p=2connection2}.
Integrating these discussions and considering the dependence of $c_{i,j}$ on $T^{(1)}_{i,j}$ and $T^{(2)}_{i,j}$ as in  \cref{eq: c_ij}, we can summarize that
 $c_{i,j}=1-p_{e} $ if 
\begin{itemize}
    \item $\ket{x_i, w_i} \bra{y_i , s_i} \in \Delta_2 $ and $\ket{x_j, w_j} \bra{y_j , s_j} \in \mathbb{S}_2-\Delta_2 $, or 
    \item $ \ket{x_i, w_i} \bra{y_i , s_i} \in \mathbb{S}-\Delta_2 $  and $ \ket{x_j, w_j} \bra{y_j , s_j} \in \Delta_2$.
\end{itemize}
$c_{i,j}=p_{e} $ 
if 
\begin{itemize}
    \item $\ket{x_i, w_i} \bra{y_i , s_i} \in \id_4-\Delta_2 $ and $ \ket{x_j, w_j} \bra{y_j , s_j} \in \mathbb{S}_2-\Delta_2$, or 
    \item $ \ket{x_i, w_i} \bra{y_i , s_i} \in \mathbb{S}_2-\Delta_2 $ and  $\ket{x_j, w_j} \bra{y_j , s_j}\in \id_4 - \Delta_2 $.
\end{itemize}
There is no case such that $c_{i,j}=0$, i.e., $N^{(3)}_{\mb{x,w,y,s}} = 0$. In the other cases, the coefficient $c_{i,j}=1 $. As an example, consider  $\ket{x_i, w_i} \bra{y_i , s_i} = \ket{0,1}\bra{0,1}, \quad
\ket{x_j, w_j} \bra{y_j , s_j} = \ket{0,1}\bra{1,0}$.  Since these two operators have no common $V_1(\pi)$, we have  $T_{i,j}^{(1)} = 1$.  
Furthermore,  $T_{i,j}^{(2)} = (x_i - y_i) + (x_j - y_j) = 0 - 1 = -1,$  
which leads to the corresponding coefficient  $c_{i,j} = p_{e}$. 

Here, one finds that given a qubit-level operator in the form of $\Delta_2^{\otimes i}\otimes (\id_4-\Delta_2)^{\otimes j}\otimes (\mbb{S}_2-\Delta_2)^{\otimes k}$, its constituent operators $\ket{\mb{x,w}}\bra{\mb{y,s}}$ all share the same coefficient $\prod_{i<j}c_{i,j}$. Specifically,  $ N^{(1)}_{\mb{x,w,y,s}}=i,k$ and $ N^{(2)}_{\mb{x,w,y,s}}=j,k$ for all operators $\ket{\mb{x,w}}\bra{\mb{y,s}}\in\Delta_2^{\otimes i}\otimes (\id_4-\Delta_2)^{\otimes j}\otimes (\mbb{S}_2-\Delta_2)^{\otimes k}$. In this way, we can group the operators into a qubit-level operator to further obtain that 
\begin{equation}
\begin{split}
   \widetilde{\mathbf{M}}_{\mc{E}_\text{phase}}^{(2)} &= 2^{-2n} \sum_{\mb{x}+\mb{w}\equiv\mb{y}+\mb{s} \pmod{4}} \ket{\mb{x,w}}\bra{\mb{y,s}} (1-p_{e})^{ N^{(1)}_{\mb{x,w,y,s}}} (p_{e})^{ N^{(2)}_{\mb{x,w,y,s}}} \\
   &=2^{-2n} \sum_{I_i+I_j+I_k=[n]} \Delta_2^{I_i}\otimes (\id_4-\Delta_2)^{I_j} \otimes (\mbb{S}_2-\Delta_2)^{I_k}(1-p_{e})^{|I_i|\cdot|I_k|}  (p_{e})^{{|I_j|\cdot|I_k|}}\\
   &=2^{-2n} \sum_{I_i+I_j+I_k=[n]} \Delta_2^{I_i}\otimes (\id_4-\Delta_2)^{I_j} \otimes (\mbb{S}_2-\Delta_2)^{I_k}(1-p_{e})^{i,k}  (p_{e})^{j,k},\\
\end{split}
\label{Eq: channel_noise_ij2}
\end{equation}
where $I_i\in\mc{I}_i,I_j\in\mc{I}_j,I_k\in\mc{I}_k$ are mutually disjoint and $I_i+I_j+I_k = [n]$. 
\end{proof}

\subsection*{B. Proof of  Eq. (7)}\label{Ap:ProofP2}

As shown in   \cref{Eq: channel_noise_ij2},  $\widetilde{\mathbf{M}}_{\mc{E}_{\text{phase}}}^{(2)}$ can be expressed as a  summation of tensor products of $\Delta_2$, $\id_4 - \Delta_2$, and $\mbb{S}_2 - \Delta_2$, which can be written as a linear combination of 2-copy Pauli operators as
\begin{equation}
    \Delta_2 = \frac{1}{2}(\id_2 \id_2 + {ZZ}), \quad \id_4 - \Delta_2 = \frac{1}{2}(\id_2 \id_2 - {ZZ}), \quad \mbb{S}_2 - \Delta_2 = \frac{1}{2}({XX} + {YY}).
\end{equation}
This implies that $\widetilde{\mathbf{M}}_{\mc{E}_\text{phase}}^{(2)}$ can be written as a linear combination of 2-copy Pauli operators
\begin{equation}\label{Eq: P2P}
    \widetilde{\mathbf{M}}_{\mc{E}_\text{phase}}^{(2)} =2^{-3n} \sum_{P\in\mb{P}_n} \sigma_P {P} \otimes {P}.
\end{equation}
Here, we aim to compute the detailed value of $\sigma_P$. Given a Pauli operator in the form ${P} =  \id_2^{I_{n_1}} \otimes {Z}_2^{I_{n_2}}\otimes \{{X},{Y}\}^{ I_{n_3}}$, where $I_{n_1}\in \mc{I}_{n_1}, I_{n_2}\in \mc{I}_{n_2}, I_{n_3}\in \mc{I}_{n_3}$ and $I_{n_1}+I_{n_2}+I_{n_3}= [n],$ we have
\begin{equation}\label{Eq:sigmaP}
\begin{split}
         \sigma_{P} &=2^n\tr(\widetilde{\mathbf{M}}_{\mc{E}_\text{phase}}^{(2)} \text{ }{P}\otimes {P} )\\
        &= 2^{-n}\sum_{I_{i},I_{j},I_{k}:I_{i}+I_{j}+I_{k}=[n]} (1-p_{e})^{{i}\times k}(p_{e})^{j\times k}\tr\left[\left(\Delta_2^{I_{i}}\otimes (\id_4-\Delta_2)^{I_{j}} \otimes (\mbb{S}_2-\Delta_2)^{I_{k}}\right)\left( {P}\otimes {P}\right)\right]\\
        &= 2^{-n}\sum_{I_{k}=I_{n_3},I_{i}+I_{j}=I_{n_1}+I_{n_2}} (1-p_{e})^{{i}\times k}(p_{e})^{j\times k}\tr[\left(\Delta_2^{I_{i}}\otimes (\id_4-\Delta_2)^{I_{j}} \otimes (\mbb{S}_2-\Delta_2)^{I_{k}}\right)\left( {P}\otimes {P}\right)]\\
        &=\sum_{I_{k}=I_{n_3},I_{i}+I_{j}=I_{n_1}+I_{n_2}}(1-p_{e})^{i\times n_3}(p_{e})^{j\times n_3}(-1)^{|I_j\cap I_{n_2}|}\\
        &=\sum_{s=0}^{n_1+n_2}\sum_{I_{n_1+n_2-s}, I_s: I_{n_1+n_2-s}+I_s = I_{n_1}+I_{n_2}}(1-p_{e})^{(n_1+n_2-s)n_3}(p_{e})^{sn_3}(-1)^{|I_{n_2}\cap I_s|}\\
        &=\sum_{s=0}^{n_1+n_2}\sum_{t:0\leq t\leq n_2, 0\leq s-t\leq n_1}(-1)^t C_{n_2}^t C_{n_1}^{s-t} (1-p_{e})^{(n_1+n_2-s)n_3}(p_{e})^{sn_3},
\end{split}
\end{equation}
which completes the proof. The coefficient $\sigma_P$ can be efficiently calculated using $\mc{O}(n^2)$ time. 

At this point, we add some remarks on the 
coefficient 
$\sigma_P$. We now prove that  \cref{Eq:sigmaP} approximately equals to $(1-p_{e})^{n_3(n-n_3)}\approx e^{-p_e n_3(n-n_3)
}$ when $np_e\ll1$.
We find that  \cref{Eq:sigmaP} can be 
interpreted as a power series in terms of the variable $s$, with each term referred to as the \textit{$s$-th order component}
\begin{equation}
\sigma_P = \sum_{s=0}^{n_1+n_2}
\underbrace{
\left[
\sum_{\substack{0 \leq t \leq n_2 \\ 0 \leq s - t \leq n_1}}
(-1)^t C_{n_2}^t C_{n_1}^{s-t} 
\right]
}_{\text{$s$-th order coefficient}} (1 - p_e)^{(n_1+n_2)n_3}(\frac{p_e}{1-p_e})^{sn_3}.    
\end{equation}
For example, when $s=0,1,2$, the $s$-th order terms are respectively $(1-p_{e})^{n_3(n-n_3)}$, $(n_1-n_2)(1-p_{e})^{n_3(n-n_3-1)}p_e^{n_3}$ and $\frac{(n_1-n_2)^2-(n_1+n_2)}{2}(1-p_{e})^{n_3(n-n_3-2)}p_e^{2n_3}$ (suppose $n_1,n_2,n_3\geq 2$). Here we discuss the $s$-th order coefficient, which is upper bounded as
\begin{equation}
    \sum_{\substack{0 \leq t \leq n_2 \\ 0 \leq s - t \leq n_1}}
(-1)^t C_{n_2}^t C_{n_1}^{s-t}\leq \sum_{\substack{0 \leq t \leq n_2 \\ 0 \leq s - t \leq n_1}}
 C_{n_2}^t C_{n_1}^{s-t} \leq sn^s,
\end{equation}
and the equality holds only when $s=0$. Note that when $n p_e \ll 1$, the upper bounded of the $s$-th order term is $ sn^s(p_e)^{sn_3}\leq s(np_e)^{s}$ , which decays exponentially as $s$ increases. 
Consequently, the series is dominated by the terms with small $s$, and higher-order components contribute negligibly. This implies that, for sufficiently small bit error probabilities $p_e$, the dominant contribution to $\sigma_P\approx(1-p_{e})^{n_3(n-n_3)}$ comes from the lowest-order terms in the series.
As a result, it is often sufficient to retain only the  $0$-th order term, which significantly simplifies the analysis and computation. 
If we only consider the $0$-th order terms, we can extend the result to the case when error rates of ${CZ}$ gates vary across different qubits as
\begin{equation}\label{Eq:appsigma}
    \sigma_P \approx(1-p_e)^{n_3(n-n_3)}=\prod_{s\in I_{n_1}+I_{n_2},t\in I_{n_3}} (1-p_e^{(s,t)}),
\end{equation}
 where $p_e^{(s,t)}$ denotes the error rate of ${CZ}$ implementing on the $s-$ and $t-$th qubit.

\section*{Supplementary Note 6. --Proof of  Theorem 2}\label{Ap:ProofTh1}

We do this in several steps.
\subsection*{A. Proof of  Eq. (8)}

Here, we prove  Eq. (8) by showing that $\widehat{\rho_f}_{\text{robust}}$ is indeed an unbiased estimator of $\rho_f$. Utilizing the result of Proposition 2, the proof is derived as 
    \begin{equation}
        \begin{split}
            \mathbb{E}_{\{U,\mathbf{b}\}} \widehat{\rho_f}_{\text{robust}} &= \sum_{\mb{b}}\mathbb{E}_{U\sim \mc{E}_\text{phase}} \Pr(\mb{b}|
\tilde{U})\sum_{P \in \mathbf{P}_n/\mathcal{Z}_n} \sigma_P^{-1} \, \tr(\Phi_{U,\mathbf{b}} P) \, P\\
&=2^n\sum_{P \in \mathbf{P}_n/\mathcal{Z}_n} \sigma_P^{-1} \, \tr( \widetilde{\mathbf{M}}_{\mc{E}_\text{phase}}^{(2)}\text{ }\rho\otimes P) \, P\\
&=2^{-2n}\sum_{P \in \mathbf{P}_n/\mathcal{Z}_n} \sigma_P^{-1} \, \tr( \sum_{P'\in\mb{P}_n}\sigma_{P'} P'\rho\otimes P'P) \, P\\
&=2^{-n}\sum_{P \in \mathbf{P}_n/\mathcal{Z}_n} \sigma_P^{-1} \,\sigma_{P} \tr(  P\rho) \, P=\rho_f.\\
        \end{split}
    \end{equation}

\subsection*{B. Third moment function of noisy phase circuits}\label{Ap:Variance}

Here, we introduce the third moment function of noisy phase circuits, which is necessary for proving the estimation variance given by Eq. (9).  We define the third moment function of noisy phase circuits as 
\begin{equation}\label{Eq:defM3App}
    \widetilde{\mathbf{M}}_{\mc{E}_\text{phase}}^{(3)} \coloneq \frac{1}{2^n}\mbb{E}_{U\sim\mc{E}_\text{phase}}\sum_\mb{b}{\Phi}_{\widetilde{U},\mb{b}}\otimes\Phi_{U,\mb{b}}\otimes \Phi_{U,\mb{b}}. 
\end{equation}
Substituting $ {\Phi}_{\widetilde{U},\mb{b}} = \widetilde{U_A}^{\dagger}{H}^{\otimes n}\ket{\mb{b}}\bra{\mb{b}}{H}^{\otimes n}\widetilde{U_A} $ and $ \Phi_{U,\mb{b}} =  U_A^{\dagger}{H}^{\otimes n}\ket{\mb{b}}\bra{\mb{b}}{H}^{\otimes n}U_A$ into \cref{Eq:defM3App}, we now arrive at
\begin{equation}\label{eq:noise_variance_first}
    \begin{split}
    \widetilde{\mathbf{M}}_{\mc{E}_\text{phase}}^{(3)} &= {2^{-n}}\mbb{E}_{U\sim\mc{E}_\text{phase}}\sum_\mb{b}{\Phi}_{\widetilde{U},\mb{b}}\otimes\Phi_{U,\mb{b}}\otimes \Phi_{U,\mb{b}}\\
        &=2^{-n}\mathbb{E}_{U\sim\mc{E}_\text{phase}}   \sum_\mb{b}\widetilde{U_A}^{\dagger}{H}^{\otimes n}\ket{\mb{b}}\bra{\mb{b}}{H}^{\otimes n}\widetilde{U}_A \otimes U_A^{\dagger}{H}^{\otimes n}\ket{\mb{b}}\bra{\mb{b}}{H}^{\otimes n}U_A\otimes U_A^{\dagger}{H}^{\otimes n}\ket{\mb{b}}\bra{\mb{b}}{H}^{\otimes n}U_A\\
    &=2^{-4n}\mathbb{E}_{U_A}   \sum_\mb{b} \sum_{\mb{x,w,z,y,s,t}}\widetilde{U_A}^{\dagger}Z^{\mb{b}}\ket{\mb{x}}\bra{\mb{y}}Z^{\mb{b}}\widetilde{U}_A \otimes U_A^{\dagger}Z^{\mb{b}}\ket{\mb{w}}\bra{\mb{s}}Z^{\mb{b}}U_A\otimes U_A^{\dagger}Z^{\mb{b}}\ket{\mb{z}}\bra{\mb{t}}Z^{\mb{b}}U_A\\
    &=2^{-4n}\mathbb{E}_{U_A}\sum_\mb{b} \sum_{\mb{x,w,z,y,s,t}}(\widetilde{U_A}^{\dagger}Z^{\mb{b}}\otimes U_A^{\dagger}Z^{\mb{b}}\otimes U_A^{\dagger}Z^{\mb{b}})\text{ }\ket{\mb{x,w,z}}\bra{\mb{y,s,t}}\text{ }(Z^{\mb{b}}\widetilde{U_A}\otimes Z^{\mb{b}}U_A\otimes Z^{\mb{b}}U_A)\\
    &=2^{-3n}\mathbb{E}_{U_A} \sum_{\mb{x,w,z,y,s,t}}\widetilde{U_A}^{\dagger}\otimes U_A^{\dagger}\otimes U_A^{\dagger}\text{ }\Lambda_{1,1}^{(3)}(\ket{\mb{x,w,z}}\bra{\mb{y,s,t}})\text{ }\widetilde{U_A}\otimes U_A\otimes U_A.
    \end{split}
\end{equation}
In \cref{eq:noise_variance_first}, $\mathbb{E}_{U_A} \widetilde{U_A}^{\dagger}\otimes U_A^{\dagger}\otimes U_A^{\dagger}(\cdot)\widetilde{U_A}\otimes U_A\otimes U_A$ is a concatenation of $\widetilde{\Lambda}_{{A,ij}}^{(3)}(\cdot)\coloneq\mathbb{E}_{A_{i,j}\in\{0,1\}} \widetilde{{CZ}}^{\dagger A_{i,j}}_{i,j}\otimes {CZ}_{i,j}^{\dagger A_{i,j}}\otimes {CZ}_{i,j}^{\dagger A_{i,j}}(\cdot)\widetilde{{CZ}}_{i,j}^{A_{i,j}}\otimes {CZ}_{i,j}^{A_{i,j}}\otimes {CZ}_{i,j}^{A_{i,j}}$, for all  $i< j$. Similar to the $m=2$ moment counterpart in  \cref{Eq: channel_noise_ij}, the action of $\widetilde{\Lambda}_{{A,ij}}^{(3)}$ on $\ket{\mb{x,w,z}}\bra{\mb{y,s,t}}$ is given by
 \begin{equation}
    \begin{split}
            \widetilde{\Lambda}^{(3)}_{{A,ij}}(\ket{\mb{x,w,z}}\bra{\mb{y,s,t}})&=\begin{cases}\ket{\mb{x,w,z}}\bra{\mb{y,s,t}}
     & \text{if } T_{i,j}^{(1)}\equiv0 \pmod{2}, T_{i,j}^{(2)}\equiv0 \pmod{2} \\
(1-p_{e})\ket{\mb{x,w,z}}\bra{\mb{y,s,t}} & \text{if } T_{i,j}^{(1)}\equiv0\pmod{2}, T_{i,j}^{(2)}\equiv1 \pmod{2}\\
0 & \text{if } T_{i,j}^{(1)}\equiv1 \pmod{2}, T_{i,j}^{(2)}\equiv0 \pmod{2}\\
p_{e}\ket{\mb{x,w,z}}\bra{\mb{y,s,t}}& \text{if } T_{i,j}^{(1)}\equiv1 \pmod{2}, T_{i,j}^{(2)}\equiv1 \pmod{2}\\
\end{cases},\\
    \end{split}
\end{equation}
where the coefficients are $T_{i,j}^{(1)} \coloneq x_i x_j +w_i w_j +z_i z_j-y_i y_j -s_i s_j-t_i t_j$ and $T_{i,j}^{(2)}\coloneq x_i+x_j-y_i-y_j$. Based on the above results, we have
\begin{equation}\label{eq:noise_variance_second}
    \begin{split}
        \widetilde{\mathbf{M}}_{\mc{E}_\text{phase}}^{(3)} 
    &=2^{-3n}\mathbb{E}_{U_A} \sum_{\mb{x,w,z,y,s,t}}\widetilde{U_A}^{\dagger}\otimes U_A^{\dagger}\otimes U_A^{\dagger}\text{ }\Lambda_{1,1}^{(3)}(\ket{\mb{x,w,z}}\bra{\mb{y,s,t}})\text{ }\widetilde{U_A}\otimes U_A\otimes U_A\\
    &=2^{-3n}\sum_{\mb{x,w,z,y,s,t}}  \widetilde{\Lambda}^{(3)}_{{A,01}}\circ\widetilde{\Lambda}^{(3)}_{{A,02}}\circ\dots \circ\widetilde{\Lambda}^{(3)}_{{A,n-2,n-1}}\circ\Lambda_{1,2}^{(3)}(\ket{\mb{x,w,z}}\bra{\mb{y,s,t}})\\
        &=2^{-3n}\sum_{\mb{x+w+z\equiv y+s+t} \pmod{4}} \widetilde{\Lambda}^{(3)}_{{A,01}}\circ\widetilde{\Lambda}^{(3)}_{{A,02}}\circ\dots \circ\widetilde{\Lambda}^{(3)}_{{A,n-2,n-1}}(\ket{\mb{x,w,z}}\bra{\mb{y,s,t}})\\
        &=2^{-3n} \sum_{\mb{x}+\mb{w}+\mb{z}\equiv\mb{y}+\mb{s}+\mb{t} \pmod{4}} \ket{\mb{x,w,z}}\bra{\mb{y,s,t}}\prod_{i<j}c_{i,j},\\
    \end{split}
\end{equation}
where $c_{i,j}$ is the coefficient that satisfies
\begin{equation}
    c_{i,j} = \begin{cases}1
     & \text{if } T_{i,j}^{(1)}\equiv0 \pmod{2}, T_{i,j}^{(2)}\equiv0 \pmod{2} \\
1-p_{e} & \text{if } T_{i,j}^{(1)}\equiv0\pmod{2}, T_{i,j}^{(2)}\equiv 1 \pmod{2}\\
0 & \text{if } T_{i,j}^{(1)}\equiv1 \pmod{2}, T_{i,j}^{(2)}\equiv0 \pmod{2}\\
p_{e}& \text{if } T_{i,j}^{(1)}\equiv1 \pmod{2}, T_{i,j}^{(2)}\equiv1 \pmod{2}\\
\end{cases}\\. \label{eq: cij}
\end{equation}
Hence, we can write $\widetilde{\mathbf{M}}_{\mc{E}_\text{phase}}^{(3)}$ as
\begin{equation}
    \widetilde{\mathbf{M}}_{\mc{E}_\text{phase}}^{(3)} = 2^{-3n} \sum_{\mb{x}+\mb{w}+\mb{z}=\mb{y}+\mb{s}+\mb{t} }\ket{\mb{x,w,z}}\bra{\mb{y,s,t}} (1-p_{e})^{N^{(1)}_{\mb{x,w,z,y,s,t}}} (p_{e})^{N^{(2)}_{\mb{x,w,z,y,s,t}}} 0^{N^{(3)}_{\mb{x,w,z,y,s,t}}},
\end{equation}
where ${N^{(1)}_{\mb{x,w,z,y,s,t}}}$, ${N^{(2)}_{\mb{x,w,z,y,s,t}}}$, ${N^{(3)}_{\mb{x,w,z,y,s,t}}}$ are the occurrence of $(1-p_e), p_e, 0$ in the product $\prod_{i<j}c_{i,j}$, respectively. 

Next, we determine the values of $N^{(1)}_{\mb{x,w,z,y,s,t}}, N^{(2)}_{\mb{x,w,z,y,s,t}}, N^{(3)}_{\mb{x,w,z,y,s,t}}$. We first express the full operator as a tensor product of local operators: \begin{equation}
    \ket{\mb{x,w,z}}\bra{\mb{y,s,t}} = \bigotimes_{l=0}^{n-1} \ket{x_l, w_l, z_l} \bra{y_l, s_l, t_l}.
\end{equation}
As shown in \cref{fig:p=3connection}, we identify a total of 20 operators in the form $\ket{x_l, w_l, z_l}\bra{ y_l ,s_l ,t_l}$ that satisfy the condition $x_l + w_l + z_l = y_l + s_l + t_l$. These operators $\ket{x_l, w_l, z_l}\bra{ y_l, s_l, t_l}$ belong to the operators in $\{V_1(\pi)|\pi\in S_3\}$. To maintain consistency with the labels in \cref{fig:p=3connection}, we introduce some aliases for the symmetric group elements, which are $\pi_{0} = \pi_{()}, \pi_{1} = \pi_{(23)}, \pi_{2} = \pi_{(12)}, \pi_{3} = \pi_{(13)}, \pi_{4} = \pi_{(123)}, \pi_{5} = \pi_{(132)}$.

\begin{figure}
    \centering
    \includegraphics[width=.9\linewidth]{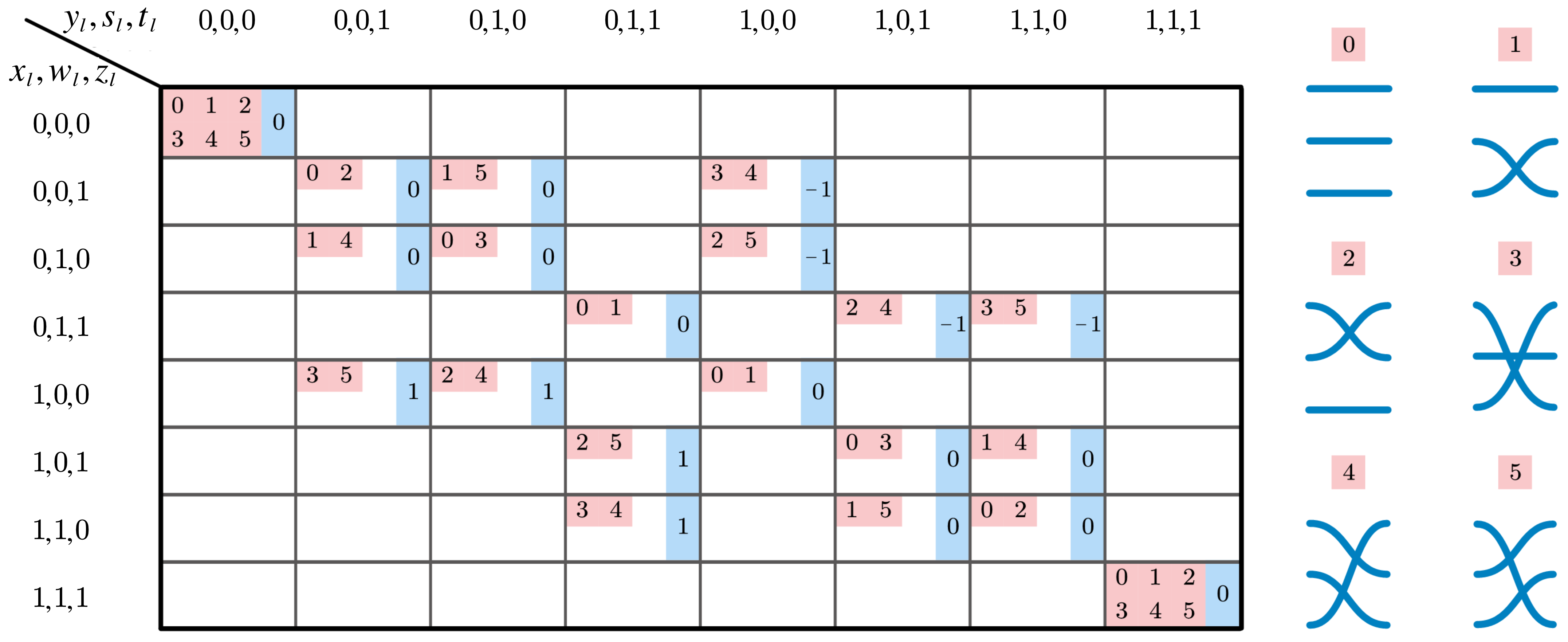}
    \caption{Illustration of the 20 possible strings of $(x_l, w_l, z_l, y_l, s_l, t_l)$. For each string, the number marked in red indicates the permutation operator(s) to which $\ket{x_l , w_l , z_l}\bra{y_l , s_l , t_l}$ belongs. For example, $\ket{0,0,1}\bra{0,0,1} \in V_1(\pi_{()})\cap V_1(\pi_{(12)})$, which are labeled as 0 and 2, respectively. The number marked in blue is the value of $x_l - y_l$.   }
    \label{fig:p=3connection}
\end{figure}
%
Based on these insights, we now observe the following:
\begin{itemize}
    \item The value of $T_{i,j}^{(1)}$ depends on whether there exist $\pi\in S_3$ such that both $\ket{x_i, w_i ,z_i} \bra{y_i, s_i ,t_i}\in V_1({\pi})$ and $\ket{x_j, w_j,z_j} \bra{y_j ,s_j, t_j}\in V_1({\pi})$. If such $\pi$ exists, then $T_{i,j}^{(1)} \equiv 0\ \pmod 2$. Otherwise, $T_{i,j}^{(1)} \equiv 1\ \pmod 2$. The symmetric operators to which $\ket{x_i ,w_i, z_i} \bra{y_i, s_i, t_i}$ belongs are indicated by the numbers marked in red in \cref{fig:p=3connection}. 
    \item The value of $T_{i,j}^{(2)}$ is determined by the difference between $x_i - y_i$ and $x_j - y_j$, which are marked in blue in \cref{fig:p=3connection}.
\end{itemize} 

With these observations, we can easily determine the values of $T_{i,j}^{(1)}$ and $T_{i,j}^{(2)}$ according to \cref{fig:p=3connection}.  For example, consider $\ket{x_i ,w_i, z_i } \bra{y_i, s_i ,t_i} = \ket{0,0,1}\bra{0,0,1}, \; \ket{x_j, w_j, z_j} \bra{y_j ,s_j, t_j} = \ket{0,1,1}\bra{1,0,1}$.  We have  $T_{i,j}^{(1)} = 0$ because both operators belongs to $\ V_1(\pi_{2})$, and $T_{i,j}^{(2)} = (x_i - y_i) + (x_j - y_j) = 0 - 1 = -1,$ which can be easily read from the figure. Then, according to \cref{eq: cij}, the corresponding coefficient $c_{i,j}$ is $1-p_e$.

Moreover, some operators share identical red and blue marks, and hence contribute equally to $c_{i,j}$. For example, both $\ket{0,0,1}\bra{1,0,0}$ and $\ket{1,1,0}\bra{0,1,1}$ belong to $V_1(\pi_{3}) \cap V_1(\pi_4)$ and have $x_l - y_l \equiv 1\ ({mod}\ 2)$. Thus, they can be grouped together. In total, the 20 local operators can be grouped into 10 categories, each corresponding to a specific pattern of red and blue labels, as summarized in \cref{Tab:1}. 
As concluded from the figure, the coefficients $c_{i,j}$ are determined as follows:  

\begin{enumerate}
    \item $c_{i,j} = p_{e}$ if either of the following conditions is satisfied:  
    \begin{itemize}
        \item $ \ket{x_i , w_i , z_i}\bra{y_i , s_i , t_i} \in B_{k,r} $ and $ \ket{x_j ,w_j ,z_j}\bra{y_j ,s_j ,t_j} \in B_{3,r} $ or $ B_{k,3} $, for all $k, r = 1, 2 $.  
        \item $ \ket{x_i , w_i , z_i}\bra{y_i , s_i , t_i} \in B_{3,3} $ and $ \ket{x_j ,w_j ,z_j}\bra{y_j ,s_j ,t_j}  \in B_{1,3} + B_{3,1} + B_{3,2} + B_{2,3} $.  
    \end{itemize}

    \item $ c_{i,j} = 1 - p_{e} $ if either of the following conditions is satisfied:  
    \begin{itemize}
        \item $\ket{x_i , w_i , z_i}\bra{y_i , s_i , t_i} \in B_{k,r} $ and $ \ket{x_j ,w_j ,z_j}\bra{y_j ,s_j ,t_j} \in B_{3, 3-r} $ or $ B_{3-k, 3} $, for all $ k, r = 1, 2 $.  
        \item $ \ket{x_i , w_i , z_i}\bra{y_i , s_i , t_i} \in B_{0,0} $ and $ \ket{x_j ,w_j ,z_j}\bra{y_j ,s_j ,t_j} \in  B_{1,3} + B_{3,1} + B_{3,2} + B_{2,3} $.  
    \end{itemize}

    \item $ c_{i,j} = 0$ if any of the following conditions holds:  
    \begin{itemize}
        \item $ \ket{x_i , w_i , z_i}\bra{y_i , s_i , t_i} \in B_{k,r} $ and $ \ket{x_j ,w_j ,z_j}\bra{y_j ,s_j ,t_j} \in B_{k, 3-r} $ or $ B_{3-k, r} $ for all $ k, r = 1, 2 $.  
        \item $ \ket{x_i , w_i , z_i}\bra{y_i , s_i , t_i} \in B_{3,1} $ and $ \ket{x_j ,w_j ,z_j}\bra{y_j ,s_j ,t_j} \in B_{3,2} $.  
        \item $ \ket{x_i , w_i , z_i}\bra{y_i , s_i , t_i} \in B_{1,3} $ and $ \ket{x_j ,w_j ,z_j}\bra{y_j ,s_j ,t_j} \in B_{2,3} $.  
    \end{itemize}
\end{enumerate}

These conclusions remain valid upon exchanging $i$ and $j$.

\begin{table}[h]
\centering
\caption{Ten categories $B_{k,r}$ of three-copy operators,  where operators of the same category have the same red and blue marks in \cref{fig:p=3connection}.}
\label{Tab:1}
\begin{tabularx}{\textwidth}{SC|SC|SC} 
\hhline{|===|}
\multicolumn{3}{c}{$B_{0,0} = \ket{0,0,0}\bra{0,0,0}+\ket{1,1,1}\bra{1,1,1} = \frac{1}{4} (\id_2 \id_2 \id_2 + \id_2 ZZ + Z \id_2 Z + ZZ \id_2)$} \\[3pt]
\hline
$\begin{aligned}
B_{1,1} &=\id_2 \id_2 X \cdot B_{0,0} \cdot \id_2 \id_2 X \\
&= \ket{0,0,1}\bra{0,0,1}+\ket{1,1,0}\bra{1,1,0}
\end{aligned}$ &
$\begin{aligned}
B_{1,2} &=\id_2 \id_2 X \cdot B_{0,0} \cdot \id_2 X \id_2 \\
&= \ket{0,0,1}\bra{0,1,0}+\ket{1,1,0}\bra{1,0,1}
\end{aligned}$ &
$\begin{aligned}
B_{1,3} &=\id_2 \id_2 X \cdot B_{0,0} \cdot X \id_2 \id_2 \\
&= \ket{0,0,1}\bra{1,0,0}+\ket{1,1,0}\bra{0,1,1}
\end{aligned}$ \\
\hline
$\begin{aligned}
B_{2,1} &=\id_2 X \id_2 \cdot B_{0,0} \cdot \id_2 \id_2 X \\
&= \ket{0,1,0}\bra{0,0,1}+\ket{1,0,1}\bra{1,1,0}
\end{aligned}$ &
$\begin{aligned}
B_{2,2} &=\id_2 X \id_2 \cdot B_{0,0} \cdot \id_2 X \id_2 \\
&= \ket{0,1,0}\bra{0,1,0}+\ket{1,0,1}\bra{1,0,1}
\end{aligned}$ &
$\begin{aligned}
B_{2,3} &=\id_2 X \id_2 \cdot B_{0,0} \cdot X \id_2 \id_2 \\
&= \ket{0,1,0}\bra{1,0,0}+\ket{1,0,1}\bra{0,1,1}
\end{aligned}$ \\
\hline
$\begin{aligned}
B_{3,1} &=X \id_2 \id_2 \cdot B_{0,0} \cdot \id_2 \id_2 X \\
&= \ket{1,0,0}\bra{0,0,1}+\ket{0,1,1}\bra{1,1,0}
\end{aligned}$ &
$\begin{aligned}
B_{3,2} &=X \id_2 \id_2 \cdot B_{0,0} \cdot \id_2 X \id_2 \\
&= \ket{1,0,0}\bra{0,1,0}+\ket{0,1,1}\bra{1,0,1}
\end{aligned}$ &
$\begin{aligned}
B_{3,3} &=X \id_2 \id_2 \cdot B_{0,0} \cdot X \id_2 \id_2 \\
&= \ket{1,0,0}\bra{1,0,0}+\ket{0,1,1}\bra{0,1,1}
\end{aligned}$ \\
\hhline{|===|}
\end{tabularx}
\end{table}

Given an operator $\ket{\mb{x,w,z}}\bra{\mb{y,s,t}}=\bigotimes_{l=0}^{n-1} \ket{x_l, w_l, z_l} \bra{y_l, s_l, t_l}$, we denote $N_I^{(k,r)}$ as the number of elements in the set $\{\ket{x_l, w_l, z_l} \bra{y_l, s_l, t_l}\}_{l=0}^{n-1}$ that belong to the subset $B_{k,r}$. Consider the $n$-qubit operator
\begin{equation}\label{eq: mbB}
    \mb{B} = \bigotimes_{k,r=1,2,3} B_{k,r}^{\otimes I^{(k,r)}} \otimes B_{0,0}^{\otimes I^{(0,0)}}, 
\end{equation}
where the index sets satisfy the condition $I^{(0,0)} + \sum_{k,r=1,2,3} I^{(k,r)} = [n]$.
All the operators $\ket{\mb{x,w,z}}\bra{\mb{y,s,t}}$ belonging to $\mb{B}$ has counts $N_I^{(k,r)} = |I^{(k,r)}|$, and the associated coefficients $(1 -  p_e)^{N^{(1)}_{\mb{x,w,z,y,s,t}}} (p_e)^{N^{(2)}_{\mb{x,w,z,y,s,t}}} 0^{N^{(3)}_{\mb{x,w,z,y,s,t}}}
$ are the same. Consequently, the values of $N^{(1)}_{\mb{x,w,z,y,s,t}}, N^{(2)}_{\mb{x,w,z,y,s,t}}, N^{(3)}_{\mb{x,w,z,y,s,t}}$ are determined by $\mb{B}$, and we denote them by $N^{(1,2,3)}_{\mb{B}}$:
\begin{equation}\label{Eq:N123}
    \begin{split}
    N^{(1)}_{\mb{B}}\coloneq  N^{(1)}_{\mb{x,w,z,y,s,t}} &= \sum_{k,r=1}^{2} N_I^{(k,r)} [N_I^{(3,3-r)} + N_I^{(3-k,3)}] + N_I^{(0,0)} [N_I^{(1,3)} + N_I^{(3,1)} + N_I^{(3,2)} + N_I^{(2,3)}], \\
    N^{(2)}_{\mb{B}}\coloneq N^{(2)}_{\mb{x,w,z,y,s,t}} &= \sum_{k,r=1}^{2} N_I^{(k,r)} [N_I^{(3,r)} + N_I^{(k,3)}] + N_I^{(3,3)} [N_I^{(1,3)} + N_I^{(3,1)} + N_I^{(3,2)} + N_I^{(2,3)}], \\
    N^{(3)}_{\mb{B}}\coloneq N^{(3)}_{\mb{x,w,z,y,s,t}} &= [N_I^{(1,1)} + N_I^{(2,2)}] [N_I^{(1,2)} + N_I^{(2,1)}] + N_I^{(3,1)} N_I^{(3,2)} + N_I^{(1,3)} N_I^{(2,3)}.
    \end{split}
\end{equation}
This leads to the third noisy moment function as
\begin{equation}\label{Eq:M3noisy}
\begin{split}
        \widetilde{\mathbf{M}}_{\mc{E}_\text{phase}}^{(3)}&=2^{-3n}\sum_{\mb{x+w+z\equiv y+s+t} \pmod{4}} \ket{\mb{x,w,z}}\bra{\mb{y,s,t}}(1-p_{e})^{N^{(1)}_{\mb{x,w,z,y,s,t}}} (p_{e})^{N^{(2)}_{\mb{x,w,z,y,s,t}}} 0^{N^{(3)}_{\mb{x,w,z,y,s,t}}}\\
        &= 2^{-3n} \sum_{\mb{B}:\substack{I^{(0,0)} + \sum_{k,r=1,2,3} I^{(k,r)} = [n]}} \mb{B} (1-p_{e})^{N^{(1)}_{\mb{B}}} (p_{e})^{N^{(2)}_{\mb{B}}} 0^{N^{(3)}_{\mb{B}}},\\
\end{split}
\end{equation}
where the second summation is over all $\mb{B}$ determined by the partition $ \substack{I^{(0,0)} + \sum_{k,r=1,2,3} I^{(k,r)} = [n]} $ according to \cref{eq: mbB}. 

\subsection*{C. Proof of  Eq. (9)}

The variance of $\widehat{o_f}_{\text{robust}} = \tr(O \widehat{\rho_f}_{\text{robust}})$ is 
    \begin{equation}
        \begin{split}
            &\text{Var}(\widehat{o_f}_{\text{robust}}) \leq\mbb{E}_{\{U,\mathbf{b}\}} [\widehat{o_f}_{\text{robust}}^2] = \mbb{E}_{\{U,\mathbf{b}\}} \tr(O\widehat{\rho_f}_{\text{robust}})^2
            =\mbb{E}_{U\sim \mc{E}_\text{phase}}\sum_{\mb{b}}\tr(\rho {\Phi}_{\widetilde{U},\mb{b}})\tr(O\widehat{\rho_f})^2\\
            &=\mbb{E}_{U\sim \mc{E}_\text{phase}}\sum_{\mb{b}}\tr[(\rho\otimes O_f\otimes O_f) \text{ }({\Phi}_{\widetilde{U},\mb{b}}\otimes \widehat{\rho_f}\otimes\widehat{\rho_f})].\\
        \end{split}
        \label{Eq:invrobustvariance0}
    \end{equation}
Recall that $\widehat{\rho_f}_{\text{robust}} = \sum_{P \in \mathbf{P}_n/\mathcal{Z}_n} \sigma_P^{-1} \, \tr(\Phi_{U,\mathbf{b}} P) \, P$ and ${\Phi}_{\widetilde{U},\mb{b}} = \frac{1}{2^n} \sum_{P\in \mathbf{P}_n} \tr({\Phi}_{\widetilde{U},\mb{b}} P)  P$, we have 
\begin{equation}
    \begin{split}
        \mbb{E}_{U\sim \mc{E}_\text{phase}}\sum_{\mb{b}}{\Phi}_{\widetilde{U},\mb{b}}\otimes \widehat{\rho_f}\otimes\widehat{\rho_f}&= \frac{1}{2^n}\mbb{E}_{U\sim \mc{E}_\text{phase}}\sum_{\mb{b}}\sum_{P_0,P_1,P_2\in{\mb{P}_n}} \sigma_{P_1}^{-1}\sigma_{P_2}^{-1}\tr({\Phi}_{\widetilde{U},\mb{b}} P_0)\tr(\Phi_{U,\mb{b}} P_1)\tr(\Phi_{U,\mb{b}} P_2)P_0 \otimes P_1 \otimes P_2\\
        &= \sum_{P_0,P_1,P_2\in{\mb{P}_n}} \sigma_{P_1}^{-1}\sigma_{P_2}^{-1}\tr( \widetilde{\mathbf{M}}_{\mc{E}_\text{phase}}^{(3)} \text{ }P_0\otimes P_1\otimes P_2)P_0\otimes P_1\otimes P_2,
    \end{split}
    \label{Eq:invrobustvariance1}
\end{equation}
where we set $\sigma_P^{-1} = 0$ for $P\in\mc{Z}_n$ for mathematical completeness. Then, by plugging  \cref{Eq:invrobustvariance1} into  \cref{Eq:invrobustvariance0}, we arrive at 
\begin{equation}
    \begin{split}
            \text{Var}(\widehat{o_f}_{\text{robust}})&=\mbb{E}_{U\sim \mc{E}_\text{phase}}\sum_{\mb{b}}[\tr(\rho\otimes O_f\otimes O_f )\text{ }({\Phi}_{\widetilde{U},\mb{b}}\otimes \widehat{\rho_f}\otimes\widehat{\rho_f})]\\
            &=\sum_{P_0,P_1,P_2\in{\mb{P}_n}} \sigma_{P_1}^{-1}\sigma_{P_2}^{-1}\tr( \widetilde{\mathbf{M}}_{\mc{E}_\text{phase}}^{(3)} \text{ }P_0\otimes P_1\otimes P_2)\tr[(\rho\otimes O_f\otimes O_f )(P_0\otimes P_1\otimes P_2)]\\
            &=\tr[2^{3n}\widetilde{\mathbf{M}}_{\mc{E}_\text{phase}}^{(3)}\text{ }(\rho\otimes \mathring{O_f}\otimes \mathring{O_f})],
        \end{split}
\end{equation}
where $\mathring{O_f} = \frac{1}{2^n} \sum_{P\in \mathbf{P}_n}\sigma_{P}^{-1}\tr(P O_f) P$. 
We now bound the variance by dividing \cref{Eq:M3noisy} into two terms based on whether $N^{(2)}_{\mathbf{x,w,z,y,s,t}} = 0$ as 
\begin{equation}    \label{Eq:MajorVariance}
    \begin{split}
        &\tr[2^{3n}\widetilde{\mathbf{M}}_{\mc{E}_\text{phase}}^{(3)}\text{ }(\rho \otimes \mathring{O_f}\otimes \mathring{O_f})] \\
        &=\mathop {\underbrace{\tr[(\rho \otimes \mathring{O_f}\otimes \mathring{O_f})\text{ }\sum_{\mb{x}+\mb{w}+\mb{z}=\mb{y}+\mb{s}+\mb{t},{N^{(2)}_{\mb{x,w,z,y,s,t}}}={N^{(3)}_{\mb{x,w,z,y,s,t}}}=0 }(1-p_{e})^{N^{(1)}_{\mb{x,w,z,y,s,t}}}\ket{\mb{x,w,z}}\bra{\mb{y,s,t}}  ]}} \limits_{\circled{1}} \\
        &+\mathop {\underbrace{\tr[(\rho \otimes \mathring{O_f}\otimes \mathring{O_f})\text{ }\sum_{\mb{x}+\mb{w}+\mb{z}=\mb{y}+\mb{s}+\mb{t},{N^{(2)}_{\mb{x,w,z,y,s,t}}}\neq0,{N^{(3)}_{\mb{x,w,z,y,s,t}}}=0 }(1-p_{e})^{N^{(1)}_{\mb{x,w,z,y,s,t}}} p_e^{N^{(2)}_{\mb{x,w,z,y,s,t}}}\ket{\mb{x,w,z}}\bra{\mb{y,s,t}}  ]}} \limits_{\circled{2}}.\\
    \end{split}
\end{equation}
We first give the upper bound of the term~$\circled{1}$. The condition $N^{(2)}_{\mathbf{x,w,z,y,s,t}} = N^{(3)}_{\mathbf{x,w,z,y,s,t}} = 0$ is equivalent to that $T_{i,j}^{(1)} = 0$ for all $i, j \in [n]$. According to Proposition 1, the summation in term~$\circled{1}$ is supported on $\bigcup_{\pi \in S_3} V_n(\pi)$.
Hence, we can rewrite term~$\circled{1}$ as
\begin{equation}    \label{Eq:MajorNoiseVar}
    \begin{split}
        \circled{1}&=\tr[(\rho \otimes \mathring{O_f}\otimes \mathring{O_f})\text{ }\sum_{\mb{x}+\mb{w}+\mb{z}=\mb{y}+\mb{s}+\mb{t},{N^{(2)}_{\mb{x,w,z,y,s,t}}}={N^{(3)}_{\mb{x,w,z,y,s,t}}}=0 }(1-p_{e})^{N^{(1)}_{\mb{x,w,z,y,s,t}}}\ket{\mb{x,w,z}}\bra{\mb{y,s,t}} ]\\
        &=\tr[(\rho \otimes \mathring{O_f}\otimes \mathring{O_f})\text{ }\sum_{\ket{\mb{x,w,z}}\bra{\mb{y,s,t}}\in\bigcup_{\pi \in S_3}V_n(\pi)}(1-p_{e})^{N^{(1)}_{\mb{x,w,z,y,s,t}}}\ket{\mb{x,w,z}}\bra{\mb{y,s,t}}  ]\\
        &=\tr[(\rho \otimes \mathring{O_f}\otimes \mathring{O_f})\text{ }\sum_{\ket{\mb{x,w,z}}\bra{\mb{y,s,t}}\in V_n(\pi_1)\cup V_n(\pi_4)\cup V_n(\pi_5)}(1-p_{e})^{N^{(1)}_{\mb{x,w,z,y,s,t}}}\ket{\mb{x,w,z}}\bra{\mb{y,s,t}}  ],\\
    \end{split}
\end{equation}
where the final equality holds because the diagonal elements of $\mathring{O}_f$ are zero, which results in the trace to be zero if $\ket{\mb{x,w,z}}\bra{\mb{y,s,t}}\in V_{n}(\pi_0)\cup V_{n}(\pi_2)\cup V_{n}(\pi_3)$. Then, similar to the derivation in \cref{Eq:varfinal}, we have

\begin{equation}
    \begin{split}
        \circled{1}&=\tr[(\rho \otimes \mathring{O_f}\otimes \mathring{O_f})\text{ }\sum_{\ket{\mb{x,w,z}}\bra{\mb{y,s,t}}\in V_n(\pi_1)\cup V_n(\pi_4)\bigcup V_n(\pi_5)}(1-p_{e})^{N^{(1)}_{\mb{x,w,z,y,s,t}}}\ket{\mb{x,w,z}}\bra{\mb{y,s,t}}  ]\\
        &=\tr[(\rho \otimes \mathring{O_f}\otimes \mathring{O_f})\text{ }(\sum_{\mb{x,w,z}}(1-p_{e})^{N^{(1)}_{\mb{x,w,z,x,z,w}}}\ket{\mb{x,w,z}}\bra{\mb{x,z,w}}+(1-p_{e})^{N^{(1)}_{\mb{x,w,z,z,x,w}}}\ket{\mb{x,w,z}}\bra{\mb{z,x,w}}+(1-p_{e})^{N^{(1)}_{\mb{x,w,z,w,z,x}}}\ket{\mb{x,w,z}}\bra{\mb{w,z,x}})]\\
        &-\tr[(\rho \otimes \mathring{O_f}\otimes \mathring{O_f})\text{ }(\sum_{\mb{x},\mb{z}}(1-p_{e})^{N^{(1)}_{\mb{x,x,z,x,z,x}}}\ket{\mb{x,x,z}}\bra{\mb{x,z,x}}+\sum_{\mb{x},\mb{w}}(1-p_{e})^{N^{(1)}_{\mb{x,w,x,x,x,w}}}\ket{\mb{x,w,x}}\bra{\mb{x,x,w}})]\\
        &=\sum_{\mb{x,w,z}}(1-p_{e})^{N^{(1)}_{\mb{x,w,z,z,x,w}}}\bra{\mb{z}}\rho\ket{\mb{x}}\bra{\mb{x}}\mathring{O_f}\ket{\mb{w}}\bra{\mb{w}}\mathring{O_f}\ket{\mb{z}}+\sum_{\mb{x,w,z}}(1-p_{e})^{N^{(1)}_{\mb{x,w,z,w,z,x}}}\bra{\mb{w}}\rho\ket{\mb{x}}\bra{\mb{z}}\mathring{O_f}\ket{\mb{w}}\bra{\mb{x}}\mathring{O_f}\ket{\mb{z}}\\
        &+\sum_{\mb{x,w,z}}(1-p_{e})^{N^{(1)}_{\mb{x,w,z,x,z,w}}}\bra{\mb{x}}\rho\ket{\mb{x}}\bra{\mb{z}}\mathring{O_f}\ket{\mb{w}}\bra{\mb{w}}\mathring{O_f}\ket{\mb{z}}\\
        &-[\sum_{\mb{x,z}}(1-p_{e})^{N^{(1)}_{\mb{x,x,z,x,z,x}}}\bra{\mb{x}}\rho\ket{\mb{x}}\bra{\mb{z}}\mathring{O_f}\ket{\mb{x}}\bra{\mb{x}}\mathring{O_f}\ket{\mb{z}}+\sum_{\mb{x,w}}(1-p_{e})^{N^{(1)}_{\mb{x,w,x,x,x,w}}}\bra{\mb{x}}\rho\ket{\mb{x}}\bra{\mb{x}}\mathring{O_f}\ket{\mb{w}}\bra{\mb{w}}\mathring{O_f}\ket{\mb{x}}].\\
    \end{split}
\end{equation}
Notice that the last two terms are non-negative, so we eliminate them by bounding them above with zero. Also, in the current setting where $ T^{(1)}_{i,j} = 0 $, the quantity $N^{(1)}_{\mathbf{x,w,z,y,s,t}} $ depends only on $ \mathbf{x} $ and $ \mathbf{y} $. This is because $N^{(1)}_{\mathbf{x,w,z,y,s,t}} $ now counts the number of index pairs $(i,j)$ such that $T_{i,j}^{(2)} = x_i + x_j - y_i - y_j \neq 0,$ which involves only $\mb{x}$ and $\mb{y}$. Therefore, we can define $\dot{\rho}$ in the computational basis from $\rho$ by setting $ \bra{\mb{x}}\dot{\rho}\ket{\mb{y}}\coloneq (1-p_{e})^{N^{(1)}_{\mb{x,w,z,y,s,t}}}\bra{\mb{x}}\rho\ket{\mb{y}}$, and obtain the final bound:
\begin{equation}\label{Eq: MajorVariance_final}
    \begin{split}
         \circled{1}&\leq\sum_{\mb{x,w,z}}(1-p_{e})^{N^{(1)}_{\mb{xwzzxw}}}\bra{\mb{z}}\rho\ket{\mb{x}}\bra{\mb{x}}\mathring{O_f}\ket{\mb{w}}\bra{\mb{w}}\mathring{O_f}\ket{\mb{z}}+\sum_{\mb{x,w,z}}(1-p_{e})^{N^{(1)}_{\mb{xwzwzx}}}\bra{\mb{w}}\rho\ket{\mb{x}}\bra{\mb{z}}\mathring{O_f}\ket{\mb{w}}\bra{\mb{x}}\mathring{O_f}\ket{\mb{z}}\\
         &+\sum_{\mb{x,w,z}}(1-p_{e})^{N^{(1)}_{\mb{xwzxzw}}}\bra{\mb{x}}\rho\ket{\mb{x}}\bra{\mb{z}}\mathring{O_f}\ket{\mb{w}}\bra{\mb{w}}\mathring{O_f}\ket{\mb{z}}\\
         &=\sum_{\mb{x,w,z}}\bra{\mb{z}}\dot{\rho}\ket{\mb{x}}\bra{\mb{x}}\mathring{O_f}\ket{\mb{w}}\bra{\mb{w}}\mathring{O_f}\ket{\mb{z}}+\sum_{\mb{x,w,z}}\bra{\mb{w}}\dot{\rho}\ket{\mb{x}}\bra{\mb{z}}\mathring{O_f}\ket{\mb{w}}\bra{\mb{x}}\mathring{O_f}\ket{\mb{z}}+\sum_{\mb{x,w,z}}\bra{\mb{x}}\rho\ket{\mb{x}}\bra{\mb{z}}\mathring{O_f}\ket{\mb{w}}\bra{\mb{w}}\mathring{O_f}\ket{\mb{z}}\\
        &= 2\tr(\dot{\rho} \mathring{O}_f^2)+\tr(O_f^2)\leq3\tr(\mathring{O}_f^2).\\
    \end{split}
\end{equation}
\begin{figure}
    \centering
    \includegraphics[width=0.78\linewidth]{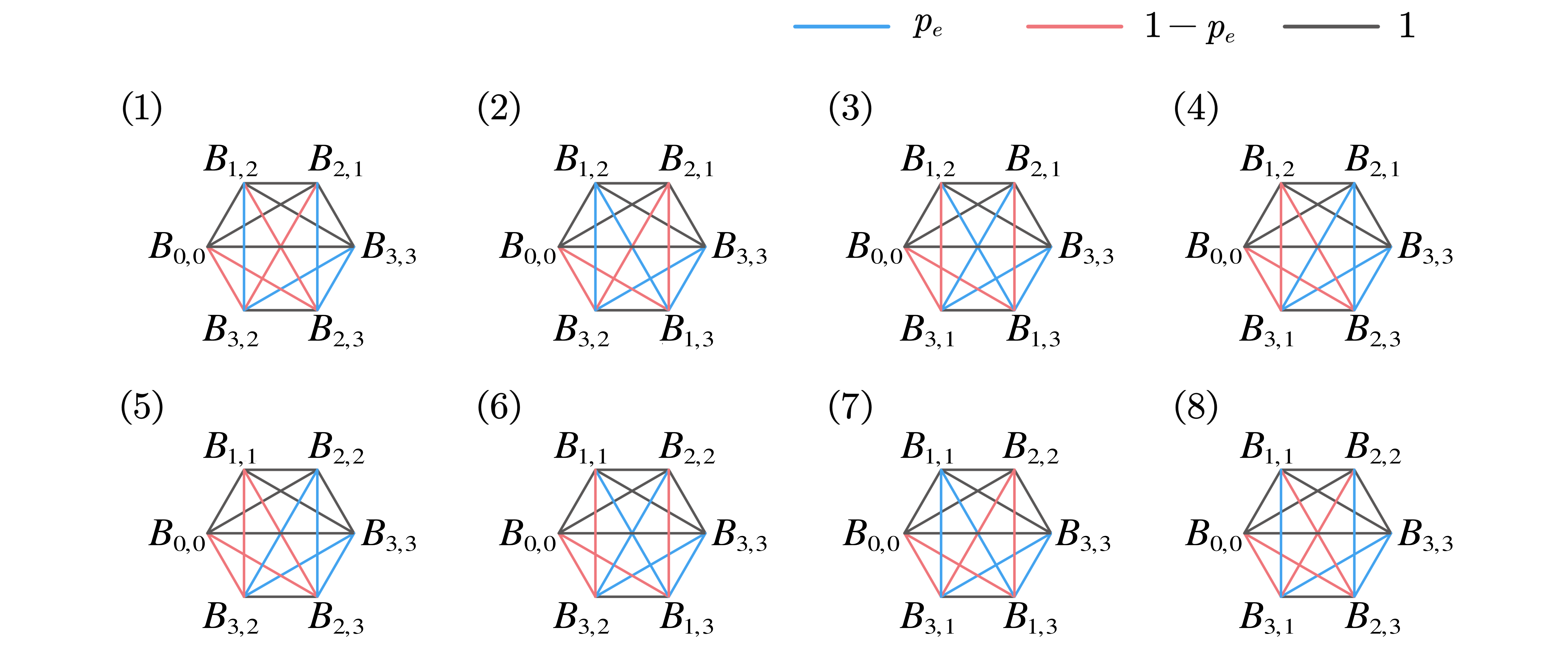}
    \caption{ Here, we illustrate all possible classes $\mc{C}_\mc{B}$ defined in  \cref{Eq:CB}.  The constraint in \cref{Eq:3or} involves three logical “or” conditions, resulting in eight distinct configurations, shown in panels (1)–(8). Each graph represents one possible case, where the nodes correspond to the subsets $B_{k,r}$, and edges are labeled by their associated coefficients $c_{i,j}$. For a valid class $\mc{C}_\mc{B}$, $\mc{B}$ corresponds to a subgraph of one of these configurations and must contain at least one blue edge.}
    \label{fig:eight_cases}
\end{figure}
Next, we provide an upper bound of term $\circled{2}$ as
 \begin{equation}\label{Eq: PhaseVariance_1}
 \begin{split}
        \circled{2}&=\tr[(\rho \otimes \mathring{O_f}\otimes \mathring{O_f})\text{ }\sum_{\mb{x}+\mb{w}+\mb{z}=\mb{y}+\mb{s}+\mb{t},{N^{(2)}_{\mb{x,w,z,y,s,t}}}\neq0,{N^{(3)}_{\mb{x,w,z,y,s,t}}}=0 }(1-p_{e})^{N^{(1)}_{\mb{x,w,z,y,s,t}}}p_e^{N^{(2)}_{\mb{x,w,z,y,s,t}}}\ket{\mb{x,w,z}}\bra{\mb{y,s,t}}  ]\\
        &=\tr[(\rho \otimes \mathring{O_f}\otimes \mathring{O_f}) \sum_{\mb{B}:\substack{I^{(0,0)} + \sum_{k,r=1,2,3} I^{(k,r)} = [n]},N^{(2)}_{\mb{B}}\neq0,N^{(3)}_{\mb{B}}=0} \mb{B} (1-p_{e})^{N^{(1)}_{\mb{B}}} (p_{e})^{N^{(2)}_{\mb{B}}} 0^{N^{(3)}_{\mb{B}}}],\\
 \end{split}
 \end{equation}
where the summation is over 
\begin{equation}\label{Eq: circ_2_set}
     \{\mb{B}|\substack{I^{(0,0)} + \sum_{k,r=1,2,3} I^{(k,r)} = [n]},N^{(2)}_{\mb{B}}\neq0,N^{(3)}_{\mb{B}}=0\},
\end{equation} with $\mb{B} $ determined by $I^{(0,0)}, I^{(k,r)}$ as \cref{eq: mbB}. 
We now analyze the condition  $N^{(2)}_{\mb{B}}\neq0,N^{(3)}_{\mb{B}}=0$. According to  \cref{Eq:N123}, $N^{(3)}_{\mb{x,w,z,y,s,t}} = 0$ if and only if the following three conditions are satisfied:
\begin{equation}\label{Eq:3or}
    \begin{split}
        N_I^{(1,1)}=N_I^{(2,2)}=0&\text{ or }N_I^{(1,2)}=N_I^{(2,1)}=0,\\
        N_I^{(3,1)}=0 &\text{ or }N_I^{(3,2)}=0,\\
        N_I^{(1,3)}=0 &\text{ or }N_I^{(2,3)}=0.\\
    \end{split}
\end{equation}
According to \cref{Eq:3or}, the condition  $N^{(3)}_{\mb{x,w,z,y,s,t}} = 0$ can be divided into 8 cases as shown in \cref{fig:eight_cases}.  For example, \cref{fig:eight_cases}~(1) corresponds to the case of $N_I^{(1,1)}=N_I^{(2,2)}=0, N_I^{(3,1)}=0, N_I^{(1,3)}=0$, and the tensor product of $\mb{B}$ contains only elements from a subset $\mc{B} \subseteq \{B_{0,0}, B_{1,2}, B_{2,1}, B_{3,3}, B_{2,3}, B_{3,2}\}$, where $\{B_{0,0}, B_{1,2}, B_{2,1}, B_{3,3}, B_{2,3}, B_{3,2}\}$ are labels of the vertices of the graph and $\mc{B}$ can be considered a subgraph. We use colored edges in the graphs to represent the coefficients introduced by a pair of $B$'s. The condition $N^{(2)}_{\mb{B}}\neq0$ means that there is at least one blue edge in the subgraphs. 

With the above analysis based on \cref{fig:eight_cases}, we can partition the set in \cref{Eq: circ_2_set} into subsets $ C_{\mc{B}} $.  Let $ \mc{B} \subseteq \{B_{0,0}\} \cup \{B_{k,r} \mid k=1,2,3,\; r=1,2,3\} $. We define the class $ C_{\mc{B}} $ as
\begin{equation}\label{Eq:CB}
    C_{\mc{B}} \coloneq \{\mb{B} = \bigotimes_{k,r=1,2,3} B_{k,r}^{\otimes I^{(k,r)}} \otimes B_{0,0}^{\otimes I^{(0,0)}}|N_I^{(k,r)} \neq 0 \text{ if }B_{k,r} \in \mc{B}, N_I^{(k,r)}=0 \text{ otherwise}\},
\end{equation}
where $N_I^{(k,r)}=|I^{(k,r)}|$ and $\mc{B}$ is a subgraph of one (or some) of the eight graphs in  \cref{fig:eight_cases} with at least one blue edge. Note that $ C_{\mc{B}} \cap C_{\mc{B}'} = \emptyset $ if $ \mc{B} \neq \mc{B}' $. As a result, we can write the term$~\circled{2}$ as

\begin{equation}\label{Eq: PhaseVariance_1p}
  \begin{split}
         \circled{2}&=\tr[(\rho \otimes \mathring{O_f}\otimes \mathring{O_f})  \sum_{\mb{B}:\substack{I^{(0,0)} + \sum_{k,r=1,2,3} I^{(k,r)} = [n]},N^{(2)}_{\mb{B}}\neq 0,N^{(3)}_{\mb{B}}=0} \mb{B} (1-p_{e})^{N^{(1)}_{\mb{B}}} (p_{e})^{N^{(2)}_{\mb{B}}} 0^{N^{(3)}_{\mb{B}}}]\\
         &=\tr[(\rho \otimes \mathring{O_f}\otimes \mathring{O_f})\text{ }\sum_{\mc{C}_{\mc{B}}}\sum_{\mb{B}\in\mc{C}_{\mb{B}} }\mb{B}(1-p_{e})^{N^{(1)}_{\mb{B}}} (p_{e})^{N^{(2)}_{\mb{B}}} 0^{N^{(3)}_{\mb{B}}}].
  \end{split}
 \end{equation}
The total number of subsets $\mc{C}_{\mc{B}}$ is constant. Therefore, our analysis will focus on bounding the values of
\begin{equation}
    \tr[(\rho \otimes \mathring{O_f}\otimes \mathring{O_f})\text{ }\sum_{\mb{B}\in\mc{C}_{\mc{B}}}\mb{B}(1-p_{e})^{N^{(1)}_{\mb{B}}} (p_{e})^{N^{(2)}_{\mb{B}}} 0^{N^{(3)}_{\mb{B}}}] = \sum_{\mb{B}\in C_{\mc{B}}}\tr[{(1-p_{e})^{N^{(1)}_{\mb{B}}} (p_{e})^{N^{(2)}_{\mb{B}}}}\mb{B}\text{ }\rho\otimes \mathring{O_f}\otimes \mathring{O_f}] \label{eq: tmptmp}
\end{equation}
for every $\mc{C}_{\mc{B}}$.
Here, we give the upper bound of $\tr[{(1-p_{e})^{N^{(1)}_{\mb{B}}} (p_{e})^{N^{(2)}_{\mb{B}}}}\mb{B}\text{ }(\rho\otimes \mathring{O_f}\otimes \mathring{O_f})] $.

An interesting property derived from \cref{Tab:1} is that any elements from $\{B_{0,0}\} \cup \{B_{k,r} \mid k,r = 1,2,3\}$ satisfy the following form of a sum of Pauli operators: 
\begin{equation}\label{Eq: Bijdef}
    B_{k,r}=\frac{1}{4}\sum_{P_0\in\mb{P}_{k,r}^{(0)},P_1\in\mb{P}_{k,r}^{(1)}} (-1)^{s_{k,r}}P_0P_1\otimes P_0\otimes P_1,
\end{equation}
where $\mb{P}_{k,r}^{(0)}$ and $\mb{P}_{k,r}^{(1)}$ are sets containing two Pauli operators and $s_{k,r}\in\{0,1\}$. In this way, $\mb{B}$ can be expressed as 
\begin{equation}
    \mb{B}=\frac{1}{4^n}\sum_{P_0\in\mb{P}_{\mb{B}}^{(0)},P_1\in\mb{P}_{\mb{B}}^{(1)}} (-1)^{\sum_{k,r}s_{k,r}|I^{(k,r)}|} P_0P_1\otimes P_0\otimes P_1,
\end{equation}
where $\mb{P}_{\mb{B}}^{(0)}$ and $\mb{P}_{\mb{B}}^{(1)}$ are subsets of $2^n$ elements from  $\mb{P}_n$.  It is therefore crucial to evaluate
\begin{equation}
    \begin{split}
    &\tr[{(1-p_{e})^{N^{(1)}_{\mb{B}}} (p_{e})^{N^{(2)}_{\mb{B}}}}\mb{B}\text{ }(\rho\otimes \mathring{O_f}\otimes \mathring{O_f})] \\
    =&4^{-n}\tr[{(1-p_{e})^{N^{(1)}_{\mb{B}}} (p_{e})^{N^{(2)}_{\mb{B}}}}\sum_{P_0\in\mb{P}_{\mb{B}}^{(0)},P_1\in\mb{P}_{\mb{B}}^{(1)}}(-1)^{\sum_{k,r}s_{k,r}|I^{(k,r)}|}  (P_0P_1\otimes P_0\otimes P_1)\text{ }(\rho\otimes \mathring{O_f}\otimes \mathring{O_f})] \\
    \leq&4^{-n}\sum_{P_0\in\mb{P}_{\mb{B}}^{(0)},P_1\in\mb{P}_{\mb{B}}^{(1)}} (1-p_{e})^{N^{(1)}_{\mb{B}}} (p_{e})^{N^{(2)}_{\mb{ B}}}|\tr(\rho P_0P_1)\tr(\mathring{O_f}P_0) \tr(\mathring{O_f}P_1)|\\
    \leq&4^{-n}(1-p_{e})^{N^{(1)}_{\mb{B}}} (p_{e})^{N^{(2)}_{\mb{B}}}\sqrt{[\sum_{P_0\in\mb{P}_{\mb{B}}^{(0)},P_1\in\mb{P}_{\mb{B}}^{(1)}}|\tr(\rho P_0P_1)|^2][\sum_{P_0\in\mb{P}_{\mb{B}}^{(0)},P_1\in\mb{P}_{\mb{B}}^{(1)}}\tr(\mathring{O_f}P_0)^2 \tr(\mathring{O_f}P_1)^2 }\\
    \leq&4^{-n}(1-p_{e})^{N^{(1)}_{\mb{B}}} (p_{e})^{N^{(2)}_{\mb{B}}}\sqrt{2^n \sum_{P\in{\mb{P}_n}}\tr(\rho P)^2}\sum_{P\in{\mb{P}_n}}\tr(P \mathring{O_f})^2\\
    =& (1-p_{e})^{N^{(1)}_{\mb{B}}} (p_{e})^{N^{(2)}_{\mb{B}}} \sqrt{\tr(\rho^2)}\tr(\mathring{O_f}^2)\\
    \leq&  (p_{e})^{N^{(2)}_{\mb{B}}}\tr( \mathring{O_f}^2).\\
    \end{split}
\label{Eq:B_bound}
\end{equation}
In the fifth line in  \cref{Eq:B_bound}, we have used $\sum_{P_0\in\mb{P}_{\mb{B}}^{(0)},P_1\in\mb{P}_{\mb{B}}^{(1)}}|\tr(\rho P_0P_1)|^2\leq \sum_{P_0\in\mb{P}_{\mb{B}}^{(0)},P_1\in\mb{P}_{n}}|\tr(\rho P_0P_1)|^2 = {2^n \sum_{P\in{\mb{P}_n}}\tr(\rho P)^2}$, because $\mb{P}_{\mb{B}}^{(0)},\mb{P}_{\mb{B}}^{(1)}$ are subset of $\mb{P}_n$ and $|\mb{P}_{\mb{B}}^{(0)}|=|\mb{P}_{\mb{B}}^{(1)}|=2^n$.

With  \cref{Eq:B_bound}, we can finally bound $\sum_{\mb{B}\in C_{\mc{B}}}\tr[{(1-p_{e})^{N^{(1)}_{\mb{B}}} (p_{e})^{N^{(2)}_{\mb{B}}}}\mb{B}\text{ }\rho\otimes \mathring{O_f}\otimes \mathring{O_f}]$ .
Here, we partition $\mc{B}$ into two mutually disjoint subsets, i.e., $\mc{B} = \mc{B}(1) + \mc{B}(2)$. The elements of $\mc{B}(2)$ correspond to vertices in the graphs of \cref{fig:eight_cases} without blue-edge connection with other elements in $\mc{B}$, while elements in $\mc{B}(1)$ correspond to vertices have at least one blue-edge connection with other elements in $\mc{B}$. We find that $\mc{B}(1)\geq 2$ due to the condition  $N^{(2)}_{\mb{B}}\neq 0$, and $|\mc{B}(1)|\leq 5$ by observing  \cref{fig:eight_cases}. 

Moreover, we can prove that $|\mc{B}(2)|\leq2$. $|\mc{B}(2)|$ reaches it maximum when $|\mc{B}(1)|$ reaches the minimum $|\mc{B}(1)|=2$. Exhausting all possible cases where $|\mc{B}(1)|=2$ from  \cref{Eq:N123}, we find that

\begin{enumerate}
    \item  suppose $\mc{B}(1) = \{B_{k,r}, B_{k,3}\}$ with $k,r \in \{1, 2\}$. \\ The configurations that maximize $ |\mc{B}(2)| $ are:
    $\mc{B}(2) = \{B_{0,0}, B_{3-k, 3-r}\}, \quad \text{or} \quad \mc{B}(2) = \{B_{0,0}, B_{3, 3-r}\}.$
    \item Suppose $\mc{B}(1) = \{B_{k,r}, B_{3,r}\}$ with $k,r \in \{1, 2\}$.\\ The configurations yielding the largest $ |\mc{B}(2)| $ are:
    $
        \mc{B}(2) = \{B_{0,0}, B_{3-k, 3-r}\}, \quad \text{or}  \quad \mc{B}(2) = \{B_{0,0}, B_{3-k, 3}\}.
    $
    \item Suppose $\mc{B}(1) = \{B_{3,3}, B_{3,r}\}$ with $r \in \{1, 2\}$.\\ The configurations that maximize $ |\mc{B}(2)| $ are:
    $
        \mc{B}(2) = \{B_{0,0}, B_{1, 3-r}\}, \quad \text{or}  \quad \mc{B}(2) = \{B_{0,0}, B_{2, 3-r}\}.
    $
    \item Suppose $\mc{B}(1) = \{B_{3,3}, B_{k,3}\}$ with $k \in \{1, 2\}$. \\The configurations with the largest $ |\mc{B}(2)| $ are:
    $
        \mc{B}(2) = \{B_{0,0}, B_{3-k, 1}\}, \quad \text{or}  \quad \mc{B}(2) = \{B_{0,0}, B_{3-k, 2}\}.
    $
\end{enumerate}
In all cases, the maximum value of $ |\mc{B}(2)| $ is 2, and if the maximum holds, then $B_{0,0}\in \mc{B}(2)$. 

Finally, we prove the upper bound of \cref{eq: tmptmp}. Given a specific operator $\mb{B}\in C_{\mc{B}}$, we define $N_{I,\mb{B}}^{(1)}\coloneq\sum_{B_{k,r}\in\mc{B}(1)}N_{\mb{B}}^{(k,r)}$ and $N_{I,\mb{B}}^{(2)}\coloneq\sum_{B_{k,r}\in\mc{B}(2)}N_{\mb{B}}^{(k,r)}$. Using the inequality $a\times b\geq a+b-1$ for any positive integers $a,b$, it follows directly from  \cref{Eq:N123} that 
\begin{equation}\label{Eq: N2}
    N^{(2)}_{\mb{B}}\geq  N_{I,\mb{B}}^{(1)}-\frac{1}{2}|\mc{B}(1)|.
\end{equation} 
Next, we use this inequality to prove that \cref{eq: tmptmp} is upper bounded as a constant of $\tr(\mathring{O_f}^2)$. We do this by considering $|\mc{B}(2)| = 0, 1, 2$ case by case.

\textbf{Case 1:} $|\mc{B}(2)|=0$.
In the case of $|\mc{B}(2)|=0$, we have $N_{I,\mb{B}}^{(1)}=n$. In this way, we can obtain that 
\begin{equation}\label{Eq:B_CB_num0}
\begin{split}
    &\sum_{\mb{B}\in C_{\mc{B}}}\tr[{(1-p_{e})^{N^{(1)}_{\mb{B}}} (p_{e})^{N^{(2)}_{\mb{B}}}}\mb{B}\text{ }\rho\otimes \mathring{O_f}\otimes \mathring{O_f}]\\
    &\leq \sum_{\mb{B}\in C_{\mc{B}}}(p_{e})^{N^{(2)}_{\mb{B}}}\tr( \mathring{O_f}^2) \text{ via  \cref{Eq:B_bound}}\\
    &\leq \sum_{\mb{B}\in C_{\mc{B}}} (p_{e})^{n-\frac{1}{2}|\mc{B}(1)|}\tr( \mathring{O_f}^2) \text{ via Eq.  \cref{Eq: N2}}\\
    &<|\mc{B}(1)|^{n} (p_{e})^{n-\frac{1}{2}|\mc{B}(1)|}\tr( \mathring{O_f}^2).\\
\end{split}
\end{equation}
Since $p_e\ll 1$ and $|\mc{B}(1)|$ is constant, this value is exponentially small and can be neglected. 

\textbf{Case 2:} $|\mc{B}(2)|=1$. For the case where $|\mc{B}(2)|=1$, we can similarly obtain
\begin{equation}\label{Eq:B_CB_num1}
\begin{split}
    &\sum_{\mb{B}\in C_{\mc{B}}}\tr[{(1-p_e)^{N^{(1)}_{\mb{B}}} (p_e)^{N^{(2)}_{\mb{B}}}}\mb{B}\text{ }\rho\otimes \mathring{O_f}\otimes \mathring{O_f}]\\
    &\leq \sum_{\mb{B}\in C_{\mc{B}}}(p_e)^{N^{(2)}_{\mb{B}}}\tr( \mathring{O_f}^2)\\
    &\leq \sum_{\mb{B}\in C_{\mc{B}}} (p_e)^{ N_{I,\mb{B}}^{(1)}-\frac{1}{2}|\mc{B}(1)|}\tr( \mathring{O_f}^2)\\
    &\leq\sum_{k=|{\mc{B}(1)}|}^{n-1}C_{n}^{k} |{\mc{B}(1)}|^{k} (p_e)^{ k-\frac{1}{2}|{\mc{B}(1)}|}\tr( \mathring{O_f}^2),\\
\end{split}
\end{equation}
where the last inequality holds because, given a class $\mc{C}_\mc{B}$, the number of $\mb{B}$ with $N_{I,\mb{B}}^{(1)}=k$ is less than $C_{n}^{k} |{\mc{B}(1)}|^{k}$, and the number $k$ ranges from $|\mc{B}(1)|$ ($N_I^{(k,r)}=1$ for all $B_{k,r}\in \mc{B}(1)$) to $n-1$ (one for the sole element in $\mc{B}(2)$). In this context, we assume that $ n^2p_e$ is a constant number, while $np_e\ll 1$. Hence, we bound the scale of  \cref{Eq:B_CB_num1} as 
\begin{equation}
    \sum_{k=|{\mc{B}(1)}|}^{n-1}C_{n}^{k} |{\mc{B}(1)}|^{k} (p_e)^{ k-\frac{1}{2}|{\mc{B}(1)}|}< \sum_{k=|{\mc{B}(1)}|}^{n-1}{n}^{k} |{\mc{B}(1)}|^{k} (p_e)^{ k-\frac{1}{2}|{\mc{B}(1)}|}<[n|{\mc{B}(1)}|\sqrt{p_e}]^{|{\mc{B}(1)}|} \frac{1}{1-n|{\mc{B}(1)}|p_e}\sim (n\sqrt{p_e})^{|{\mc{B}(1)}|},
\end{equation}
which can also be considered a constant.

\textbf{Case 3:} $|\mc{B}(2)|=2$. In this case, we reorganize the sum over $\mb{B}\in C_{\mc{B}}$ by collecting all valid labelings $\{I^{(k,r)}\}$ such that $\bigcup_{B_{k,r}\in \mc{B}} I^{(k,r)} = [n]$. Since there are exactly two elements in $\mc{B}(2)$, we define $ I^{(2)}_{\mb{B}} \coloneq \bigcup_{B_{k,r}\in \mc{B}(2)} I^{(k,r)} $, and write

\begin{equation}\label{Eq:B_CB_num2_1}
\begin{split}
    &\sum_{\mb{B}\in C_{\mc{B}}}\tr[{(1-p_{e})^{N^{(1)}_{\mb{B}}} (p_{e})^{N^{(2)}_{\mb{B}}}}\mb{B}\text{ }(\rho\otimes \mathring{O_f}\otimes \mathring{O_f})]\\
    =&\sum_{\mb{B}:\sum_{B_{k,r}\in \mc{B}}I^{(k,r)}=[n],|I^{(k,r)}|>0}\tr[{(1-p_{e})^{N^{(1)}_{\mb{B}}} (p_{e})^{N^{(2)}_{\mb{B}}}}\bigotimes_{B_{k,r}\in \mc{B}} B_{k,r}^{\otimes I^{(k,r)}}\text{ }(\rho\otimes \mathring{O_f}\otimes \mathring{O_f})]\\
     =&\sum_{I^{(2)}_{\mb{B}}+\sum_{B_{k,r}\in \mc{B}(1)}I^{(k,r)}=[n],|I^{(k,r)}|>0}\tr[{ (1-p_{e})^{N^{(1)}_{\mb{B}}}(p_{e})^{N^{(2)}_{\mb{B}}}}\bigotimes_{B_{k,r}\in \mc{B}(1)} B_{k,r}^{\otimes I^{(k,r)}}\otimes (\sum_{B_{k,r}\in \mc{B}(2)} B_{k,r})^{\otimes I_{\mb{B}}^{(2)}}\text{ }(\rho\otimes \mathring{O_f}\otimes \mathring{O_f})].\\
\end{split}
\end{equation} 
Since we already know that $B_{0,0}\in \mc{B}(2)$ if $|\mc{B}(2)|=2$, we now consider what the other element is. If the other element of $\mc{B}(2)$ is $B_{k,k}$ for $k=1,2,3$, then from  \cref{Tab:1}, we can see that $B_{0,0}+B_{k,k}=\frac{1}{2}(\mbb{I}_2\mbb{I}_2\mbb{I}_2+Z^{\delta_{k\neq3}}Z^{\delta_{k\neq2}}Z^{\delta_{k\neq1}})$.

If $k=3$, i.e., $\mc{B}(2)=\{B_{0,0},B_{3,3}\}$, then as observed from \cref{fig:eight_cases}, there exists no $\mc{B}(1)$ such that $|\mc{B}(1)|\geq 2$. 
If $k=1$ or $2$, then $\mc{B}(1)\subseteq\{B_{3-k,3},B_{3,3-k},B_{3,3},B_{3-k,3-k}\}$, and we find that: the second Pauli operator in the tensor of $B\in \mc{B}(1)$ is within $\{\id_2, Z\}$ if $k=2$, and the third Pauli operator in the tensor of $B\in \mc{B}(1)$ is within $\{\id_2, Z\}$ if $k=1$. 
In other words, when we decompose  $\bigotimes_{B_{k,r}\in \mc{B}(1)} B_{k,r}^{\otimes I^{(k,r)}}\otimes (\sum_{B_{k,r}\in \mc{B}_2} B_{k,r})^{\otimes I_{\mb{B}}^{(2)}}$ as a sum of 3-copy Pauli operators, one of its last two copies consists solely of Pauli elements from $\mc{Z}_n$. Hence, for the off-diagonal operator $\mathring{O_f}$, the  trace 
\begin{equation}
        \tr[\bigotimes_{B_{k,r}\in \mc{B}(1)} B_{k,r}^{\otimes I^{(k,r)}}\otimes (\sum_{B_{k,r}\in \mc{B}(2)} B_{k,r})^{\otimes I_{\mb{B}}^{(2)}}\text{ }(\rho\otimes \mathring{O_f}\otimes \mathring{O_f})]=0
\end{equation}
vanishes. Therefore, only the mixed case $\mc{B}(2)=\{B_{0,0},B_{k,r}\}$ with $k\neq r$ remains.  According to \cref{Tab:1}, such a sum admits a decomposition:
\begin{equation}
    B_{0,0}+B_{k,r} = \frac{1}{4}\sum_{P_0\in\{\mbb{I}_2,Z\},P_1\in\{\mbb{I}_2,X,Y,Z\}} (-1)^{s_{k,r}} V_1(\pi) (P_0P_1\otimes P_0\otimes P_1),
\end{equation} 
for some $\pi\in S_3$. Similarly, an element $B_{k',r'}\in\mc{B}(1)$ can also be written as
\begin{equation}
    B_{k',r'}=\frac{1}{4}\sum_{P_0\in\mb{P}_{k',r'}^{(0)},P_1\in\mb{P}_{k',r'}^{(1)}} (-1)^{s_{k',r'}} V_1(\pi) (P_0P_1\otimes P_0\otimes P_1).
\end{equation}
Based on this, we calculate the $n$-qubit operator to find
\begin{equation}
   \bigotimes_{B_{k,r}\in \mc{B}(1)} B_{k,r}^{\otimes I^{(k,r)}}\otimes (\sum_{B_{k,r}\in \mc{B}_2} B_{k,r})^{\otimes I_{\mb{B}}^{(2)}} = \frac{1}{4^n}\sum_{P_0\in \mb{P}_{\mb{B}}^{(0)},P_1\in \mb{P}_{\mb{B}}^{(1)}}(-1)^{s_{k,r}|I_{\mb{B}}^{(2)}|+\sum_{B_{k',r'}\in\mc{B}(1)}s_{k',r'}|I^{(k',r')}|}V_n(\pi)  P_0P_1\otimes P_0\otimes P_1,
\end{equation}
where $|P_0|=2^n$, and $|P_1|=2^{n+|I_{\mb{B}}^{(2)}|}\leq 4^n$. Then, we denote $(O_a,O_b,O_c)=\pi (\rho,\mathring{O_f},\mathring{O_f})$, and
\begin{equation}
\begin{split}\label{Eq: B2_num2_single}
    &\tr[\bigotimes_{B_{k,r}\in \mc{B}(1)} B_{k,r}^{\otimes I^{(k,r)}}\otimes (\sum_{B_{k,r}\in \mc{B}_2} B_{k,r})^{\otimes I_{\mb{B}}^{(2)}}\text{ }(\rho\otimes \mathring{O_f}\otimes \mathring{O_f})]\\
    &=4^{-n}\sum_{P_0\in \mb{P}_{\mb{B}}^{(0)},P_1\in \mb{P}_{\mb{B}}^{(1)}}(-1)^{s_{k,r}|I_{\mb{B}}^{(2)}|+\sum_{B_{k',r'}\in\mc{B}(1)}s_{k',r'}|I^{(k',r')}|} \tr(P_0P_1 O_{a}) \tr(P_0 O_{b})\tr(P_1 O_{c}) \\
    &\leq 4^{-n}\sqrt{\sum_{P_0\in \mb{P}_{\mb{B}}^{(0)},P_1\in \mb{P}_{\mb{B}}^{(1)}}|\tr(P_0P_1 O_{a})|^2}\sqrt{\sum_{P_0\in \mb{P}_{\mb{B}}^{(0)},P_1\in \mb{P}_{\mb{B}}^{(1)}}\tr(P_0 O_{b})^2\tr(P_1 O_{c})^2}\\
    &\leq 2^{-n}\sqrt{\sum_{P_0\in \mb{P}_{\mb{B}}^{(0)},P_1\in \mb{P}_{\mb{B}}^{(1)}}|\tr(P_0P_1 O_{a})|^2}\sqrt{[2^{-n}\sum_{P_0\in \mb{P}_{\mb{B}}^{(0)}} \tr(P_0 O_{b})^2 ][2^{-n}\sum_{\mb{P}_1\in P_{\mb{B}}^{(1)}} \tr(P_1 O_{c})^2 ]} \\
    &\leq 2^{-n}\sqrt{\sum_{P_0\in \mb{P}_{\mb{B}}^{(0)},P_1\in \mb{P}_{n}}|\tr(P_0P_1 O_{a})|^2}\sqrt{[2^{-n}\sum_{P_0\in \mb{P}_{\mb{B}}^{(0)}} \tr(P_0 O_{b})^2 ][2^{-n}\sum_{\mb{P}_1\in P_{\mb{B}}^{(1)}} \tr(P_1 O_{c})^2]} \\
    &= 2^{-n}\sqrt{2^n\sum_{P\in \mb{P}_{n}}|\tr(P O_{a})|^2}\sqrt{[2^{-n}\sum_{P_0\in \mb{P}_{\mb{B}}^{(0)}} \tr(P_0 O_{b})^2 ][2^{-n}\sum_{\mb{P}_1\in P_{\mb{B}}^{(1)}} \tr(P_1 O_{c})^2 ]} \\
    &=\sqrt{\tr(O_{a}^2)\tr(O_{b}^2)\tr(O_{c}^2)}=\sqrt{\tr(\rho^2)\tr(\mathring{O_f}^2)\tr(\mathring{O_f}^2)}\leq \tr(\mathring{O_f}^2).
\end{split}
\end{equation}
In this way, we plug  \cref{Eq: B2_num2_single} into  \cref{Eq:B_CB_num2_1} to obtain
\begin{equation}\label{Eq:B_CB_num2}
\begin{split}
    &\sum_{\mb{B}\in C_{\mc{B}}}\tr[{(1-p_{e})^{N^{(1)}_{\mb{B}}} (p_{e})^{N^{(2)}_{\mb{B}}}}\mb{B}\text{ }(\rho\otimes \mathring{O_f}\otimes \mathring{O_f})]\\
    &=\sum_{I^{(2)}_{\mb{B}}+\sum_{B_{k,r}\in \mc{B}(1)}I^{(k,r)}=[n],|I^{(k,r)}|>0}\tr[{ (1-p_{e})^{N^{(1)}_{\mb{B}}}(p_{e})^{N^{(2)}_{\mb{B}}}}\bigotimes_{B_{k,r}\in \mc{B}(1)} B_{k,r}^{\otimes I^{(k,r)}}\otimes (\sum_{B_{k,r}\in \mc{B}(2)} B_{k,r})^{\otimes I_{\mb{B}}^{(2)}}\text{ }(\rho\otimes \mathring{O_f}\otimes \mathring{O_f})]\\
    &\leq \sum_{I^{(2)}_{\mb{B}}+\sum_{B_{k,r}\in \mc{B}(1)}I^{(k,r)}=[n],|I^{(k,r)}|>0}(p_{e})^{N^{(2)}_{\mb{B}}}\tr( \mathring{O_f}^2)\\
    &\leq \sum_{I^{(2)}_{\mb{B}}+\sum_{B_{k,r}\in \mc{B}(1)}I^{(k,r)}=[n],|I^{(k,r)}|>0} (p_e)^{ N_{I,\mb{B}}^{(1)}-\frac{1}{2}|\mc{B}(1)|}\tr( \mathring{O_f}^2)\\
    &\leq\sum_{k=|{\mc{B}(1)}|}^{n-1}C_{n}^{k} |{\mc{B}(1)}|^{k} (p_e)^{ k-\frac{1}{2}|{\mc{B}(1)}|}\tr( \mathring{O_f}^2) = \Theta(1)\tr( \mathring{O_f}^2).\\
\end{split}
\end{equation}
Since there exist only finite numbers of different subsets $C_{\mc{B}}$, and the number is independent of the qubit number, the total contribution to the variance of the term $\circled{2}$ is bounded by a constant multiple of $\tr(\mathring{O_f}^2)$. Combine the upper bound of both terms $\circled{1}$ and $\circled{2}$, we finally obtain
\begin{equation}\label{eq:VarTh}  
\text{Var}_{\mc{E}_{{\text{phase}}},\text{robust}}\left(\widehat{O_f}\right)\lesssim [3+\Theta(1)]\tr( \mathring{O_f}^2)\approx[3+\Theta(1)]e^{n^2 p_e/2}\tr(O_f^2). 
\end{equation}
We employ the asymptotic notation `$\lesssim$' under the conditions $np_e \ll 1$ and $n^2p_e = \Theta(1)$, as discussed in the main text. Under the weak noise regime where $np_e \ll 1$, the coefficient $\sigma_P$ admits the approximation
\begin{equation}
    \sigma_P^{-1} \approx (1-p_e)^{-n_3(n-n_3)} \approx e^{n_3(n-n_3)p_e} \leq e^{\frac{n^2p_e}{4}},
\end{equation}
as derived in  \cref{Eq:appsigma}. This approximation enables quantitative characterization of the estimation variance, with numerical simulations confirming its validity for the parameter regime of interest in Supplementary Note 8.

\section*{Supplementary Note 7. --Efficient post-processing of the robust phase shadow}\label{Ap:EC}
In shadow estimation, the post-processing scheme can be efficiently implemented on classical computers. In this section, we find that although the reversed measurement map is now a combination of exponential terms, we can still efficiently compute the value of the estimator $\widehat{o_f}=\tr(\widehat{\rho_f} O)$, where the observable $O$ is a Pauli operator or a stabilizer state.

\textit{Pauli operators.} Suppose that $Q$ denotes an $n$-qubit non-${Z}$-type Pauli operator $Q \in \mathbf{P}_n/\mathcal{Z}_n$, the estimator of $Q$ with the RPS protocol is
\begin{equation}
\begin{split}
\tr(\widehat{\rho_f} Q) &= \sum_{P\in\mathbf{P}_n/\mathcal{Z}_n}\sigma_P^{-1}\tr( \Phi_{U,\mb{b}}P) \tr(PQ)\\
&=D \sigma_Q^{-1}\bra{\mb{b}}UQU^{\dagger}\ket{\mb{b}},
         \label{Eq: InvNoisychannelonPauli}
\end{split}
\end{equation}
where we should choose $P=Q$ in the summation. Note that it takes $\mc{O}(n^2)$ time to compute the value of $\sigma_Q$ according to Proposition 2. In addition, calculating the inner product of stabilizer states $\bra{\mb{b}}UQU^{\dagger}\ket{\mb{b}}$ requires $\mc{O}(n^3)$ time using the Tableau formalism ~\cite{aaronson2004improved}. Therefore, we summarize that the overall post-processing time for robust phase shadows with Pauli observables is $\mc{O}(n^3)$. This efficient approach can naturally extend to polynomial terms of the Pauli operators in $O$.

\textit{Stabilizer states.} On the other hand, suppose that $\Psi$ is a stabilizer state $\Psi = V^{\dagger}|\mb{0}\rangle \langle\mb{0}|V$ with $V$ being a Clifford unitary, we have
\begin{equation}
\begin{split}
\tr(\widehat{\rho_f} \Psi) &= \sum_{P\in\mathbf{P}_n/\mathcal{Z}_n}\sigma_P^{-1}\tr( \Phi_{U,\mb{b}}P) \tr(P\Psi)\\
&=\sum_{P\in\mathbf{P}_n/\mathcal{Z}_n}\sigma_P^{-1}\tr[ (\Phi_{U,\mb{b}}\otimes V^{\dagger}|\mb{0}\rangle \langle\mb{0}|V )\text{ }(P\otimes P)].
         \label{Eq: InvNoisychannelPauli}
\end{split}
\end{equation}
Since $|\mb{0}\rangle \langle\mb{0}| = D^{-1}\sum_{\mb{a}}Z^{\mb{a}}$, we have $\Psi =  D^{-1}\sum_{\mb{a}}V^{\dag}Z^{\mb{a}}V=D^{-1}\sum_{S_V\in\mb{S}_V}S_V$, where $S_V$ are Pauli operators that stabilize $\Psi$, and $\mb{S}_V$ denotes the stabilizer states. In addition, suppose that $\Phi_{U,\mb{b}} = U^{\dagger}\ket{\mb{b}}\bra{\mb{b}}U=U^{\dagger}X^{\mb{b}}\ket{\mb{0}}\bra{\mb{0}}X^{\mb{b}}U=D^{-1}\sum_{S_{{X}^{\mb{b}}U}\in \mb{S}_{{X}^{\mb{b}}U}}S_{{X}^{\mb{b}}U}$, we arrive at
    \begin{equation}
    \begin{split}
       \tr(\widehat{\rho_f} \Psi) &= D^{-2}\sum_{P\in\mathbf{P}_n/\mathcal{Z}_n}\sigma_P^{-1}\sum_{\mb{a,a'}\in\{0,1\}^n}\tr[( U^{\dagger}X^{\mb{b}}Z^{\mb{a}}X^{\mb{b}}U\otimes V^{\dagger}Z^{\mb{a'}}V) \text{ }(P\otimes P)]\\
       &=D^{-2}\sum_{P\in\mathbf{P}_n/\mathcal{Z}_n}\sigma_P^{-1}\sum_{\mb{a,a'}\in\{0,1\}^n}\tr( U^{\dagger}X^{\mb{b}}Z^{\mb{a}}X^{\mb{b}}U P)\tr( V^{\dagger}Z^{\mb{a'}}V P)\\
       &=D^{-2}\sum_{P\in\mathbf{P}_n/\mathcal{Z}_n}\sigma_P^{-1}\sum_{S_{{X}^{\mb{b}}U}\in \mb{S}_{{X}^{\mb{b}}U},S_V\in \mb{S}_V}\tr( S_{{X}^{\mb{b}}U} P)\tr( S_V P).\\
    \end{split}
    \label{Eq:CliffrodEff1}
\end{equation}
Here, for a 
Pauli element $P\in\{\pm1,\pm i\}\times \{\id_2,{X},{Y},{Z}\}^{\otimes n }$, we denote $[{P}] $ as the corresponding phaseless Pauli operator that $[{P}]\in\{\id_2,{X},{Y},{Z}\}^{\otimes n }$, and $[\mb{S}_U]\coloneq\{[P]|P\in\mb{S}_U\}$. Since ${X}^{\mb{b}}{Z}^{\mb{a}}{X}^{\mb{b}} = (-1)^{\mb{a}\cdot\mb{b}}{Z}^{\mb{a}}$, when disregarding the phases, we find $[\mb{S}_U] = [\mb{S}_{{X}^{\mb{b}}U}]$. As such,   \cref{Eq:CliffrodEff1} can be transformed to 
\begin{equation}\label{Eq:CliffrodEff2}
\begin{split}
        \tr(\widehat{\rho_f} \Psi) &=\sum_{P\in[\mb{S}_U]\cap[\mb{S}_V]\cap\mathbf{P}_n/\mathcal{Z}_n}\sigma_P^{-1}\bra{\mb{b}}UPU^{\dagger}\ket{\mb{b}}\bra{\mb{0}}VPV^{\dagger}\ket{\mb{0}}.\\
\end{split} 
\end{equation}
We observe that to obtain $\tr(\widehat{\rho}_f \Psi)$, one needs to consider at most $|[\mb{S}_U]\cap[\mb{S}_V]|$ terms by considering the operators $P$ in the intersection of $|[\mb{S}_U]\cap[\mb{S}_V]|$. To determine these terms, we utilize the Tableau formalism ~\cite{aaronson2004improved}, where an $n$-qubit Clifford element can be specified with four $n \times n$  binary matrices ($\mb{\alpha},\mb{\beta},\mb{\gamma},\mb{\delta}$) and two $n$-dimensional binary vectors ($\mb{r,s}$), such that 
\begin{equation}
 {{U}^{\dag }}{{X}_{i}}U={{\left( -1 \right)}^{{{r}_{i}}}}\underset{j=0}{\overset{n-1}{\prod }}\,\sqrt{-1}^{\alpha_{i,j}\beta_{i,j}}{X}_{j}^{{{\alpha }_{i,j}}}{Z}_{j}^{{{\beta }_{i,j}}}\text{   }\!\!\And\!\!\text{   }{{U}^{\dag }}{{Z}_{i}}U={{\left( -1 \right)}^{{{s}_{i}}}}\underset{j=0}{\overset{n-1}{ \prod }}\,\sqrt{-1}^{\gamma_{i,j}\delta_{i,j}}{X}_{j}^{{{\gamma }_{i,j}}}{Z}_{j}^{{{\delta }_{i,j}}}.
 \label{CLF1}
\end{equation}
The operators ${X}_i$, ${Z}_i$ are respectively defined as a Pauli matrix with $X$ applied to the $i$-th qubit, and as a Pauli matrix with $Z$ applied to the $i$-th qubit. The parameters form a $2n \times (2n+1)$ binary matrix, which is called the Tableau of a Clifford element
\begin{equation}
\left( \begin{matrix}
   \begin{matrix}
   \text{A} \\
  \text{C} \\
\end{matrix} & \begin{matrix}
   \text{B}  \\
   \text{D}  \\
\end{matrix} & \begin{matrix}
   \mb{r}  \\
   \mb{s}  \\
\end{matrix}  \\
\end{matrix} \right) = \left( \begin{matrix}
   \begin{matrix}
   {{[{{\alpha }_{i,j}}]}_{n\times n}}  \\
   {{[{{\gamma }_{i,j}}]}_{n\times n}}  \\
\end{matrix} & \begin{matrix}
   {{[{{\beta }_{i,j}}]}_{n\times n}}  \\
   {{[{{\delta }_{i,j}}]}_{n\times n}}  \\
\end{matrix} & \begin{matrix}
   {{[{{r}_{i,j}}]}_{n\times 1}}  \\
   {{[{{s}_{i,j}}]}_{n\times 1}}  \\
\end{matrix}  \\
\end{matrix} \right).
 \label{Tableau}
\end{equation}
Using the Tableau formalism~\cite{aaronson2004improved}, we can efficiently determine the generators of the elements of $[\mb{S}_U]\cap[\mb{S}_V]$, thus computing the estimator $\tr(\widehat{\rho_f} \Psi) $, which is shown in   Algorithm 1 in the \emph{End Matter}.


 Algorithm 1 requires a $\mc{O}(n^3)+\mc{O}(n^3 \cdot 2^{n_g(U,V)})$ time cost. Here, $n_g(U,V)$ denotes the rank of the subspace $\{\mb{a}\}$ that satisfies $[UV^{\dag}(Z^{\mb{a}})VU^{\dag}]\in \mc{Z}_n$, and $|[\mb{S}_U]\cap[\mb{S}_V]|=2^{n_g(U,V)}$. In phase shadow, we sample random rotations $U$ uniformly from the phase circuit ensemble $\mc{E}_\text{phase}$. In this way, we prove Lemma 1 in the \textit{End Matter} with a stronger formulation as follows.

\begin{lemma}[Lemma 1 in the \textit{End Matter}]
\label{lemma:time_complexity}
Let $\Psi = V^\dagger\ket{\mathbf{0}}\bra{\mathbf{0}}V$ be an $n$-qubit stabilizer observable and let $U$ be sampled uniformly from the phase ensemble $\mathcal{E}_{\text{phase}}$. The variable  $2^{n_g(U,V)}$ for calculating ${\tr}(\widehat{\rho_f} \Psi)$ satisfies
\begin{equation}\label{Eq:time_com_1}
    \mathbb{E}_{U\sim\mathcal{E}_{\text{phase}}}[2^{n_g(U,V)}] = \mathcal{O}(1),
\end{equation}
and there exist a constant $c>0$ such that
\begin{equation}\label{Eq:chev}
\mathbb{P}_{U\sim\mathcal{E}_{\text{phase}}}\left(2^{n_g(U,V)}\geq 2^c\right) = \mc{O}(2^{-2c}).
\end{equation}
\end{lemma}
\begin{proof}

Sampling $U$ uniformally from $\mc{E}_\text{phase}$, the expectation is
\begin{equation}
\begin{split}
        \mathbb{E}_{U\sim\mc{E}_\text{phase}} 2^{n_g(U,V)} &= \sum_{U\in\mc{E}_\text{phase}} \frac{1}{|\mc{E}_\text{phase}|}2^{n_g(U,V)}\\
        &=\frac{|\mc{E}_\text{Cl}|}{|\mc{E}_\text{phase}|}\sum_{U\in\mc{E}_\text{phase}} \frac{1}{|\mc{E}_\text{Cl}|}2^{n_g(U,V)}\\
        &<\frac{|\mc{E}_\text{Cl}|}{|\mc{E}_\text{phase}|}\sum_{U\in\mc{E}_\text{Cl}} \frac{1}{|\mc{E}_\text{Cl}|}2^{n_g(U,V)} = \frac{|\mc{E}_\text{Cl}|}{|\mc{E}_\text{phase}|} \mathbb{E}_{U\sim\mc{E}_\text{Cl}} 2^{n_g(U,V)},\\
\end{split}
\end{equation}
where $\mc{E}_\text{Cl}$ is the the ensemble of $n$-qubit stabilizer circuits. To complete the random phase gates, one should impose $\frac{n(n-1)}{2}$ random ${CZ}$ gates and $n$ random $S$ gates. Therefore, $|\mc{E}_\text{phase}|=2^{\frac{n(n-1)}{2}}4^n=2^{\frac{n(n+3)}{2}}$. Now we determine the value of $|\mc{E}_\text{Cl}|$. From Ref.~\cite{aaronson2004improved}, we obtain that $|\mc{E}_\text{Cl}| = 2^n\prod_{k=0}^{n-1}(2^{n-k}+1)$. So 
\begin{equation}
    \frac{|\mc{E}_\text{Cl}|}{|\mc{E}_\text{phase}|} = \frac{2^n\prod_{k=0}^{n-1}(2^{n-k}+1)}{2^{\frac{n(n+3)}{2}}}=\prod_{k=0}^{n-1}(1+\frac{1}{2^{n-k}})<\prod_{k=0}^{n-1}\exp({2^{-n+k}})<e,
\end{equation}
which is a constant. This means that from now on, we only need to show that $\mathbb{E}_{U\sim\mc{E}_\text{Cl}} 2^{n_g(U,V)} = \mathcal{O}(1)$, which requires to figure out the number of $U\in\mc{E}_{\text{Cl}}$ that satisfy $|[\mb{S}_U]\cap[\mb{S}_V]|=2^{n_g(U,V)}$. 

We use a result from Ref.~\cite{garcia2017geometry}: Given two arbitrary stabilizer states, $\ket{\psi} = U^{\dagger}\ket{\mb{0}}$ and $\ket{\phi} = V^{\dagger}\ket{\mb{0}}$, their inner product satisfies  $|\langle\phi|\psi\rangle|^2 = 2^{n_g(U,V)-n}$ if and only if 1) $|[\mb{S}_U]\cap[\mb{S}_V]|=2^{n_g(U,V)}$, and 2) $[\mb{S}_U\cap\mb{S}_V]=[\mb{S}_U]\cap[\mb{S}_V]$. Also, if the second constraint is violated, then $ |\langle\phi|\psi\rangle| = 0$, which happens when there exists a $P$ satisfying $P\in\mb{S}_U$ and $-P\in \mb{S}_V$. 
Therefore, given $V$, the number of $U$ satisfying $|[\mb{S}_U]\cap[\mb{S}_V]|=2^{n_g(U,V)}$ is $2^{n_g(U,V)}$ times the number of $V$ with $|\langle\phi|\psi\rangle|^2 = 2^{n_g(U,V)-n}$.

Theorem 15 of Ref.~\cite{garcia2017geometry} states that, given a stabilizer state $\ket{\phi} = V^{\dagger} \ket{\mb{0}}$, the number of stabilizer states $\ket{\psi} = U^{\dagger}\ket{\mb{0}}$ with inner product $|\langle\phi|\psi\rangle|^2 = 2^{k-n}$ is given by  
\begin{equation}
  \mc{L}_n(n-k) = \prod_{j=1}^{n-k} \frac{2}{2^k} \frac{4^n/2^{j-1}-2^n}{2^{n-k}-2^{j-1}}.  
\end{equation}
Thus, the number of stabilizer states $\ket{\psi} = U^{\dagger}\ket{\mb{0}}$ satisfying $ n_g(U,V) = k$ is $2^{k} \mc{L}_n(n-k) $. Theorem 16 of Ref.~\cite{garcia2017geometry} gives a limit of $\mc{L}_n(n-k)$ that $\frac{\mc{L}_n(n-k)}{|\mc{E}_{\text{Cl}}|}<\frac{\mc{O}(1)}{2^{k(k+3)/2}}$.
Using this result, we conclude that given an arbitrary moment $m$ of the random variable $2^{n_g(U,V)}$,
\begin{equation}\label{Eq:postconst}
    \begin{split}
        \mathbb{E}_{U\sim\mc{E}_\text{phase}} 2^{n_g(U,V) m} &<e \mathbb{E}_{U\sim\mc{E}_\text{Cl}} 2^{n_g(U,V)m}\\
        &= e\sum_{k=0}^{n} \frac{2^k\mc{L}_n(n-k)}{|\mc{E}_\text{Cl}|} 2^{km}\\
        &<e\sum_{k=0}^{n} 2^{k(m+1)} \frac{\mc{O}(1)}{2^{k(k+3)/2}}\\
        &<e\sum_{k=0}^{n} \frac{\mc{O}(1)}{2^{k(k+1-2m)/2}} = \mc{O}(1),
    \end{split}
\end{equation}
when $m$ is a constant. Specially, when $m=1$, $\mathbb{E}_{U\sim\mc{E}_\text{phase}} 2^{n_g(U,V)}$ is a constant, thus  \cref{Eq:time_com_1} is proved. When $m=2$, $\mathbb{E}_{U\sim\mc{E}_\text{phase}} 2^{2n_g(U,V)}$ is also a constant, then we utilize the Markov inequality to obtain that for all $c>0$.
\begin{equation}
    \mathbb{P}_{U\sim\mathcal{E}_\text{phase}}\left(2^{n_g(U,V)}\geq 2^c\right) = \mathbb{P}_{U\sim\mathcal{E}_\text{phase}}\left(2^{2n_g(U,V)}\geq 2^{2c}\right)\leq \frac{\mbb{E}_{U\sim\mathcal{E}_\text{phase}}2^{2n_g(U,V)}}{2^{2c}}=\mc{O}(2^{-2c}),
\end{equation} 
which proves  \cref{Eq:chev}.
\end{proof}

Therefore, the average total computational complexity is still $\mc{O}(n^3)$ for all stabilizer states. One may be concerned that the computational complexity will be exponential for a single snapshot when $n_g\rightarrow n$. But in shadow estimation, we evaluate the average complexity depending on the method of random measurement, and the probability of $n_g\rightarrow n$ is exponentially small according to  \cref{Eq:chev}. 


\section*{Supplementary Note 8. --Details on the numerical examples}\label{Ap:Numerical}
In this section, we numerically evaluate the performance of the RPS and the PS protocols. Specifically, we assume that each two-qubit gate in the circuit is subject to gate-dependent noises with an error rate $p_e$, as defined in  \cref{Eq:noimodel}. The performances depend on two key variables: the number of qubits $n$ in the system and the gate error rate $p_e$.

In the extensive numerical simulations, we evolve some input state $\rho$ by a random (noisy) unitary $U(\widetilde{U})$ sampled from $\mc{E}_{{\text{phase}}}$, measure the involved state in the computational basis for one shot and record the measurement outcome $\mathbf{b}$. This process gives us the snapshot $ \Phi_{U,\mathbf{b}} = U^{\dag} \ket{\mathbf{b}}\bra{\mathbf{b}} U $. Then, we compute ${\tr}(O\widehat{\rho_f})$, where $O$ is either a Pauli operator or a stabilizer state and $\widehat{\rho_f}$ is a shadow snapshot as in Theorem 1 (noiseless case) and Theorem 2 (noisy case) in the main text. When $O$ is a stabilizer state, we employ Algorithm 1 to efficiently compute the value of ${\tr}(O_f \widehat{\rho})$ in the noisy case. We repeat this process for $N$ rounds and compute the empirical expectation and variance
according to
\begin{align}
\text{Bias}[{\tr}(\widehat{\rho} O_f)] &=  \left| \frac{1}{N}\sum_{i=1}^{N}{\tr}(O \rho_f^{(i)}) - {\tr}(O_f \rho) \right|, \\
\text{Var}[{\tr}(\widehat{\rho} O_f)] &= \frac{1}{N}\sum_{i=1}^{N} \left[ {\tr}(O \rho_f^{(i)}) - \frac{1}{N}\sum_{i=1}^{N}{\tr}(O \rho_f^{(i)}) \right]^2,
\end{align}
where $\rho_f^{(i)}$ is the $i$-th shadow estimation. Three input states are considered in our simulation, including the graph GHZ state vector $\ket{\text{GHZ*}} = \prod_{j=1}^{n-1}\text{CZ}_{0,j}\ket{+}^{\otimes n}$, which is equivalent to the canonical GHZ state vector $\ket{\text{GHZ}}=
\frac{1}{\sqrt{2}}[\ket{0}^{\otimes n}+\ket{1}^{\otimes n}] $
up to local transformations~\cite{van2004graphical,Graphs}, the one-dimensional cluster state vector $\ket{\text{C}_1} = \prod_{j=0}^{n-1}\text{CZ}_{j,j+1}\ket{+}^{\otimes n}$, and the product state vector $|+\rangle^{\otimes n}$. Note that since the input states are stabilizer states, and the phase gates are Clifford operations and the noise is stochastic Pauli noise, classical simulation is efficient using the Clifford Tableau techniques~\cite{aaronson2004improved}.

In the beginning, we simulate the bias and the variance of the PS protocol, say, the noiseless case. As is shown in  \cref{fig:ApHadamard-bias}(a), the estimation variance of the PS protocol is constant (independent of the qubit number), which is similar to the Clifford measurement. In contrast, the estimation variance of the Pauli measurement scales exponentially with $n$. This result is consistent with  Theorem 1 in the main text.

Next, we focus on the RPS protocol. This simulation encompasses two aspects. First, we compared the RPS protocol with the PS protocol in terms of the estimation bias in the noisy scenario. We simulate fidelity estimation using the RPS and PS protocol for $|\text{GHZ}^*\rangle$ with $n=25, 35, 45$ qubits. As shown in \cref{fig:ApHadamard-bias}(b),  when the circuits are noisy, the RPS protocol still provides an unbiased estimate, whereas the PS protocol does not. These simulated results highlight the protocol's reliability in NISQ-era applications, where devices typically have tens of qubits and error rates in the range of $10^{-3}$ to $10^{-2}$.

Secondly, to evaluate post-processing efficiency, \cref{fig:ApHadamard-bias}(c) shows the cube root of the average post-processing time to estimate the fidelity of $ \ket{\text{C}_1} $,  $\ket{\text{GHZ*}} $ and $ \ket{+}^{\otimes n} $ on $N = 10,\!000$ snapshots. The cubic scaling confirms that post-processing remains efficient even under noise, as mentioned in Proposition 3. 

\begin{figure}
    \centering
    \includegraphics[width=0.5\linewidth]{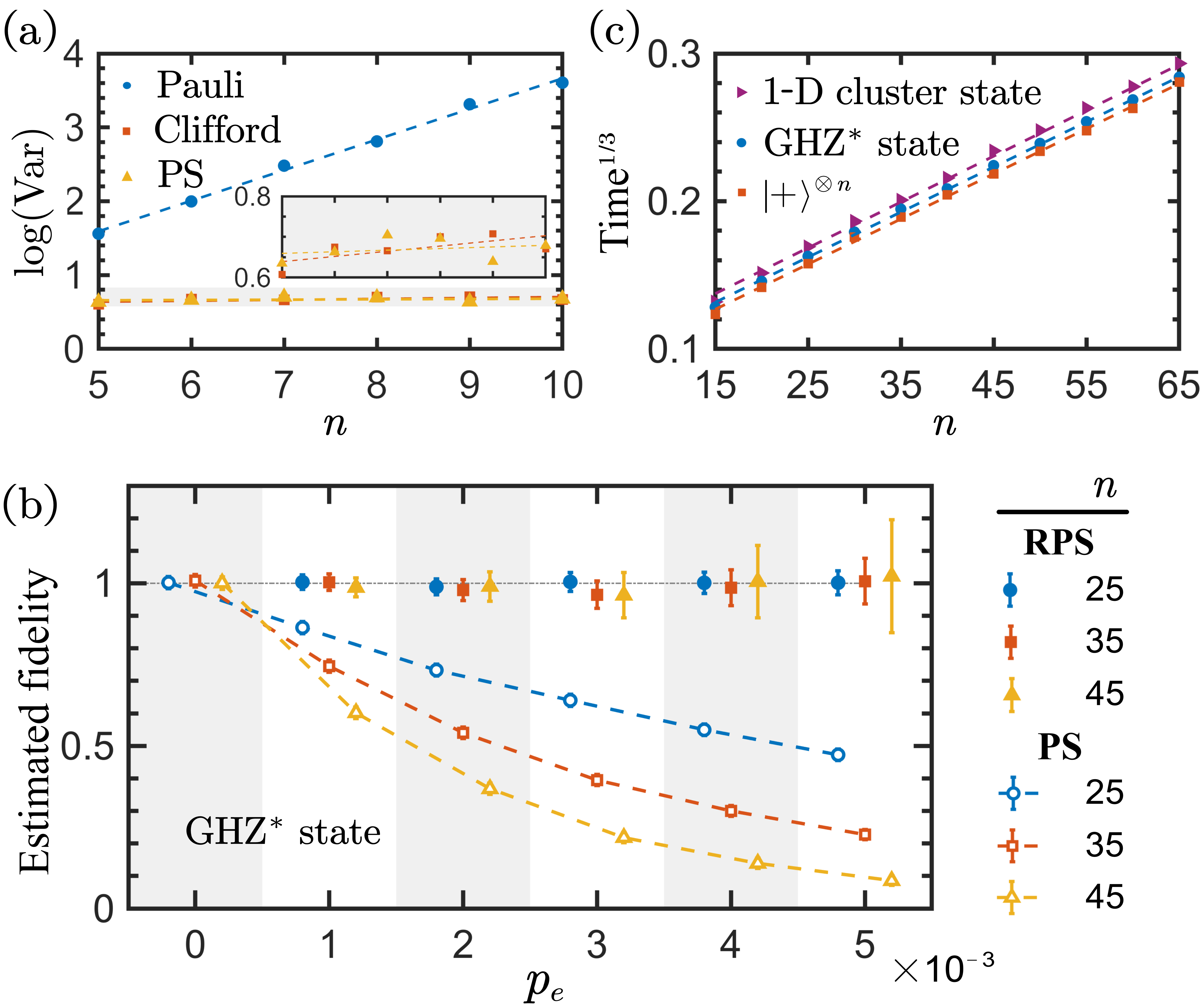}
    \caption{(Fig. 3 in the main text) (a) Variance of fidelity estimation for the GHZ* state (noiseless case), comparing Pauli, Clifford measurements, and PS~\cite{huang2020predicting} under different qubit number $n$ of the system.  (b) Fidelity estimation of the GHZ* state using  RPS and PS measurements for $n \in \{25,35,45\}$ qubits under different noise level $p_e$ ($N = 50,\!000$ snapshots for each data point).  (c) 
    Cubic root of the post-processing time for fidelity estimation of three graph states using RPS ($n = 15$--$65$, $N = 10,\!000$). }
    \label{fig:ApHadamard-bias}
\end{figure}

Lastly, we examine the estimation variance of the RPS protocol. \cref{fig:variance}(a) shows the variance of the fidelity estimation for 10 random stabilizer states across $n = 20\text{--}30$ qubits and $p_e = 0\text{--}0.01$. The magnitude of estimated variance is numerically below the theoretical upper bound in Eq. (9), where we choose $\Theta(1) = 3$.
Furthermore, it is important to quantitatively verify that the estimation variance exhibits the predicted exponential scaling behavior $e^{n^2 p_e / 2}$, as this provides precise guidance for practical implementations. To do so, we examine the dependence of the variance on the qubit number $n$ and the error rate $p_e$, respectively. As shown in \cref{fig:variance}(c), we find that $\sqrt{\log(\text{Var})}$ increases linearly with $n$, with a fitted slope of $0.044$, which is slightly smaller than the theoretical upper bound $\sqrt{p_e/2} = 0.05$. This indicates that the growth rate of the variance is below the theoretical worst-case scenario. Similarly, \cref{fig:variance}(c) shows that $\log(\text{Var})$ increases linearly with $p_e$, with a fitted slope of $172$. This again falls below the theoretical slope $n^2 / 2 = 200$, further confirming that the empirical growth rate of the variance with respect to $p_e$ remains consistently below the upper bound. Taken together, these results validate the scaling behaviors of the variance bound derived in Eq. (9), which provides a more favorable variance profile than the worst-case bound suggests, reinforcing the practical utility of the RPS protocol for realistic applications. 

\begin{figure}
    \centering
    \includegraphics[width=0.8\linewidth]{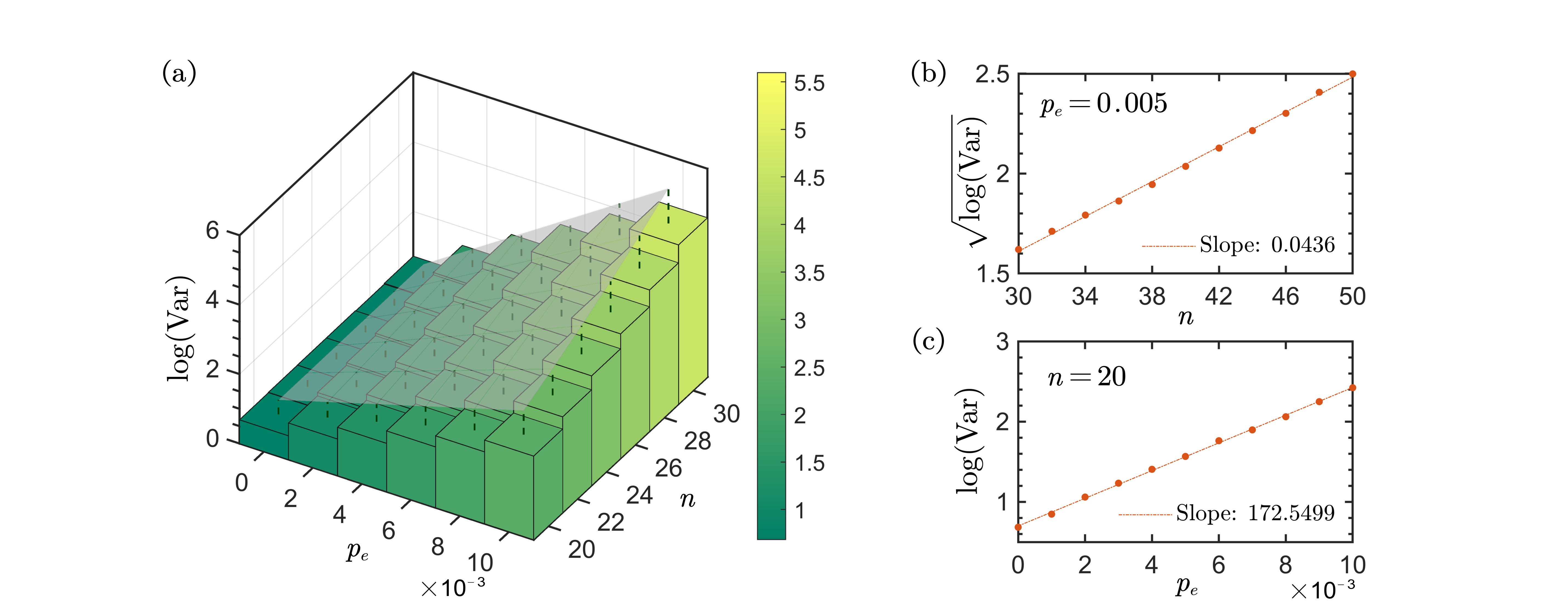}
    \caption{ An illustration of the estimation variance using robust phase shadow. For each task, it is repeated for $N=50,\!000$ snapshots. (a)  The variance of estimated fidelity for random stabilizer states using the RPS protocol as a function of error rate $p_e$ and $n$. The number of qubits ranges from $20$ to $30$, and the error rate $p_e$ varies from $0$ to $0.01$. The numerical results are shown as three-dimensional colored bars, and the theoretical bound $3e^{n^2p_e/2}$ is shown in lightgray.
    (b) Fidelity estimation of the GHZ state using the RPS protocol for systems with $n = 30$ to $50$ qubits, at a fixed error rate of $p_e = 0.005$. The dashed line represents a linear fit to the numerical data, yielding a slope of $0.0436$, which is close to the theoretical upper bound of $0.05$. 
    (c) Fidelity estimation of the GHZ state using the RPS protocol at a fixed qubit number $n = 20$, with error rate $p_e$ ranging from $0$ to $0.01$. The fitted slope from numerical results is $172.54$, which is close to but remains under the theoretical upper bound $200$. 
  }
    \label{fig:variance}
\end{figure}

\section*{Supplementary Note 9. --Extension of noise models}\label{Ap:NoiseEx}

\subsection*{A. Pauli-Z-type noise model}

In the RPS protocol, we model the gate-dependent noise model as the ZZ Pauli error. Such noise naturally arises in realistic physical systems and exhibits a favorable mathematical structure.  Here, we present an extended biased noise model that incorporates IZ and ZI errors, and the corresponding robust implementation of the phase shadow tomography. The extended noise model is native in Rydberg atom platforms~\cite{sahay2023high,cong2022hardware}. 
In this way, we denote the noisy two-qubit gates
\begin{equation}\label{Eq:erasure_errmodel}
    \Lambda(\rho) = (1-\frac{3}{4}p_e)\rho +\frac{p_e}{4} {ZI} \rho {ZI} +\frac{p_e}{4} {IZ} \rho 
    {IZ} +\frac{p_e}{4} {ZZ} \rho {ZZ}.
\end{equation}
Here, we aim to obtain the measurement channel. The derivation should follow the procedures in Supplementary Note 5. First, we first utilize the twirling $\widetilde{\Lambda}^{(2)}_{{A,ij}}$ defined in Supplementary Note 5(A) to calculate that
\begin{equation}
\begin{aligned}
\widetilde{\Lambda}^{(2)}_{{A,ij}}(\ket{\mb{x,w}}\bra{\mb{y,s}})
&=\mathbb{E}_{A_{i,j}\in\{0,1\}} \widetilde{{CZ}}^{\dagger A_{i,j}}_{i,j}\otimes {CZ}_{i,j}^{\dagger A_{i,j}}(\ket{\mb{x,w}}\bra{\mb{y,s}})\widetilde{{CZ}}_{i,j}^{A_{i,j}}\otimes {CZ}_{i,j}^{A_{i,j}}\\
&= \{\frac{1}{2} +\frac{1}{2}(-1)^{x_i x_j +w_i w_j-y_i y_j -s_i s_j}[\frac{p_e}{4}(-1)^{x_i  -y_i}+\frac{p_e}{4}(-1)^{x_i+x_j -y_i-y_j}+\frac{p_e}{4}(-1)^{x_j  -y_j}+(1-\frac{3p_e}{4})]\}  \ket{\mb{x,w}}\bra{\mb{y,s}}\\
&= \frac{1}{2}  \ket{\mb{x,w}}\bra{\mb{y,s}}+\frac{1}{2}(-1)^{x_i x_j +w_i w_j-y_i y_j -s_i s_j}[\frac{p_e}{4}[(-1)^{x_i  -y_i}+1][(-1)^{x_j  -y_j}+1]+(1-p_e)]  \ket{\mb{x,w}}\bra{\mb{y,s}}\\
&=\begin{cases}
    \ket{\mb{x,w}}\bra{\mb{y,s}} & \text{if } T_{i,j}^{(1)}\equiv0 \pmod{2}, \widetilde{T}_{i,j}^{(2)}\equiv0 \pmod{2} \\
(1-\frac{p_e}{2})\ket{\mb{x,w}}\bra{\mb{y,s}} & \text{if } T_{i,j}^{(1)}\equiv0\pmod{2}, \widetilde{T}_{i,j}^{(2)}\equiv1 \pmod{2}\\
0& \text{if } T_{i,j}^{(1)}\equiv1\pmod{2}, \widetilde{T}_{i,j}^{(2)}\equiv0 \pmod{2}\\
\frac{p_e}{2}\ket{\mb{x,w}}\bra{\mb{y,s}}& \text{if } T_{i,j}^{(1)}\equiv1\pmod{2}, \widetilde{T}_{i,j}^{(2)}\equiv1 \pmod{2}\\
\end{cases}.\\
\end{aligned}
\label{Eq: channel_noise_ij_fulldephasing}
\end{equation}
We now define the coefficient $\widetilde{T}_{i,j}^{(2)}\coloneq (x_i  -y_i)(x_j  -y_j)$, instead of ${T}_{i,j}^{(2)}=x_i  -y_i+x_j  -y_j$ in \cref{Eq: channel_noise_ij}. In this way, $c_{i,j} = 1-\frac{p_e}{2}$ instead of $1$ when $\ket{x_i, w_i}\bra{y_i,s_i}$ and $\ket{x_j ,w_j}\bra{y_j,s_j}$ both belong to $\mbb{S}_2-\Delta_2$. Moreover, the coefficient changes from $p_e$ to $\frac{p_e}{2}$. Consequently, the noisy moment function becomes

\begin{equation}
    \widetilde{\mathbf{M}}_{\mc{E}_\text{phase}}^{(2)} =2^{-2n} \sum_{I_i,I_j,I_k} (1-\frac{p_e}{2})^{i\times k+\frac{k(k-1)}{2}}(\frac{p_e}{2})^{j\times k}\Delta_2^{I_i}\otimes (\id_4-\Delta_2)^{I_j} \otimes (\mbb{S}_2-\Delta_2)^{I_k},
\label{Eq: channel_noise_ij2prime}
\end{equation}

Mention that following the proof of Eq. (8) in Supplementary Note 6, we can set 
\begin{equation}
    \sigma_P =  \sum_{s=0}^{n_1 + n_2} \sum_{t=\max(0,s-n_1)}^{\min(s, n_2)} (-1)^t \binom{n_1}{s-t} \binom{n_2}{t}  (1 - \frac{p_e}{2})^{(n_1 + n_2 - s)n_3+\frac{n_3(n_3-1)}{2}} (\frac{p_e}{2})^{sn_3}
\end{equation}
for Pauli operators in the form $ P =\id_2^{\otimes n_1}\otimes  Z^{\otimes n_2}\otimes \{X,Y\}^{\otimes n_3}$. In this way, one can also achieve an unbiased estimation.

Furthermore, we numerically investigate the variance of the robust protocol for the extended noise model. The numerical results show that the upper bound of the variance also scales exponentially with $n^2p_e$, following the result in Theorem 2 in the \emph{main text}. This scaling is demonstrated via the numerical simulation to estimate the fidelity of stabilizer states. As shown in  \cref{fig:furthermore}, the logarithmic estimation variance increases linearly with the error rate $p_e$, with a slope close to a linear function of $n^2$.

\begin{figure}
    \centering
    \includegraphics[width=.8\linewidth]{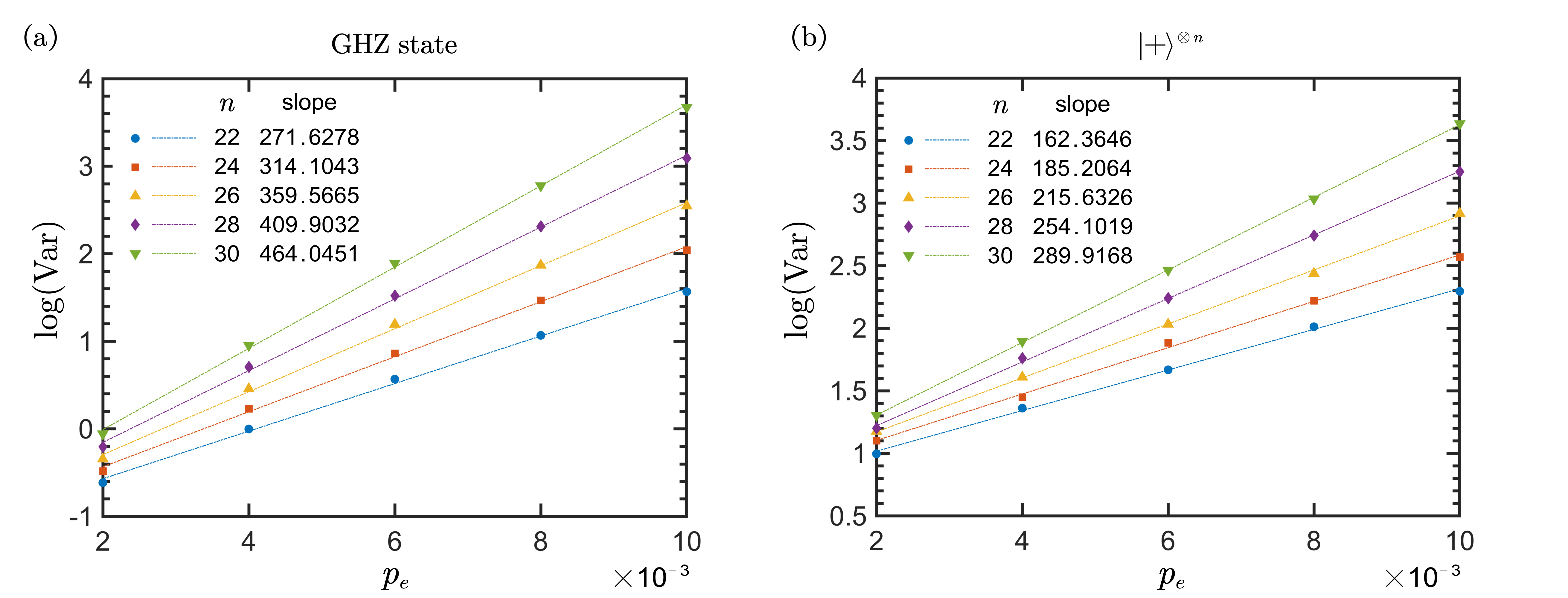}
    \caption{Estimation variance of the RPS protocol using the noise model in  \cref{Eq:erasure_errmodel}, with snapshot $N=50,\!000$ for each task. (a) The estimated fidelity of the GHZ state. Here, the logarithmic variance scales linearly with $p_e$ with the slope close to $n^2/2$. (b) The estimated fidelity to $\ket{+}^{\otimes n}$. The logarithmic variance also scales linearly with $p_e$, and the slope fits a linear function of $n^2$.} 
    \label{fig:furthermore}
\end{figure}

\section*{Supplementary Note 10. Generalized robust shadow estimation}\label{Ap:Paulinoise}

One of the key obstacles to the broader applicability of the RPS method lies in the restricted noise model considered in the current analysis. As detailed in the \emph{main text} and in Supplementary Note 9, the theoretical guarantees of RPS are derived only for Pauli-$Z$–type noises. To address this limitation, we introduce an extended formulation of the protocol that accommodates general noise models.  While our extension is primarily developed with Pauli noise in mind, it is applicable to arbitrary noise if the noise model is invariant under Pauli twirling. Specifically, we can leverage randomized compiling~\cite{wallman2016noise} to convert arbitrary noise into Pauli noise, thereby enabling the extended method to be applied to scenarios involving arbitrary noise. 

To this end, we first propose a modified estimator whose only difference from Algorithm 1 in the \emph{main text} is that the coefficient $\sigma_{P}$ is replaced by a gate-dependent coefficient $\sigma(P,U)$, which satisfies that 
\begin{equation}
    \widetilde{U}^\dagger P \widetilde{U}
  = \sigma(P,U)\, U^\dagger P U ,
\end{equation}
where $\tilde{U}$ means the noisy Clifford circuit with gate-dependent Pauli noises. In the rest of this section, we first describe how $\sigma(P,U)$ is computed. Then we prove that one can obtain an unbiased estimation using the estimator
\begin{equation}
  \widehat{\rho_f}_{\mathrm{gen}}
  = D^{-1} \sum_{P \in \mb{P}_n / \mathcal{Z}_n}
    \sigma(P,U)^{-1}\, \tr(\Phi_{U,\mathbf{b}} P)\, P ,
\end{equation}
as well as proving the upper bound of the estimation variance. Finally, we provide the details on the numerical simulations.

\subsection*{A. Computation of the gate-dependent coefficient $\sigma(P,U)$}
Consider a Pauli operator $P\in\mathcal{P}_n$ and a Clifford circuit $U$ consisting of a sequence of gates $g_1,\ldots,g_l$. Under ideal evolution, the conjugation $U^\dagger P U$ is again a Pauli operator. In the presence of gate-level Pauli noise, we assume that each gate $g_j$ is followed by a Pauli channel $\Lambda_j$, which acts as $\tilde{g_j} = g_j\circ\Lambda_j$, where $\Lambda_j(Q) = \alpha_{j,Q}Q$ for all Pauli operators $Q\in\mbb{P}_k$, and $k$ is the qubit numbers on which $g_j$ acts. It satisfy that $0<\alpha_{j,Q}\leq1$ for all $Q$. Tracking the propagation of Pauli observable $P$ through the noisy circuit therefore yields a sequence
\begin{equation}
P_0 = P,\qquad 
P_{j} = g_j P_{j-1} g_j^\dagger,\qquad j=1,\ldots,m ,
\end{equation}
together with a multiplicative attenuation factor $\alpha_{j,P_j}$. In this way, the total coefficient associated with the propagation of $P$ through $U$ is 
\begin{equation}\label{Eq:sigmaPU}
    \sigma(P,U) = \prod_{j=1}^{m}\alpha_{j,P_j}.
\end{equation}

This attenuation factor $\sigma(P, U)$ can be efficiently computed using Algorithm~\ref{Algo:sigmaPU} given the circuit and noise models. For example, in RPS, if each two-qubit CZ gate in $U$ is followed by an independent depolarizing channel $\Lambda(\rho)= (1-p)\rho + p\,\mathrm{Tr}(\rho)\id_{4}/4$, then any non-identity Pauli operator supported on the affected qubits acquires a factor of $(1-p)$. Thus, if $P$ is propagated through three such CZ layers and remains non-identity on the corresponding supports, then $\sigma(P,U)=(1-p)^3$.
In~\cref{fig:PauliNoise}, we provide another example on computing the coefficient $\sigma(P,U)$, where we assume all CZ gates are subject to depolarizing noise $\Lambda(\rho) = (1-p)\rho + p\,\mathrm{Tr}(\rho)\id_{4}/4$, and all single-qubit gates are subject to depolarizing noise $\Lambda'(\rho) = (1-p_s)\rho + p_s\,\mathrm{Tr}(\rho)\id_{2}/2$.

\begin{figure}
    \centering
\includegraphics[width=0.9\linewidth]{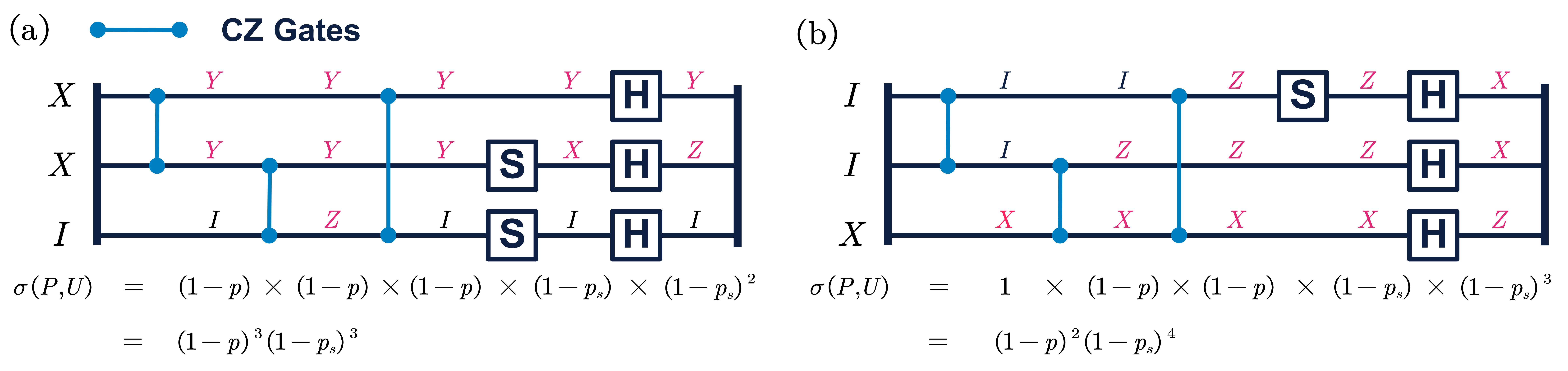}
    \caption{Calculation of $\sigma(P,U)$ under depolarizing noise.  Non-identity Pauli operators (red) accumulate a decay factor [$(1-p)$ or $(1-p_s)$], while identity operators (black) contribute a factor of 1. The total coefficient is the product of all decay terms in the circuit instance.}
    \label{fig:PauliNoise}
\end{figure}

\setcounter{algocf}{1}
\begin{algorithm}[H]
\caption{Computation of $\sigma(P,U)$ for arbitrary Pauli noise}
\label{Algo:sigmaPU}
\KwInput{Pauli operator $P$ and Clifford circuit $U=(g_1,\ldots,g_l)$ with corresponding Pauli channels $(\Lambda_1,\ldots,\Lambda_l)$.}
\KwOutput{Coefficient $\sigma(P,U)$.}

$\sigma\leftarrow 1,\qquad P_0\leftarrow P$\;

\For{$j=1$ {\bf to} $l$}{
    Compute $P_j=g_j P_{j-1} g_j^\dagger$\;
    Extract $\alpha_{j,P_j}$ from the Pauli channel $\Lambda_j$\;
    $\sigma\leftarrow \sigma\cdot \alpha_{j,P_j}$\;
}
\KwRet{$\sigma$}\;
\end{algorithm}

\medskip

Since the propagation of a single Pauli operator through a Clifford circuit requires only $O(l)$ operations, the computation of $\sigma(P,U)$ has complexity $O(l)$. For RPS, the length of the circuit satisfies $l=O(n^2)$. Therefore, replacing $\sigma_P$ by $\sigma(P,U)$ does not change the overall post-processing complexity of RPS, which remains $O(n^3)$. This ensures that the generalized RPS protocol retains the same asymptotic efficiency while supporting arbitrary Pauli noise channels.

The proposed generalized RPS framework substantially extends the utility of the original RPS protocol. By relaxing the structural and noise assumptions, it provides a versatile tool for robust shadow estimation.

\begin{enumerate}
\item \textbf{Extension to more noise models.}
    The method no longer requires the idealized assumption that all quantum gates share an identical noise channel. Instead, each gate $g_j$ in the circuit-including single-qubit gates like Hadamard and Phase gates-may be assigned an individual Pauli channel $\Lambda_j$ that reflects its specific physical error profile. This implies that, unlike previous robust shadow estimation ~\cite{chen2021robust} approaches which rely on a single global average noise parameter, our framework permits the individual addressing of errors for each gates. One can characterize these gate noise channels using established protocols~\cite{blume2017demonstration, greenbaum2015introduction,erhard2019characterizing}, and subsequently utilize this fine-grained information to construct a strictly unbiased robust shadow estimator.
    Furthermore, while the generalized RPS framework is primarily developed with Pauli noise in mind, it can be extended to handle arbitrary noise. This capability is achieved by integrating randomized compiling—provided that the noise model remains independent of single-qubit gates. 

\item \textbf{Extension to more circuit structures.}
    The introduction of the instance-specific correction factor $\sigma(P,U)$ removes the structural constraint that the underlying circuit must follow the specific ${\rm S}{\rm\text{-}}{\rm CZ}{\rm\text{-}}{\rm H}$ architecture. Since the propagation rule $P_j = g_j P_{j-1} g_j^\dagger$ applies to arbitrary Clifford gates, our framework can be directly applied to random Clifford measurements, and even \emph{shallow circuit invariants} (e.g., finite-depth brick-wall architectures) or other hardware-native connectivity graphs. This capability allows for a reduction in circuit depth from $O(n^2)$ to $O(n\log n)$ or even $O(n)$, addressing the gate-count limitations of near-term devices while maintaining robustness.
\end{enumerate}

Next, we utilize the coefficient $\sigma(P,U)$ in RPS and prove Theorem 3 in the \emph{main text}.

\subsection*{B. Proof of Theorem 3}

We now formally establish the statistical guarantees of the generalized RPS framework in Theorem 3 in the \emph{main text}. This theorem encapsulates two fundamental results. First, Proposition~\ref{prop:GRPSunbias} demonstrates the strict unbiasedness of the generalized estimator under gate-dependent noise. Second, Proposition~\ref{prop:GRPSVar} provides a rigorous upper bound on the estimation variance. 

\begin{prop}[Unbiasedness of the generalized RPS estimator]\label{prop:GRPSunbias}
Suppose that one conducts 
robust phase shadow estimation using the random unitary ensemble $\mc{E}_{{\text{phase}}}$ with noisy ${CZ},S$ and $H$ gates, then the unbiased estimator of \(\rho_f\) in a single-shot is given by
\begin{equation} \label{Eq:GenInvNoisychannelPauli}  
\widehat{\rho_f}_{\text{gen}} := \sum_{P \in \mb{P}_n/\mathcal{Z}_n} \sigma(P,U)^{-1} \, \tr(\Phi_{U,\mathbf{b}} {P}) \, {P},  
\end{equation}  

\end{prop}
where \(\sigma(P,U)\) is the Pauli coefficient of the forward noisy channel from \cref{Eq:sigmaPU}, ensuring \(\mathbb{E}_{\{U,\mathbf{b}\}} \widehat{\rho_f}_{\text{robust}} = \rho_f\). 
\begin{proof}
We begin from the noisy snapshot expansion in the Pauli basis.
For each $U\in\mathcal{E}_{\mathrm{phase}}$, the noisy measurement
operator admits the decomposition
\begin{equation}
    \Phi_{\widetilde{U},\mathbf{b}}
    = 2^{-n}\!\sum_{P\in\mathbb{P}_n}
      \sigma(P,U)\,
      \tr(\Phi_{U,\mathbf{b}}P)\,P,
\end{equation}
and the probability of measuring $\mathbf{b}$ is 
\begin{equation}
\Pr(\mathbf{b}\mid U)
= \tr(\Phi_{\widetilde{U},\mathbf{b}}\rho)
= 2^{-n}\!\sum_{P}\sigma(P,U)\,
  \tr(\rho P)\,\tr(\Phi_{U,\mathbf{b}}P). 
  \label{eq: prob b grps}
\end{equation}
Now, we consider the generalized RPS estimator
\begin{equation}
    \widehat{\rho_f}_{\mathrm{gen}}
    = \sum_{P\in\mathbb{P}_n/\mathcal{Z}_n}
      \sigma(P,U)^{-1}\,\tr(\Phi_{U,\mathbf{b}}P)\,P.
\end{equation}
Taking the expectation over both $U$ and $\mathbf{b}$, we obtain
\begin{equation}
\mathbb{E}_{U,\mathbf{b}}
\bigl[\widehat{\rho_f}_{\mathrm{gen}}\bigr]
=
\sum_{P\notin\mathcal{Z}_n}
\mathbb{E}_{U,\mathbf{b}}
\bigl[
  \sigma(P,U)^{-1}\tr(\Phi_{U,\mathbf{b}}P)
\bigr]\,P.
\label{Eq:gen-exp-start}
\end{equation}
Using the probability given by Eq.~\ref{eq: prob b grps}, we have 
\begin{equation}\label{Eq:genUnbias1}
\begin{aligned}
\mathbb{E}_{U,\mathbf{b}}
\bigl[
  \sigma(P,U)^{-1}\tr(\Phi_{U,\mathbf{b}}P)
\bigr]
&=
\mathbb{E}_{U}\sum_{\mathbf{b}}
\Pr(\mathbf{b}\mid U)\,
\sigma(P,U)^{-1}\,\tr(\Phi_{U,\mathbf{b}}P)
\\
&=
2^{-n}\mathbb{E}_{U}\sum_{\mathbf{b}}
\sum_{Q}\sigma(Q,U)\tr(\rho Q)
\,\tr(\Phi_{U,\mathbf{b}}Q)
\,\sigma(P,U)^{-1}\tr(\Phi_{U,\mathbf{b}}P).
\end{aligned}
\end{equation}

At this point, the key property is the following
orthogonality relation 
\begin{equation}
\begin{split}
    \sum_{\mb{b}}
\bigl[
  \tr(\Phi_{U,\mathbf{b}}Q)\,
  \tr(\Phi_{U,\mathbf{b}}P)
\bigr]&
=\sum_{\mb{b}}\bra{\mb{b}}UQU^{\dagger}\ket{\mb{b}}\bra{\mb{b}}UPU^{\dagger}\ket{\mb{b}}\\
&=\tr\bigl[(UPU^{\dagger}\otimes UQU^{\dagger})\text{ }\sum_{\mb{b}}\ket{\mb{b}}\bra{\mb{b}}\otimes\ket{\mb{b}}\bra{\mb{b}} \bigl]\\
&=2^{-n}\tr\bigl[(UPU^{\dagger}\otimes UQU^{\dagger})\text{ }\Delta_2^{\otimes n} \bigl],\\
\end{split}
\label{Eq:orthogonality}
\end{equation}
where $\Delta_2=(\id_2\otimes\id_2+Z\otimes Z)/2$.
Notice that $\Delta_2$ is a linear combination of two-copy Pauli tensors. Since $U$ is Clifford, both $UPU^\dagger$ and $UQU^\dagger$ are Pauli operators. Therefore, by Pauli orthogonality, the trace in \cref{Eq:orthogonality} vanishes whenever $P\neq Q$
\begin{equation}
    \tr\!\left[(UPU^{\dagger}\otimes UQU^{\dagger})\;\Delta_2^{\otimes n}\right]=0
\qquad (P\neq Q).
\end{equation}
On the other hand, when $P=Q$, we have $UPU^{\dagger}=UQU^{\dagger}$, and hence

 \begin{equation}
     \tr\bigl[(UPU^{\dagger}\otimes UQU^{\dagger})\text{ }\Delta_2^{\otimes n} \bigl]=2^n\tr\bigl[{UPU^{\dagger}}^{\otimes 2}\text{ }\sum_{\mb{b}}\ket{\mb{b}}\bra{\mb{b}}\otimes\ket{\mb{b}}\bra{\mb{b}} \bigl]=2^n\sum_{\mb{b}}\bra{\mb{b}}UPU^{\dagger}\ket{\mb{b}}^2.
 \end{equation} 
 
 In summary, we can conclude from \cref{Eq:orthogonality} that $2^{-n}\tr\bigl[(UPU^{\dagger}\otimes UQU^{\dagger})\text{ }\Delta_2^{\otimes n} \bigl]=\delta_{P=Q}\sum_{\mb{b}}\bra{\mb{b}}UPU^{\dagger}\ket{\mb{b}}^2$. Substituting this identity into the expectation in \cref{Eq:genUnbias1}, we obtain

\begin{equation}
\begin{split}
    \mathbb{E}_{U,\mathbf{b}}
\bigl[\widehat{\rho_f}_{\mathrm{gen}}\bigr]
&=
\sum_{P\notin\mathcal{Z}_n}
2^{-n}\mathbb{E}_{U}\sum_{\mathbf{b}}
\sum_{Q}\sigma(Q,U)\tr(\rho Q)
\,\tr(\Phi_{U,\mathbf{b}}Q)
\,\sigma(P,U)^{-1}\tr(\Phi_{U,\mathbf{b}}P)\,P\\
&= \sum_{P\notin\mathcal{Z}_n}2^{-n}\mathbb{E}_{U}\sum_{\mb{b}}\sigma(P,U)\sigma(P,U)^{-1}
\bra{\mb{b}}UPU^{\dagger}\ket{\mb{b}}^2\tr(\rho P)P\\
&=2^{-n}\sum_{P\notin\mathcal{Z}_n}\mathbb{E}_{U}\sum_{\mb{b}}\bra{\mb{b}}UPU^{\dagger}\ket{\mb{b}}^2\tr(\rho P)\,P\\
&=2^{-2n}\sum_{P\notin\mathcal{Z}_n}\tr[(\id_4^{\otimes n}+\mbb{S}_2^{\otimes n}-\Delta_2^{\otimes n})\text{ }P^{\otimes 2}]\tr(\rho P)\,P=2^{-n}\sum_{P\notin\mathcal{Z}_n}\tr(\rho P)\,P.
\end{split}
\end{equation}
This is precisely the Pauli expansion of the off-diagonal component
$\rho_f=\rho-\rho_d$, namely
\begin{equation}
\rho_f
= \sum_{P\notin\mathcal{Z}_n}2^{-n}\tr(\rho P)\,P.
\end{equation}
Therefore
\begin{equation}
\boxed{
\mathbb{E}_{U,\mathbf{b}}
\bigl[\widehat{\rho_f}_{\mathrm{gen}}\bigr]
= \rho_f,
}
\end{equation}
which completes the proof.
\end{proof}

\begin{prop}[Variance of the generalized RPS estimator for stabilizer observables]\label{prop:GRPSVar}
Let the observable  $O = V^\dagger\ket{\mathbf{0}}\bra{\mathbf{0}}V$ be an $n$-qubit stabilizer state with $V$ being a Clifford operator, then the estimation variance of the generalized RPS
\begin{equation}
  \mathrm{Var}\bigl(\widehat{o_f}_{\mathrm{gen}}\bigr)
  \;\le\;
  \Theta(1)\,
  \sqrt{\,
    \mathbb{E}_{U\sim\mathcal{E}_{\mathrm{phase}}}
    \Bigl[
      \max_{P\in\mathbb{P}_n/\mathcal{Z}_n}\sigma(P,U)^{-4}
    \Bigr]
  }.
  \label{Eq:stab-var-final}
\end{equation}
\end{prop}

\begin{proof}
Fix a random unitary $U\sim\mathcal{E}_{\mathrm{phase}}$.  
As in Suppoelemtary Note~7, let $n_g(U,V)$ denote the dimension of the overlap subspace between the stabilizer groups associated with $U$ and $V$.  
For this fixed $U$, the generalized RPS estimator for the stabilizer fidelity can be written as a finite signed sum
\begin{equation}
  \widehat{o_f}_{\mathrm{gen}}(U,\mathbf{b})
  =
  \sum_{j=0}^{2^{n_g(U,V)}-1} \sigma(P_j,U)^{-1}\, s_j(U,\mathbf{b}),
\end{equation}
where each $P_j\in\mathbf{P}_n/\mathcal{Z}_n$, and
$s_j(U,\mathbf{b})\in\{-1,+1\}$ collects the stabilizer signs
$\bra{\mathbf{b}}U^\dagger P_j U\ket{\mathbf{b}}\bra{\mathbf{0}}V^\dagger P_j V\ket{\mathbf{0}}$.  
As $0<\sigma(P_j,U)\le 1$, all coefficients $\sigma(P_j,U)^{-1}$ lie in $[1,+\infty)$.

For such a signed sum, we have the deterministic bound
\begin{equation}
  \bigl|\widehat{o_f}_{\mathrm{gen}}(U,\mathbf{b})\bigr|
  = \Bigl|\sum_{j=0}^{2^{n_g(U,V)}-1} \sigma(P_j,U)^{-1}s_j(U,\mathbf{b})\Bigr|
  \le \sum_{j=0}^{2^{n_g(U,V)}-1} |\sigma(P_j,U)^{-1}|
  \le 2^{n_g(U,V)}\max_{P\in \mathbf{P}_n/\mathcal{Z}_n } \sigma(P,U)^{-1},
\end{equation}
hence, we arrive at
\begin{equation}
  \bigl(\widehat{o_f}_{\mathrm{gen}}(U,\mathbf{b})\bigr)^2
  \le4^{\,n_g(U,V)}\Bigl(\max_{P\in\mathbb{P}_n/\mathcal{Z}_n}\sigma(P,U)^{-1}\Bigr)^2.
\end{equation}
Taking the expectation over $\mathbf{b}$ for fixed $U$ and using $\mathrm{Var}(X)\le\mathbb{E}[X^2]$ gives
\begin{equation}
  \mathrm{Var}\bigl(\widehat{o_f}_{\mathrm{gen}}\mid U\bigr)
  \le
  \mathbb{E}_{\mathbf{b}\mid U}
  \bigl[\bigl(\widehat{o_f}_{\mathrm{gen}}(U,\mathbf{b})\bigr)^2\bigr]\le4^{\,n_g(U,V)}\Bigl(\max_{P\in\mathbb{P}_n/\mathcal{Z}_n}\sigma(P,U)^{-1}\Bigr)^2.
\end{equation}
Averaging over $U$ then yields
\begin{equation}
  \mathrm{Var}\bigl(\widehat{o_f}_{\mathrm{gen}}\bigr)
  = \mathbb{E}_{U}\bigl[\mathrm{Var}(\widehat{o_f}_{\mathrm{gen}}\mid U)\bigr]
  \le
  \mathbb{E}_{U\sim\mathcal{E}_{\mathrm{phase}}}
  \Bigl[
    4^{\,n_g(U,V)}
    \Bigl(\max_{P\in\mathbb{P}_n/\mathcal{Z}_n}\sigma(P,U)^{-1}\Bigr)^2
  \Bigr].
  \label{Eq:stab-var-CS-step0}
\end{equation}

To bound the right-hand side of \cref{Eq:stab-var-CS-step0}, we apply the Cauchy--Schwarz inequality to the average over $U$:
\begin{equation}
\begin{split}
  &\mathbb{E}_{U}
  \Bigl[
    4^{\,n_g(U,V)}
    \Bigl(\max_{P}\sigma(P,U)^{-2}\Bigr)
  \Bigr] \\
  &\qquad\le
  \Bigl(
    \mathbb{E}_{U}\bigl[4^{\,2n_g(U,V)}\bigr]
  \Bigr)^{1/2}
  \Bigl(
    \mathbb{E}_{U}
    \Bigl[
      \max_{P}\sigma(P,U)^{-4}
    \Bigr]
  \Bigr)^{1/2}.
\end{split}
\end{equation}
We now define
\begin{equation}
  C_{\mathrm{stab}}
  :=
  \Bigl(
    \mathbb{E}_{U\sim\mathcal{E}_{\mathrm{phase}}}
    \bigl[4^{\,2n_g(U,V)}\bigr]
  \Bigr)^{1/2},
\end{equation}
which depends only on the ensemble $\mathcal{E}_{\mathrm{phase}}$ and the fixed stabilizer observable $O = V^\dagger\ket{\mathbf{0}}\bra{\mathbf{0}}V$.  
As shown in \cref{Eq:postconst}, $C_{\mathrm{stab}}$ is finite (and in fact $\Theta(1)$).  
Substituting this definition into the previous inequality yields \cref{Eq:stab-var-final}, completing the proof.
\end{proof}

For detailed noise models, one can bound the coefficient$\Bigl(\mathbb{E}_{U\sim\mathcal{E}_{\mathrm{phase}}}\bigl[\max_{P\in\mathbb{P}_n/\mathcal{Z}_n}\sigma(P,U)^{-4}
    \bigr]
\Bigr)^{1/2}
$ as follows. In the case of the $ZZ$-type Pauli noise 
$\Lambda(\rho)=(1-p)\rho+p\,ZZ\rho ZZ$, 
each noisy $\mathrm{CZ}$ gate contributes a multiplicative factor
$(1-2p)$ whenever the propagated Pauli operator anti-commutes with the ZZ Pauli.  
Thus $\sigma(P,U)^{-1}=(1-2p)^{-k(P,U)}\leq(1-2p)^{-L(U)}$, where $k(P,U)$ counts the number of ZZ-error events along the Pauli propagation through $U$, and $L(U)$ counts the number of CZ gates in circuit $U$.  For random $U$, the variable $L(U)$ follows a binomial distribution 
$L(U)\sim\mathrm{Bin}(L,q)$, where 
$L=n(n-1)/2$ is the maximal number of $\mathrm{CZ}$ gates in the phase-shadow circuit.  Hence, one arrives at
\begin{equation}
\begin{aligned}
\mathbb{E}_U\bigl[\sigma(P,U)^{-4}\bigr]\leq\mathbb{E}_U\bigl[(1-2p)^{-4L(U)}\bigr]
  &= 2^{-L}\sum_{q=0}^{L} (1-2p)^{-4q}
     C_L^q = \bigr[\frac{1+(1-2p)^{-4}}{2}\bigr]^L.
\end{aligned}
\label{Eq:BinMGF}
\end{equation}
For small noise $p\ll 1$, using the Taylor expansion $ (1-2p)^{-4} = 1 + 8p + O(p^2)$, one can bound that the estimation variance scales

\begin{equation}
\begin{aligned}
\bigr(\mathbb{E}_U\bigl[\sigma(P,U)^{-4}\bigr]\bigr)^{1/2}\leq \bigl[1+4p+O(p^2)\bigl]^{L/2}\lesssim e^{n^2p}.
\end{aligned}
\label{Eq:BinMGF2}
\end{equation}
which matches the square of the variance scaling $\exp(n^2 p_e/2)$ used in the main text up to constant factors in the exponent.

\begin{figure*}
    \centering
    \includegraphics[width=\linewidth]{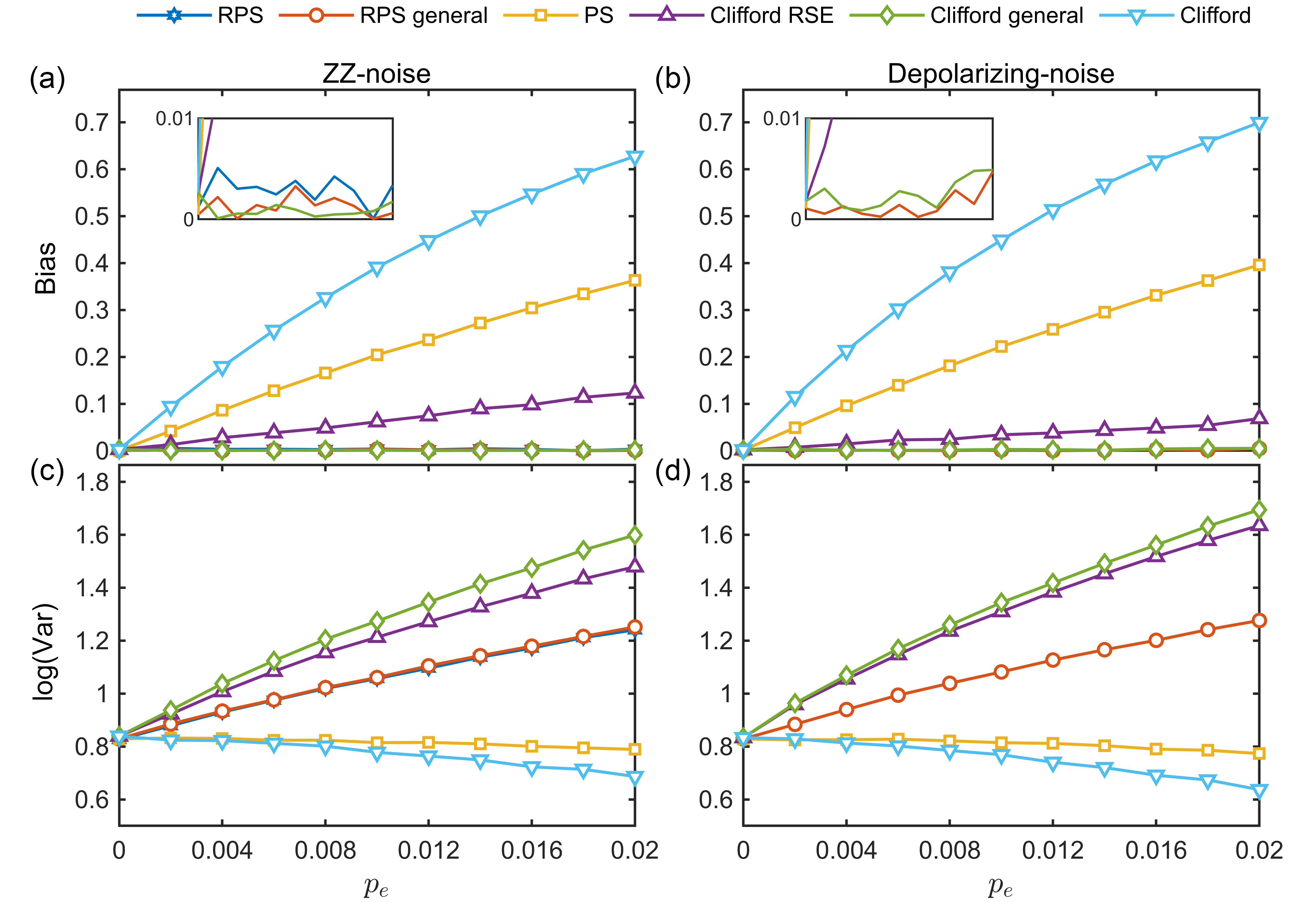}
    \caption{
    Numerical comparison of different shadow protocols under gate-dependent noise for $n=10$ random stabilizer states, with $N=10,\!000$ shots per state and per protocol, averaged over 100 random instances.  
    The upper panels show the average estimation bias as a function of the reference gate error rate $p_e$, and the lower panels show the logarithm of the empirical variance.  
    (a,c) Gate-dependent $ZZ$ Pauli noise acting only on two-qubit ${\rm CZ}$ gates.  
    (b,d) Gate-dependent depolarizing noise acting on both single- and two-qubit gates, with position-dependent error rate.  }
    \label{fig:ApGateDep}
\end{figure*}

\subsection*{C. Details on the numerical results}

To rigorously validate the efficacy of the proposed framework, we present a comprehensive benchmarking analysis against existing protocols. We choose the input state and the observable to be the same random stabilizer state, $\rho = O$, on $n=10$ qubits. For each reference error rate $p_e \in \{0, 0.002, \dots, 0.02\}$, we sample 100 independent instances and collect $N=10,\!000$ measurement shots per protocol.

We employ two families of gate-dependent Pauli noise models to explicitly test robustness against realistic noises. For the Pauli-ZZ error model, we consider that only the two-qubit entangling gates are noisy. Specifically, each ${\rm CZ}$ gate is followed by a $ZZ$ channel $\Lambda_{ZZ}(\rho) = (1-p_e)\rho + p_e \, ZZ\,\rho\, ZZ$, while single-qubit gates are assumed to be ideal. For the gate-dependent depolarizing noise model, we consider a more challenging scenario where both single-qubit and two-qubit gates are noisy. To introduce strong heterogeneity, we assign position-dependent error rates: $p_e^{(2)}(i,j) = p_e\,[1+\cos (ij)]$ for two-qubit gates acting on qubits $i,j$, and $p_e^{(1)}(i) = \frac{p_e}{10}\,[1+\cos (i^2)]$ for single-qubit gates on qubit $i$. We emphasize that under this depolarizing model, the error rates vary significantly across different gates in the device. This setup mimics the non-uniform noise landscape of realistic hardware, creating a regime where standard approximations (which assume uniform error rates) are expected to fail.

We compare the following six schemes:

(i) the original RPS protocol tailored only to $ZZ$ noise (“RPS”); 

(ii) the proposed generalized robust phase shadow on phase-shadow circuits (“generalized RPS”); 

(iii) the unmitigated phase shadow protocol (“PS”); 

(iv) the robust shadow estimation protocol of Ref.~\cite{chen2021robust} (“RSE”), implemented with an explicit calibration step; 

(v) the Clifford variant of generalized RPS (“Clifford general”), obtained by applying the generalized RPS formalism to random Clifford measurement circuits; and 

(vi) the unmitigated shadow estimation with random Clifford measurement (“Clifford”).

We summarize the compared schemes in Table~\ref{tab:schemes}.

\begin{table}[h]
\centering
\caption{Summary of shadow estimation protocols compared in~\cref{fig:ApGateDep}.}
\label{tab:schemes}
\begin{tabular}{l l l}
\hline \hline
\textbf{Abbr.} & \textbf{Circuit Ensemble} & \textbf{Mitigation Strategy} \\ \hline
RPS & Phase Shadow & Original RPS for $ZZ$-noise inversion (Ours) \\
Generalized RPS & Phase shadow & Generalized gate-dependent noise inversion (Ours) \\
PS & Phase shadow & Unmitigated \\ \hline
RSE & Clifford & Gate-independent noise inversion~\cite{chen2021robust} \\
Clifford general & Clifford & Generalized gate-dependent noise inversion (Ours) \\
Clifford & Clifford & Unmitigated \\
\hline \hline
\end{tabular}
\end{table}

For RSE, we use $10^6$ calibration shots per setting, so that the calibration error is at the negligible $10^{-3}$ level.

We illustrate the numerical performance on the native phase-shadow circuit architecture in Fig.~\ref{fig:ApGateDep}(a) and (b). As shown in the bias plots, the proposed \textbf{generalized RPS} (orange) maintains strict unbiasedness across all error rates, effectively neutralizing both the coherent ZZ noise and the highly heterogeneous depolarizing noise. In contrast, the unmitigated baseline (yellow) exhibits a linear accumulation of bias. The insets display the logarithmic estimation variance, indicating that the Generalized RPS incurs a controlled sampling overhead to achieve this strictly unbiased recovery.

We further explicitly benchmark our method against the \emph{robust shadow estimation} (RSE) protocol~\cite{chen2021robust}. The results highlight the critical advantage of our approach under realistic conditions. While RSE (purple) suppresses noise compared to the unmitigated baseline, it retains a measurable residual systematic bias. This occurs because standard RSE relies on a gate-independent assumption. In contrast, our generalized Clifford estimator (orange, green) accounts for the specific error channel of each individual gate, thereby eliminating the residual bias completely.

To demonstrate the universality of our framework, we extend the numerical validation beyond the native -S-CZ-H- structure to standard random Clifford measurements. This step is crucial to show that the proposed generalized estimator is not limited to specific circuits but is applicable to arbitrary ensembles. By applying our generalized formalism to the random Clifford ensemble (green), we verify that the property of strict unbiasedness holds regardless of the underlying circuit structure. Furthermore, as shown in Fig.~\ref{fig:ApGateDep}(c, d), the generalized RPS (orange) exhibits a significantly lower variance compared to the generalized Clifford estimator. This advantage arises because, compared with Clifford circuits, the phase shadow ensemble features a reduction in circuit depth.

\section*{Supplementary Note 11. --Comment on fermionic analogues}

While the scheme of the main text has been concerned with qubits, as they are common in a wealth of physical platforms, it should be 
clear that one can also devise fermionic analogues \cite{MatchgateShadows,MatchgateHelsen}. The analog of Clifford circuits in this setting are matchgate circuits or fermionic Gaussian unitaries.
For randomly picked such unitaries, one can show that the 
first three moments of the 
Haar distribution over the continuous group of matchgate circuits are equal to those of the discrete uniform distribution over only the matchgate circuits that are also Clifford unitaries. For this reason, the latter forms what is called a matchgate 3-design \cite{MatchgateShadows}, useful for fermionic classical shadows. A reasonable analog of the present scheme would 
be one that takes the form that first a layer of gates is applied that act as $U(f_1,\dots , 
f_n)^T$ on the fermionic operators $f_1,\dots f_n$ as 
\begin{equation}
U=\frac{1}{2}(\id + i Y)
 \oplus \id,
\end{equation} 
where non-trivial unitaries are applied to any pairs of modes $1,\dots, n$ and the identity otherwise. This is being followed by a diagonal 
unitary $V={\rm diag}(e^{i\theta_1},\dots,e^{i\theta_n})$ for real $\theta_1,\dots, \theta_n$, so that the overall circuit is captured by $UV$. Similarly as above, such restricted fermionic circuits give rise to classical shadows.

\end{appendix}

\input{output.bbl}

\end{document}

%% file: output.bbl
%

%% file: arxiv_2.bbl
\begin{thebibliography}{74}%
\makeatletter
\providecommand \@ifxundefined [1]{%
 \@ifx{#1\undefined}
}%
\providecommand \@ifnum [1]{%
 \ifnum #1\expandafter \@firstoftwo
 \else \expandafter \@secondoftwo
 \fi
}%
\providecommand \@ifx [1]{%
 \ifx #1\expandafter \@firstoftwo
 \else \expandafter \@secondoftwo
 \fi
}%
\providecommand \natexlab [1]{#1}%
\providecommand \enquote  [1]{``#1''}%
\providecommand \bibnamefont  [1]{#1}%
\providecommand \bibfnamefont [1]{#1}%
\providecommand \citenamefont [1]{#1}%
\providecommand \href@noop [0]{\@secondoftwo}%
\providecommand \href [0]{\begingroup \@sanitize@url \@href}%
\providecommand \@href[1]{\@@startlink{#1}\@@href}%
\providecommand \@@href[1]{\endgroup#1\@@endlink}%
\providecommand \@sanitize@url [0]{\catcode `\\12\catcode `\$12\catcode `\&12\catcode `\#12\catcode `\^12\catcode `\_12\catcode `\%12\relax}%
\providecommand \@@startlink[1]{}%
\providecommand \@@endlink[0]{}%
\providecommand \url  [0]{\begingroup\@sanitize@url \@url }%
\providecommand \@url [1]{\endgroup\@href {#1}{\urlprefix }}%
\providecommand \urlprefix  [0]{URL }%
\providecommand \Eprint [0]{\href }%
\providecommand \doibase [0]{https://doi.org/}%
\providecommand \selectlanguage [0]{\@gobble}%
\providecommand \bibinfo  [0]{\@secondoftwo}%
\providecommand \bibfield  [0]{\@secondoftwo}%
\providecommand \translation [1]{[#1]}%
\providecommand \BibitemOpen [0]{}%
\providecommand \bibitemStop [0]{}%
\providecommand \bibitemNoStop [0]{.\EOS\space}%
\providecommand \EOS [0]{\spacefactor3000\relax}%
\providecommand \BibitemShut  [1]{\csname bibitem#1\endcsname}%
\let\auto@bib@innerbib\@empty
\bibitem [{\citenamefont {Eisert}\ \emph {et~al.}(2020)\citenamefont {Eisert}, \citenamefont {Hangleiter}, \citenamefont {Walk}, \citenamefont {Roth}, \citenamefont {Markham}, \citenamefont {Parekh}, \citenamefont {Chabaud},\ and\ \citenamefont {Kashefi}}]{eisert2020quantum}%
  \BibitemOpen
  \bibfield  {author} {\bibinfo {author} {\bibfnamefont {J.}~\bibnamefont {Eisert}}, \bibinfo {author} {\bibfnamefont {D.}~\bibnamefont {Hangleiter}}, \bibinfo {author} {\bibfnamefont {N.}~\bibnamefont {Walk}}, \bibinfo {author} {\bibfnamefont {I.}~\bibnamefont {Roth}}, \bibinfo {author} {\bibfnamefont {D.}~\bibnamefont {Markham}}, \bibinfo {author} {\bibfnamefont {R.}~\bibnamefont {Parekh}}, \bibinfo {author} {\bibfnamefont {U.}~\bibnamefont {Chabaud}},\ and\ \bibinfo {author} {\bibfnamefont {E.}~\bibnamefont {Kashefi}},\ }\bibfield  {title} {\bibinfo {title} {Quantum certification and benchmarking},\ }\href {https://www.nature.com/articles/s42254-020-0186-4} {\bibfield  {journal} {\bibinfo  {journal} {Nature Rev. Phys.}\ }\textbf {\bibinfo {volume} {2}},\ \bibinfo {pages} {382} (\bibinfo {year} {2020})}\BibitemShut {NoStop}%
\bibitem [{\citenamefont {Kliesch}\ and\ \citenamefont {Roth}(2021)}]{kliesch2021theory}%
  \BibitemOpen
  \bibfield  {author} {\bibinfo {author} {\bibfnamefont {M.}~\bibnamefont {Kliesch}}\ and\ \bibinfo {author} {\bibfnamefont {I.}~\bibnamefont {Roth}},\ }\bibfield  {title} {\bibinfo {title} {Theory of quantum system certification},\ }\href {https://journals.aps.org/prxquantum/abstract/10.1103/PRXQuantum.2.010201} {\bibfield  {journal} {\bibinfo  {journal} {PRX quantum}\ }\textbf {\bibinfo {volume} {2}},\ \bibinfo {pages} {010201} (\bibinfo {year} {2021})}\BibitemShut {NoStop}%
\bibitem [{\citenamefont {Ohliger}\ \emph {et~al.}(2013)\citenamefont {Ohliger}, \citenamefont {Nesme},\ and\ \citenamefont {Eisert}}]{Efficient}%
  \BibitemOpen
  \bibfield  {author} {\bibinfo {author} {\bibfnamefont {M.}~\bibnamefont {Ohliger}}, \bibinfo {author} {\bibfnamefont {V.}~\bibnamefont {Nesme}},\ and\ \bibinfo {author} {\bibfnamefont {J.}~\bibnamefont {Eisert}},\ }\bibfield  {title} {\bibinfo {title} {Efficient and feasible state tomography of quantum many-body systems},\ }\href {https://doi.org/10.48550/arXiv.1204.5735} {\bibfield  {journal} {\bibinfo  {journal} {New J. Phys.}\ }\textbf {\bibinfo {volume} {15}},\ \bibinfo {pages} {015024} (\bibinfo {year} {2013})}\BibitemShut {NoStop}%
\bibitem [{\citenamefont {Elben}\ \emph {et~al.}(2023)\citenamefont {Elben}, \citenamefont {Flammia}, \citenamefont {Huang}, \citenamefont {Kueng}, \citenamefont {Preskill}, \citenamefont {Vermersch},\ and\ \citenamefont {Zoller}}]{elben2023randomized}%
  \BibitemOpen
  \bibfield  {author} {\bibinfo {author} {\bibfnamefont {A.}~\bibnamefont {Elben}}, \bibinfo {author} {\bibfnamefont {S.~T.}\ \bibnamefont {Flammia}}, \bibinfo {author} {\bibfnamefont {H.-Y.}\ \bibnamefont {Huang}}, \bibinfo {author} {\bibfnamefont {R.}~\bibnamefont {Kueng}}, \bibinfo {author} {\bibfnamefont {J.}~\bibnamefont {Preskill}}, \bibinfo {author} {\bibfnamefont {B.}~\bibnamefont {Vermersch}},\ and\ \bibinfo {author} {\bibfnamefont {P.}~\bibnamefont {Zoller}},\ }\bibfield  {title} {\bibinfo {title} {The randomized measurement toolbox},\ }\href {https://www.nature.com/articles/s42254-022-00535-2} {\bibfield  {journal} {\bibinfo  {journal} {Nature Rev. Phys.}\ }\textbf {\bibinfo {volume} {5}},\ \bibinfo {pages} {9} (\bibinfo {year} {2023})}\BibitemShut {NoStop}%
\bibitem [{\citenamefont {Cie{\'s}li{\'n}ski}\ \emph {et~al.}(2024)\citenamefont {Cie{\'s}li{\'n}ski}, \citenamefont {Imai}, \citenamefont {Dziewior}, \citenamefont {G{\"u}hne}, \citenamefont {Knips}, \citenamefont {Laskowski}, \citenamefont {Meinecke}, \citenamefont {Paterek},\ and\ \citenamefont {V{\'e}rtesi}}]{cieslinski2024analysing}%
  \BibitemOpen
  \bibfield  {author} {\bibinfo {author} {\bibfnamefont {P.}~\bibnamefont {Cie{\'s}li{\'n}ski}}, \bibinfo {author} {\bibfnamefont {S.}~\bibnamefont {Imai}}, \bibinfo {author} {\bibfnamefont {J.}~\bibnamefont {Dziewior}}, \bibinfo {author} {\bibfnamefont {O.}~\bibnamefont {G{\"u}hne}}, \bibinfo {author} {\bibfnamefont {L.}~\bibnamefont {Knips}}, \bibinfo {author} {\bibfnamefont {W.}~\bibnamefont {Laskowski}}, \bibinfo {author} {\bibfnamefont {J.}~\bibnamefont {Meinecke}}, \bibinfo {author} {\bibfnamefont {T.}~\bibnamefont {Paterek}},\ and\ \bibinfo {author} {\bibfnamefont {T.}~\bibnamefont {V{\'e}rtesi}},\ }\bibfield  {title} {\bibinfo {title} {Analysing quantum systems with randomised measurements},\ }\href {https://www.sciencedirect.com/science/article/pii/S0370157324003326} {\bibfield  {journal} {\bibinfo  {journal} {Phys. Rep.}\ }\textbf {\bibinfo {volume} {1095}},\ \bibinfo {pages} {1} (\bibinfo {year} {2024})}\BibitemShut {NoStop}%
\bibitem [{\citenamefont {Aaronson}(2019)}]{aaronson2019shadow}%
  \BibitemOpen
  \bibfield  {author} {\bibinfo {author} {\bibfnamefont {S.}~\bibnamefont {Aaronson}},\ }\bibfield  {title} {\bibinfo {title} {Shadow tomography of quantum states},\ }\href {https://epubs.siam.org/doi/abs/10.1137/18M120275X} {\bibfield  {journal} {\bibinfo  {journal} {SIAM J. Comp.}\ }\textbf {\bibinfo {volume} {49}},\ \bibinfo {pages} {STOC18} (\bibinfo {year} {2019})}\BibitemShut {NoStop}%
\bibitem [{\citenamefont {Huang}\ \emph {et~al.}(2020)\citenamefont {Huang}, \citenamefont {Kueng},\ and\ \citenamefont {Preskill}}]{huang2020predicting}%
  \BibitemOpen
  \bibfield  {author} {\bibinfo {author} {\bibfnamefont {H.-Y.}\ \bibnamefont {Huang}}, \bibinfo {author} {\bibfnamefont {R.}~\bibnamefont {Kueng}},\ and\ \bibinfo {author} {\bibfnamefont {J.}~\bibnamefont {Preskill}},\ }\bibfield  {title} {\bibinfo {title} {Predicting many properties of a quantum system from very few measurements},\ }\href {https://www.nature.com/articles/s41567-020-0932-7} {\bibfield  {journal} {\bibinfo  {journal} {Nature Phys.}\ }\textbf {\bibinfo {volume} {16}},\ \bibinfo {pages} {1050} (\bibinfo {year} {2020})}\BibitemShut {NoStop}%
\bibitem [{\citenamefont {Hu}\ and\ \citenamefont {You}(2022)}]{Hu2022Hamiltonian}%
  \BibitemOpen
  \bibfield  {author} {\bibinfo {author} {\bibfnamefont {H.-Y.}\ \bibnamefont {Hu}}\ and\ \bibinfo {author} {\bibfnamefont {Y.-Z.}\ \bibnamefont {You}},\ }\bibfield  {title} {\bibinfo {title} {Hamiltonian-driven shadow tomography of quantum states},\ }\href {https://doi.org/10.1103/PhysRevResearch.4.013054} {\bibfield  {journal} {\bibinfo  {journal} {Phys. Rev. Research}\ }\textbf {\bibinfo {volume} {4}},\ \bibinfo {pages} {013054} (\bibinfo {year} {2022})}\BibitemShut {NoStop}%
\bibitem [{\citenamefont {Hu}\ \emph {et~al.}(2023)\citenamefont {Hu}, \citenamefont {Choi},\ and\ \citenamefont {You}}]{hu2023classical}%
  \BibitemOpen
  \bibfield  {author} {\bibinfo {author} {\bibfnamefont {H.-Y.}\ \bibnamefont {Hu}}, \bibinfo {author} {\bibfnamefont {S.}~\bibnamefont {Choi}},\ and\ \bibinfo {author} {\bibfnamefont {Y.-Z.}\ \bibnamefont {You}},\ }\bibfield  {title} {\bibinfo {title} {Classical shadow tomography with locally scrambled quantum dynamics},\ }\href {https://journals.aps.org/prresearch/abstract/10.1103/PhysRevResearch.5.023027} {\bibfield  {journal} {\bibinfo  {journal} {Phys. Rev. Res.}\ }\textbf {\bibinfo {volume} {5}},\ \bibinfo {pages} {023027} (\bibinfo {year} {2023})}\BibitemShut {NoStop}%
\bibitem [{\citenamefont {Zhang}\ \emph {et~al.}(2024)\citenamefont {Zhang}, \citenamefont {Liu},\ and\ \citenamefont {Zhou}}]{zhang2024minimal}%
  \BibitemOpen
  \bibfield  {author} {\bibinfo {author} {\bibfnamefont {Q.}~\bibnamefont {Zhang}}, \bibinfo {author} {\bibfnamefont {Q.}~\bibnamefont {Liu}},\ and\ \bibinfo {author} {\bibfnamefont {Y.}~\bibnamefont {Zhou}},\ }\bibfield  {title} {\bibinfo {title} {{Minimal-Clifford shadow estimation by mutually unbiased bases}},\ }\href {https://journals.aps.org/prapplied/abstract/10.1103/PhysRevApplied.21.064001} {\bibfield  {journal} {\bibinfo  {journal} {Phys. Rev. Appl.}\ }\textbf {\bibinfo {volume} {21}},\ \bibinfo {pages} {064001} (\bibinfo {year} {2024})}\BibitemShut {NoStop}%
\bibitem [{\citenamefont {Park}\ \emph {et~al.}(2023)\citenamefont {Park}, \citenamefont {Teo},\ and\ \citenamefont {Jeong}}]{park2023resource}%
  \BibitemOpen
  \bibfield  {author} {\bibinfo {author} {\bibfnamefont {G.}~\bibnamefont {Park}}, \bibinfo {author} {\bibfnamefont {Y.~S.}\ \bibnamefont {Teo}},\ and\ \bibinfo {author} {\bibfnamefont {H.}~\bibnamefont {Jeong}},\ }\bibfield  {title} {\bibinfo {title} {Resource-efficient shadow tomography using equatorial measurements},\ }\href {https://inspirehep.net/literature/2726189} {\bibfield  {journal} {\bibinfo  {journal} {arXiv:2311.14622}\ } (\bibinfo {year} {2023})}\BibitemShut {NoStop}%
\bibitem [{\citenamefont {Zhang}\ \emph {et~al.}(2021)\citenamefont {Zhang}, \citenamefont {Sun}, \citenamefont {Fang}, \citenamefont {Zhang}, \citenamefont {Yuan},\ and\ \citenamefont {Lu}}]{zhang2021experimental}%
  \BibitemOpen
  \bibfield  {author} {\bibinfo {author} {\bibfnamefont {T.}~\bibnamefont {Zhang}}, \bibinfo {author} {\bibfnamefont {J.}~\bibnamefont {Sun}}, \bibinfo {author} {\bibfnamefont {X.-X.}\ \bibnamefont {Fang}}, \bibinfo {author} {\bibfnamefont {X.-M.}\ \bibnamefont {Zhang}}, \bibinfo {author} {\bibfnamefont {X.}~\bibnamefont {Yuan}},\ and\ \bibinfo {author} {\bibfnamefont {H.}~\bibnamefont {Lu}},\ }\bibfield  {title} {\bibinfo {title} {Experimental quantum state measurement with classical shadows},\ }\href {https://journals.aps.org/prl/abstract/10.1103/PhysRevLett.127.200501} {\bibfield  {journal} {\bibinfo  {journal} {Phys. Rev. Lett.}\ }\textbf {\bibinfo {volume} {127}},\ \bibinfo {pages} {200501} (\bibinfo {year} {2021})}\BibitemShut {NoStop}%
\bibitem [{\citenamefont {Hadfield}\ \emph {et~al.}(2022)\citenamefont {Hadfield}, \citenamefont {Bravyi}, \citenamefont {Raymond},\ and\ \citenamefont {Mezzacapo}}]{hadfield2022measurements}%
  \BibitemOpen
  \bibfield  {author} {\bibinfo {author} {\bibfnamefont {C.}~\bibnamefont {Hadfield}}, \bibinfo {author} {\bibfnamefont {S.}~\bibnamefont {Bravyi}}, \bibinfo {author} {\bibfnamefont {R.}~\bibnamefont {Raymond}},\ and\ \bibinfo {author} {\bibfnamefont {A.}~\bibnamefont {Mezzacapo}},\ }\bibfield  {title} {\bibinfo {title} {{Measurements of quantum Hamiltonians with locally-biased classical shadows}},\ }\href {https://link.springer.com/article/10.1007/s00220-022-04343-8} {\bibfield  {journal} {\bibinfo  {journal} {Comm. Math. Phys.}\ }\textbf {\bibinfo {volume} {391}},\ \bibinfo {pages} {951} (\bibinfo {year} {2022})}\BibitemShut {NoStop}%
\bibitem [{\citenamefont {Nguyen}\ \emph {et~al.}(2022)\citenamefont {Nguyen}, \citenamefont {B{\"o}nsel}, \citenamefont {Steinberg},\ and\ \citenamefont {G{\"u}hne}}]{nguyen2022optimizing}%
  \BibitemOpen
  \bibfield  {author} {\bibinfo {author} {\bibfnamefont {H.~C.}\ \bibnamefont {Nguyen}}, \bibinfo {author} {\bibfnamefont {J.~L.}\ \bibnamefont {B{\"o}nsel}}, \bibinfo {author} {\bibfnamefont {J.}~\bibnamefont {Steinberg}},\ and\ \bibinfo {author} {\bibfnamefont {O.}~\bibnamefont {G{\"u}hne}},\ }\bibfield  {title} {\bibinfo {title} {Optimizing shadow tomography with generalized measurements},\ }\href {https://journals.aps.org/prl/abstract/10.1103/PhysRevLett.129.220502} {\bibfield  {journal} {\bibinfo  {journal} {Phys. Rev. Lett.}\ }\textbf {\bibinfo {volume} {129}},\ \bibinfo {pages} {220502} (\bibinfo {year} {2022})}\BibitemShut {NoStop}%
\bibitem [{\citenamefont {Wu}\ \emph {et~al.}(2023)\citenamefont {Wu}, \citenamefont {Sun}, \citenamefont {Huang},\ and\ \citenamefont {Yuan}}]{wu2023overlapped}%
  \BibitemOpen
  \bibfield  {author} {\bibinfo {author} {\bibfnamefont {B.}~\bibnamefont {Wu}}, \bibinfo {author} {\bibfnamefont {J.}~\bibnamefont {Sun}}, \bibinfo {author} {\bibfnamefont {Q.}~\bibnamefont {Huang}},\ and\ \bibinfo {author} {\bibfnamefont {X.}~\bibnamefont {Yuan}},\ }\bibfield  {title} {\bibinfo {title} {Overlapped grouping measurement: A unified framework for measuring quantum states},\ }\href {https://doi.org/10.22331/q-2023-01-13-896} {\bibfield  {journal} {\bibinfo  {journal} {Quantum}\ }\textbf {\bibinfo {volume} {7}},\ \bibinfo {pages} {896} (\bibinfo {year} {2023})}\BibitemShut {NoStop}%
\bibitem [{\citenamefont {Van~Kirk}\ \emph {et~al.}(2022)\citenamefont {Van~Kirk}, \citenamefont {Cotler}, \citenamefont {Huang},\ and\ \citenamefont {Lukin}}]{van2022hardware}%
  \BibitemOpen
  \bibfield  {author} {\bibinfo {author} {\bibfnamefont {K.}~\bibnamefont {Van~Kirk}}, \bibinfo {author} {\bibfnamefont {J.}~\bibnamefont {Cotler}}, \bibinfo {author} {\bibfnamefont {H.-Y.}\ \bibnamefont {Huang}},\ and\ \bibinfo {author} {\bibfnamefont {M.~D.}\ \bibnamefont {Lukin}},\ }\bibfield  {title} {\bibinfo {title} {Hardware-efficient learning of quantum many-body states},\ }\href@noop {} {\bibfield  {journal} {\bibinfo  {journal} {arXiv preprint arXiv:2212.06084}\ } (\bibinfo {year} {2022})}\BibitemShut {NoStop}%
\bibitem [{\citenamefont {Zhou}\ and\ \citenamefont {Liu}(2023)}]{liu2023group}%
  \BibitemOpen
  \bibfield  {author} {\bibinfo {author} {\bibfnamefont {Y.}~\bibnamefont {Zhou}}\ and\ \bibinfo {author} {\bibfnamefont {Q.}~\bibnamefont {Liu}},\ }\bibfield  {title} {\bibinfo {title} {Performance analysis of multi-shot shadow estimation},\ }\href {https://quantum-journal.org/papers/q-2023-06-29-1044/pdf/} {\bibfield  {journal} {\bibinfo  {journal} {Quantum}\ }\textbf {\bibinfo {volume} {7}},\ \bibinfo {pages} {1044} (\bibinfo {year} {2023})}\BibitemShut {NoStop}%
\bibitem [{\citenamefont {Gresch}\ and\ \citenamefont {Kliesch}(2025)}]{gresch2025guaranteed}%
  \BibitemOpen
  \bibfield  {author} {\bibinfo {author} {\bibfnamefont {A.}~\bibnamefont {Gresch}}\ and\ \bibinfo {author} {\bibfnamefont {M.}~\bibnamefont {Kliesch}},\ }\bibfield  {title} {\bibinfo {title} {{Guaranteed efficient energy estimation of quantum many-body Hamiltonians using ShadowGrouping}},\ }\href {https://www.nature.com/articles/s41467-024-54859-x.pdf} {\bibfield  {journal} {\bibinfo  {journal} {Nature Comm.}\ }\textbf {\bibinfo {volume} {16}},\ \bibinfo {pages} {689} (\bibinfo {year} {2025})}\BibitemShut {NoStop}%
\bibitem [{\citenamefont {Helsen}\ and\ \citenamefont {Walter}(2023)}]{helsen2023thrifty}%
  \BibitemOpen
  \bibfield  {author} {\bibinfo {author} {\bibfnamefont {J.}~\bibnamefont {Helsen}}\ and\ \bibinfo {author} {\bibfnamefont {M.}~\bibnamefont {Walter}},\ }\bibfield  {title} {\bibinfo {title} {{Thrifty shadow estimation: Reusing quantum circuits and bounding tails}},\ }\href {https://journals.aps.org/prl/abstract/10.1103/PhysRevLett.131.240602} {\bibfield  {journal} {\bibinfo  {journal} {Phys. Rev. Lett.}\ }\textbf {\bibinfo {volume} {131}},\ \bibinfo {pages} {240602} (\bibinfo {year} {2023})}\BibitemShut {NoStop}%
\bibitem [{\citenamefont {Mao}\ \emph {et~al.}(2025)\citenamefont {Mao}, \citenamefont {Yi},\ and\ \citenamefont {Zhu}}]{mao2025qudit}%
  \BibitemOpen
  \bibfield  {author} {\bibinfo {author} {\bibfnamefont {C.}~\bibnamefont {Mao}}, \bibinfo {author} {\bibfnamefont {C.}~\bibnamefont {Yi}},\ and\ \bibinfo {author} {\bibfnamefont {H.}~\bibnamefont {Zhu}},\ }\bibfield  {title} {\bibinfo {title} {Qudit shadow estimation based on the clifford group and the power of a single magic gate},\ }\href {https://doi.org/10.1103/PhysRevLett.134.160801} {\bibfield  {journal} {\bibinfo  {journal} {Phys. Rev. Lett.}\ }\textbf {\bibinfo {volume} {134}},\ \bibinfo {pages} {160801} (\bibinfo {year} {2025})}\BibitemShut {NoStop}%
\bibitem [{\citenamefont {Bertoni}\ \emph {et~al.}(2024)\citenamefont {Bertoni}, \citenamefont {Haferkamp}, \citenamefont {Hinsche}, \citenamefont {Ioannou}, \citenamefont {Eisert},\ and\ \citenamefont {Pashayan}}]{bertoni2024shallow}%
  \BibitemOpen
  \bibfield  {author} {\bibinfo {author} {\bibfnamefont {C.}~\bibnamefont {Bertoni}}, \bibinfo {author} {\bibfnamefont {J.}~\bibnamefont {Haferkamp}}, \bibinfo {author} {\bibfnamefont {M.}~\bibnamefont {Hinsche}}, \bibinfo {author} {\bibfnamefont {M.}~\bibnamefont {Ioannou}}, \bibinfo {author} {\bibfnamefont {J.}~\bibnamefont {Eisert}},\ and\ \bibinfo {author} {\bibfnamefont {H.}~\bibnamefont {Pashayan}},\ }\bibfield  {title} {\bibinfo {title} {{Shallow shadows: Expectation estimation using low-depth random Clifford circuits}},\ }\href {https://journals.aps.org/prl/abstract/10.1103/PhysRevLett.133.020602} {\bibfield  {journal} {\bibinfo  {journal} {Phys. Rev. Lett.}\ }\textbf {\bibinfo {volume} {133}},\ \bibinfo {pages} {020602} (\bibinfo {year} {2024})}\BibitemShut {NoStop}%
\bibitem [{\citenamefont {Akhtar}\ \emph {et~al.}(2023)\citenamefont {Akhtar}, \citenamefont {Hu},\ and\ \citenamefont {You}}]{akhtar2023scalable}%
  \BibitemOpen
  \bibfield  {author} {\bibinfo {author} {\bibfnamefont {A.~A.}\ \bibnamefont {Akhtar}}, \bibinfo {author} {\bibfnamefont {H.-Y.}\ \bibnamefont {Hu}},\ and\ \bibinfo {author} {\bibfnamefont {Y.-Z.}\ \bibnamefont {You}},\ }\bibfield  {title} {\bibinfo {title} {Scalable and flexible classical shadow tomography with tensor networks},\ }\href {https://quantum-journal.org/papers/q-2023-06-01-1026/} {\bibfield  {journal} {\bibinfo  {journal} {Quantum}\ }\textbf {\bibinfo {volume} {7}},\ \bibinfo {pages} {1026} (\bibinfo {year} {2023})}\BibitemShut {NoStop}%
\bibitem [{\citenamefont {Schuster}\ \emph {et~al.}(2025)\citenamefont {Schuster}, \citenamefont {Haferkamp},\ and\ \citenamefont {Huang}}]{schuster2024random}%
  \BibitemOpen
  \bibfield  {author} {\bibinfo {author} {\bibfnamefont {T.}~\bibnamefont {Schuster}}, \bibinfo {author} {\bibfnamefont {J.}~\bibnamefont {Haferkamp}},\ and\ \bibinfo {author} {\bibfnamefont {H.-Y.}\ \bibnamefont {Huang}},\ }\bibfield  {title} {\bibinfo {title} {Random unitaries in extremely low depth},\ }\href {https://www.science.org/doi/abs/10.1126/science.adv8590} {\bibfield  {journal} {\bibinfo  {journal} {Science}\ }\textbf {\bibinfo {volume} {389}},\ \bibinfo {pages} {92} (\bibinfo {year} {2025})}\BibitemShut {NoStop}%
\bibitem [{\citenamefont {Chen}\ \emph {et~al.}(2021)\citenamefont {Chen}, \citenamefont {Yu}, \citenamefont {Zeng},\ and\ \citenamefont {Flammia}}]{chen2021robust}%
  \BibitemOpen
  \bibfield  {author} {\bibinfo {author} {\bibfnamefont {S.}~\bibnamefont {Chen}}, \bibinfo {author} {\bibfnamefont {W.}~\bibnamefont {Yu}}, \bibinfo {author} {\bibfnamefont {P.}~\bibnamefont {Zeng}},\ and\ \bibinfo {author} {\bibfnamefont {S.~T.}\ \bibnamefont {Flammia}},\ }\bibfield  {title} {\bibinfo {title} {Robust shadow estimation},\ }\href {https://journals.aps.org/prxquantum/pdf/10.1103/PRXQuantum.2.030348} {\bibfield  {journal} {\bibinfo  {journal} {PRX Quantum}\ }\textbf {\bibinfo {volume} {2}},\ \bibinfo {pages} {030348} (\bibinfo {year} {2021})}\BibitemShut {NoStop}%
\bibitem [{\citenamefont {Koh}\ and\ \citenamefont {Grewal}(2022)}]{koh2022classical}%
  \BibitemOpen
  \bibfield  {author} {\bibinfo {author} {\bibfnamefont {D.~E.}\ \bibnamefont {Koh}}\ and\ \bibinfo {author} {\bibfnamefont {S.}~\bibnamefont {Grewal}},\ }\bibfield  {title} {\bibinfo {title} {Classical shadows with noise},\ }\href {http://quantum-journal.org/papers/q-2022-08-16-776/} {\bibfield  {journal} {\bibinfo  {journal} {Quantum}\ }\textbf {\bibinfo {volume} {6}},\ \bibinfo {pages} {776} (\bibinfo {year} {2022})}\BibitemShut {NoStop}%
\bibitem [{\citenamefont {Onorati}\ \emph {et~al.}(2024)\citenamefont {Onorati}, \citenamefont {Kitzinger}, \citenamefont {Helsen}, \citenamefont {Ioannou}, \citenamefont {Werner}, \citenamefont {Roth},\ and\ \citenamefont {Eisert}}]{onorati2024noise}%
  \BibitemOpen
  \bibfield  {author} {\bibinfo {author} {\bibfnamefont {E.}~\bibnamefont {Onorati}}, \bibinfo {author} {\bibfnamefont {J.}~\bibnamefont {Kitzinger}}, \bibinfo {author} {\bibfnamefont {J.}~\bibnamefont {Helsen}}, \bibinfo {author} {\bibfnamefont {M.}~\bibnamefont {Ioannou}}, \bibinfo {author} {\bibfnamefont {A.}~\bibnamefont {Werner}}, \bibinfo {author} {\bibfnamefont {I.}~\bibnamefont {Roth}},\ and\ \bibinfo {author} {\bibfnamefont {J.}~\bibnamefont {Eisert}},\ }\bibfield  {title} {\bibinfo {title} {Noise-mitigated randomized measurements and self-calibrating shadow estimation},\ }\href {https://arxiv.org/pdf/2403.04751} {\bibfield  {journal} {\bibinfo  {journal} {arXiv:2403.04751}\ } (\bibinfo {year} {2024})}\BibitemShut {NoStop}%
\bibitem [{\citenamefont {Hu}\ \emph {et~al.}(2025)\citenamefont {Hu}, \citenamefont {Gu}, \citenamefont {Majumder}, \citenamefont {Ren}, \citenamefont {Zhang}, \citenamefont {Wang}, \citenamefont {You}, \citenamefont {Minev}, \citenamefont {Yelin},\ and\ \citenamefont {Seif}}]{hu2025demonstration}%
  \BibitemOpen
  \bibfield  {author} {\bibinfo {author} {\bibfnamefont {H.-Y.}\ \bibnamefont {Hu}}, \bibinfo {author} {\bibfnamefont {A.}~\bibnamefont {Gu}}, \bibinfo {author} {\bibfnamefont {S.}~\bibnamefont {Majumder}}, \bibinfo {author} {\bibfnamefont {H.}~\bibnamefont {Ren}}, \bibinfo {author} {\bibfnamefont {Y.}~\bibnamefont {Zhang}}, \bibinfo {author} {\bibfnamefont {D.~S.}\ \bibnamefont {Wang}}, \bibinfo {author} {\bibfnamefont {Y.-Z.}\ \bibnamefont {You}}, \bibinfo {author} {\bibfnamefont {Z.}~\bibnamefont {Minev}}, \bibinfo {author} {\bibfnamefont {S.~F.}\ \bibnamefont {Yelin}},\ and\ \bibinfo {author} {\bibfnamefont {A.}~\bibnamefont {Seif}},\ }\bibfield  {title} {\bibinfo {title} {Demonstration of robust and efficient quantum property learning with shallow shadows},\ }\href {https://www.nature.com/articles/s41467-025-57349-w.pdf} {\bibfield  {journal} {\bibinfo  {journal} {Nature Comm.}\ }\textbf {\bibinfo {volume} {16}},\ \bibinfo {pages} {2943} (\bibinfo {year} {2025})}\BibitemShut {NoStop}%
\bibitem [{\citenamefont {Farias}\ \emph {et~al.}(2025)\citenamefont {Farias}, \citenamefont {Peddinti}, \citenamefont {Roth},\ and\ \citenamefont {Aolita}}]{farias2024robust}%
  \BibitemOpen
  \bibfield  {author} {\bibinfo {author} {\bibfnamefont {R.~M.~S.}\ \bibnamefont {Farias}}, \bibinfo {author} {\bibfnamefont {R.~D.}\ \bibnamefont {Peddinti}}, \bibinfo {author} {\bibfnamefont {I.}~\bibnamefont {Roth}},\ and\ \bibinfo {author} {\bibfnamefont {L.}~\bibnamefont {Aolita}},\ }\bibfield  {title} {\bibinfo {title} {Robust ultra-shallow shadows},\ }\href {https://doi.org/10.1088/2058-9565/adc14f} {\bibfield  {journal} {\bibinfo  {journal} {Quant. Sc. Tech.}\ }\textbf {\bibinfo {volume} {10}},\ \bibinfo {pages} {025044} (\bibinfo {year} {2025})}\BibitemShut {NoStop}%
\bibitem [{\citenamefont {Helsen}\ \emph {et~al.}(2022{\natexlab{a}})\citenamefont {Helsen}, \citenamefont {Roth}, \citenamefont {Onorati}, \citenamefont {Werner},\ and\ \citenamefont {Eisert}}]{helsen2022general}%
  \BibitemOpen
  \bibfield  {author} {\bibinfo {author} {\bibfnamefont {J.}~\bibnamefont {Helsen}}, \bibinfo {author} {\bibfnamefont {I.}~\bibnamefont {Roth}}, \bibinfo {author} {\bibfnamefont {E.}~\bibnamefont {Onorati}}, \bibinfo {author} {\bibfnamefont {A.~H.}\ \bibnamefont {Werner}},\ and\ \bibinfo {author} {\bibfnamefont {J.}~\bibnamefont {Eisert}},\ }\bibfield  {title} {\bibinfo {title} {General framework for randomized benchmarking},\ }\href {https://journals.aps.org/prxquantum/abstract/10.1103/PRXQuantum.3.020357} {\bibfield  {journal} {\bibinfo  {journal} {PRX Quantum}\ }\textbf {\bibinfo {volume} {3}},\ \bibinfo {pages} {020357} (\bibinfo {year} {2022}{\natexlab{a}})}\BibitemShut {NoStop}%
\bibitem [{\citenamefont {Liu}\ \emph {et~al.}(2024{\natexlab{a}})\citenamefont {Liu}, \citenamefont {Xie}, \citenamefont {Xu},\ and\ \citenamefont {Ma}}]{liu2024group}%
  \BibitemOpen
  \bibfield  {author} {\bibinfo {author} {\bibfnamefont {G.}~\bibnamefont {Liu}}, \bibinfo {author} {\bibfnamefont {Z.}~\bibnamefont {Xie}}, \bibinfo {author} {\bibfnamefont {Z.}~\bibnamefont {Xu}},\ and\ \bibinfo {author} {\bibfnamefont {X.}~\bibnamefont {Ma}},\ }\bibfield  {title} {\bibinfo {title} {Group twirling and noise tailoring for multiqubit controlled phase gates},\ }\href {https://journals.aps.org/prresearch/abstract/10.1103/PhysRevResearch.6.043221} {\bibfield  {journal} {\bibinfo  {journal} {Phys. Rev. Res.}\ }\textbf {\bibinfo {volume} {6}},\ \bibinfo {pages} {043221} (\bibinfo {year} {2024}{\natexlab{a}})}\BibitemShut {NoStop}%
\bibitem [{\citenamefont {Cioli}\ \emph {et~al.}(2024)\citenamefont {Cioli}, \citenamefont {Ercolessi}, \citenamefont {Ippoliti}, \citenamefont {Turkeshi},\ and\ \citenamefont {Piroli}}]{cioli2024approximate}%
  \BibitemOpen
  \bibfield  {author} {\bibinfo {author} {\bibfnamefont {R.}~\bibnamefont {Cioli}}, \bibinfo {author} {\bibfnamefont {E.}~\bibnamefont {Ercolessi}}, \bibinfo {author} {\bibfnamefont {M.}~\bibnamefont {Ippoliti}}, \bibinfo {author} {\bibfnamefont {X.}~\bibnamefont {Turkeshi}},\ and\ \bibinfo {author} {\bibfnamefont {L.}~\bibnamefont {Piroli}},\ }\bibfield  {title} {\bibinfo {title} {Approximate inverse measurement channel for shallow shadows},\ }\href {https://arxiv.org/pdf/2407.11813} {\bibfield  {journal} {\bibinfo  {journal} {arXiv:2407.11813}\ } (\bibinfo {year} {2024})}\BibitemShut {NoStop}%
\bibitem [{\citenamefont {Brieger}\ \emph {et~al.}(2025)\citenamefont {Brieger}, \citenamefont {Heinrich}, \citenamefont {Roth},\ and\ \citenamefont {Kliesch}}]{brieger2025stability}%
  \BibitemOpen
  \bibfield  {author} {\bibinfo {author} {\bibfnamefont {R.}~\bibnamefont {Brieger}}, \bibinfo {author} {\bibfnamefont {M.}~\bibnamefont {Heinrich}}, \bibinfo {author} {\bibfnamefont {I.}~\bibnamefont {Roth}},\ and\ \bibinfo {author} {\bibfnamefont {M.}~\bibnamefont {Kliesch}},\ }\bibfield  {title} {\bibinfo {title} {Stability of classical shadows under gate-dependent noise},\ }\href {https://journals.aps.org/prl/abstract/10.1103/PhysRevLett.134.090801} {\bibfield  {journal} {\bibinfo  {journal} {Phys. Rev. Lett.}\ }\textbf {\bibinfo {volume} {134}},\ \bibinfo {pages} {090801} (\bibinfo {year} {2025})}\BibitemShut {NoStop}%
\bibitem [{\citenamefont {Mandel}\ \emph {et~al.}(2003)\citenamefont {Mandel}, \citenamefont {Greiner}, \citenamefont {Widera}, \citenamefont {Rom}, \citenamefont {H{\"a}nsch},\ and\ \citenamefont {Bloch}}]{Mandel-Nature-2003}%
  \BibitemOpen
  \bibfield  {author} {\bibinfo {author} {\bibfnamefont {O.}~\bibnamefont {Mandel}}, \bibinfo {author} {\bibfnamefont {M.}~\bibnamefont {Greiner}}, \bibinfo {author} {\bibfnamefont {A.}~\bibnamefont {Widera}}, \bibinfo {author} {\bibfnamefont {T.}~\bibnamefont {Rom}}, \bibinfo {author} {\bibfnamefont {T.~W.}\ \bibnamefont {H{\"a}nsch}},\ and\ \bibinfo {author} {\bibfnamefont {I.}~\bibnamefont {Bloch}},\ }\bibfield  {title} {\bibinfo {title} {Controlled collisions for multi-particle entanglement of optically trapped atoms},\ }\href {https://doi.org/10.1038/nature02008} {\bibfield  {journal} {\bibinfo  {journal} {Nature}\ }\textbf {\bibinfo {volume} {425}},\ \bibinfo {pages} {937} (\bibinfo {year} {2003})}\BibitemShut {NoStop}%
\bibitem [{\citenamefont {Evered}\ \emph {et~al.}(2023)\citenamefont {Evered}, \citenamefont {Bluvstein}, \citenamefont {Kalinowski}, \citenamefont {Ebadi}, \citenamefont {Manovitz}, \citenamefont {Zhou}, \citenamefont {Li}, \citenamefont {Geim}, \citenamefont {Wang}, \citenamefont {Maskara} \emph {et~al.}}]{evered2023high}%
  \BibitemOpen
  \bibfield  {author} {\bibinfo {author} {\bibfnamefont {S.~J.}\ \bibnamefont {Evered}}, \bibinfo {author} {\bibfnamefont {D.}~\bibnamefont {Bluvstein}}, \bibinfo {author} {\bibfnamefont {M.}~\bibnamefont {Kalinowski}}, \bibinfo {author} {\bibfnamefont {S.}~\bibnamefont {Ebadi}}, \bibinfo {author} {\bibfnamefont {T.}~\bibnamefont {Manovitz}}, \bibinfo {author} {\bibfnamefont {H.}~\bibnamefont {Zhou}}, \bibinfo {author} {\bibfnamefont {S.~H.}\ \bibnamefont {Li}}, \bibinfo {author} {\bibfnamefont {A.~A.}\ \bibnamefont {Geim}}, \bibinfo {author} {\bibfnamefont {T.~T.}\ \bibnamefont {Wang}}, \bibinfo {author} {\bibfnamefont {N.}~\bibnamefont {Maskara}}, \emph {et~al.},\ }\bibfield  {title} {\bibinfo {title} {High-fidelity parallel entangling gates on a neutral-atom quantum computer},\ }\href {http://nature.com/articles/s41586-023-06481-y} {\bibfield  {journal} {\bibinfo  {journal} {Nature}\ }\textbf {\bibinfo {volume} {622}},\ \bibinfo {pages} {268} (\bibinfo {year} {2023})}\BibitemShut {NoStop}%
\bibitem [{\citenamefont {Figgatt}\ \emph {et~al.}(2019)\citenamefont {Figgatt}, \citenamefont {Ostrander}, \citenamefont {Linke}, \citenamefont {Landsman}, \citenamefont {Zhu}, \citenamefont {Maslov},\ and\ \citenamefont {Monroe}}]{figgatt2019parallel}%
  \BibitemOpen
  \bibfield  {author} {\bibinfo {author} {\bibfnamefont {C.}~\bibnamefont {Figgatt}}, \bibinfo {author} {\bibfnamefont {A.}~\bibnamefont {Ostrander}}, \bibinfo {author} {\bibfnamefont {N.~M.}\ \bibnamefont {Linke}}, \bibinfo {author} {\bibfnamefont {K.~A.}\ \bibnamefont {Landsman}}, \bibinfo {author} {\bibfnamefont {D.}~\bibnamefont {Zhu}}, \bibinfo {author} {\bibfnamefont {D.}~\bibnamefont {Maslov}},\ and\ \bibinfo {author} {\bibfnamefont {C.}~\bibnamefont {Monroe}},\ }\bibfield  {title} {\bibinfo {title} {Parallel entangling operations on a universal ion-trap quantum computer},\ }\href {https://www.nature.com/articles/s41586-019-1427-5} {\bibfield  {journal} {\bibinfo  {journal} {Nature}\ }\textbf {\bibinfo {volume} {572}},\ \bibinfo {pages} {368} (\bibinfo {year} {2019})}\BibitemShut {NoStop}%
\bibitem [{\citenamefont {Nakata}\ \emph {et~al.}(2014)\citenamefont {Nakata}, \citenamefont {Koashi},\ and\ \citenamefont {Murao}}]{nakata2014generating}%
  \BibitemOpen
  \bibfield  {author} {\bibinfo {author} {\bibfnamefont {Y.}~\bibnamefont {Nakata}}, \bibinfo {author} {\bibfnamefont {M.}~\bibnamefont {Koashi}},\ and\ \bibinfo {author} {\bibfnamefont {M.}~\bibnamefont {Murao}},\ }\bibfield  {title} {\bibinfo {title} {Generating a state t-design by diagonal quantum circuits},\ }\href {https://iopscience.iop.org/article/10.1088/1367-2630/16/5/053043/pdf} {\bibfield  {journal} {\bibinfo  {journal} {New J. Phys.}\ }\textbf {\bibinfo {volume} {16}},\ \bibinfo {pages} {053043} (\bibinfo {year} {2014})}\BibitemShut {NoStop}%
\bibitem [{\citenamefont {Nechita}\ and\ \citenamefont {Singh}(2021)}]{nechita2021graphical}%
  \BibitemOpen
  \bibfield  {author} {\bibinfo {author} {\bibfnamefont {I.}~\bibnamefont {Nechita}}\ and\ \bibinfo {author} {\bibfnamefont {S.}~\bibnamefont {Singh}},\ }\bibfield  {title} {\bibinfo {title} {A graphical calculus for integration over random diagonal unitary matrices},\ }\href {https://www.sciencedirect.com/science/article/pii/S0024379520305681} {\bibfield  {journal} {\bibinfo  {journal} {Lin. Alg. Appl.}\ }\textbf {\bibinfo {volume} {613}},\ \bibinfo {pages} {46} (\bibinfo {year} {2021})}\BibitemShut {NoStop}%
\bibitem [{\citenamefont {Liu}\ \emph {et~al.}(2024{\natexlab{b}})\citenamefont {Liu}, \citenamefont {Hao},\ and\ \citenamefont {Hu}}]{liu2024predicting}%
  \BibitemOpen
  \bibfield  {author} {\bibinfo {author} {\bibfnamefont {Z.}~\bibnamefont {Liu}}, \bibinfo {author} {\bibfnamefont {Z.}~\bibnamefont {Hao}},\ and\ \bibinfo {author} {\bibfnamefont {H.-Y.}\ \bibnamefont {Hu}},\ }\bibfield  {title} {\bibinfo {title} {Predicting arbitrary state properties from single hamiltonian quench dynamics},\ }\href {https://journals.aps.org/prresearch/pdf/10.1103/PhysRevResearch.6.043118} {\bibfield  {journal} {\bibinfo  {journal} {Phys. Rev. Res.}\ }\textbf {\bibinfo {volume} {6}},\ \bibinfo {pages} {043118} (\bibinfo {year} {2024}{\natexlab{b}})}\BibitemShut {NoStop}%
\bibitem [{\citenamefont {Mele}(2024)}]{mele2024introduction}%
  \BibitemOpen
  \bibfield  {author} {\bibinfo {author} {\bibfnamefont {A.~A.}\ \bibnamefont {Mele}},\ }\bibfield  {title} {\bibinfo {title} {Introduction to haar measure tools in quantum information: A beginner's tutorial},\ }\href {https://quantum-journal.org/papers/q-2024-05-08-1340/pdf/} {\bibfield  {journal} {\bibinfo  {journal} {Quantum}\ }\textbf {\bibinfo {volume} {8}},\ \bibinfo {pages} {1340} (\bibinfo {year} {2024})}\BibitemShut {NoStop}%
\bibitem [{\citenamefont {Bremner}\ \emph {et~al.}(2016)\citenamefont {Bremner}, \citenamefont {Montanaro},\ and\ \citenamefont {Shepherd}}]{bremner2016average}%
  \BibitemOpen
  \bibfield  {author} {\bibinfo {author} {\bibfnamefont {M.~J.}\ \bibnamefont {Bremner}}, \bibinfo {author} {\bibfnamefont {A.}~\bibnamefont {Montanaro}},\ and\ \bibinfo {author} {\bibfnamefont {D.~J.}\ \bibnamefont {Shepherd}},\ }\bibfield  {title} {\bibinfo {title} {Average-case complexity versus approximate simulation of commuting quantum computations},\ }\href {https://journals.aps.org/prl/abstract/10.1103/PhysRevLett.117.080501} {\bibfield  {journal} {\bibinfo  {journal} {Phys. Rev. Lett.}\ }\textbf {\bibinfo {volume} {117}},\ \bibinfo {pages} {080501} (\bibinfo {year} {2016})}\BibitemShut {NoStop}%
\bibitem [{\citenamefont {Zhou}\ \emph {et~al.}(2019)\citenamefont {Zhou}, \citenamefont {Zhao}, \citenamefont {Yuan},\ and\ \citenamefont {Ma}}]{zhou2019detecting}%
  \BibitemOpen
  \bibfield  {author} {\bibinfo {author} {\bibfnamefont {Y.}~\bibnamefont {Zhou}}, \bibinfo {author} {\bibfnamefont {Q.}~\bibnamefont {Zhao}}, \bibinfo {author} {\bibfnamefont {X.}~\bibnamefont {Yuan}},\ and\ \bibinfo {author} {\bibfnamefont {X.}~\bibnamefont {Ma}},\ }\bibfield  {title} {\bibinfo {title} {Detecting multipartite entanglement structure with minimal resources},\ }\href {https://www.nature.com/articles/s41534-019-0200-9.pdf} {\bibfield  {journal} {\bibinfo  {journal} {npj Quant. Inf.}\ }\textbf {\bibinfo {volume} {5}},\ \bibinfo {pages} {83} (\bibinfo {year} {2019})}\BibitemShut {NoStop}%
\bibitem [{\citenamefont {Gluza}\ \emph {et~al.}(2020)\citenamefont {Gluza}, \citenamefont {Schweigler}, \citenamefont {Rauer}, \citenamefont {Krumnow}, \citenamefont {Schmiedmayer},\ and\ \citenamefont {Eisert}}]{gluza2020quantum}%
  \BibitemOpen
  \bibfield  {author} {\bibinfo {author} {\bibfnamefont {M.}~\bibnamefont {Gluza}}, \bibinfo {author} {\bibfnamefont {T.}~\bibnamefont {Schweigler}}, \bibinfo {author} {\bibfnamefont {B.}~\bibnamefont {Rauer}}, \bibinfo {author} {\bibfnamefont {C.}~\bibnamefont {Krumnow}}, \bibinfo {author} {\bibfnamefont {J.}~\bibnamefont {Schmiedmayer}},\ and\ \bibinfo {author} {\bibfnamefont {J.}~\bibnamefont {Eisert}},\ }\bibfield  {title} {\bibinfo {title} {Quantum read-out for cold atomic quantum simulators},\ }\href {https://www.nature.com/articles/s42005-019-0273-y.pdf} {\bibfield  {journal} {\bibinfo  {journal} {Comm. Phys.}\ }\textbf {\bibinfo {volume} {3}},\ \bibinfo {pages} {12} (\bibinfo {year} {2020})}\BibitemShut {NoStop}%
\bibitem [{\citenamefont {Lotshaw}\ \emph {et~al.}(2023)\citenamefont {Lotshaw}, \citenamefont {Battles}, \citenamefont {Gard}, \citenamefont {Buchs}, \citenamefont {Humble},\ and\ \citenamefont {Herold}}]{lotshaw2023modeling}%
  \BibitemOpen
  \bibfield  {author} {\bibinfo {author} {\bibfnamefont {P.~C.}\ \bibnamefont {Lotshaw}}, \bibinfo {author} {\bibfnamefont {K.~D.}\ \bibnamefont {Battles}}, \bibinfo {author} {\bibfnamefont {B.}~\bibnamefont {Gard}}, \bibinfo {author} {\bibfnamefont {G.}~\bibnamefont {Buchs}}, \bibinfo {author} {\bibfnamefont {T.~S.}\ \bibnamefont {Humble}},\ and\ \bibinfo {author} {\bibfnamefont {C.~D.}\ \bibnamefont {Herold}},\ }\bibfield  {title} {\bibinfo {title} {{Modeling noise in global M{\o}lmer-S{\o}rensen interactions applied to quantum approximate optimization}},\ }\href {https://arxiv.org/pdf/2211.00133} {\bibfield  {journal} {\bibinfo  {journal} {Phys. Rev. A}\ }\textbf {\bibinfo {volume} {107}},\ \bibinfo {pages} {062406} (\bibinfo {year} {2023})}\BibitemShut {NoStop}%
\bibitem [{\citenamefont {Cong}\ \emph {et~al.}(2022)\citenamefont {Cong}, \citenamefont {Levine}, \citenamefont {Keesling}, \citenamefont {Bluvstein}, \citenamefont {Wang},\ and\ \citenamefont {Lukin}}]{cong2022hardware}%
  \BibitemOpen
  \bibfield  {author} {\bibinfo {author} {\bibfnamefont {I.}~\bibnamefont {Cong}}, \bibinfo {author} {\bibfnamefont {H.}~\bibnamefont {Levine}}, \bibinfo {author} {\bibfnamefont {A.}~\bibnamefont {Keesling}}, \bibinfo {author} {\bibfnamefont {D.}~\bibnamefont {Bluvstein}}, \bibinfo {author} {\bibfnamefont {S.-T.}\ \bibnamefont {Wang}},\ and\ \bibinfo {author} {\bibfnamefont {M.~D.}\ \bibnamefont {Lukin}},\ }\bibfield  {title} {\bibinfo {title} {{Hardware-efficient, fault-tolerant quantum computation with Rydberg atoms}},\ }\href {https://journals.aps.org/prx/pdf/10.1103/PhysRevX.12.021049} {\bibfield  {journal} {\bibinfo  {journal} {Phys. Rev. X}\ }\textbf {\bibinfo {volume} {12}},\ \bibinfo {pages} {021049} (\bibinfo {year} {2022})}\BibitemShut {NoStop}%
\bibitem [{\citenamefont {Cai}\ \emph {et~al.}(2023)\citenamefont {Cai}, \citenamefont {Babbush}, \citenamefont {Benjamin}, \citenamefont {Endo}, \citenamefont {Huggins}, \citenamefont {Li}, \citenamefont {McClean},\ and\ \citenamefont {O’Brien}}]{cai2023quantum}%
  \BibitemOpen
  \bibfield  {author} {\bibinfo {author} {\bibfnamefont {Z.}~\bibnamefont {Cai}}, \bibinfo {author} {\bibfnamefont {R.}~\bibnamefont {Babbush}}, \bibinfo {author} {\bibfnamefont {S.~C.}\ \bibnamefont {Benjamin}}, \bibinfo {author} {\bibfnamefont {S.}~\bibnamefont {Endo}}, \bibinfo {author} {\bibfnamefont {W.~J.}\ \bibnamefont {Huggins}}, \bibinfo {author} {\bibfnamefont {Y.}~\bibnamefont {Li}}, \bibinfo {author} {\bibfnamefont {J.~R.}\ \bibnamefont {McClean}},\ and\ \bibinfo {author} {\bibfnamefont {T.~E.}\ \bibnamefont {O’Brien}},\ }\bibfield  {title} {\bibinfo {title} {Quantum error mitigation},\ }\href {https://journals.aps.org/rmp/abstract/10.1103/RevModPhys.95.045005} {\bibfield  {journal} {\bibinfo  {journal} {Rev. Mod. Phys.}\ }\textbf {\bibinfo {volume} {95}},\ \bibinfo {pages} {045005} (\bibinfo {year} {2023})}\BibitemShut {NoStop}%
\bibitem [{\citenamefont {Aaronson}\ and\ \citenamefont {Gottesman}(2004)}]{aaronson2004improved}%
  \BibitemOpen
  \bibfield  {author} {\bibinfo {author} {\bibfnamefont {S.}~\bibnamefont {Aaronson}}\ and\ \bibinfo {author} {\bibfnamefont {D.}~\bibnamefont {Gottesman}},\ }\bibfield  {title} {\bibinfo {title} {Improved simulation of stabilizer circuits},\ }\href {https://arxiv.org/pdf/quant-ph/0406196.pdf} {\bibfield  {journal} {\bibinfo  {journal} {Phys. Rev. A}\ }\textbf {\bibinfo {volume} {70}},\ \bibinfo {pages} {052328} (\bibinfo {year} {2004})}\BibitemShut {NoStop}%
\bibitem [{\citenamefont {Wallman}\ and\ \citenamefont {Emerson}(2016)}]{wallman2016noise}%
  \BibitemOpen
  \bibfield  {author} {\bibinfo {author} {\bibfnamefont {J.~J.}\ \bibnamefont {Wallman}}\ and\ \bibinfo {author} {\bibfnamefont {J.}~\bibnamefont {Emerson}},\ }\bibfield  {title} {\bibinfo {title} {Noise tailoring for scalable quantum computation via randomized compiling},\ }\href {https://journals.aps.org/pra/abstract/10.1103/PhysRevA.94.052325} {\bibfield  {journal} {\bibinfo  {journal} {Phys. Rev. A}\ }\textbf {\bibinfo {volume} {94}},\ \bibinfo {pages} {052325} (\bibinfo {year} {2016})}\BibitemShut {NoStop}%
\bibitem [{\citenamefont {Blume-Kohout}\ \emph {et~al.}(2017)\citenamefont {Blume-Kohout}, \citenamefont {Gamble}, \citenamefont {Nielsen}, \citenamefont {Mizrahi}, \citenamefont {Sterk},\ and\ \citenamefont {Maunz}}]{blume2017demonstration}%
  \BibitemOpen
  \bibfield  {author} {\bibinfo {author} {\bibfnamefont {R.}~\bibnamefont {Blume-Kohout}}, \bibinfo {author} {\bibfnamefont {J.~K.}\ \bibnamefont {Gamble}}, \bibinfo {author} {\bibfnamefont {E.}~\bibnamefont {Nielsen}}, \bibinfo {author} {\bibfnamefont {J.}~\bibnamefont {Mizrahi}}, \bibinfo {author} {\bibfnamefont {J.~D.}\ \bibnamefont {Sterk}},\ and\ \bibinfo {author} {\bibfnamefont {P.}~\bibnamefont {Maunz}},\ }\bibfield  {title} {\bibinfo {title} {Demonstration of qubit operations below a rigorous fault tolerance threshold with gate set tomography},\ }\href {https://www.nature.com/articles/ncomms14485} {\bibfield  {journal} {\bibinfo  {journal} {Nature Comm.}\ }\textbf {\bibinfo {volume} {8}},\ \bibinfo {pages} {14485} (\bibinfo {year} {2017})}\BibitemShut {NoStop}%
\bibitem [{\citenamefont {Erhard}\ \emph {et~al.}(2019)\citenamefont {Erhard}, \citenamefont {Wallman}, \citenamefont {Postler}, \citenamefont {Meth}, \citenamefont {Stricker}, \citenamefont {Martinez}, \citenamefont {Schindler}, \citenamefont {Monz}, \citenamefont {Emerson},\ and\ \citenamefont {Blatt}}]{erhard2019characterizing}%
  \BibitemOpen
  \bibfield  {author} {\bibinfo {author} {\bibfnamefont {A.}~\bibnamefont {Erhard}}, \bibinfo {author} {\bibfnamefont {J.~J.}\ \bibnamefont {Wallman}}, \bibinfo {author} {\bibfnamefont {L.}~\bibnamefont {Postler}}, \bibinfo {author} {\bibfnamefont {M.}~\bibnamefont {Meth}}, \bibinfo {author} {\bibfnamefont {R.}~\bibnamefont {Stricker}}, \bibinfo {author} {\bibfnamefont {E.~A.}\ \bibnamefont {Martinez}}, \bibinfo {author} {\bibfnamefont {P.}~\bibnamefont {Schindler}}, \bibinfo {author} {\bibfnamefont {T.}~\bibnamefont {Monz}}, \bibinfo {author} {\bibfnamefont {J.}~\bibnamefont {Emerson}},\ and\ \bibinfo {author} {\bibfnamefont {R.}~\bibnamefont {Blatt}},\ }\bibfield  {title} {\bibinfo {title} {Characterizing large-scale quantum computers via cycle benchmarking},\ }\href {https://www.nature.com/articles/s41467-019-13068-7} {\bibfield  {journal} {\bibinfo  {journal} {Nature Comm.}\ }\textbf {\bibinfo {volume} {10}},\ \bibinfo {pages} {5347} (\bibinfo {year} {2019})}\BibitemShut {NoStop}%
\bibitem [{\citenamefont {Magesan}\ \emph {et~al.}(2012)\citenamefont {Magesan}, \citenamefont {Gambetta}, \citenamefont {C{\'o}rcoles},\ and\ \citenamefont {Chow}}]{magesan2012efficient}%
  \BibitemOpen
  \bibfield  {author} {\bibinfo {author} {\bibfnamefont {E.}~\bibnamefont {Magesan}}, \bibinfo {author} {\bibfnamefont {J.~M.}\ \bibnamefont {Gambetta}}, \bibinfo {author} {\bibfnamefont {A.}~\bibnamefont {C{\'o}rcoles}},\ and\ \bibinfo {author} {\bibfnamefont {J.~M.}\ \bibnamefont {Chow}},\ }\bibfield  {title} {\bibinfo {title} {Efficient measurement of quantum gate error by interleaved randomized benchmarking},\ }\href {https://doi.org/10.1103/PhysRevLett.109.080505} {\bibfield  {journal} {\bibinfo  {journal} {Phys. Rev. Lett.}\ }\textbf {\bibinfo {volume} {109}},\ \bibinfo {pages} {080505} (\bibinfo {year} {2012})}\BibitemShut {NoStop}%
\bibitem [{\citenamefont {Bluvstein}\ \emph {et~al.}(2024)\citenamefont {Bluvstein}, \citenamefont {Evered}, \citenamefont {Geim}, \citenamefont {Li}, \citenamefont {Zhou}, \citenamefont {Manovitz}, \citenamefont {Ebadi}, \citenamefont {Cain}, \citenamefont {Kalinowski}, \citenamefont {Hangleiter} \emph {et~al.}}]{bluvstein2024logical}%
  \BibitemOpen
  \bibfield  {author} {\bibinfo {author} {\bibfnamefont {D.}~\bibnamefont {Bluvstein}}, \bibinfo {author} {\bibfnamefont {S.~J.}\ \bibnamefont {Evered}}, \bibinfo {author} {\bibfnamefont {A.~A.}\ \bibnamefont {Geim}}, \bibinfo {author} {\bibfnamefont {S.~H.}\ \bibnamefont {Li}}, \bibinfo {author} {\bibfnamefont {H.}~\bibnamefont {Zhou}}, \bibinfo {author} {\bibfnamefont {T.}~\bibnamefont {Manovitz}}, \bibinfo {author} {\bibfnamefont {S.}~\bibnamefont {Ebadi}}, \bibinfo {author} {\bibfnamefont {M.}~\bibnamefont {Cain}}, \bibinfo {author} {\bibfnamefont {M.}~\bibnamefont {Kalinowski}}, \bibinfo {author} {\bibfnamefont {D.}~\bibnamefont {Hangleiter}}, \emph {et~al.},\ }\bibfield  {title} {\bibinfo {title} {Logical quantum processor based on reconfigurable atom arrays},\ }\href {https://www.nature.com/articles/s41586-023-06927-3_reference.pdf} {\bibfield  {journal} {\bibinfo  {journal} {Nature}\ }\textbf {\bibinfo {volume} {626}},\ \bibinfo {pages} {58} (\bibinfo {year} {2024})}\BibitemShut {NoStop}%
\bibitem [{\citenamefont {Gao}\ \emph {et~al.}(2025)\citenamefont {Gao}, \citenamefont {Fan}, \citenamefont {Zha}, \citenamefont {Bei}, \citenamefont {Cai}, \citenamefont {Cai}, \citenamefont {Cao}, \citenamefont {Chen}, \citenamefont {Chen}, \citenamefont {Chen} \emph {et~al.}}]{gao2025establishing}%
  \BibitemOpen
  \bibfield  {author} {\bibinfo {author} {\bibfnamefont {D.}~\bibnamefont {Gao}}, \bibinfo {author} {\bibfnamefont {D.}~\bibnamefont {Fan}}, \bibinfo {author} {\bibfnamefont {C.}~\bibnamefont {Zha}}, \bibinfo {author} {\bibfnamefont {J.}~\bibnamefont {Bei}}, \bibinfo {author} {\bibfnamefont {G.}~\bibnamefont {Cai}}, \bibinfo {author} {\bibfnamefont {J.}~\bibnamefont {Cai}}, \bibinfo {author} {\bibfnamefont {S.}~\bibnamefont {Cao}}, \bibinfo {author} {\bibfnamefont {F.}~\bibnamefont {Chen}}, \bibinfo {author} {\bibfnamefont {J.}~\bibnamefont {Chen}}, \bibinfo {author} {\bibfnamefont {K.}~\bibnamefont {Chen}}, \emph {et~al.},\ }\bibfield  {title} {\bibinfo {title} {{Establishing a new benchmark in quantum computational advantage with 105-qubit Zuchongzhi 3.0 processor}},\ }\href {https://journals.aps.org/prl/abstract/10.1103/PhysRevLett.134.090601} {\bibfield  {journal} {\bibinfo  {journal} {Phys. Rev. Lett.}\ }\textbf {\bibinfo {volume} {134}},\ \bibinfo {pages} {090601} (\bibinfo {year} {2025})}\BibitemShut
  {NoStop}%
\bibitem [{\citenamefont {Zhang}\ \emph {et~al.}(2023)\citenamefont {Zhang}, \citenamefont {He}, \citenamefont {Sun}, \citenamefont {Zheng}, \citenamefont {Liu}, \citenamefont {Luo}, \citenamefont {Wang}, \citenamefont {Zhu}, \citenamefont {Qiu}, \citenamefont {Shen} \emph {et~al.}}]{zhang2023scalable}%
  \BibitemOpen
  \bibfield  {author} {\bibinfo {author} {\bibfnamefont {W.-Y.}\ \bibnamefont {Zhang}}, \bibinfo {author} {\bibfnamefont {M.-G.}\ \bibnamefont {He}}, \bibinfo {author} {\bibfnamefont {H.}~\bibnamefont {Sun}}, \bibinfo {author} {\bibfnamefont {Y.-G.}\ \bibnamefont {Zheng}}, \bibinfo {author} {\bibfnamefont {Y.}~\bibnamefont {Liu}}, \bibinfo {author} {\bibfnamefont {A.}~\bibnamefont {Luo}}, \bibinfo {author} {\bibfnamefont {H.-Y.}\ \bibnamefont {Wang}}, \bibinfo {author} {\bibfnamefont {Z.-H.}\ \bibnamefont {Zhu}}, \bibinfo {author} {\bibfnamefont {P.-Y.}\ \bibnamefont {Qiu}}, \bibinfo {author} {\bibfnamefont {Y.-C.}\ \bibnamefont {Shen}}, \emph {et~al.},\ }\bibfield  {title} {\bibinfo {title} {Scalable multipartite entanglement created by spin exchange in an optical lattice},\ }\href {https://journals.aps.org/prl/abstract/10.1103/PhysRevLett.131.073401} {\bibfield  {journal} {\bibinfo  {journal} {Phys. Rev. Lett.}\ }\textbf {\bibinfo {volume} {131}},\ \bibinfo {pages} {073401} (\bibinfo {year}
  {2023})}\BibitemShut {NoStop}%
\bibitem [{\citenamefont {Dalzell}\ \emph {et~al.}(2024)\citenamefont {Dalzell}, \citenamefont {Hunter-Jones},\ and\ \citenamefont {Brand{\~a}o}}]{dalzell2024random}%
  \BibitemOpen
  \bibfield  {author} {\bibinfo {author} {\bibfnamefont {A.~M.}\ \bibnamefont {Dalzell}}, \bibinfo {author} {\bibfnamefont {N.}~\bibnamefont {Hunter-Jones}},\ and\ \bibinfo {author} {\bibfnamefont {F.~G.}\ \bibnamefont {Brand{\~a}o}},\ }\bibfield  {title} {\bibinfo {title} {Random quantum circuits transform local noise into global white noise},\ }\href {https://link.springer.com/article/10.1007/s00220-024-04958-z} {\bibfield  {journal} {\bibinfo  {journal} {Comm. Math. Phys.}\ }\textbf {\bibinfo {volume} {405}},\ \bibinfo {pages} {78} (\bibinfo {year} {2024})}\BibitemShut {NoStop}%
\bibitem [{\citenamefont {Hinsche}\ \emph {et~al.}(2025)\citenamefont {Hinsche}, \citenamefont {Ioannou}, \citenamefont {Jerbi}, \citenamefont {Leone}, \citenamefont {Eisert},\ and\ \citenamefont {Carrasco}}]{CrossDevice}%
  \BibitemOpen
  \bibfield  {author} {\bibinfo {author} {\bibfnamefont {M.}~\bibnamefont {Hinsche}}, \bibinfo {author} {\bibfnamefont {M.}~\bibnamefont {Ioannou}}, \bibinfo {author} {\bibfnamefont {S.}~\bibnamefont {Jerbi}}, \bibinfo {author} {\bibfnamefont {L.}~\bibnamefont {Leone}}, \bibinfo {author} {\bibfnamefont {J.}~\bibnamefont {Eisert}},\ and\ \bibinfo {author} {\bibfnamefont {J.}~\bibnamefont {Carrasco}},\ }\bibfield  {title} {\bibinfo {title} {{Efficient distributed inner product estimation via Pauli sampling}},\ }\href {https://doi.org/10.48550/arXiv.2405.06544} {\bibfield  {journal} {\bibinfo  {journal} {PRX Quantum}\ }\textbf {\bibinfo {volume} {6}},\ \bibinfo {pages} {in press} (\bibinfo {year} {2025})}\BibitemShut {NoStop}%
\bibitem [{\citenamefont {Elben}\ \emph {et~al.}(2020)\citenamefont {Elben}, \citenamefont {Vermersch}, \citenamefont {Van~Bijnen}, \citenamefont {Kokail}, \citenamefont {Brydges}, \citenamefont {Maier}, \citenamefont {Joshi}, \citenamefont {Blatt}, \citenamefont {Roos},\ and\ \citenamefont {Zoller}}]{elben2020cross}%
  \BibitemOpen
  \bibfield  {author} {\bibinfo {author} {\bibfnamefont {A.}~\bibnamefont {Elben}}, \bibinfo {author} {\bibfnamefont {B.}~\bibnamefont {Vermersch}}, \bibinfo {author} {\bibfnamefont {R.}~\bibnamefont {Van~Bijnen}}, \bibinfo {author} {\bibfnamefont {C.}~\bibnamefont {Kokail}}, \bibinfo {author} {\bibfnamefont {T.}~\bibnamefont {Brydges}}, \bibinfo {author} {\bibfnamefont {C.}~\bibnamefont {Maier}}, \bibinfo {author} {\bibfnamefont {M.~K.}\ \bibnamefont {Joshi}}, \bibinfo {author} {\bibfnamefont {R.}~\bibnamefont {Blatt}}, \bibinfo {author} {\bibfnamefont {C.~F.}\ \bibnamefont {Roos}},\ and\ \bibinfo {author} {\bibfnamefont {P.}~\bibnamefont {Zoller}},\ }\bibfield  {title} {\bibinfo {title} {Cross-platform verification of intermediate scale quantum devices},\ }\href {https://journals.aps.org/prl/abstract/10.1103/PhysRevLett.124.010504} {\bibfield  {journal} {\bibinfo  {journal} {Phys. Rev. Lett.}\ }\textbf {\bibinfo {volume} {124}},\ \bibinfo {pages} {010504} (\bibinfo {year} {2020})}\BibitemShut {NoStop}%
\bibitem [{\citenamefont {Gullans}\ \emph {et~al.}(2021)\citenamefont {Gullans}, \citenamefont {Krastanov}, \citenamefont {Huse}, \citenamefont {Jiang},\ and\ \citenamefont {Flammia}}]{gullans2021quantum}%
  \BibitemOpen
  \bibfield  {author} {\bibinfo {author} {\bibfnamefont {M.~J.}\ \bibnamefont {Gullans}}, \bibinfo {author} {\bibfnamefont {S.}~\bibnamefont {Krastanov}}, \bibinfo {author} {\bibfnamefont {D.~A.}\ \bibnamefont {Huse}}, \bibinfo {author} {\bibfnamefont {L.}~\bibnamefont {Jiang}},\ and\ \bibinfo {author} {\bibfnamefont {S.~T.}\ \bibnamefont {Flammia}},\ }\bibfield  {title} {\bibinfo {title} {Quantum coding with low-depth random circuits},\ }\href {https://journals.aps.org/prx/abstract/10.1103/PhysRevX.11.031066} {\bibfield  {journal} {\bibinfo  {journal} {Phys. Rev. X}\ }\textbf {\bibinfo {volume} {11}},\ \bibinfo {pages} {031066} (\bibinfo {year} {2021})}\BibitemShut {NoStop}%
\bibitem [{\citenamefont {Eisert}\ and\ \citenamefont {Preskill}(2025)}]{MindTheGaps}%
  \BibitemOpen
  \bibfield  {author} {\bibinfo {author} {\bibfnamefont {J.}~\bibnamefont {Eisert}}\ and\ \bibinfo {author} {\bibfnamefont {J.}~\bibnamefont {Preskill}},\ }\bibfield  {title} {\bibinfo {title} {Mind the gaps: The fraught road to quantum advantage},\ }\href {https://arxiv.org/abs/2510.19928} {\bibfield  {journal} {\bibinfo  {journal} {arXiv:2510.19928}\ } (\bibinfo {year} {2025})}\BibitemShut {NoStop}%
\bibitem [{\citenamefont {Wu}\ \emph {et~al.}(2022)\citenamefont {Wu}, \citenamefont {Kolkowitz}, \citenamefont {Puri},\ and\ \citenamefont {Thompson}}]{wu2022erasure}%
  \BibitemOpen
  \bibfield  {author} {\bibinfo {author} {\bibfnamefont {Y.}~\bibnamefont {Wu}}, \bibinfo {author} {\bibfnamefont {S.}~\bibnamefont {Kolkowitz}}, \bibinfo {author} {\bibfnamefont {S.}~\bibnamefont {Puri}},\ and\ \bibinfo {author} {\bibfnamefont {J.~D.}\ \bibnamefont {Thompson}},\ }\bibfield  {title} {\bibinfo {title} {Erasure conversion for fault-tolerant quantum computing in alkaline earth rydberg atom arrays},\ }\href {https://www.nature.com/articles/s41467-022-32094-6.pdf} {\bibfield  {journal} {\bibinfo  {journal} {Nature Comm.}\ }\textbf {\bibinfo {volume} {13}},\ \bibinfo {pages} {4657} (\bibinfo {year} {2022})}\BibitemShut {NoStop}%
\bibitem [{RPS()}]{RPSData}%
  \BibitemOpen
  \href@noop {} {}\bibinfo {howpublished} {Available at \href{https://github.com/Yomure/RPS-Data/}{https://github.com/Yomure/RPS-Data/}}\BibitemShut {NoStop}%
\bibitem [{\citenamefont {Roffe}\ \emph {et~al.}(2023)\citenamefont {Roffe}, \citenamefont {Cohen}, \citenamefont {Quintavalle}, \citenamefont {Chandra},\ and\ \citenamefont {Campbell}}]{roffe2023bias}%
  \BibitemOpen
  \bibfield  {author} {\bibinfo {author} {\bibfnamefont {J.}~\bibnamefont {Roffe}}, \bibinfo {author} {\bibfnamefont {L.~Z.}\ \bibnamefont {Cohen}}, \bibinfo {author} {\bibfnamefont {A.~O.}\ \bibnamefont {Quintavalle}}, \bibinfo {author} {\bibfnamefont {D.}~\bibnamefont {Chandra}},\ and\ \bibinfo {author} {\bibfnamefont {E.~T.}\ \bibnamefont {Campbell}},\ }\bibfield  {title} {\bibinfo {title} {{Bias-tailored quantum LDPC codes}},\ }\href {https://quantum-journal.org/papers/q-2023-05-15-1005/} {\bibfield  {journal} {\bibinfo  {journal} {Quantum}\ }\textbf {\bibinfo {volume} {7}},\ \bibinfo {pages} {1005} (\bibinfo {year} {2023})}\BibitemShut {NoStop}%
\bibitem [{\citenamefont {Bonilla~Ataides}\ \emph {et~al.}(2021)\citenamefont {Bonilla~Ataides}, \citenamefont {Tuckett}, \citenamefont {Bartlett}, \citenamefont {Flammia},\ and\ \citenamefont {Brown}}]{bonilla2021xzzx}%
  \BibitemOpen
  \bibfield  {author} {\bibinfo {author} {\bibfnamefont {J.~P.}\ \bibnamefont {Bonilla~Ataides}}, \bibinfo {author} {\bibfnamefont {D.~K.}\ \bibnamefont {Tuckett}}, \bibinfo {author} {\bibfnamefont {S.~D.}\ \bibnamefont {Bartlett}}, \bibinfo {author} {\bibfnamefont {S.~T.}\ \bibnamefont {Flammia}},\ and\ \bibinfo {author} {\bibfnamefont {B.~J.}\ \bibnamefont {Brown}},\ }\bibfield  {title} {\bibinfo {title} {{The XZZX surface code}},\ }\href {https://www.nature.com/articles/s41467-021-22274-1.pdf} {\bibfield  {journal} {\bibinfo  {journal} {Nature Comm.}\ }\textbf {\bibinfo {volume} {12}},\ \bibinfo {pages} {2172} (\bibinfo {year} {2021})}\BibitemShut {NoStop}%
\bibitem [{\citenamefont {Maslov}\ and\ \citenamefont {Nam}(2018)}]{maslov2018use}%
  \BibitemOpen
  \bibfield  {author} {\bibinfo {author} {\bibfnamefont {D.}~\bibnamefont {Maslov}}\ and\ \bibinfo {author} {\bibfnamefont {Y.}~\bibnamefont {Nam}},\ }\bibfield  {title} {\bibinfo {title} {Use of global interactions in efficient quantum circuit constructions},\ }\href {https://arxiv.org/pdf/1707.06356} {\bibfield  {journal} {\bibinfo  {journal} {New J. Phys.}\ }\textbf {\bibinfo {volume} {20}},\ \bibinfo {pages} {033018} (\bibinfo {year} {2018})}\BibitemShut {NoStop}%
\bibitem [{\citenamefont {Bravyi}\ \emph {et~al.}(2022)\citenamefont {Bravyi}, \citenamefont {Maslov},\ and\ \citenamefont {Nam}}]{bravyi2022constant}%
  \BibitemOpen
  \bibfield  {author} {\bibinfo {author} {\bibfnamefont {S.}~\bibnamefont {Bravyi}}, \bibinfo {author} {\bibfnamefont {D.}~\bibnamefont {Maslov}},\ and\ \bibinfo {author} {\bibfnamefont {Y.}~\bibnamefont {Nam}},\ }\bibfield  {title} {\bibinfo {title} {{Constant-cost implementations of Clifford operations and multiply-controlled gates using global interactions}},\ }\href {https://arxiv.org/pdf/2207.08691} {\bibfield  {journal} {\bibinfo  {journal} {Phys. Rev. Lett.}\ }\textbf {\bibinfo {volume} {129}},\ \bibinfo {pages} {230501} (\bibinfo {year} {2022})}\BibitemShut {NoStop}%
\bibitem [{\citenamefont {Zhang}\ \emph {et~al.}(2025)\citenamefont {Zhang}, \citenamefont {Tang}, \citenamefont {Liu}, \citenamefont {Yuan}, \citenamefont {Shen}, \citenamefont {Wu},\ and\ \citenamefont {Zhang}}]{zhang2025robust}%
  \BibitemOpen
  \bibfield  {author} {\bibinfo {author} {\bibfnamefont {W.}~\bibnamefont {Zhang}}, \bibinfo {author} {\bibfnamefont {G.}~\bibnamefont {Tang}}, \bibinfo {author} {\bibfnamefont {K.}~\bibnamefont {Liu}}, \bibinfo {author} {\bibfnamefont {X.}~\bibnamefont {Yuan}}, \bibinfo {author} {\bibfnamefont {Y.}~\bibnamefont {Shen}}, \bibinfo {author} {\bibfnamefont {Y.}~\bibnamefont {Wu}},\ and\ \bibinfo {author} {\bibfnamefont {X.-M.}\ \bibnamefont {Zhang}},\ }\bibfield  {title} {\bibinfo {title} {{Robust M{\o}lmer-S{\o}rensen gate against symmetric and asymmetric errors}},\ }\href {https://doi.org/10.1088/2058-9565/adce29} {\bibfield  {journal} {\bibinfo  {journal} {Quant. Sc. Tech.}\ }\textbf {\bibinfo {volume} {10}},\ \bibinfo {pages} {035009} (\bibinfo {year} {2025})}\BibitemShut {NoStop}%
\bibitem [{\citenamefont {Mart{\'\i}nez-Garc{\'\i}a}\ \emph {et~al.}(2022)\citenamefont {Mart{\'\i}nez-Garc{\'\i}a}, \citenamefont {Gerster}, \citenamefont {Vodola}, \citenamefont {Hrmo}, \citenamefont {Monz}, \citenamefont {Schindler},\ and\ \citenamefont {M{\"u}ller}}]{martinez2022analytical}%
  \BibitemOpen
  \bibfield  {author} {\bibinfo {author} {\bibfnamefont {F.}~\bibnamefont {Mart{\'\i}nez-Garc{\'\i}a}}, \bibinfo {author} {\bibfnamefont {L.}~\bibnamefont {Gerster}}, \bibinfo {author} {\bibfnamefont {D.}~\bibnamefont {Vodola}}, \bibinfo {author} {\bibfnamefont {P.}~\bibnamefont {Hrmo}}, \bibinfo {author} {\bibfnamefont {T.}~\bibnamefont {Monz}}, \bibinfo {author} {\bibfnamefont {P.}~\bibnamefont {Schindler}},\ and\ \bibinfo {author} {\bibfnamefont {M.}~\bibnamefont {M{\"u}ller}},\ }\bibfield  {title} {\bibinfo {title} {{Analytical and experimental study of center-line miscalibrations in M{\o}lmer-S{\o}rensen gates}},\ }\href {https://link.aps.org/accepted/10.1103/PhysRevA.105.032437} {\bibfield  {journal} {\bibinfo  {journal} {Phys. Rev. A}\ }\textbf {\bibinfo {volume} {105}},\ \bibinfo {pages} {032437} (\bibinfo {year} {2022})}\BibitemShut {NoStop}%
\bibitem [{\citenamefont {Manovitz}\ \emph {et~al.}(2017)\citenamefont {Manovitz}, \citenamefont {Rotem}, \citenamefont {Shaniv}, \citenamefont {Cohen}, \citenamefont {Shapira}, \citenamefont {Akerman}, \citenamefont {Retzker},\ and\ \citenamefont {Ozeri}}]{manovitz2017fast}%
  \BibitemOpen
  \bibfield  {author} {\bibinfo {author} {\bibfnamefont {T.}~\bibnamefont {Manovitz}}, \bibinfo {author} {\bibfnamefont {A.}~\bibnamefont {Rotem}}, \bibinfo {author} {\bibfnamefont {R.}~\bibnamefont {Shaniv}}, \bibinfo {author} {\bibfnamefont {I.}~\bibnamefont {Cohen}}, \bibinfo {author} {\bibfnamefont {Y.}~\bibnamefont {Shapira}}, \bibinfo {author} {\bibfnamefont {N.}~\bibnamefont {Akerman}}, \bibinfo {author} {\bibfnamefont {A.}~\bibnamefont {Retzker}},\ and\ \bibinfo {author} {\bibfnamefont {R.}~\bibnamefont {Ozeri}},\ }\bibfield  {title} {\bibinfo {title} {{Fast dynamical decoupling of the M{\o}lmer-S{\o}rensen entangling gate}},\ }\href {https://link.aps.org/accepted/10.1103/PhysRevLett.119.220505} {\bibfield  {journal} {\bibinfo  {journal} {Phys. Rev. Lett.}\ }\textbf {\bibinfo {volume} {119}},\ \bibinfo {pages} {220505} (\bibinfo {year} {2017})}\BibitemShut {NoStop}%
\bibitem [{\citenamefont {Sahay}\ \emph {et~al.}(2023)\citenamefont {Sahay}, \citenamefont {Jin}, \citenamefont {Claes}, \citenamefont {Thompson},\ and\ \citenamefont {Puri}}]{sahay2023high}%
  \BibitemOpen
  \bibfield  {author} {\bibinfo {author} {\bibfnamefont {K.}~\bibnamefont {Sahay}}, \bibinfo {author} {\bibfnamefont {J.}~\bibnamefont {Jin}}, \bibinfo {author} {\bibfnamefont {J.}~\bibnamefont {Claes}}, \bibinfo {author} {\bibfnamefont {J.~D.}\ \bibnamefont {Thompson}},\ and\ \bibinfo {author} {\bibfnamefont {S.}~\bibnamefont {Puri}},\ }\bibfield  {title} {\bibinfo {title} {High-threshold codes for neutral-atom qubits with biased erasure errors},\ }\href {https://journals.aps.org/prx/pdf/10.1103/PhysRevX.13.041013} {\bibfield  {journal} {\bibinfo  {journal} {Phys. Rev. X}\ }\textbf {\bibinfo {volume} {13}},\ \bibinfo {pages} {041013} (\bibinfo {year} {2023})}\BibitemShut {NoStop}%
\bibitem [{\citenamefont {Garc{\'\i}a}\ \emph {et~al.}(2017)\citenamefont {Garc{\'\i}a}, \citenamefont {Markov},\ and\ \citenamefont {Cross}}]{garcia2017geometry}%
  \BibitemOpen
  \bibfield  {author} {\bibinfo {author} {\bibfnamefont {H.~J.}\ \bibnamefont {Garc{\'\i}a}}, \bibinfo {author} {\bibfnamefont {I.~L.}\ \bibnamefont {Markov}},\ and\ \bibinfo {author} {\bibfnamefont {A.~W.}\ \bibnamefont {Cross}},\ }\bibfield  {title} {\bibinfo {title} {On the geometry of stabilizer states},\ }\href {https://arxiv.org/pdf/1711.07848} {\bibfield  {journal} {\bibinfo  {journal} {arXiv:1711.07848}\ } (\bibinfo {year} {2017})}\BibitemShut {NoStop}%
\bibitem [{\citenamefont {Van~den Nest}\ \emph {et~al.}(2004)\citenamefont {Van~den Nest}, \citenamefont {Dehaene},\ and\ \citenamefont {De~Moor}}]{van2004graphical}%
  \BibitemOpen
  \bibfield  {author} {\bibinfo {author} {\bibfnamefont {M.}~\bibnamefont {Van~den Nest}}, \bibinfo {author} {\bibfnamefont {J.}~\bibnamefont {Dehaene}},\ and\ \bibinfo {author} {\bibfnamefont {B.}~\bibnamefont {De~Moor}},\ }\bibfield  {title} {\bibinfo {title} {{Graphical description of the action of local Clifford transformations on graph states}},\ }\href {https://journals.aps.org/pra/abstract/10.1103/PhysRevA.69.022316} {\bibfield  {journal} {\bibinfo  {journal} {Phys. Rev. A}\ }\textbf {\bibinfo {volume} {69}},\ \bibinfo {pages} {022316} (\bibinfo {year} {2004})}\BibitemShut {NoStop}%
\bibitem [{\citenamefont {Hein}\ \emph {et~al.}(2004)\citenamefont {Hein}, \citenamefont {Eisert},\ and\ \citenamefont {Briegel}}]{Graphs}%
  \BibitemOpen
  \bibfield  {author} {\bibinfo {author} {\bibfnamefont {M.}~\bibnamefont {Hein}}, \bibinfo {author} {\bibfnamefont {J.}~\bibnamefont {Eisert}},\ and\ \bibinfo {author} {\bibfnamefont {H.~J.}\ \bibnamefont {Briegel}},\ }\bibfield  {title} {\bibinfo {title} {Multi-particle entanglement in graph states},\ }\href {https://doi.org/10.1103/PhysRevA.69.062311} {\bibfield  {journal} {\bibinfo  {journal} {Phys. Rev. A}\ }\textbf {\bibinfo {volume} {69}},\ \bibinfo {pages} {062311} (\bibinfo {year} {2004})}\BibitemShut {NoStop}%
\bibitem [{\citenamefont {Greenbaum}(2015)}]{greenbaum2015introduction}%
  \BibitemOpen
  \bibfield  {author} {\bibinfo {author} {\bibfnamefont {D.}~\bibnamefont {Greenbaum}},\ }\bibfield  {title} {\bibinfo {title} {Introduction to quantum gate set tomography},\ }\href {https://arxiv.org/abs/1509.02921} {\bibfield  {journal} {\bibinfo  {journal} {arXiv preprint arXiv:1509.02921}\ } (\bibinfo {year} {2015})}\BibitemShut {NoStop}%
\bibitem [{\citenamefont {Wan}\ \emph {et~al.}(2023)\citenamefont {Wan}, \citenamefont {Huggins}, \citenamefont {Lee},\ and\ \citenamefont {Babbush}}]{MatchgateShadows}%
  \BibitemOpen
  \bibfield  {author} {\bibinfo {author} {\bibfnamefont {K.}~\bibnamefont {Wan}}, \bibinfo {author} {\bibfnamefont {W.~J.}\ \bibnamefont {Huggins}}, \bibinfo {author} {\bibfnamefont {J.}~\bibnamefont {Lee}},\ and\ \bibinfo {author} {\bibfnamefont {R.}~\bibnamefont {Babbush}},\ }\bibfield  {title} {\bibinfo {title} {Matchgate shadows for fermionic quantum simulation},\ }\href {https://doi.org/10.1007/s00220-023-04844-0} {\bibfield  {journal} {\bibinfo  {journal} {Comm. Math. Phys.}\ }\textbf {\bibinfo {volume} {404}},\ \bibinfo {pages} {629} (\bibinfo {year} {2023})}\BibitemShut {NoStop}%
\bibitem [{\citenamefont {Helsen}\ \emph {et~al.}(2022{\natexlab{b}})\citenamefont {Helsen}, \citenamefont {Nezami}, \citenamefont {Reagor},\ and\ \citenamefont {Walter}}]{MatchgateHelsen}%
  \BibitemOpen
  \bibfield  {author} {\bibinfo {author} {\bibfnamefont {J.}~\bibnamefont {Helsen}}, \bibinfo {author} {\bibfnamefont {S.}~\bibnamefont {Nezami}}, \bibinfo {author} {\bibfnamefont {M.}~\bibnamefont {Reagor}},\ and\ \bibinfo {author} {\bibfnamefont {M.}~\bibnamefont {Walter}},\ }\bibfield  {title} {\bibinfo {title} {Matchgate benchmarking: Scalable benchmarking of a continuous family of many-qubit gates},\ }\href {https://doi.org/10.22331/q-2022-02-21-657} {\bibfield  {journal} {\bibinfo  {journal} {Quantum}\ }\textbf {\bibinfo {volume} {6}},\ \bibinfo {pages} {657} (\bibinfo {year} {2022}{\natexlab{b}})}\BibitemShut {NoStop}%
\end{thebibliography}
